\documentclass[11pt]{article}

\usepackage[utf8]{inputenc}
\pdfoutput=1
\usepackage[margin=1.1in,letterpaper]{geometry}
\usepackage{tikz}

\makeatletter
\newlength\aftertitskip     \newlength\beforetitskip
\newlength\interauthorskip  \newlength\aftermaketitskip

\def\maketitle{\par
	\begingroup
	\def\thefootnote{\fnsymbol{footnote}}
	\def\@makefnmark{\hbox to 4pt{$^{\@thefnmark}$\hss}}
	\@maketitle \@thanks
	\endgroup
	\let\maketitle\relax \let\@maketitle\relax
	\gdef\@thanks{}\gdef\@author{}\gdef\@title{}\let\thanks\relax}

\def\@startauthor{\noindent \normalsize\bf}
\def\@endauthor{}
\def\@starteditor{\noindent \small {\bf Editor:~}}
\def\@endeditor{\normalsize}
\def\@maketitle{\vbox{\hsize\textwidth
		\linewidth\hsize \vskip \beforetitskip
		{\begin{center} \LARGE\@title \par \end{center}} \vskip \aftertitskip
		{\def\and{\unskip\enspace{\rm and}\enspace}%
			\def\addr{\small\it}%
			\def\email{\hfill\small\tt}%
			\def\name{\normalsize\bf}%
			\def\AND{\@endauthor\rm\hss \vskip \interauthorskip \@startauthor}
			\@startauthor \@author \@endauthor}
}}

\makeatother

\pdfoutput=1                    

\usepackage{amsmath,amsthm}
\usepackage{graphicx}
\usepackage{color}
\usepackage{amssymb,amsfonts,amsxtra,mathrsfs,bm}
\usepackage{multicol}
\usepackage{multirow}
\usepackage{paralist}
\usepackage{xspace}
\usepackage{algorithm}
\usepackage{algpseudocode}
\usepackage[colorlinks=true, urlcolor = blue, citecolor=black]{hyperref}
\usepackage[round]{natbib}
\usepackage{bbm}
\usepackage{nicefrac}
\usepackage{subcaption}
\newsavebox{\imagebox}
\usepackage{wrapfig}

\newcommand{\reals}{\mathbb{R}}

\newcommand{\posdef}{\mathbb{P}}
\newcommand{\pd}{\posdef}

\newcommand{\frc}{FRC\xspace}
\newcommand{\orc}{ORC\xspace}
\newcommand{\Ric}{{\rm Ric}}

\newcommand{\w}{\omega}

\newcommand{\Bp}{\mathcal{B}}

\newcommand\restr[2]{{
		\left.\kern-\nulldelimiterspace 
		#1 
		\vphantom{\big|} 
		\right|_{#2} 
}}

\newtheorem{theorem}{Theorem}[section]
\newtheorem{lem}[theorem]{Lemma}

\theoremstyle{definition}
\newtheorem{defn}[theorem]{Definition}

\newtheorem{rmk}[theorem]{Remark}
\newtheorem{example}[theorem]{Example}

\usepackage{booktabs} 
\usepackage{makecell} 

\numberwithin{equation}{section}

\newcommand{\EE}{\mathbb{E}}
\newcommand{\RR}{\mathbb{R}}

\newcommand\blfootnote[1]{%
	\begingroup
	\renewcommand\thefootnote{}\footnote{#1}%
	\addtocounter{footnote}{-1}%
	\endgroup
}

\usepackage{titlesec}

\title{Curvature-based Clustering on Graphs}
\author{\name Yu Tian$^*$ \email yu.tian@su.se \\
	\addr Nordita\\
	Stockholm University and KTH Royal Institute of Technology\\
	Stockholm, SE-106 91, Sweden
	\AND
	\name Zachary Lubberts$^*$ \email zlubber1@jhu.edu \\
	\addr Department of Applied Mathematics and Statistics\\
	Johns Hopkins University\\
	Baltimore, MD 21218, USA
	\AND
	\name Melanie Weber \email{mweber@seas.harvard.edu}\\
	  \addr{Harvard University\\
	  Cambridge, MA 02138, USA}
	}

\begin{document}
	\maketitle
	\blfootnote{$^*$equal contribution}
	\blfootnote{A preliminary version of part of this work was presented at the NeurIPS Workshop on Symmetry and Geometry in Neural Representations 2022~\citep{tian2022mixed}.}
	
	
	\begin{abstract}
		Unsupervised node clustering (or \emph{community detection}) is a classical graph learning task. In this paper, we study algorithms, which exploit the geometry of the graph to identify densely connected substructures, which form clusters or communities.  Our method implements discrete Ricci curvatures and their associated geometric flows, under which the edge weights of the graph evolve to reveal its community structure. We consider several discrete curvature notions and analyze the utility of the resulting algorithms. In contrast to prior literature, we study not only single-membership community detection, where each node belongs to exactly one community, but also mixed-membership community detection, where communities may overlap.
		For the latter, we argue that it is beneficial to perform community detection on the line graph, i.e., the graph's dual. We provide both theoretical and empirical evidence for the utility of our curvature-based clustering algorithms. In addition, we give several results on the relationship between the curvature of a graph and that of its dual, which enable the efficient implementation of our proposed mixed-membership community detection approach and which may be of independent interest for curvature-based network analysis.
	\end{abstract}
	
	\tableofcontents
	
	\section{Introduction}
	Relational data, such as graphs or networks, is ubiquitous in machine learning and data science.  Consequently, a large body of literature devoted to studying the structure of such data has been developed.  
	Clustering on graphs, also known as \emph{community detection} or \emph{(unsupervised) node clustering}, is of central importance to the study of relational data. It seeks to identify  densely interconnected substructures (\emph{clusters}) in a given graph. Such structure is ubiquitous in relational data: We may think of friend circles in social networks, metabolic pathways in biochemical networks or article categories in Wikipedia. A standard mathematical model for analyzing community structure in graphs is the \emph{Stochastic Block Model (SBM)}, which has led to fundamental insights into the detectability of communities~\citep{abbe2017community}.
	Classically, communities are identified by clustering the nodes of the graph. Popular methods include the Louvain algorithm~\citep{Blondel_2008}, the Girvan-Newman algorithm~\citep{girvan-newman} and Spectral clustering~\citep{Cheeger,fiedler1973algebraic,spielman1996spectral}. Recently, there has been growing interest in another class of algorithms, which apply a geometric lens to community detection. 
	
	\emph{Graph curvatures} provide a means to understanding the structure of networks through a characterization of their geometry. More generally, curvature is a classical tool in differential geometry, which is used to characterize the local and global properties of geodesic spaces. While originally defined in continuous spaces, discrete notions of curvature have recently seen a surge of interest. Of particular interest are discretizations of \emph{Ricci curvature}, which is a local notion of curvature that relates to the volume growth rate of the unit ball in the space of interest (geodesic dispersion). 
	Curvature-based analysis of relational data (see, e.g.,~\citep{WSJ1,WSJ2}) has been applied in many domains, including to biological~\citep{elumalai2022graph,weber2017curvature,tannenbaum2015ricci}, chemical~\citep{leal2021forman,saucan2018discrete}, social~\citep{jost_social} and financial networks~\citep{sandhu2016ricci}. 
	
	Community detection via graph curvatures is motivated by the observation that edges between communities (so called \emph{bridges}, see Fig.~\ref{fig:cluster-mixed}(a)) have low Ricci curvature. By identifying and removing such bridges, one can learn a partition of the graph into its communities. This simple idea has given rise to a series of curvature-based approaches for community detection~\citep{ni2019community,sia2019ollivier,weber2018detecting,tian2022mixed}. 
	However, the picture is far from complete. The proposed algorithms vary significantly in their implementation, in particular in the choice of the curvature notion that they utilize. In this paper, we will elucidate commonalities and differences among variants of the two most common curvature notions in terms of their utility in community detection. To this end, we propose a unifying framework for curvature-based community detection algorithms and perform a systematic analysis of the advantages and disadvantages of each curvature notion. The insights gained from this analysis give rise to several improvements, which mitigate the identified shortcomings in accuracy and scalability. 

	\begin{figure}[t]
		\centering
		\begin{subfigure}[t]{.45\linewidth}
			\centering
			\includegraphics[width=0.5\linewidth]{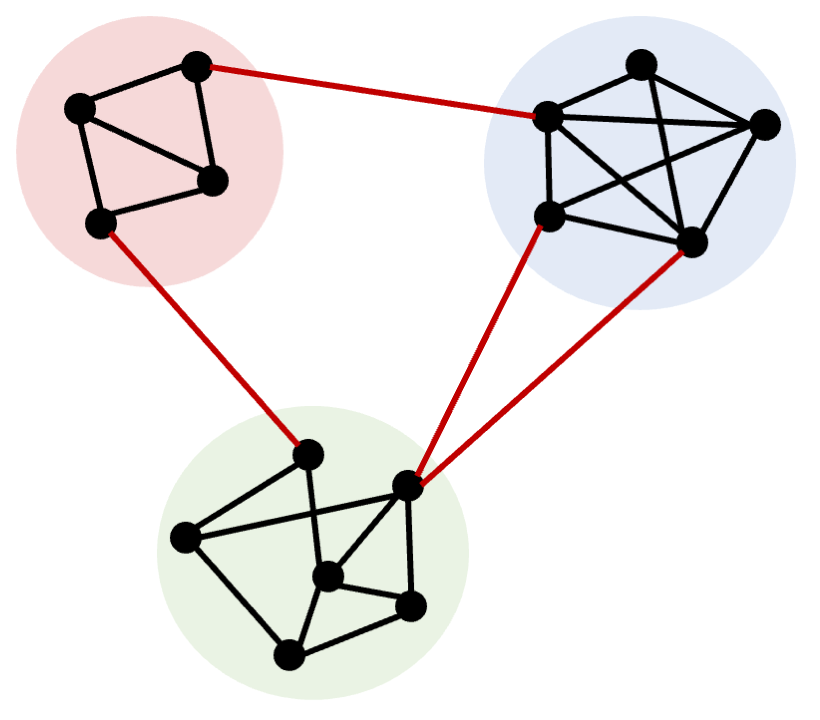}
			\caption{Single-Membership}\label{fig:cluster-single}
		\end{subfigure}
		\hfill
		\begin{subfigure}[t]{.5\linewidth}
			\centering
			\includegraphics[width=0.4\linewidth]{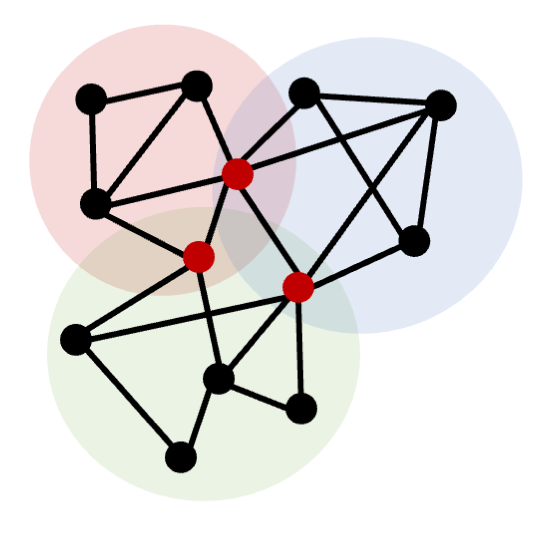}
			\caption{Mixed-Membership}\label{fig:cluster-mixed}
		\end{subfigure}
		\caption{Clustering on Graphs.}
	\end{figure}
	The present paper also addresses a second important gap in the literature. 
	Existing work on partition-based community detection, specifically on curvature-based approaches, focuses almost exclusively on detecting communities with \emph{unique} membership, i.e., each node can belong to only one community (\emph{single-membership community detection}). However, in many complex systems that generate relational data, nodes may belong to more than one community: A member of a social network might belong to a circle of high school friends and a circle of college friends, a protein may have different functional roles in a biochemical network etc. This creates a need for algorithms that can cluster nodes with \emph{mixed} membership~\citep{airoldi2008mixed,yang2013overlapping,zhang2020overlap}, so-called \emph{mixed-membership community detection}. Observe that the mixed-membership structure precludes the existence of bridges between communities, because the communities overlap (see Fig.~\ref{fig:cluster-mixed}(b)). This renders community detection methods that rely on graph partitions, such as the curvature-based methods discussed above, inapplicable. In this paper, we propose a principled curvature-based approach for detecting mixed-membership community structure. We show that while partition-based approaches do not recover the underlying community structure when applied to the original graph, they perform well on its dual, the \emph{line graph}. The line graph encodes the connectivity between the edges of the original graph (see Fig.~\ref{fig:construction}). This reformulation of the relational information in the graph allows for disentangling the overlapping community structure (see Fig.~\ref{fig:intro}). When applied to the line graph, our curvature-based community detection approach correctly identifies
	edges in the line graph, which can be cut to partition the graph into its communities and assign node labels that reflect mixed-membership community structure. Utilizing fundamental relationships between the curvature of a graph and its dual, we demonstrate that curvature-based mixed-membership community detection results in scalable and accurate methods.
	
	\subsection{Overview and Summary of Contributions}
	Our main contributions are as follows:
	\begin{enumerate}
		\item We propose a unifying framework for curvature-based community detection (unsupervised node clustering) algorithms (sec.~\ref{sec:overview}). Utilizing this framework, we systematically investigate the strengths and weaknesses of different discrete curvature notions (specifically, variants of Forman's and Ollivier's Ricci curvature, sec.~\ref{sec:exp}, sec.~\ref{sec:comparison}).
		\item We adapt our proposed framework to mixed-membership community detection (sec.~\ref{sec:algorithm}), providing the first curvature-based algorithms for this setting. 
		\item We elucidate significant scalability issues in community detection approaches that utilize Ollivier's notion of discrete curvature. To overcome these issues, we propose an effective approximation, which may be of independent interest in the broader area of curvature-based network analysis (sec.~\ref{sec:orc-approx}, sec.~\ref{sec:line-weighted}). 
		\item We demonstrate the utility of our framework through theoretical analysis (sec.~\ref{sec:orc-theory}) and through a series of systematic experiments with synthetic and real network data. We further benchmark curvature-based methods against classical community detection approaches (sec.~\ref{sec:exp}).
		\item As a byproduct, we derive several fundamental relations between discrete curvatures in a graph and its dual. Those relations may be of independent interest (sec.~\ref{sec:line-graph-curv}). 
	\end{enumerate}

	\section{Clustering on Graphs}\label{sec:overview}
	
	Community structure is a hallmark feature of networks, characterized by clusters of nodes that have more internal connections than connections to nodes in other clusters. 
	In networks with mixed-membership structure, which quantifies the idea of overlapping communities, nodes may belong to more than one cluster. 
	
	In the following, let $G=\left(V,E,W\right)$ denote a graph with vertex set $V$ and edge set $E$.  We may associate weights with both vertices and edges, defined as functions $\w(v): V \rightarrow \reals_+$
	and $\w(e): E \rightarrow \reals_+$. 
	The distance measure on $G$ is the standard \emph{path distance} $d$, i.e.,
	\begin{equation}
		d_G(u,v) := \inf_{\mathcal{P}} \sum_{i=1}^{n-1} w(\{z_i, z_{i+1}\}),
	\end{equation}
	where the infimum is taken over all $u-v$ paths $\mathcal{P}=\{z_i\}_{i=1}^n\subseteq V$, with $\{z_i,z_{i+1}\}\in E$ for each $i$, and $z_1=u, z_n=v$.
	We further define the \emph{degree} $d_u$ of a node $u$ as the weighted sum of its adjacent edges, i.e.,
	\begin{equation}
		d_u := \sum_{u \sim e} w(e) \; .
	\end{equation}
	We remark that the terms ``network'' and 	``graph'' are often used interchangeably. Following the convention in the graph learning literature, we use 	``graph'', if we refer to the mathematical object and 	``network'' when referring to a data set or specific instance of a graph. We will use the terms ``node" and 	``vertex'' interchangeably.
	
	\subsection{Single-Membership Node Clustering}
	
	\subsubsection{Single-Membership Community Structure} 
	In classical community detection, each node can only belong to one community, which characterizes the community structure as being single-membership. Various algorithms exist by virtue of interdisciplinary expertise \citep{Blondel_2008,girvan-newman,vonLuxburg2007spectra}, generally aiming to optimise a quality function with respect to different partitions of the network. Mathematically, if we denote a quality function as $Q$ which increases as the resulting community structure becomes clearer from some perspective, such as the modularity \citep{girvan-newman}, and denote a partition of the vertices as $\{\Bp_1,\dots, \Bp_{n_b}\}$ that are mutually exclusive and exhaustive, i.e., $\Bp_h\cap \Bp_l = \emptyset$ when $h\ne l$ and $\cup_{h=1}^{n_b}\Bp_h = V$, with $n_b$ being the number of communities, then the classical community detection can be formulated as the following combinatorial optimization problem
	\begin{align*}
		\max_{n_b}\max_{\{\Bp_1,\dots, \Bp_{n_b}\}}Q(\{\Bp_1,\dots, \Bp_{n_b}\}).
	\end{align*}
	In this generality, the community detection problem is NP-hard, so heuristics-based algorithms are typically applied in order to provide a sufficiently good solution. 
	
	The \emph{Stochastic Block Model} (short: \emph{SBM})~\citep{Holland_SBM_1983} is a random graph model, which emulates the community structure found in many real networks. It is a popular model for studying complex networks and, in particular, clustering on graphs.
	In the SBM, vertices are partitioned into subgroups called blocks, and the distribution of the edges between vertices is dependent on the blocks to which the vertices belong. Formally:
	\begin{defn}[Stochastic Block Model]
		\emph{
			Let $\{\Bp_1,\dots, \Bp_{n_b}\}$ be a partition of a graph's vertices into $n_b$ \emph{blocks}, and let $\mathbf{B} = (B_{hl})$ denote the characteristic matrix between the blocks, where $B_{hl}$ shows the probability of an edge existing between a vertex in block $\Bp_h$ and one in $\Bp_l$. Then, if we denote the random adjacency matrix of the network by $\mathbf{A} = (A_{ij})$ where $A_{ij} = 1$ if there is an edge between vertices $v_i$ and $v_j$,
			\begin{equation*}
				\pd\left(A_{ij} = 1\right) = B_{\sigma(i) \sigma(j)} \; ,
			\end{equation*}
			where $\sigma: V\to \{1,\dots, n_b\}$ indicates the block membership. A \emph{planted SBM} has the additional structure that $B_{hh} = p_{in},\, \forall h$, and $B_{hl} = p_{out},\, \forall h\ne l$, which we denote by $SBM(p_{in}, p_{out})$.}
	\end{defn}
	
	\subsubsection{Partition-based clustering methods}
	Of particular importance for classical community detection are edges between clusters, often called \emph{bridges}.
	Bridges connect two nodes in distinct communities (see Fig.~\ref{fig:cluster-single}). In graphs with single-membership community structure, i.e., where each node belongs to one community only, such edges always exist (otherwise each community is a connected component). 
	This observation has motivated a range of \emph{partition-based} community detection approaches~\citep{Blondel_2008,girvan-newman,vonLuxburg2007spectra}, which rely on identifying and cutting such bridges to partition the graph into communities. 
	Among the most popular community detection methods is the Louvain algorithm~\citep{Blondel_2008}, which is a heuristic method proposed to maximize modularity hierarchically, starting from each node as a community and then successively merging with its neighbouring communities whenever it improves modularity. Spectral clustering~\citep{vonLuxburg2007spectra}, another widely used method, utilizes the spectrum of the graph Laplacian to find a partition of the graph with a minimal number of edge cuts (\emph{mincut problem}). This optimization problem too is NP-hard, but it can be solved efficiently after relaxing some constraints. 
	We give a brief overview of other classical partition-based community detection methods in section~\ref{sec:related-work}. In this paper, we discuss methods that utilize discrete curvature to identify bridges: Below, we introduce a blueprint for partition-based clustering via discrete Ricci curvature.

	\subsection{Mixed-Membership Node Clustering}
	
	\subsubsection{Mixed-membership Community Structure}
	In graphs with mixed-membership communities, nodes may belong to more than one community. Such community structure can be formalized via an extension of the SBM: The \emph{Mixed-Membership Stochastic Block Model (short: \emph{MMB})}~\citep{airoldi2008mixed}. 
	In both models, each element of the adjacency matrix $\mathbf{A}$ above the diagonal is an independent Bernoulli random variable whose expectation only depends on the block memberships of the corresponding nodes. In the MMB, it is possible that nodes belong to more than one block, with various affiliation strengths. Formally: 
	\begin{defn}[Mixed-Membership Block Model]
		\emph{
			We assume that the expectation $\mathbb{E}[\mathbf{A}]$ has the form
			$\mathbb{E}[\mathbf{A}] = \mathbf{X}\mathbf{B}\mathbf{X}^T$,
			where $\mathbf{X}\in[0,1]^{n\times k}$ is the community membership matrix with $X_{il}$ indicating the affiliation of node $i$ with community $l$ and $\sum_{l}X_{il} = 1$, and $n$ denotes the size of the network. The matrix $\mathbf{B} \in [0,1]^{k\times k}$ encodes the block connection probabilities. If $X_{il}=1$ for some $l$, we call vertex $i$ a \emph{pure node}; otherwise $i$ is a \emph{mixed node}.
		} 
	\end{defn}
	When all nodes are pure, we recover the ordinary SBM model. In our experiments, we consider a planted version where $\mathbf{B}_{hh} = p_{in},\ \forall h$, and $\mathbf{B}_{hl} = p_{out},\ \forall h\ne l$, which we denote by $MMB(p_{in}, p_{out})$.
	
	\subsubsection{Line Graphs}
	\begin{wrapfigure}{R}{0.4\linewidth}
		\centering
		\includegraphics[width=6cm]{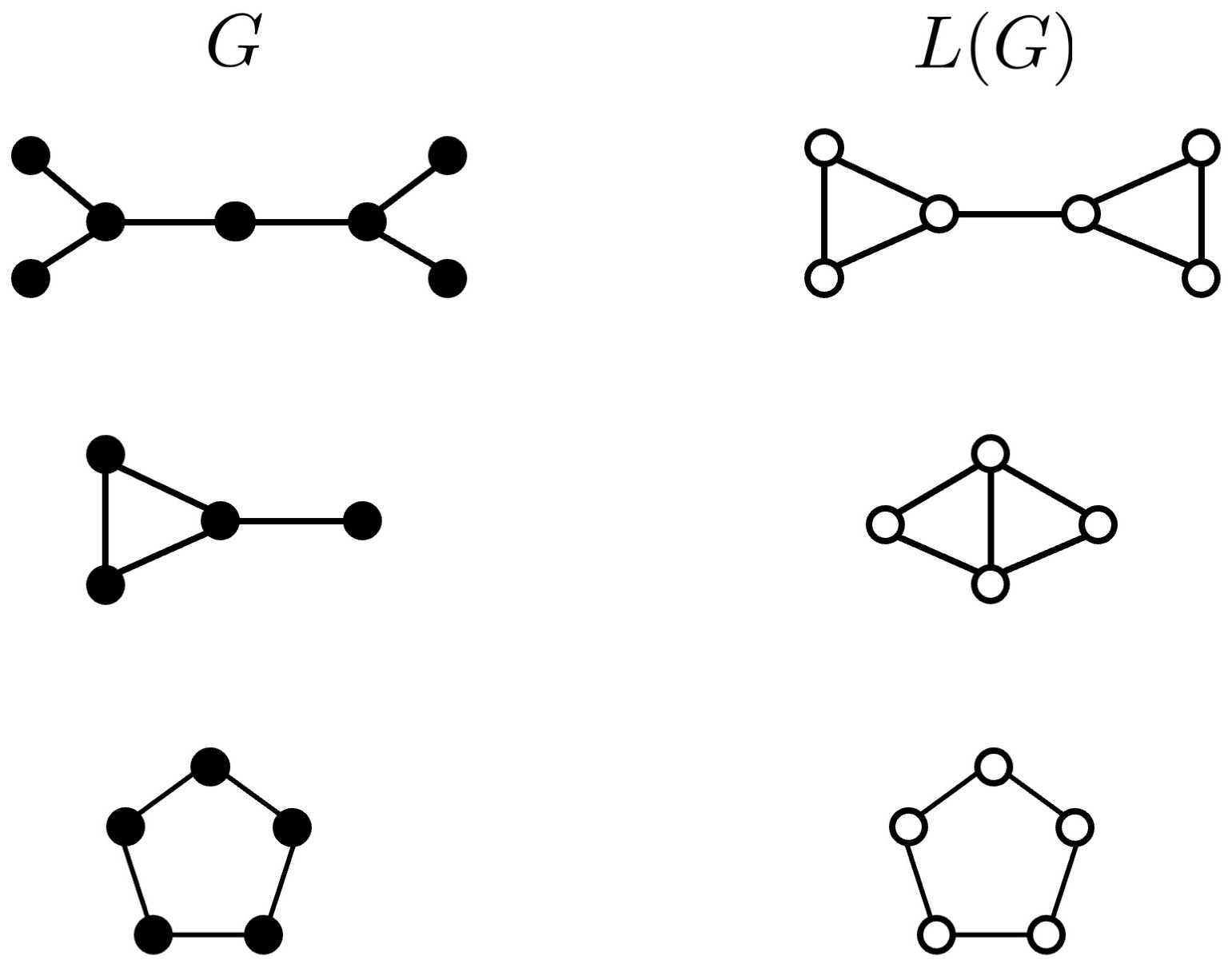}
		\caption{Graphs $G$ with their corresponding line graphs $L(G)$. Note that vertices in $G$ with degree at least 3 result in cliques in $L(G)$, and that cycles are preserved.}
		\label{fig:construction}
		\vspace{-0.5em}
	\end{wrapfigure}
	We denote the \emph{dual} of a graph as its \emph{line graph}:
	\begin{defn}
		Given an unweighted, undirected graph $G=(V,E)$, its line graph is defined as $L(G)=(E,\mathcal{E})$, where each edge $\{\{u,v\},\{r,s\}\}\in\mathcal{E}$ appears when $|\{u,v\}\cap \{r,s\}|=1$. In other words, $e,e'$ are adjacent in $L(G)$ when they are both incident on the same vertex. 
	\end{defn}
	The line graph can be seen as a re-parametrization of the input graph, which encodes higher-order information about its edges' connectivity.  Examples of graphs and their line graphs can be found in Fig.~\ref{fig:construction}. Importantly, cycles are preserved under this construction (e.g., Fig.~\ref{fig:construction}(bottom)). Cliques and cycles in the line graph may also arise from nodes of degree three or higher (see, e.g., Fig.~\ref{fig:construction}(top)).
	
	Notice that the line graph is typically larger in size than the original graph. In a graph with $n$ vertices and average degree $n\alpha_n$, we expect to have $n^2\alpha_n$ edges. In the line graph, this means $n^2\alpha_n$ vertices and $n^3\alpha_n^2$ edges, though this number increases with greater variation of the vertex degrees in the original graph. For the applications to community detection that we will discuss below, this necessitates a careful consideration of the scalability of different community detection approaches. 
	
	The classical line graph construction gives an unweighted graph, in the sense that edge weights in the original graph may be incorporated as node weights in the line graph. We will discuss below (sec.~\ref{sec:line-weighted}) how weights can be imposed on line graph edges that encode meaningful structural information in the original graph.
	
	
	\subsubsection{Partition-based Methods on the Line Graph.}
	\begin{wrapfigure}{R}{0.4\linewidth}
		\includegraphics[width=6cm]{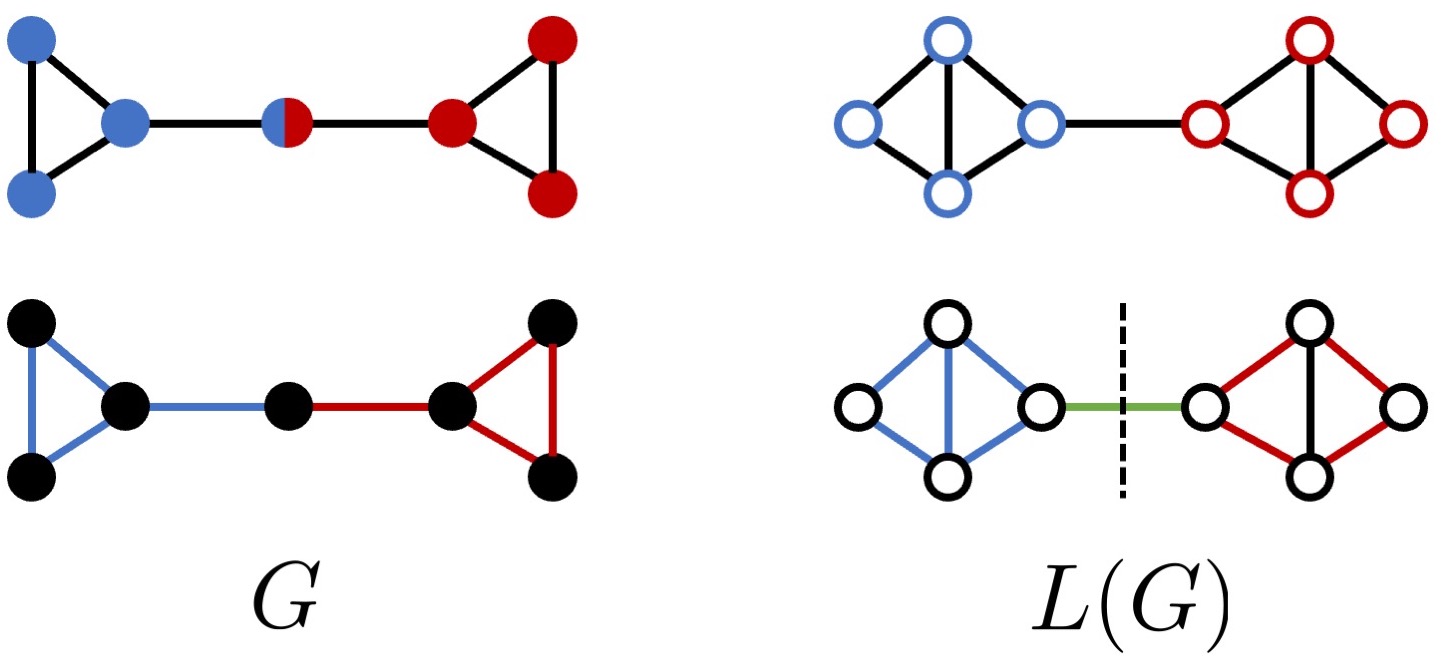}
		\caption{Community structure in a graph $G$ and its line graph $L(G)$ with node labels (top) and edge labels (bottom). 
		}
		\label{fig:intro}
		\vspace{-10pt}
	\end{wrapfigure}
	We have seen above that of particular importance for single-membership community detection are edges between clusters, often called \emph{bridges}. However, overlaps between communities preclude the existence of these bridges (see Fig.~\ref{fig:cluster-mixed}), which limits the applicability of partition-based approaches.
	In this paper we argue that the line graph provides a natural input for partition-based mixed-membership community detection. 
	While nodes may not have a unique label in this model, the adjacent edges may still be internal, connecting two nodes that are in the same community (at least partially). In this case, each edge is associated with a single community (see Fig.~\ref{fig:intro}(left)). Consequently, by representing the relationships among edges in a line graph, we disentangle the overlapping communities. Each edge in the line graph appears at a vertex in the original graph. When the vertex has mixed membership, a bridge between the communities arises in the line graph (Fig.~\ref{fig:intro}(right)).

	\subsection{Curvature-based Clustering Algorithm}
	%
	%
	\begin{algorithm}[ht]
		\begin{algorithmic}[1]
			\State \textbf{Input:} Graph $G=(V,E,W_0)$, hyperparameters $\nu$ (step size, flow), $\epsilon$ (tolerance), $\epsilon_d$ (drop threshold), $T$ (number of iterations)
			\For{$t=1, \dots, T$}
			\For{$\{u,v \} \in E$}
			\State Compute curvature $\kappa( \{u,v\})$.
			\State Evolve weight under Ricci flow: $w^{t}(\{u,v \}) \gets (1- \nu \cdot \kappa( \{u,v\}))d_G(u,v)$.
			\EndFor
			\State Renormalize edge weights $W^{t}$: $w^{t}(\{u,v\}) \gets \frac{\vert E \vert d_G(u,v)}{\sum_{\{u',v'\} \in E} d_G(u',v')}$ for all $\{u,v \} \in E$.
			\EndFor 
			\State Construct cut-off points $\lbrace x_0, x_1, \dots, x_{n_f} \rbrace$.
			\item Initialize $\mathbf{l} \in \mathbb{R}^{\vert V \vert}$ (list of node labels), $Q_{-1}, Q_{best} := \epsilon$.
			\For{$i = 0, \dots n_f$} 
			\State Construct $\tilde{G}_i = (V, E_i, W_i^T)$ with $E_i = \{\{u,v\}: w^T(\{u,v\}) > x_i\}$ and $W_i^T = \restr{W^T}{E_i}$.
			\State Determine connected components of $\tilde{G}_i$. Assign node labels by components.
			\State Compute the modularity $Q_i$ of the resulting  community assignment. 
			\If{$Q_i > Q_{best}$, $\frac{Q_i - Q_{i-1}}{Q_i} > \epsilon_d$}
			\State Store label assignments in $\mathbf{l}$. $Q_{best}=Q_i$.
			\EndIf
			\EndFor
			\State If $Q_{best}>\epsilon$, return $\mathbf{l}$.
		\end{algorithmic}
		\caption{Curvature-based Clustering via Ricci Flow} \label{alg.1}
	\end{algorithm}

	In this section, we will describe the blueprint of our \emph{curvature-based clustering algorithm}, which we will then specialize to the single- and mixed-membership problem described above. In the community detection literature, approaches based on discrete Ricci curvature have been studied for unique-membership communities~\citep{ni2019community,sia2019ollivier,weber2018detecting}. Here, we propose a curvature-based method for \emph{mixed-membership community detection}.
	
	Like previous approaches that utilize Ricci curvature (e.g.,~\citet{ni2019community}), we build on a notion of discrete Ricci flow first proposed by~\citet{Ol2}:
	\begin{equation}
		\frac{d}{d t} d_G(u,v)(t) = -\kappa(\{u,v\})(t) \cdot d_G(u,v)(t) \quad (\{u, v\} \in E)\; ,
	\end{equation}
	where $d_G(u,v)$ denotes the shortest path distance between adjacent nodes $u,v \in E$ and $\kappa(\{u,v\})$ the Ricci curvature along that edge. In this work, $\kappa(\{u,v\})$ may denote either Forman's or Ollivier's Ricci curvature. We will describe the two resulting methods in detail below.
	
	Curvature-based clustering implements \emph{discrete Ricci flow} via a combinatorial evolution equation, which evolves edge weights according to the local geometry of the graph. 
	Consider a family of weighted graphs $G^t = \{V, E, W^t \}$, which is constructed from an input graph $G$ by evolving its edge weights as
	\begin{equation}\label{eq:ricci-flow}
		w^t(\{u,v\}) \gets (1-\kappa(\{u,v\}))d_G(u,v) \qquad (\{u, v \} \in E)\; ,
	\end{equation}
	where the curvature $\kappa(\{u,v\})$ and shortest-path distance $d_G(u,v)$ is computed on the graph $G^{t-1}$. Equation~\eqref{eq:ricci-flow} can be viewed as a discrete analog of Ricci flow on Riemannian manifolds; it was first introduced by~\citet{Ol} with $\kappa$ denoting ORC. Since then, several variants have been studied in the context of clustering, utilizing ORC~\citep{ni2019community,sia2019ollivier,tian2022mixed}, as well as FRC~\citep{weber2018detecting}. Each of these algorithms is a special case of the general framework that we describe below.
	
	The clustering procedure is initialized with the (possibly unweighted) input graph, i.e., $G^0 := G$. We then evolve the edge weights for $T$ iteration under Ricci flow (Eq.~\eqref{eq:ricci-flow}). In each iteration, the edge weights are renormalized using
	\begin{equation}
		w^t(\{u,v\}) \gets \frac{\vert E \vert d_G(u,v)}{\sum_{\{u',v'\} \in E} d_G(u',v')} \; .
	\end{equation}
	Over time, the negative curvature of edges that bridge communities becomes stronger, since edges with a lower Ricci curvature contract slower under Ricci flow. On the other hand, edges with higher Ricci curvature contract faster. This results in a decrease of the weight of internal edges over time, while the weight of the bridges increases (see Fig.~\ref{fig:weights}). With that, the discrete Ricci flow reinforces the meso-scale structure of the network. 
	\begin{figure}[htb]
		\centering
		\begin{tabular}{ccc}
			\includegraphics[width=.3\textwidth]{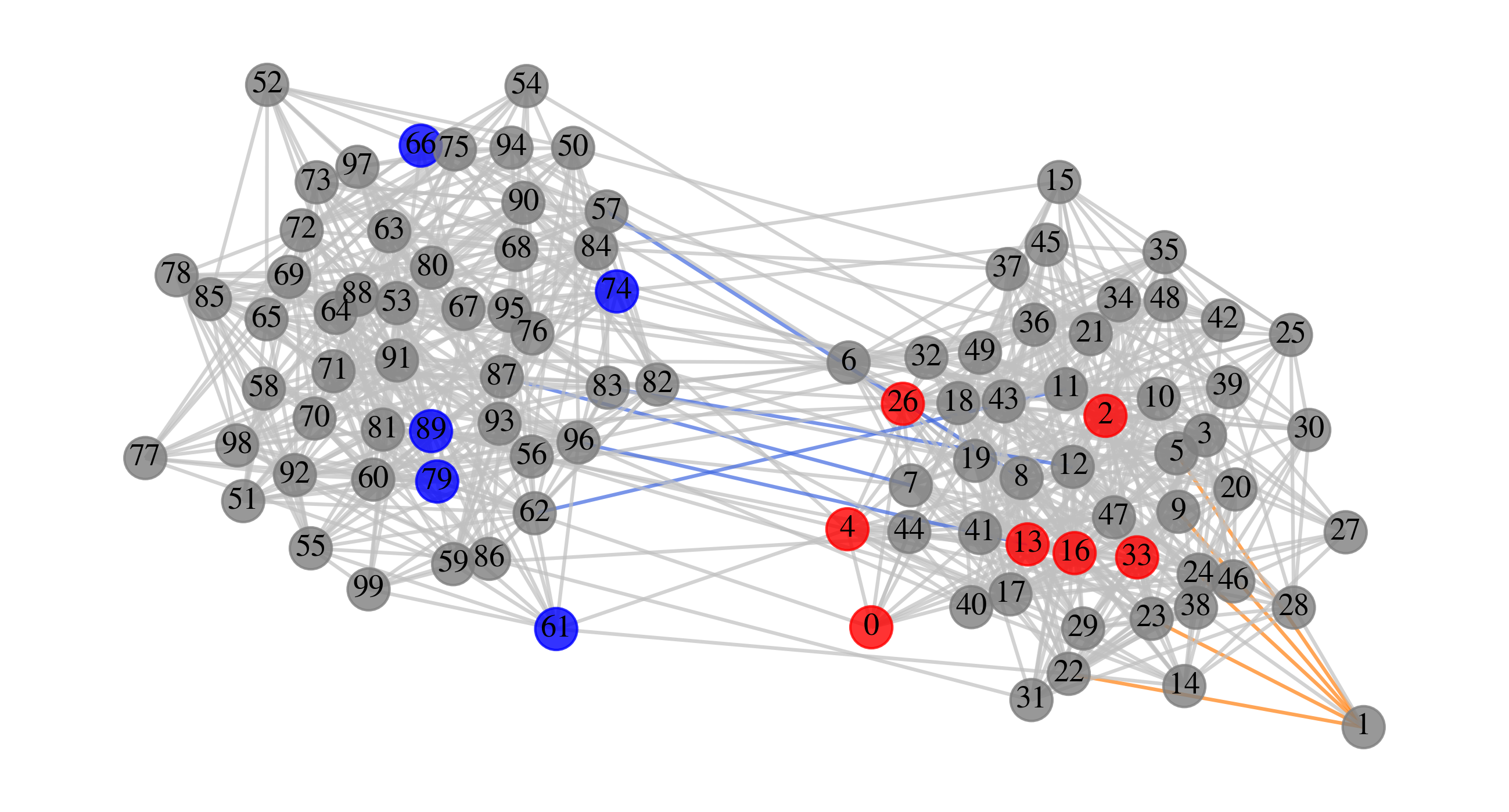} & \includegraphics[width=.3\textwidth]{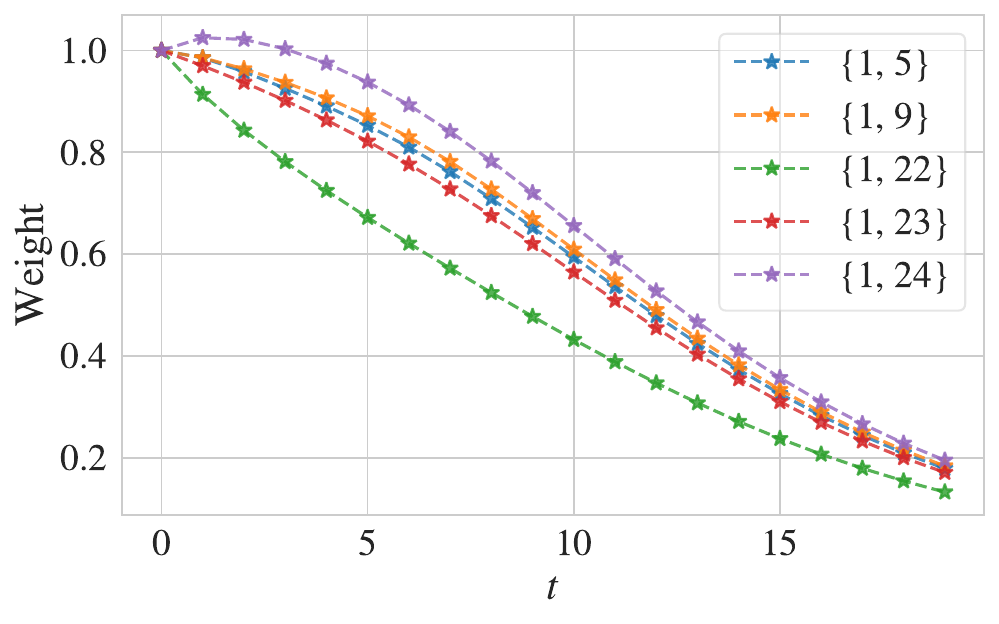} & \includegraphics[width=.3\textwidth]{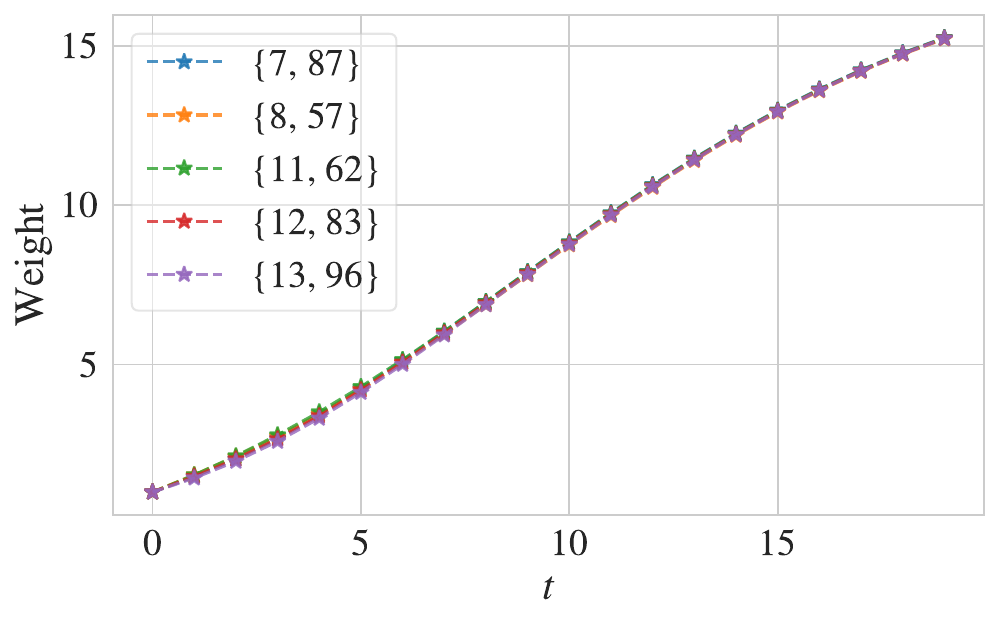}
		\end{tabular}
		\caption{To illustrate the effect of the discrete Ricci flow in Alg.~\ref{alg.1}, we consider a planted SBM of size $n=100$, two equally-sized blocks, $p_{in}=0.3$ and $p_{out}=0.02$, visualized in the left panel. We select representative edges within a community, shown in orange on the graph, and plot the evolution of their weights in the middle panel. In the right panel, we show this evolution for selected edges between communities (i.e., bridges), shown in blue on the graph.}
		\label{fig:weights}
	\end{figure}
	
	
	In order to recover single-membership labels, we cut the edges in the graph with high weights (or, equivalently, with very negative curvature) after evolving edge weights for $T$ iterations. The resulting partition delivers a node clustering. In the mixed-membership case, we cut the edges in the line graph with the highest weight (or, equivalently, the most negative curvature) after evolving edge weights for $T$ iterations. The resulting partition of the graph delivers an edge clustering from which we infer the community labels: We obtain a mixed-membership label vector $y$ for each node $v$ by computing 
	\begin{equation}\label{eq:labeling}
		y_l(v)=\frac{1}{|E_v|}\sum_{e\in E_v} \chi_{l}(e) \; ,
	\end{equation}
	where $\chi_{l}$ is the indicator for the edge cluster $C_l$. Intuitively, each edge belonging to cluster $l$ that is incident on a node $v$ adds more evidence of affinity between the node $v$ and the cluster $C_l$.
	
	In both cases, the success of the method depends crucially on identifying a good weight threshold for cutting edges in the original graph (single-membership case) or the line graph (mixed-membership case). We propose to perform a hyperparameter search over \emph{cut-off points} $\lbrace x_0, x_1, \dots, x_{n_f} \rbrace$. The construction of the cut-off points is an important design choice and depends on the curvature notion used in the approach.
	We then compute the \emph{modularity}, a classical quality metric in the community detection literature, for the label assignments corresponding to each cut-off point $x_i$. Specifically, modularity~\citep{girvan-newman,gomez2009analysis} is defined as
	\begin{equation}\label{eq:mod}
		Q := \frac{1}{2 w} \sum_{\{u,v\}\in E} \left( w(\{u,v\}) - \frac{d_u d_v}{2 w}\right) \delta(\sigma(u),\sigma(v)) \; ,
	\end{equation}
	where $2w:=\sum_{ij} w(\{i,j\})$. The larger the modularity, the better we expect the clustering to reflect the underlying community structure induced by the graph's connectivity.
	Our approach is schematically shown in Alg.~\ref{alg.1}. 
	

	\section{Related Work}
	\label{sec:related-work}
	
	\paragraph{Community detection.} 
	Mixed-membership community detection is widely studied in the network science and data mining communities. Notable approaches include Bayesian methods~\citep{airoldi2008mixed,hopkins_bayesian_2017}, matrix factorization~\citep{yang2013overlapping}, spectral clustering~\citep{zhang2020overlap}, and vertex hunting \citep{jin2017estimating}, among others. 
	In addition to the mixed-membership model that we consider here~\citep{airoldi2008mixed}, there is a significant body of literature on closely related overlapping community models~\citep{Lancichinetti2009nmi,xie2013overlapping}, which also study the problem of learning non-unique node labels. Curvature-based community detection methods for non-overlapping communities have recently received growing interest~\citep{ni2019community,sia2019ollivier,gosztolai2021unfolding,weber2018detecting}. 
	Such approaches utilize notions of discrete curvature~\citep{Ol,Ol2,forman2003bochner} 
	to partition networks into communities, based on the observation that edges between communities have low curvature. The absence of such bridges in overlapping communities renders these approaches inapplicable to the setting studied in this paper. To the best of our knowledge, our algorithm is the first to study mixed-membership community structure with curvature-based methods.
	
	\paragraph{Higher-order structure in Networks.} Historically, much of the network analysis literature has focused on the structural information encoded in nodes and edges. Recently, the analysis of higher-order structures has received increasing attention. A plethora of methods for analyzing higher-order structure in relational data has been developed~\citep{benson2016higher,battiston2020networks}.
	Our work follows this line of thought, in that it focuses on the relations \emph{between edges} and the structural information encoded therein. Curvature of higher-order structure has previously been studied in~\citep{WSJ2,hypernets,leal2021forman}. Here, networks have been studied as polyhedral complexes~\citep{WSJ2} or interactions between group of nodes have been encoded in hypergraphs~\citep{hypernets,leal2021forman}. To the best of our knowledge, this is the first work that studies the curvature of line graphs. 
	
	\paragraph{Line Graphs.} The notion of the \emph{line graph} goes back at least to~\citep{whitney1932congruent}, where it is shown that whenever $|V|\geq 5$, two graphs are isomorphic if and only if their line graphs are. An effective version of this result, which reports whether a given graph is a line graph, and if it is, returns the base graph, appears in \cite{lehot1974optimal}. Community detection and more generally network analysis via the line graph has been recently studied in~\citep{chen2017supervised,krzakala2013spectral,lubberts2021beyond,evans_line_2010}. 
	
	\paragraph{Discrete Curvature.}
	There exists a large body of literature on discrete notions of curvature. Notable examples include Gromov's $\delta$-hyperbolicity~\citep{gromov}, Bakry-Emre curvature~\citep{erbar2012ricci}, as well as Forman's and Ollivier's Ricci curvatures~\citep{forman2003bochner,Ol}. While there is some work on studying those curvature notions on higher-order networks~\citep{bloch2014combinatorial,WSJ2} and hypernetworks~\citep{hypernets,leal2021forman}, they have not been studied on line graphs. Curvature-based analysis of relational data (see, e.g.,~\citep{WSJ1,WSJ2}) has been applied in many domains, including to biological~\citep{elumalai2022graph,weber2017curvature,tannenbaum2015ricci}, chemical~\citep{leal2021forman,saucan2018discrete}, social~\citep{jost_social} and financial networks~\citep{sandhu2016ricci}. Discrete curvature has also found applications in Representation Learning, in particular for identifying representation spaces in graph embeddings~\citep{lubold2023identifying,pmlr-v108-weber20a,weber2018curvature}.

	\section{Discrete Graph Curvature}
	\label{sec:graph-curv}
	Curvature is a classical tool in Differential Geometry, which is used to characterize the local and global properties of geodesic spaces. In this paper, we investigate \emph{discrete} notions of curvature that are defined on graphs. Specifically, we focus on discretizations of \emph{Ricci curvature}, which is a local notion of curvature that relates to the volume growth rate of the unit ball in the space of interest (geodesic dispersion).
	
	In the following,  we only consider \emph{undirected} networks.  We introduce two classical notions of \emph{discrete Ricci curvature}, which were originally introduced by~\citet{Ol} and~\citet{forman2003bochner}, respectively.
	\begin{figure}[t]
		\begin{subfigure}[t]{.5\linewidth}
			\centering
			\includegraphics[width=0.65\linewidth]{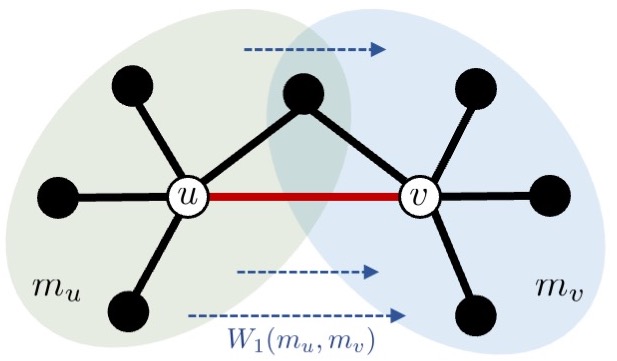}
			\caption{ORC}\label{fig:orc}
		\end{subfigure}
		\hfill
		\begin{subfigure}[t]{.5\linewidth}
			\centering
			\includegraphics[width=0.65\linewidth]{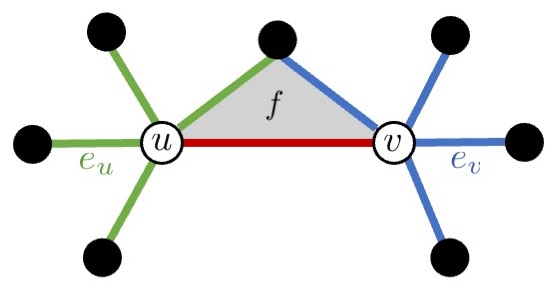}
			\caption{FRC}
		\end{subfigure}
		\caption{Computation of discrete Ricci curvature.}\label{fig:frc}
	\end{figure}
	%
	
	\subsection{Ollivier's Ricci Curvature}\label{sec:orc}
	
	\noindent Our first notion of discrete curvature relates geodesic dispersion to optimal mass transport on the graph. 
	
	\subsubsection{Formal Definition}  
	Consider the transportation cost between two distance balls (i.e.,  vertex neighborhoods) along an edge in the network. In an unweighted graph, we endow the neighborhoods of vertices $u,v$ adjacent to an edge $e=\{u,v\}$ with a uniform measure, i.e.,
	\begin{equation}
		m_{u}(z) := \frac{1}{d_u} \qquad \forall z, \; {s.t.} \; u \sim z ,
	\end{equation}
	and analogously for $m_v(z)$.
	Here, $u \sim z$ indicates that $u,z$ are neighbors. In a weighted graph, for $\alpha\in[0,1]$ and $p\geq 0$, we set
	\begin{equation}
		\label{eq:nodemeasure}
		m_u^{\alpha,p} (z) := \begin{cases}
			\alpha&\text{if }z=u,\\
			\frac{1-\alpha}{C_u}\exp(-d_G(u,z)^p) &\text{if }z \sim u, \\
			0 &{\rm else}.
		\end{cases} \; 
	\end{equation}
	where the normalizing constant $C_u=\sum_{z\sim u} \exp(-d_G(u,z)^p)$. Notice that for neighboring vertices $u,z$ we have $d_G(u,z)=w(\{u,z\}) := w_{u,z}$. When $\alpha=0$, and $p=0$ or we have an unweighted graph, this reduces to the uniform measure.  We define \emph{Ollivier's Ricci curvature}~\citep{Ol} (short: \orc) with respect to the Wasserstein-1 distance $W_1$ between those measures, i.e.,
	\begin{equation}\label{eq:orc-e}
		\Ric_O (e) := 1 - \frac{W_1 (m_{u}, m_{v})}{d_G(u,v)} \; .
	\end{equation}
	The computation of \orc is illustrated in Fig.~\ref{fig:orc}.
	We can further define \orc for vertices with respect to the curvature of its adjacent edges. Formally, let  $E_v := \{ e \in E: v\in e\}$ denote the set of edges adjacent to a vertex $v$. Then its curvature is given by
	\begin{equation}\label{eq:orc-v}
		\Ric_O(v) = \sum_{e_z \in E_z} \Ric_O (e_z) \; .
	\end{equation}
	
	\subsubsection{Computational Considerations} 
	\begin{figure}[!t]
		\centering
		\begin{tabular}{ccc}
			\includegraphics[width=.3\textwidth]{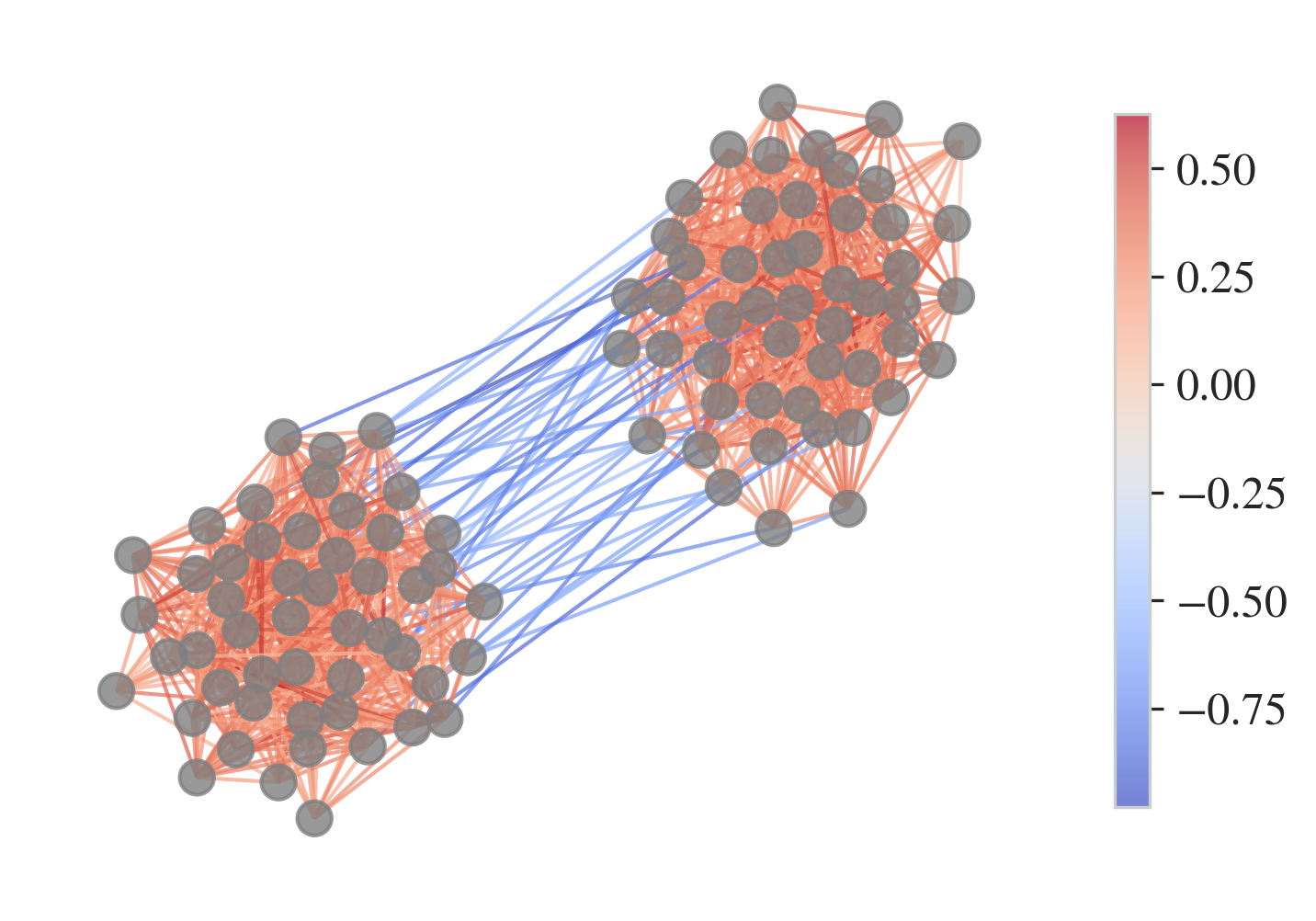} & \includegraphics[width=.25\textwidth]{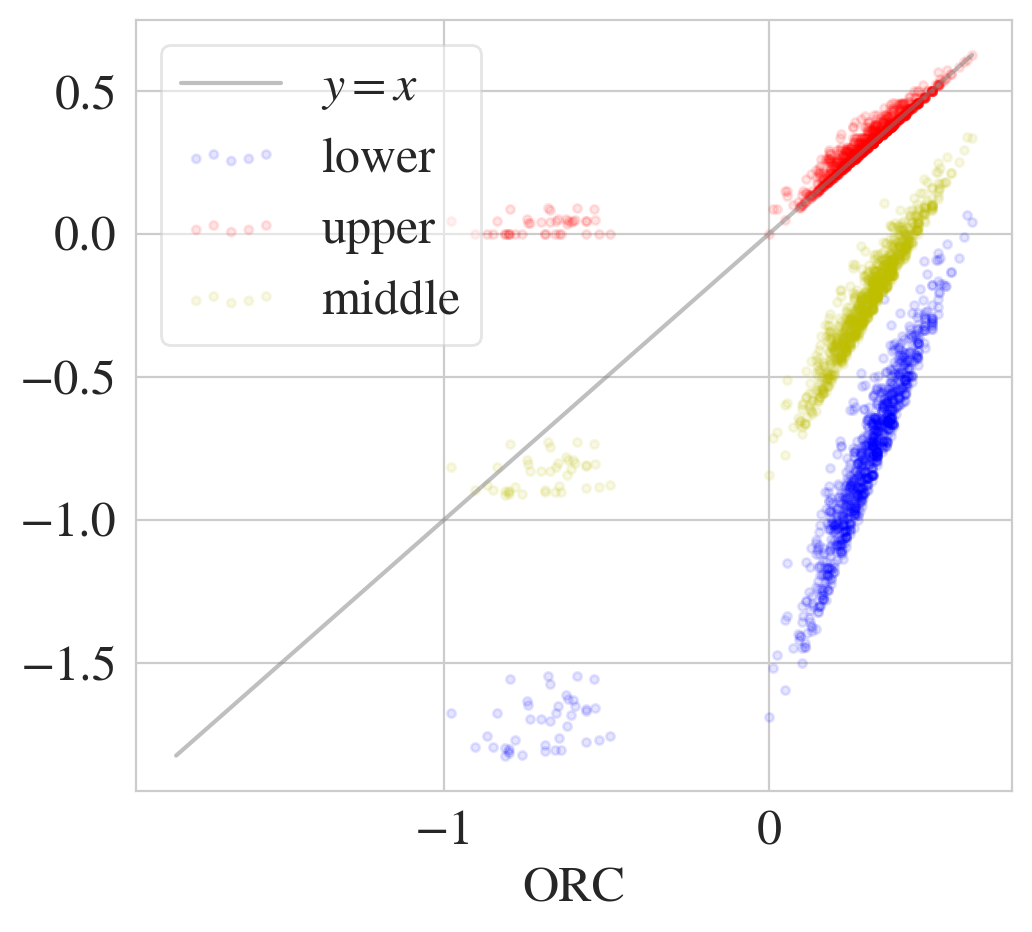} &  \includegraphics[width=.25\textwidth]{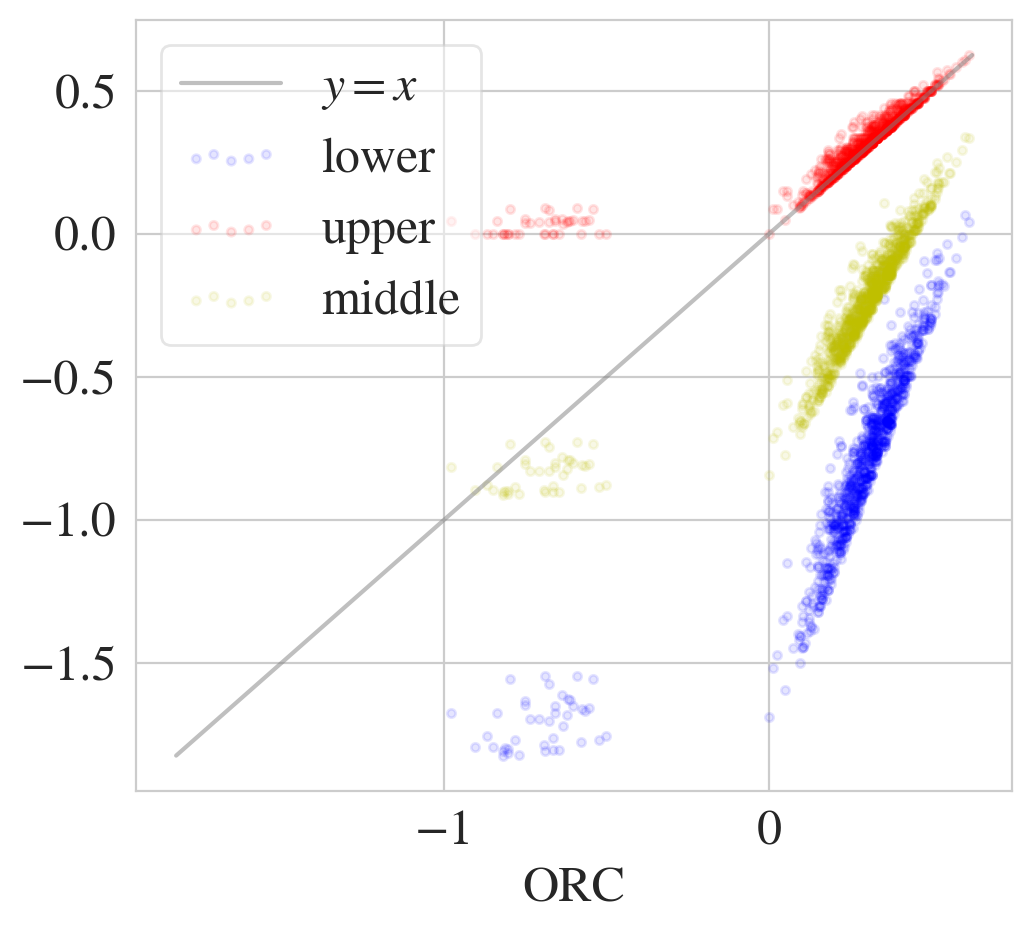}\\
			\includegraphics[width=.3\textwidth]{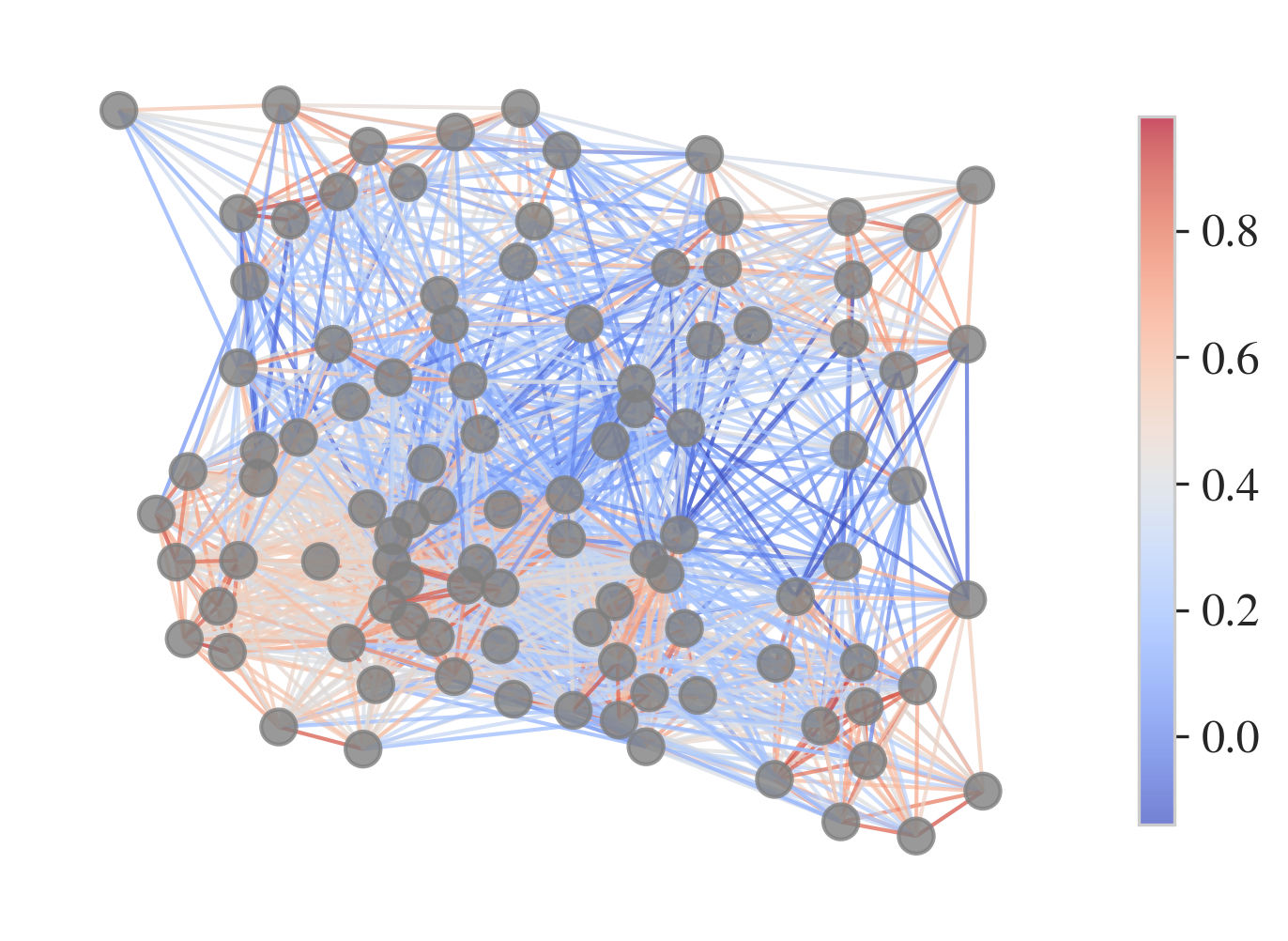} & \includegraphics[width=.25\textwidth]{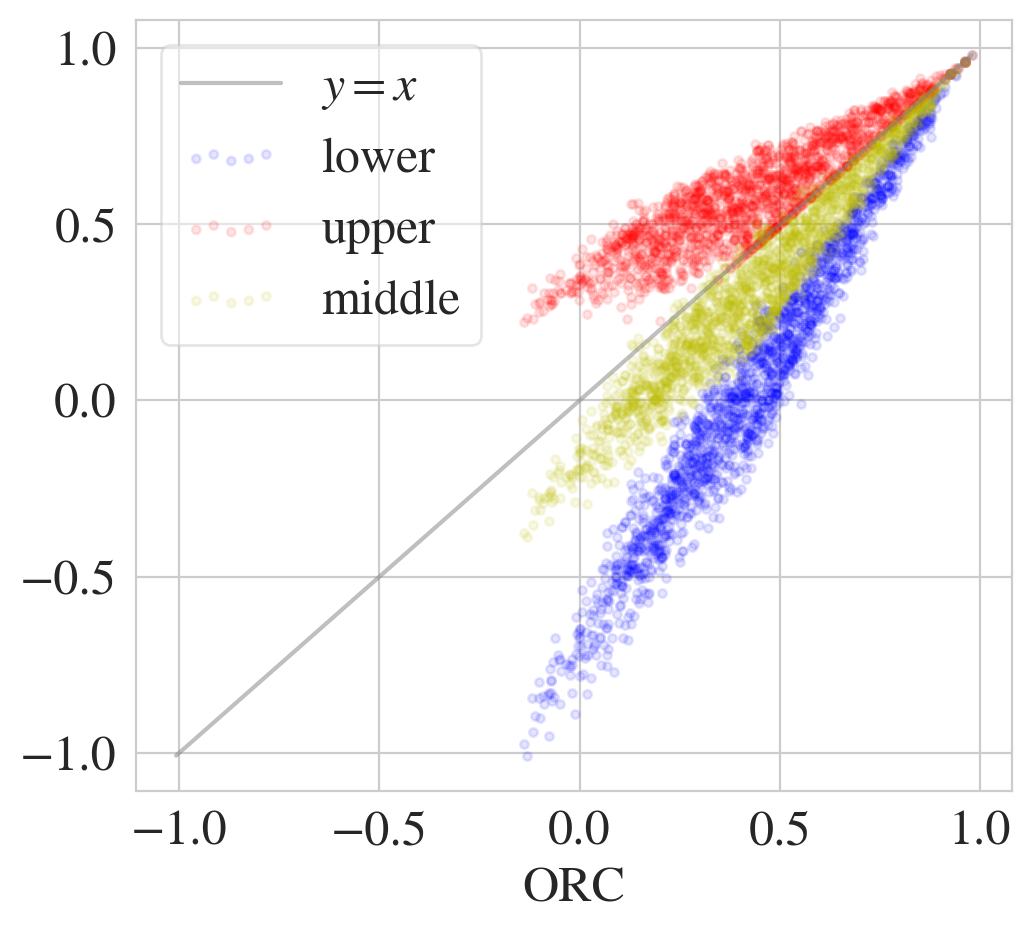} &  \includegraphics[width=.25\textwidth]{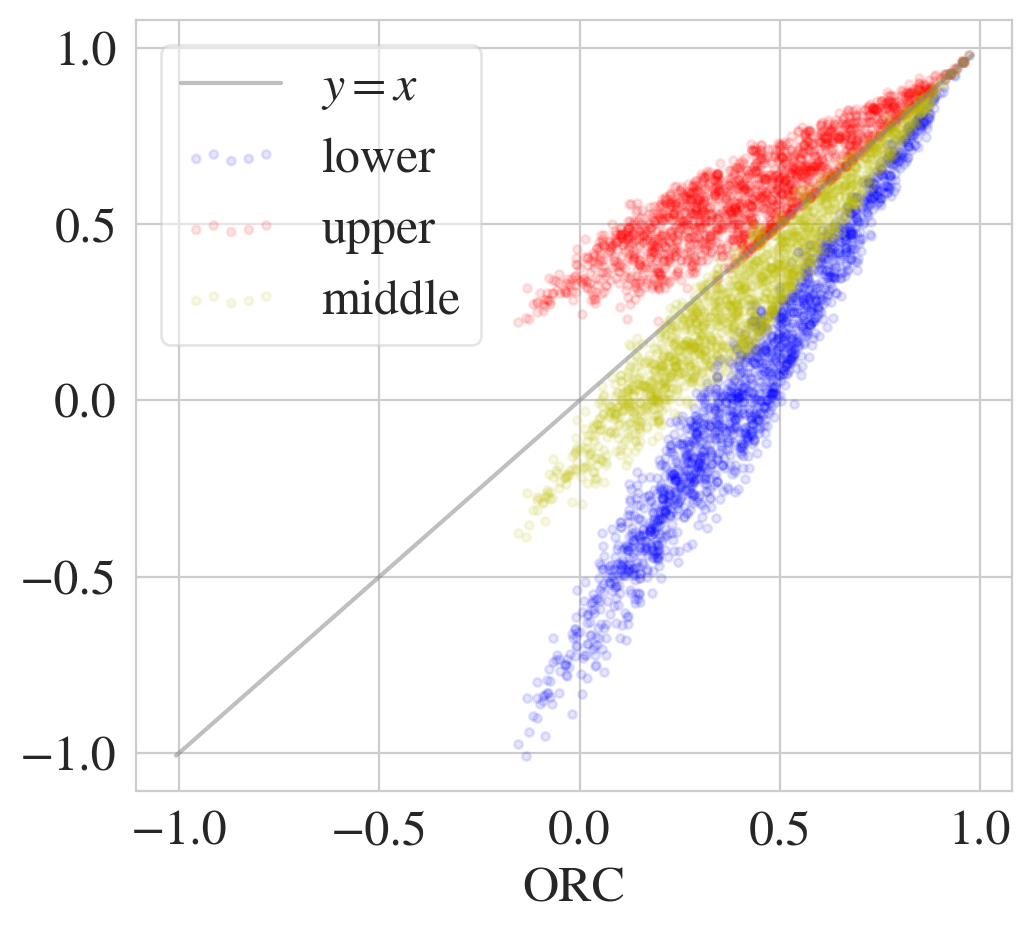} \\
			\includegraphics[width=.3\textwidth]{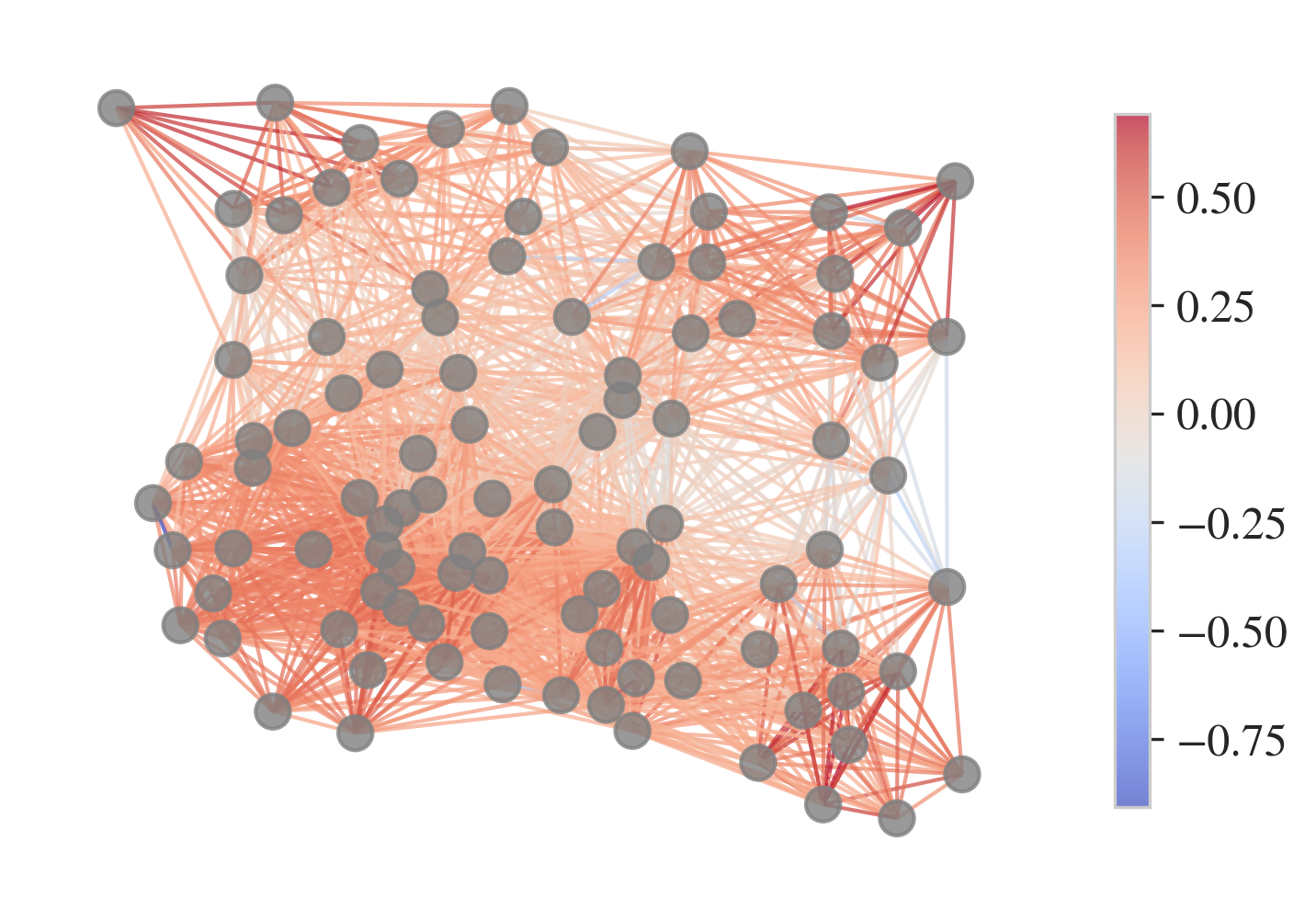} & \includegraphics[width=.25\textwidth]{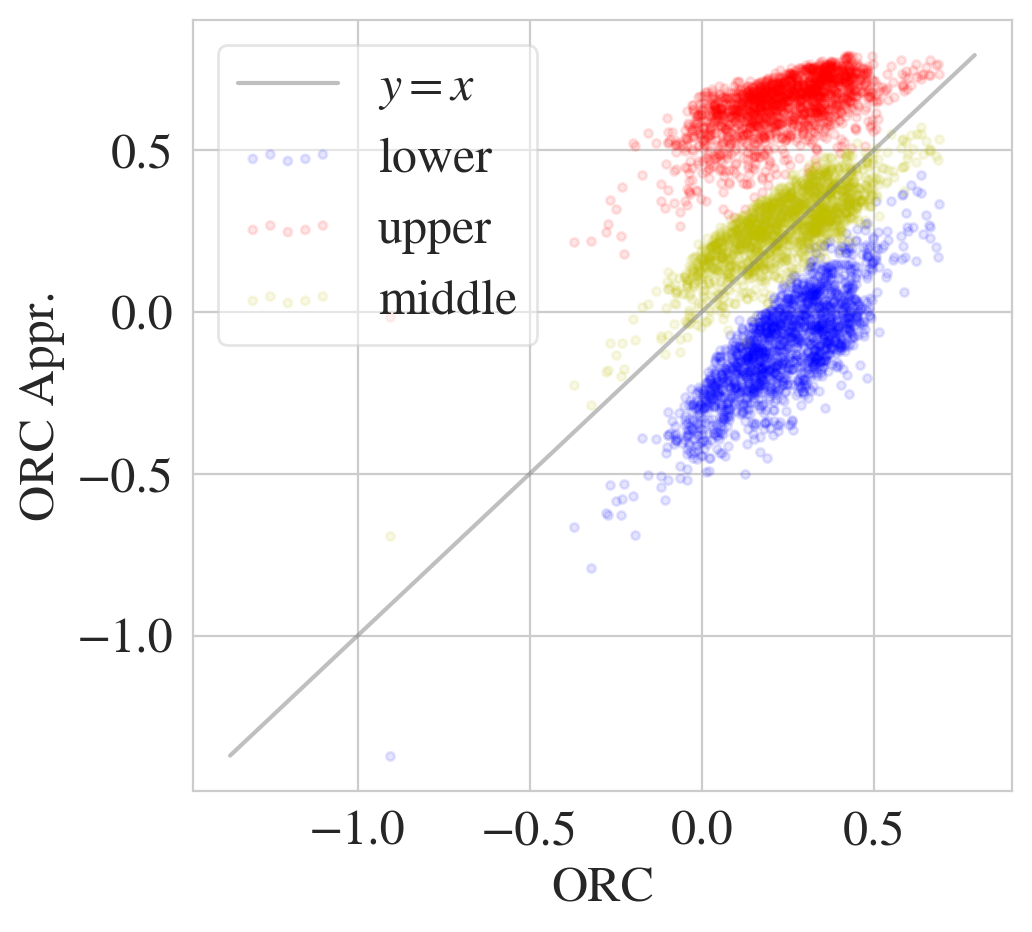} &  \includegraphics[width=.25\textwidth]{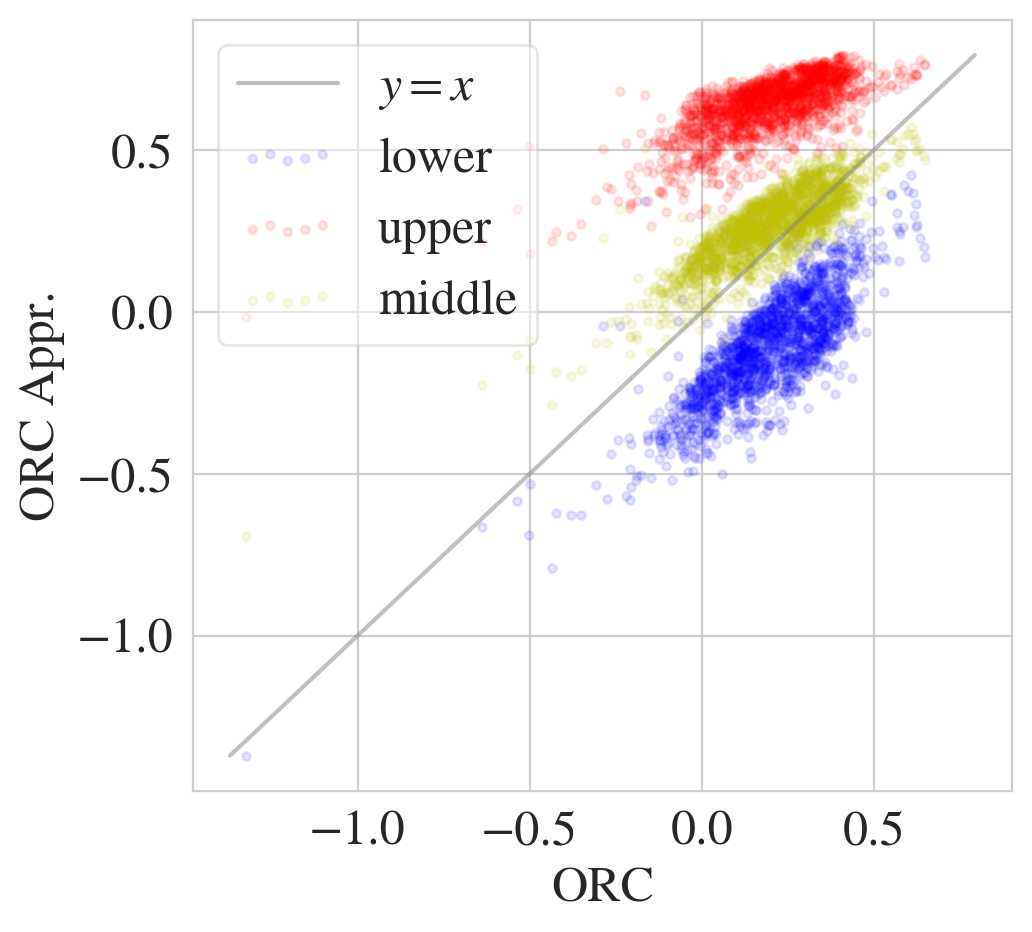}
		\end{tabular}
		\caption{Scatter plot of the ORC on $G$ obtained from the optimization versus its approximation, where the optimization is done by solving the earth mover's distance exactly (middle) and approximately with Sinkhorn (right), and $G$ is obtained from the planted SBM of size $n=100$, two equally-sized blocks, $p_{out}=0.02$ and $p_{in}=0.4$, with weights being $1$ (top), and the $2$-d RGG of size $n=100$ and radius $r = 0.4$ with weights being $1$ (middle), and proportional to the distance (bottom). $G$ is visualised in the left column, with the color indicating the ORC obtained by solving the earth mover's distance exactly.} 
		\label{fig:ORC-approx}
	\end{figure}
	Computing ORC on a graph is quite expensive. The high cost arises mainly from the computation of the $W_1$-distance, which, in the discrete setting, is also known as \emph{earth mover's distance}. Computing $W_1(m_u,m_v)$ between the neighborhoods of two vertices $u,v$ that are connected by an edge, corresponds to solving the linear program
	\begin{align}
		\min\quad &W_1(m_u, m_v) = \inf_{\Lambda} \sum_{i \sim x} \sum_{j \sim y} \lambda_{ij} d_G(i,j)\label{eq:lp}\\
		{\rm s.t.} \quad &\sum_{i \sim x} \lambda_{ij} = m_y^j \;\; \forall \; j \sim y;\qquad 
		\sum_{j \sim y} \lambda_{ij} = m_x^i \; \forall \; i \sim x; \notag \\
		&\lambda_{ij} \geq 0 \;\; \forall \; i,j \; . \notag
	\end{align}
	Classically, $W_1(m_u,m_v)$ is computed using the \emph{Hungarian algorithm}~\citep{kuhn1955hungarian}, the fastest variant of which has a cubic complexity. With growing graph size, the computational cost of the Hungarian algorithm becomes prohibitively large. Alternatively, $W_1(m_u,m_v)$ can be approximated with the faster Sinkhorn algorithm~\citep{sinkhorn1967concerning}, whose complexity is only quadratic. However, even a quadratic cost introduces a significant computational bottleneck on the large-scale network data that we encounter in machine learning and data science applications. Here, we propose a different approach. Instead of approximating the $W_1$-distance, we propose an approximation of ORC. Our approximation can be written as a simple combinatorial quantity, which can be computed in linear time. 
	
	We begin by recalling a few classical bounds on ORC, first proven by~\citet{jost2014ollivier}.
	Let $\#(u,v)$ be the number of triangles that include the edge $(x,y)$, $a \wedge b := \min \{ a,b \}$ and $a \vee b := \max \{ a,b \}$. With respect to these quantities, upper and lower bounds on ORC are given as follows:
	\begin{theorem}[Unweighted case~\citep{jost2014ollivier}]
		\label{the:ollilow}
		\hfill
		\begin{enumerate}
			\item Lower bound on \orc:
			\begin{align}\label{eq:ollilow}
				\Ric_O (\{u,v\}) &\geq - \left(
				1 - \frac{1}{d_v} - \frac{1}{d_u} - \frac{\#(u,v)}{d_u \wedge d_v}\right)_+ \\
				&\quad - \left(
				1 - \frac{1}{d_v} - \frac{1}{d_u} - \frac{\#(u,v)}{d_u \vee d_v}\right)_+ + \frac{\#(u,v)}{d_u \vee d_v} \; . \nonumber
			\end{align}
			Note that a simpler, but less tight lower bound with respect to node degrees only is given by
			\begin{equation*}
				\Ric_O (\{u,v\}) \geq -2 \left(
				1 - \frac{1}{d_v} - \frac{1}{d_u} \right) \; .
			\end{equation*}
			\item Upper bound on \orc:
			\begin{equation}\label{eq:olliup}
				\Ric_O (\{u,v\}) \leq \frac{\#(u,v)}{d_u \vee d_v}\; .
			\end{equation}
		\end{enumerate}
	\end{theorem}
	\noindent In the following section, we will derive an extension of these results to weighted graphs.\\
	
	\noindent Now, let $\Ric_O^{up}$ and $\Ric_O^{low}$ denote the upper and lower bounds, respectively. We propose to approximate $\Ric_O(\{u,v\})$ as the arithmetic mean of the upper and lower bounds, i.e.,
	\begin{equation}\label{eq:orc-approx}
		\widehat{\Ric_O}(\{u,v\}) := \frac{1}{2} \left(\Ric_O^{up}(\{u,v\}) + \Ric_O^{low}(\{u,v\}) \right) \; .
	\end{equation}
	Like the bounds themselves, $\widehat{\Ric_O}$ has a simple combinatorial form and can be computed in linear time. Since the computation relies only on local structural information, the procedure can be parallelized easily. Fig.~\ref{fig:ORC-approx} demonstrates the utility of this approximation in practise (in both the SBM and Random Geometric Graphs (short: RGG); see Appendix \ref{apx:data} for more explanation).

	\subsubsection{Bounds on Ollivier's curvature in weighted graphs}\label{sec:orc-approx}
	
	Since our curvature-based algorithms work by modifying the edge weights in a graph according to their curvature values, we require an analogue of Theorem~\ref{the:ollilow} for weighted graphs, in order to use a simple combinatorial approximation to the curvature as in Equation~\eqref{eq:orc-approx}. To do so, we make use of the dual description of the Wasserstein distance between two measures $m_1, m_2$:
	$$W_1(m_1,m_2)= \inf_{\pi \in \Pi(m_1,m_2)} \EE_{(x,y)\sim \pi} d(x,y) = \sup_{\|h\|_L\leq 1} \left[\EE_{x\sim m_1} [h(x)] - \EE_{y\sim m_2} [h(y)] \right],$$ where $\Pi(m_1,m_2)$ is the set of couplings of $m_1, m_2$, and the supremum is over all functions $h:V\rightarrow\RR$ which are Lipschitz with constant at most 1 (with respect to shortest path distance). We may obtain an upper bound on $W_1(m_1,m_2)$ by finding any coupling $\pi$ of $m_1,m_2$, and a lower bound by finding any 1-Lipschitz function $h$. 
	\paragraph{Upper bound on $W_1(m_1,m_2)$.} By constructing a transport plan taking $m_x$ to $m_y$, we obtain the following upper bound:
	\begin{lem}
		\label{lem:up}
		Let $m_x, m_y$ denote the measures on the node neighborhoods of $x$ and $y$ given by taking $\alpha=0$ and $p=1$ in Equation~\eqref{eq:nodemeasure}. Denote the vertices in $N(x)\setminus N(y)$ with $\ell$, $N(y)\setminus N(x)$ with $r$, and $N(x)\cap N(y)$ with $c$. Introduce the shorthands $L_x := \sum_\ell m_x(\ell)$, $L_y := \sum_\ell m_y(\ell)$, as well as
		$X_x:=m_x(x)$, $X_y:=m_y(x)$. Then
		\begin{align*}
			W_1(m_x,m_y)&\leq \sum_\ell w_{\ell,x} m_x(\ell)+\sum_r w_{r,y} m_y(r)\\
			&\quad+\sum_c w_{c,y}(m_x(c)-m_y(c))_+ +w_{c,x}(m_y(c)-m_x(c))_+\\
			&\quad+ \left|L_x+X_x - X_y - \sum_c (m_y(c)-m_x(c))_+\right|w_{x,y} \; .
		\end{align*}
	\end{lem}

	\begin{rmk}
		This upper bound is exact for weighted trees. Consider $x\sim y$ in the weighted tree $G$, and define $\epsilon=\mathrm{sgn}(L_x+X_x-X_y)$. Define the 1-Lipschitz function $h$ as follows:
		$$
		h(v)= \begin{cases}
			w_{\ell,x}+\epsilon w_{x,y}&\text{if }v=\ell\in N(x)\setminus N(y)\\
			\epsilon w_{x,y}&\text{if }v=x\\
			0&\text{if }v=y\\
			-w_{r,y}&\text{if }v=r\in N(y)\setminus N(x)\\
			w_{\ell,x}+\epsilon w_{x,y}&\text{if }v\in \mathcal{X}_\ell\\
			-w_{r,y}&\text{if }v\in \mathcal{Y}_r,
		\end{cases}
		$$
		where $\mathcal{X}_\ell$ is the set of vertices $v\in V\setminus(N(x)\cup N(y))$ such that $d_G(v,x)<d_G(v,y)$, and the shortest $(v,x)$ path in $G$ passes through $\ell$; $\mathcal{Y}_r$ is the set of vertices $v\in V\setminus (N(x)\cup N(y))$ such that $d_G(v,x)>d_G(v,y)$, and the shortest $(v,y)$ path in $G$ passes through $r$. In other words, $h$ is constant on any branches of the tree starting from an edge $\{\ell,x\}$ or $\{y,r\}$. This function $h$ is 1-Lipschitz since there are no additional paths in $G$ that have not been considered in the subgraph $H$: this would create a cycle, violating the assumption that $G$ is a tree. Let the upper bound from the lemma be denoted $U_{x,y}$. We have
		\begin{align*}
			W_1(m_x,m_y)&\geq \EE_{v\sim m_x} [h(v)] - \EE_{v\sim m_y} [h(v)]\\
			&= \sum_{v\in V} h(v) (m_x(v)-m_y(v))\\
			&= \sum_{\ell} (w_{\ell,x}+\epsilon w_{x,y})m_x(\ell)+ \epsilon w_{x,y}(m_x(x)-m_y(x))- \sum_r w_{r,y} (-m_y(r))\\
			&= U_{x,y}\geq W_1(m_x,m_y).
		\end{align*}
		Thus, $W_1(m_x,m_y)=U_{x,y}$ in this setting, as we wanted to show.
	\end{rmk}
	As a computational tool, this bound is a step in the right direction: Consider the usual procedure for calculating ORC. We need to compute $m_x, m_y$, and $d_G(u,v)$ for each pair of vertices $u,v\in V_H$. Then we need to carry out the optimization, which can be realized as a linear program as in Equation~\eqref{eq:lp}. This lower bound still requires the first step, but avoids the need for the optimization over couplings, which is the major bottleneck in the curvature computation.
	
	\paragraph{Lower bound on $W_1(m_1,m_2)$.} 
	We now provide a lower bound on $W_1(m_x,m_y)$:

	\begin{lem}
		\label{lem:low}
		Let $m_x, m_y$ be as defined above, and let $\mathcal{P}= \{v\in N(x)\cup N(y): m_x(v)-m_y(v)>0\}$, $\mathcal{N}=\{v\in N(x)\cup N(y): m_x(v)-m_y(v)<0\}$. For $S\subseteq V$ and $v\in V$, let $d_G(v,S)=\min\{d_G(v,u): u\in S\}$. Then
		$$W_1(m_x,m_y)\geq \max\left\{ \sum_{v\in\mathcal{P}} d_G(v,\mathcal{N})(m_x(v)-m_y(v)), \sum_{v\in\mathcal{N}} d_G(v,\mathcal{P})(m_y(v)-m_x(v))\right\}.$$
	\end{lem}
	
	\begin{rmk}\normalfont
		A lower bound for $W_1(m_x,m_y)$ was previously given by~\citet{jost2014ollivier}. However, the ORC notion considered therein differed from ours in the choice of the mass distribution imposed on node neighborhoods. In addition, we believe that the result in~\citet{jost2014ollivier} has a slight inaccuracy; for details, see the discussion in Remark \ref{rmk:jost-liu-issue}.
	\end{rmk}

	\paragraph{Combinatorial bounds on ORC in the weighted case.}
	We now utilize the combinatorial upper and lower bounds on $W_1(m_x,m_y)$ to obtain combinatorial bounds on ORC itself.
	
	\begin{theorem}[ORC bounds (weighted case)]
		\label{thm:weightollibnds}
		For $x,y\in V, x\sim y$, we have the following bounds: 
		\begin{enumerate}
			\item Upper bound on \orc:
			\begin{equation*}
				\Ric_O(\{x,y\})\leq 1- \max \left\{\sum_v \frac{d_G(v,\mathcal{N})}{w_{xy}}(m_x(v)-m_y(v))_+, \sum_v \frac{d_G(v,\mathcal{P})}{w_{xy}}(m_y(v)-m_x(v))_+ \right\}.
			\end{equation*}
			\item Lower bound on \orc:
			\begin{align*}
				\Ric_O(\{x,y\})&\geq 1- \sum_\ell \frac{w_{\ell,x}}{w_{x,y}} m_x(\ell)-\sum_r \frac{w_{r,y}}{w_{x,y}} m_y(r)\\
				&\quad-\sum_c \left[\frac{w_{c,y}}{w_{x,y}}(m_x(c)-m_y(c))_+ +\frac{w_{c,x}}{w_{x,y}}(m_y(c)-m_x(c))_+\right]\\
				&\quad- \left|L_x+X_x - X_y - \sum_c (m_y(c)-m_x(c))_+\right| \; .
			\end{align*}
		\end{enumerate}
	\end{theorem}
	While technical, both bounds are purely combinatorial and can be computed in linear time without the need to solve a linear program.

	\begin{example}[$d$-regular asymmetric unweighted graphs]
		Suppose $G$ is a $d$-regular asymmetric unweighted graph. Now take any $x\sim y$ in $G$. 
		
		We have:
		\begin{align*}
			X_x+L_x-X_y-\sum_c (m_y(c)-m_x(c))_+&= \alpha+(1-\alpha)\left(\frac{\#\ell}{d}-\frac{1}{d}- \frac{\#c}{d}+\frac{\#c}{d}\right)\\
			&= \alpha+(1-\alpha)\frac{\#\ell-1}{d}\geq0,
		\end{align*}
		since the assumption of asymmetry ensures that $x$ and $y$ each have at least one unshared neighbor (otherwise there would be a nontrivial automorphism of $G$ swapping $x$ and $y$). Now suppose $\alpha=1/(d+1)$, so that $m_x(x)=m_x(v)=1/(d+1)$ for each $v\sim x$. Note that $m_x(c)=m_y(c)$ for any common vertices, and $x\not\in \mathcal{N}$. Then a shortest path for a vertex of type $\ell$ to the set $\mathcal{N}$ must take one of three forms:
		\begin{enumerate}[(a)]
			\item $(\ell,x,y,r)$
			\item $(\ell,c,y,r)$
			\item $(\ell,c,r)$: Note that this case is impossible if $d=3$, since it requires $c$ to be adjacent to $x,y,\ell,r$.
			\item $(\ell,r)$: This corresponds to a chord-free 4-cycle in $G$ with vertices $(x,y,r,\ell)$.
		\end{enumerate}
		Note that the only vertices in $\mathcal{P}$ must be of type $\ell$, and the only vertices in $\mathcal{N}$ must be of type $r$. If $v\in \mathcal{P}$ has a path of type (d) to $\mathcal{N}$, then $d(\ell,\mathcal{N})=1$; if it has no path of type (d), but has a path of type (c) to $\mathcal{N}$, then $d(\ell,\mathcal{N})=2$; otherwise, it must have $d(\ell,\mathcal{N})=3$. If we partition the vertices of type $\ell$ into the sets $D, C, A$ ($A$ for all other vertices of type $\ell$) based on which type of path they have to $\mathcal{N}$, we get
		$$\Ric_O(\{x,y\})\leq 1-\frac{|D|}{d+1}-2\frac{|C|}{d+1}-3\frac{|A|}{d+1}.$$
		For a lower bound, we get 
		$$\Ric_O(\{x,y\})\geq 1- 2\frac{\#\ell}{d+1}-\frac{\#r}{d+1}=1-3\frac{|A|+|C|+|D|}{d+1},$$ since $\#r=\#\ell=|A|+|C|+|D|.$ In particular, when $|C|=|D|=0$, these are equal.
	\end{example}

	\subsection{Forman's Ricci Curvature}
	
	\noindent Our second notion of discrete curvature utilizes an analogy between spectral properties of manifolds and CW complexes to define a discrete Ricci curvature.
	
	\subsubsection{Formal Definition}
	Forman's discrete Ricci curvature (short: \frc) was originally defined as a measure of geodesic dispersion on CW complexes~\citep{forman2003bochner}.  In its most general form,  \frc is defined via a discrete analogue of the Bochner-Weitzenböck identity
	\begin{align*}
		\square_d = B_d + F_d \; ,
	\end{align*}
	which establishes a connection between the (discrete) Riemann-Laplace operator $\square_d$,  the Bochner Laplacian $F_d$ and Forman's curvature tensor $F_d$.  Networks can be viewed as polyhedral complexes, which is one instance of a CW complex.  In particular, by viewing a network's edges as 1-cells,  we can define \frc for edges $e=(v_1,v_2)$  as 
	\begin{align}\label{eq:frc-e}
		\Ric_{F} (e) = \w_e \left( \frac{\w_{v_1}}{\w_e} + \frac{\w_{v_2}}{\w_e} - \sum_{e_{v_1} \sim e} \frac{\w_{v_1}}{\sqrt{\w_e \w_{e_{v_1}}}} -\sum_{e_{v_2}\sim e} \frac{\w_{v_2}}{\sqrt{\w_e \w_{e_{v_2}}}}
		\right)
	\end{align}
	To arrive at this expression, we have specialized Forman's general notion for $F_d$ to the case of 1-cells (see also~\citep{WSJ1,WSJ2}).  Here,  $e_{v}$ denotes an edge that shares the vertex $v$ with $e$. 
	Note that if $G$ is unweighted, then this reduces to
	\begin{align*}
		\Ric_F (e) = 4 - d_{v_1} - d_{v_2} \; .
	\end{align*}
	We can define \frc for vertices (i.e., 0-cells) with respect to the curvature of its adjacent edges $E_v$ (i.e.,  $E_v := \lbrace e \in E: \; e=(\cdot,v) \; {\rm or} \; e=(v,\cdot)$):
	\begin{align}\label{eq:frc-v}
		\Ric_F(v) &= \sum_{e_u \in E_v} \Ric_F (e_u) \; .
	\end{align}
	
	It is well-known that higher-order structures impact the community structure in a graph, as well as the ability of many well-known methods to detect community structure. For example, the clustering coefficient, a well-known graph characteristic, can be defined via triangle counts (\citet{watts-strogatz}, see also sec.~\ref{sec:clustering-coeff}). Therefore, we
	additionally define an FRC notion, which takes contributions of higher-order structures into account. To characterize such structure, we introduce the following notion: We say that two edges $e,\hat{e}$ are \emph{parallel}, denoted as $e \| \hat{e}$, if they are adjacent to either the same vertex ($v \sim e,\hat{e}$) or the same face ($f \sim e,\hat{e}$), but not both.
	In the following, we consider 
	the 2-complex version of FRC~\citep{WSJ2}:
	\begin{align}\label{eq:frc-f}
		\Ric_{F} (e) &:= \w_{e} \left( \Big(\sum_{f \sim e}\frac{\w_{e}}{\w_{f}} \Big) + \frac{\w_{v_1}}{\w_{e}} + \frac{\w_{v_2}}{\w_{e}} 
		- \sum_{ \hat{e} \parallel e}\left\vert
		\sum_{f \sim \hat{e},e} \frac{\sqrt{\w_{e} \cdot \w_{\hat{e}}}}{\w_f} - \sum_{v \sim e, \hat{e}} \frac{\w_v}{\sqrt{\w_e \cdot \w_{\hat{e}}}} \right\vert \right)  \; .
	\end{align}
	Here, $f$ denotes a 2-face of order $k$, such as a triangle ($k=3$), a quadrangle ($k=4$), a pentagon ($k=5$), etc.

	\subsubsection{Computational Considerations} Notice that Eq.~\eqref{eq:frc-e} and Eq.~\eqref{eq:frc-v} give a simple, combinatorial curvature notion, which can be efficiently computed even on large-scale graphs. 
	To compute the 2-complex FRC of an edge $e$ exactly, one needs to identify all higher order faces in the neighborhood of an edge $e$ and compute their respective contribution in Eq.~\eqref{eq:frc-f}. This limits the scalability of this notion significantly. However, notice that the likelihood of finding a face of order $k$ in the neighborhood of an edge decreases rapidly as $k$ increases.  Fig.~\ref{fig:k-faces} illustrates this observation with summary statistics for simulations of two popular graph models (the SBM and RGG). We propose to approximate the 2-complex FRC by only considering triangular faces and ignoring structural information involving $k$-faces with $k>3$. This choice reduces the computational burden of the 2-complex FRC significantly. We demonstrate below that, experimentally, the performance of our FRC-based clustering algorithm does not improve, if higher-order faces are taking into account for the curvature contribution (see Table \ref{tab:sbm4-nmi}). Specifically, we will utilize the following \emph{augmented FRC}:
	\begin{equation}\label{eq:frc-2-complex}
		\Ric_{F} (e) := \w_e \Biggl( \Big(\sum_{ \triangle \sim e}\frac{\w_e}{\w_{\triangle}}\Big) + \frac{\w_{v_1}}{\w_{e}} + \frac{\w_{v_2}}{\w_{e}} - \Biggl[
		\sum_{\substack{e_{v_1} \sim e\\ e_{v_1} \parallel e}}\frac{\w_{v_1}}{\sqrt{\w_{e} \w_{e_{v_1}}}} + \sum_{\substack{e_{v_2} \sim e\\ e_{v_2} \parallel e}}\frac{\w_{v_2}}{\sqrt{\w_e  \w_{e_{v_2}}}}
		\Biggr] \Biggr) \; .
	\end{equation}
	We consider the following weighting scheme for faces, which utilizes Heron's formula to determine face weights via edge weights: Let $f=(e_i,e_j,e_k)$ denote a triangle in $G$, i.e., $e_i \sim e_j$, $e_j \sim e_k$ and $e_i \sim e_k$. Then we set
	\begin{align}
		\label{equ:heron}
		\w_f &:= \sqrt{s(s-\w_{e_i})(s-\w_{e_j})(s-\w_{e_k})} \\
		s &= \frac{\w_{e_i}+\w_{e_j}+\w_{e_k}}{2} \; . \nonumber
	\end{align}
	This weighting scheme was previously used in~\citep{WSJ2}.
	\begin{figure}[tp]
		\centering
		\begin{tabular}{cc}
			\includegraphics[width=.35\textwidth]{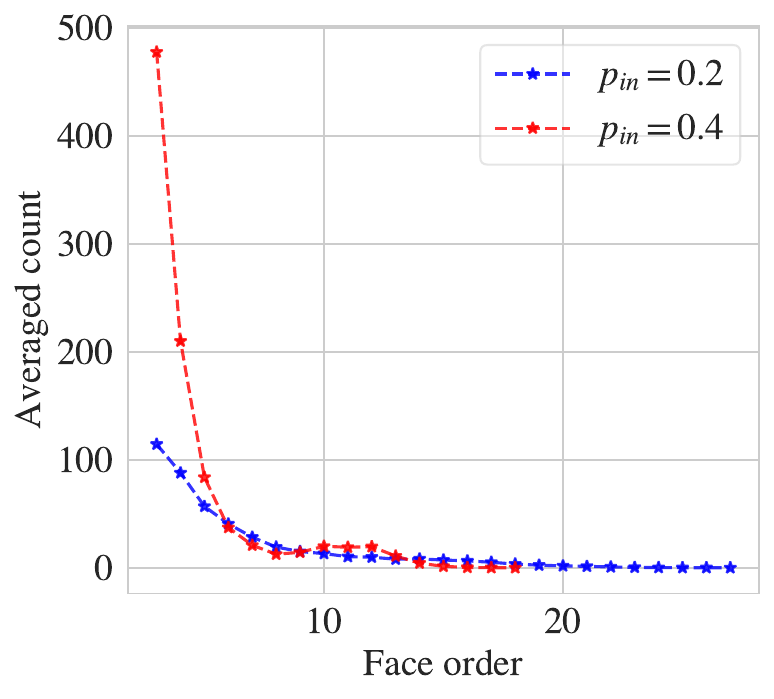} & \includegraphics[width=.35\textwidth]{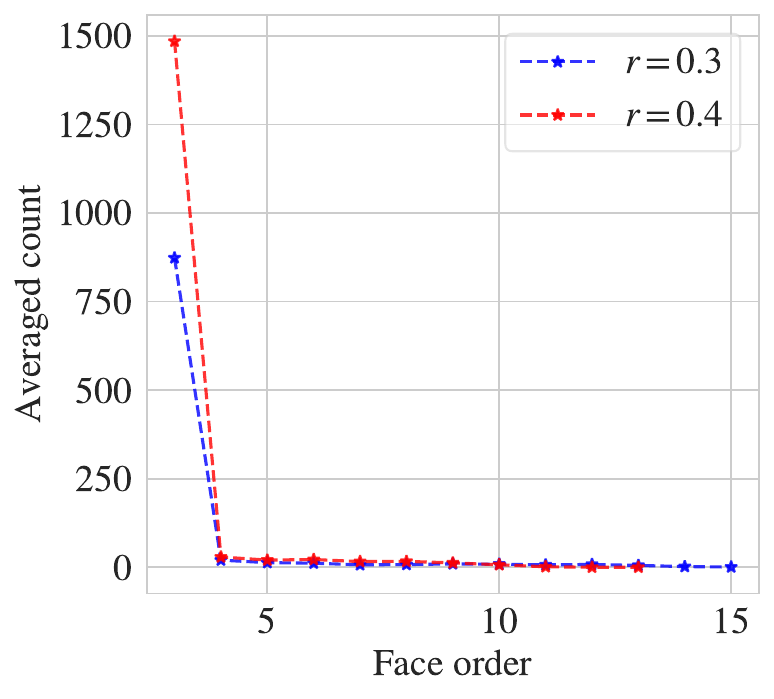} 
		\end{tabular}
		\caption{Number of $k$-faces in graphs obtained from the planted SBM of size $n=100$, two equally-sized blocks, $p_{out}=0.02$, and the $2$-d RGG of size $n=100$, where the results are averaged over $n_s = 50$ realisations of each set of parameters.}
		\label{fig:k-faces}
	\end{figure}
	%
	
	\section{Discrete Curvature on the Line Graph}\label{sec:line-graph-curv}
	In this section, we investigate the computation of discrete Ricci curvature on the line graph. Again, we consider two notions of curvature, Ollivier's (ORC) and Forman's (FRC). Specifically, we want to describe the relationship between a graph's curvature and the curvature of its line graph.
	
	We have seen in Section~\ref{sec:graph-curv} that curvature can be defined on the node- and edge-level. Since nodes in the line graph correspond to edges in the underlying graph, it is natural to study the relationship between node-level curvature in the line graph and edge-level curvature in the original graph. In the context of the clustering applications considered in this paper, we are further interested in the insight that curvature can provide on differences between node clustering (based on connectivity in the original graph) and edge clustering (based on connectivity in the line graph). In this context, it is instructive to study the relationship between edge-level curvature in the line graph and edge-level curvature in the underlying graph.
	
	Recall that, given an unweighted, undirected graph $G=(V,E)$, its line graph is defined as $L(G)=(E,\mathcal{E})$, where an edge $\{\{u,v\},\{r,s\}\}\in\mathcal{E}$ in the line graph implies that $|\{u,v\}\cap \{r,s\}|=1$. We will make use of the following simple observation throughout the section:
	\begin{lem}\label{lem:degree}
		Node degrees in the line graph are given as
		\begin{equation*}
			\label{eq:degreerelation}
			d_{\{u,v\}} = d_u + d_v -2 \; ,
		\end{equation*}
		where the left hand side denotes node degrees in the line graph $L(G)$, and the right hand side the degrees of the vertices $u$ and $v$ in $G$.
	\end{lem}
	This formula follows from the fact that the edges connecting to $\{u,v\}$ in $L(G)$ are precisely those edges of the form $\{u,z\}\in E$ or $\{v,z\}\in E$, where $z\neq u,v$.
	
	\subsection{Numerical observations}
	\noindent We explore the relationship of (1) edge-level curvature in the original graph $G$ and node-level curvature in the corresponding line graph $L$ numerically, as well as (2) edge-level curvatures in the original graph $G$ and its line graph $L$. Our empirical results consider the two notions of curvature introduced above: Forman's Ricci curvature (FRC) and Ollivier's Ricci curvature (ORC). Our data sets are described in Appendix~\ref{apx:data}.
	To compute the classical FRC and ORC notions in the original graph $G$ and its line graph $L$ numerically, we use the built-in ORC computation both for edges and for nodes in \citep{ni2019community}. Our implementations for augmented FRC and our ORC approximation build on the code in this library. 
	Results from each random graph model are obtained from $n_s = 50$ realisations.

	\begin{figure}[!th]
		\centering
		\hspace*{-1.5em}
		\begin{tabular}{cccc}
			\includegraphics[width=.24\textwidth]{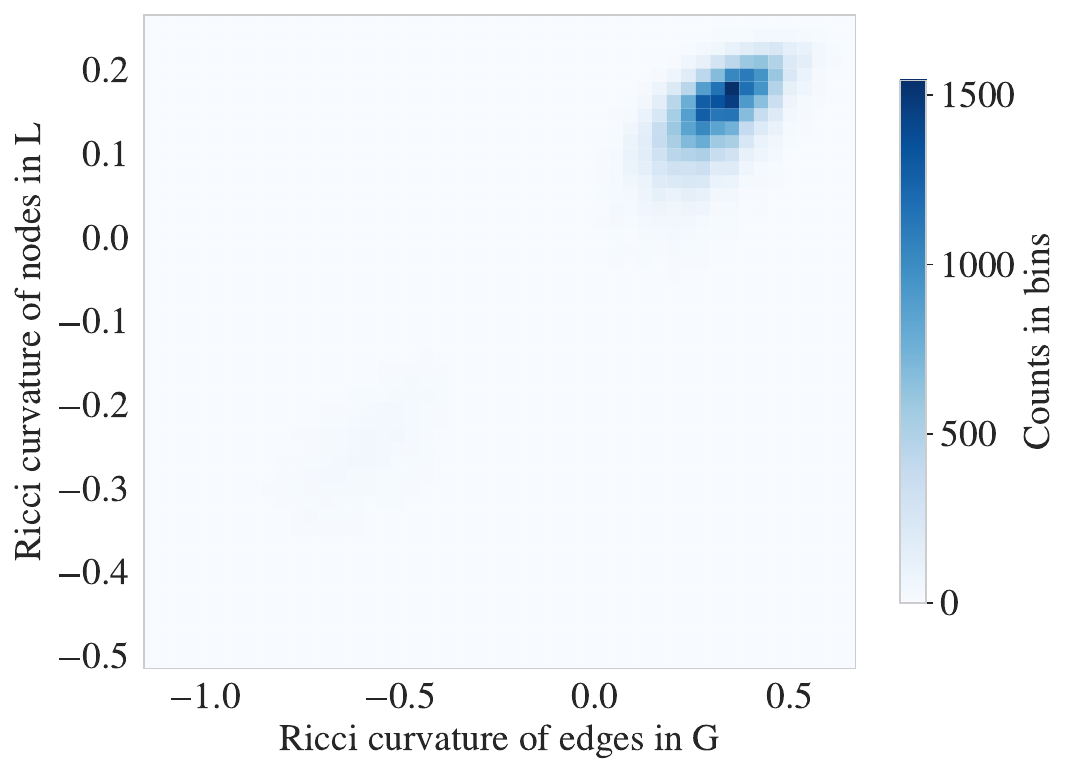} & \includegraphics[width=.24\textwidth]{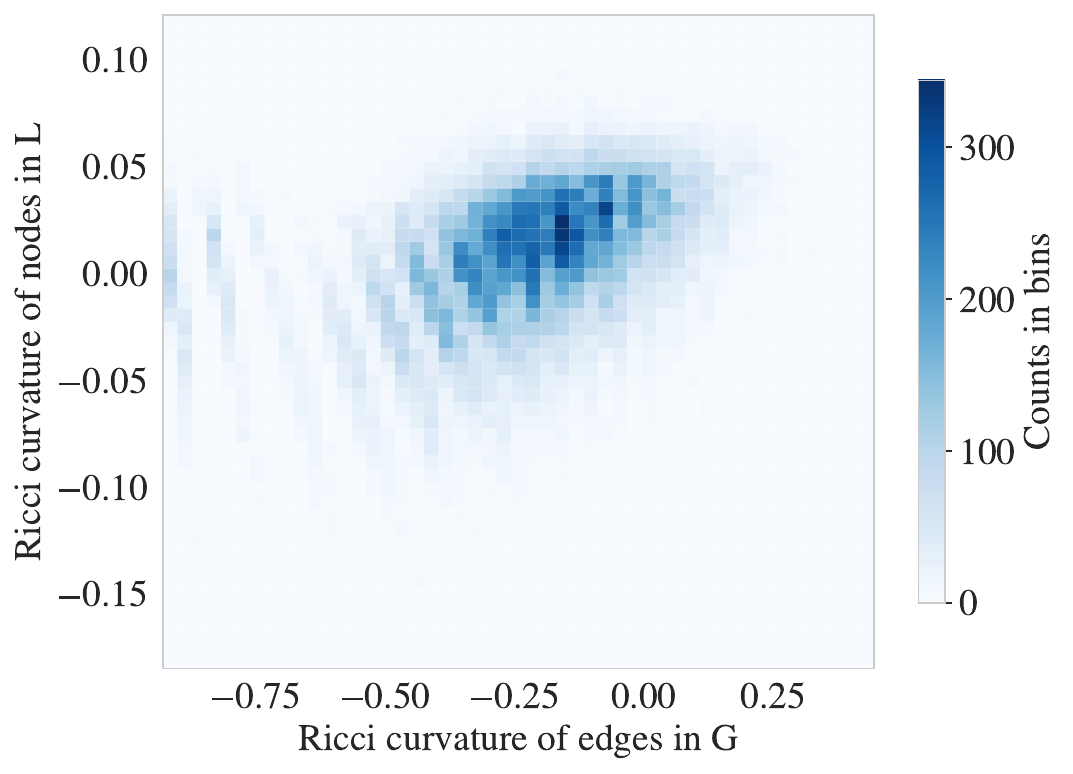} & \includegraphics[width=.24\textwidth]{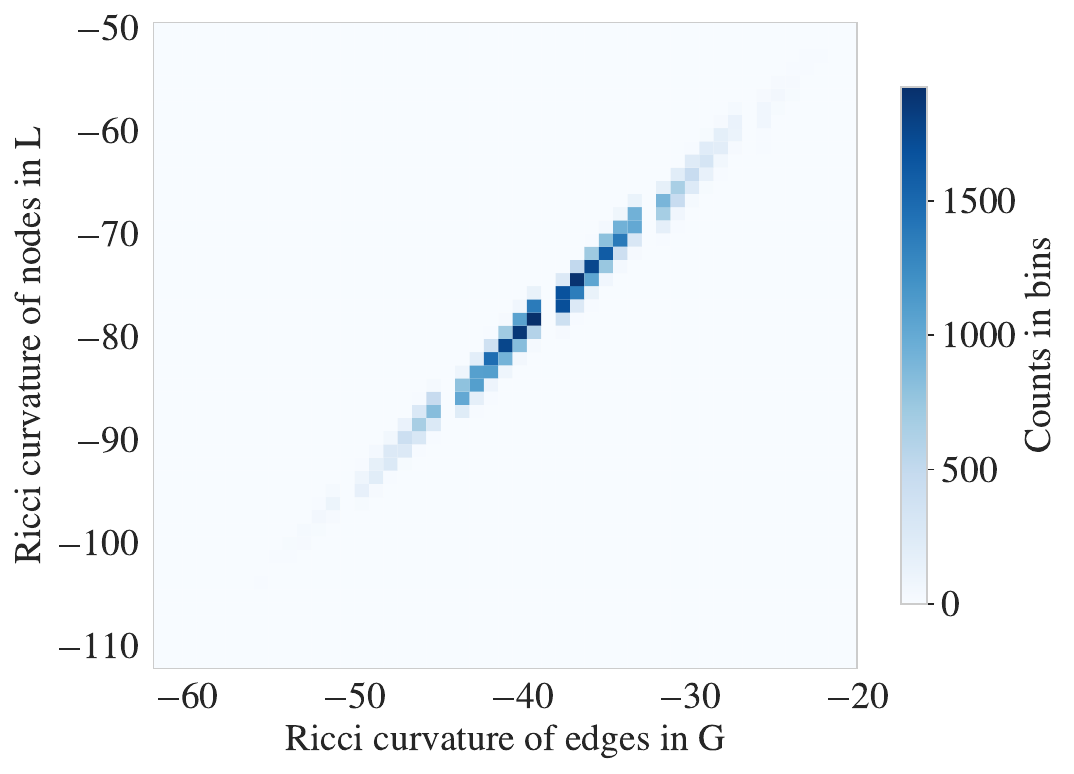} & \includegraphics[width=.24\textwidth]{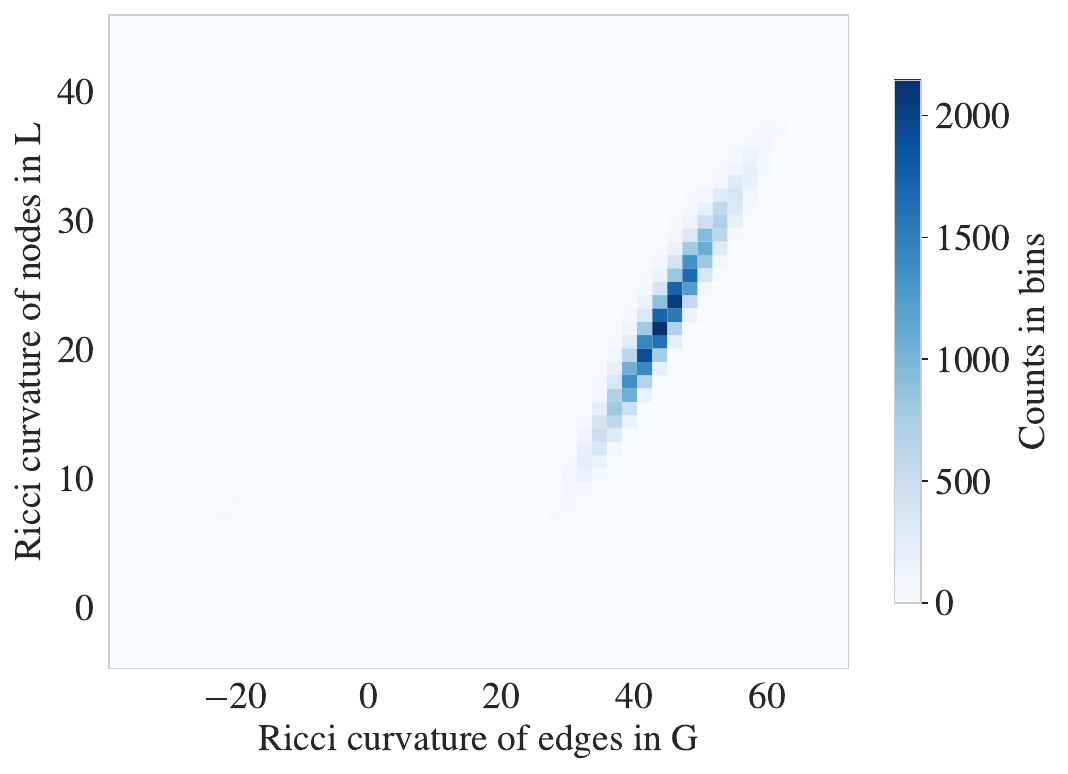}\\
			\includegraphics[width=.24\textwidth]{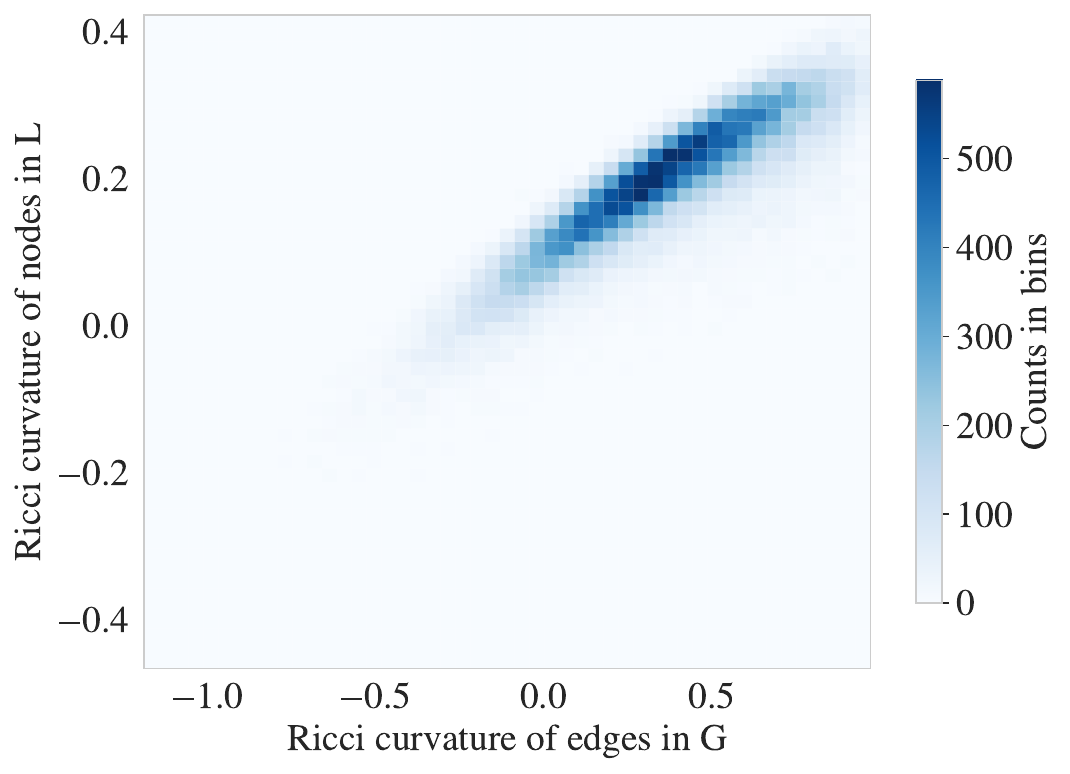} & \includegraphics[width=.24\textwidth]{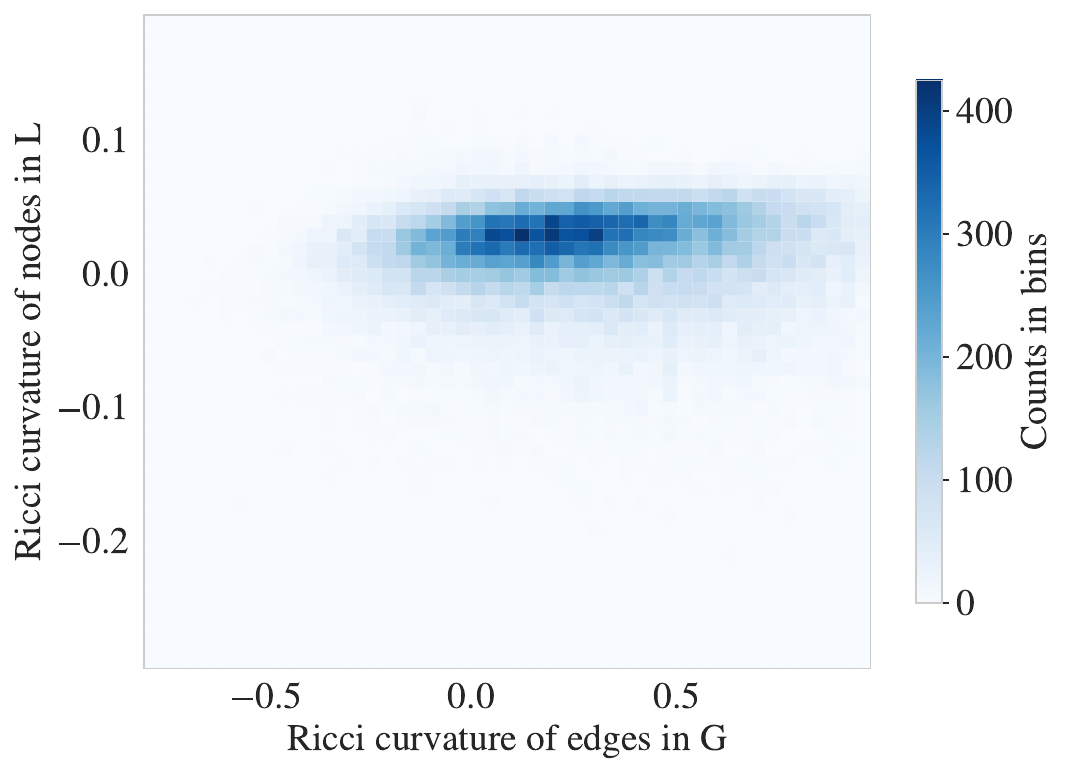} & \includegraphics[width=.24\textwidth]{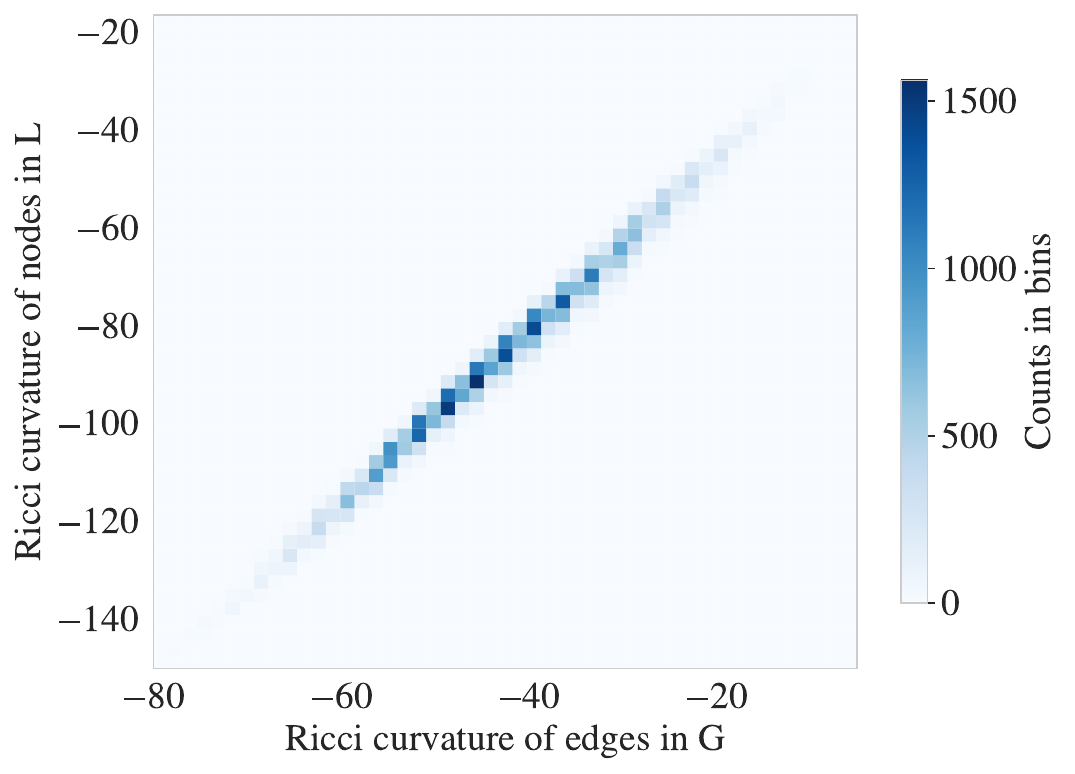} & \includegraphics[width=.24\textwidth]{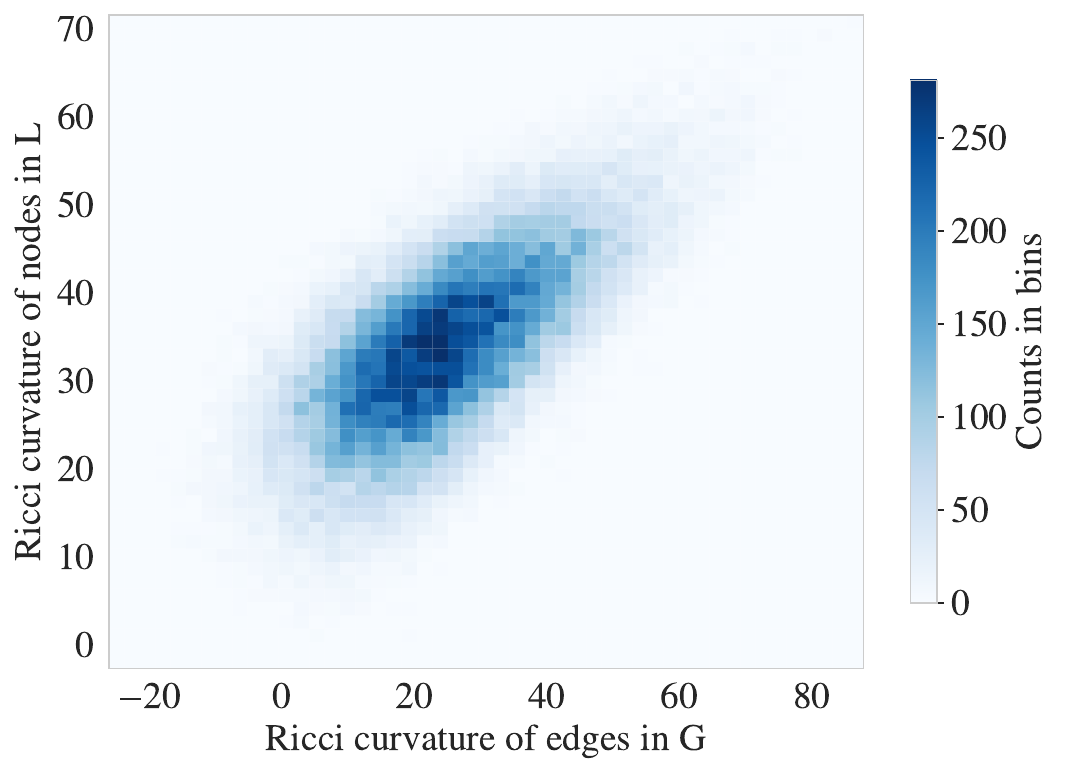}\\
			\includegraphics[width=.24\textwidth]{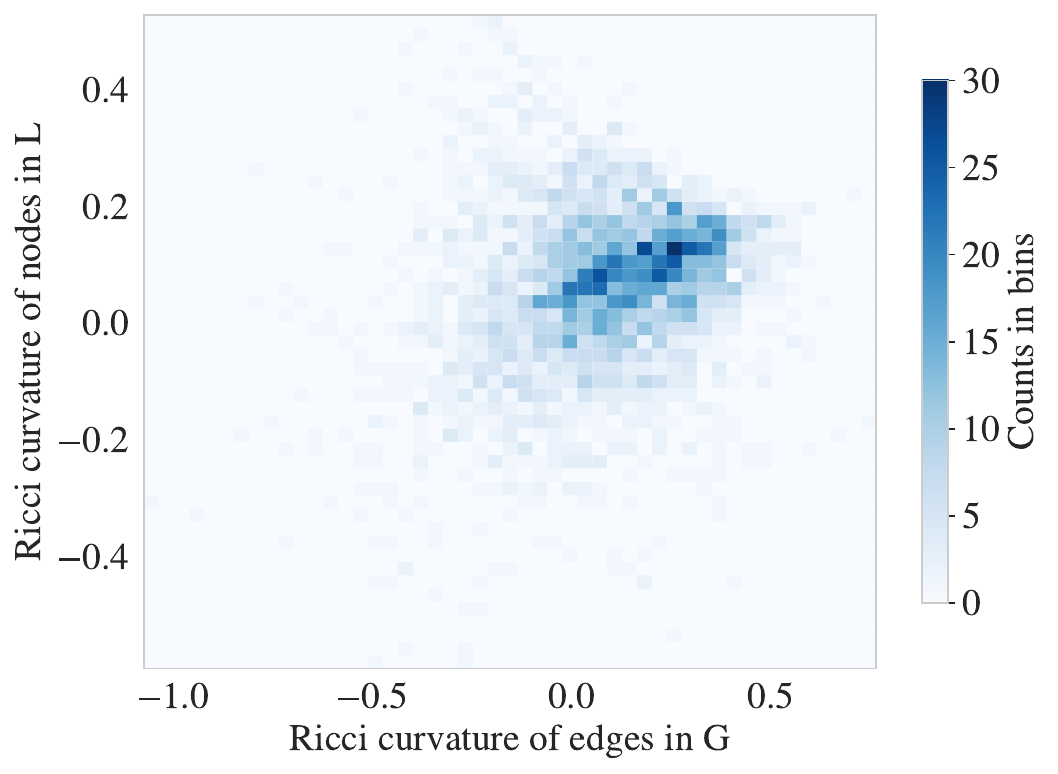} & \includegraphics[width=.24\textwidth]{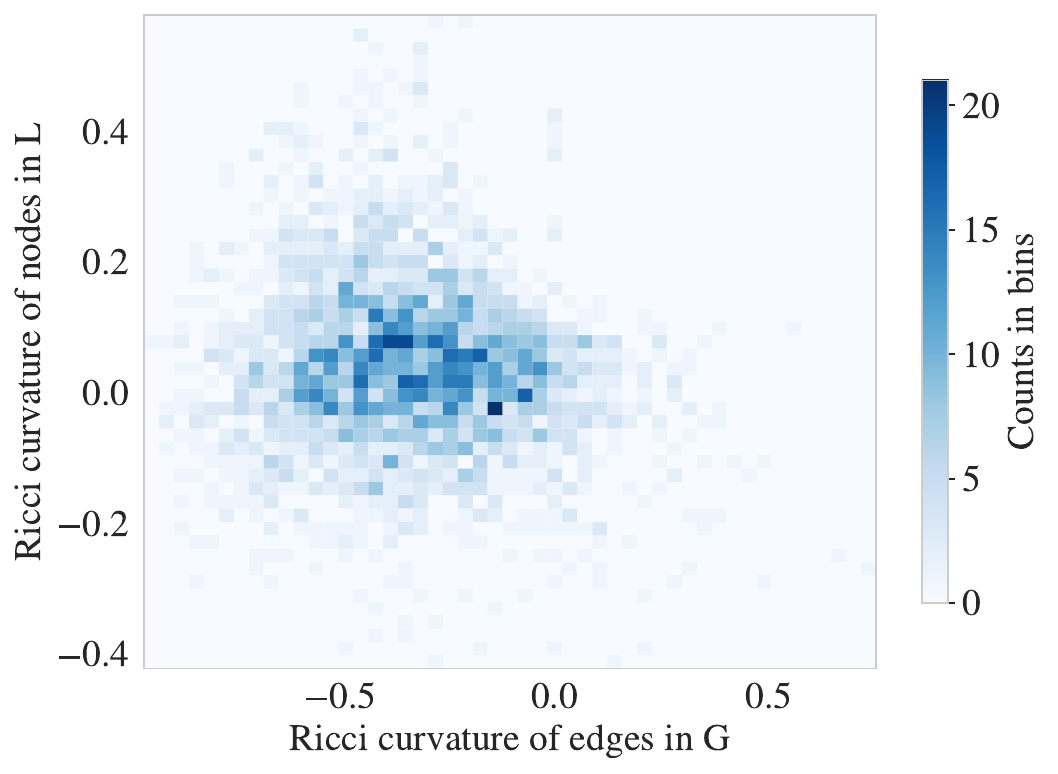} & \includegraphics[width=.24\textwidth]{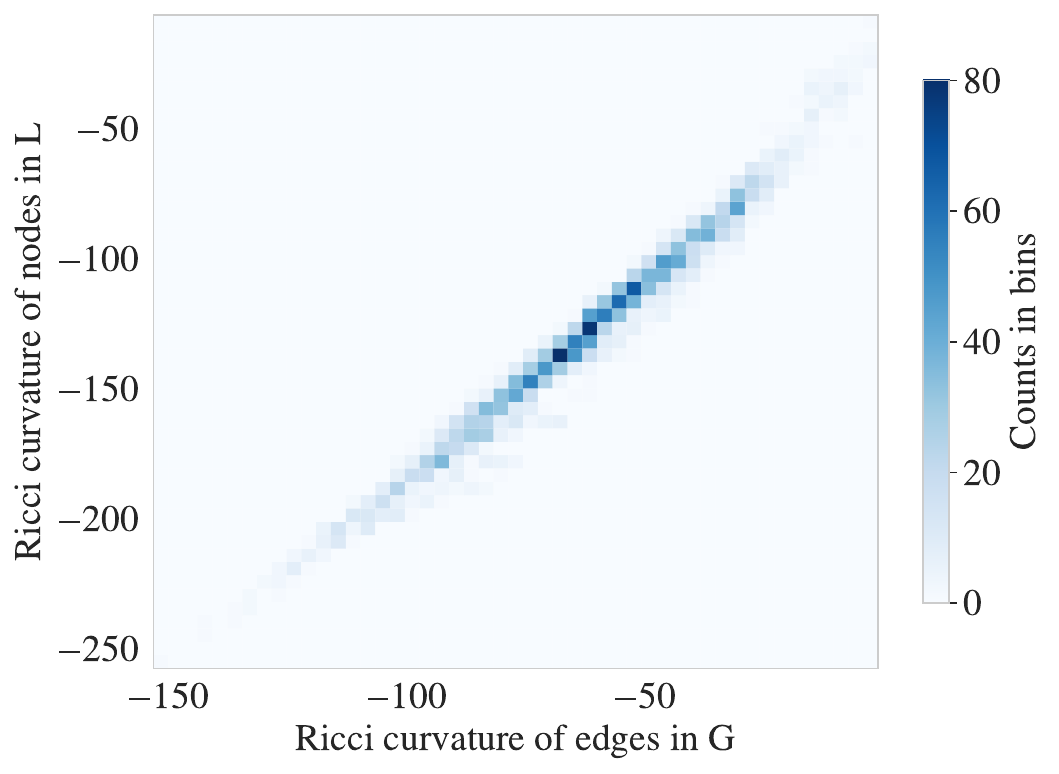} & \includegraphics[width=.24\textwidth]{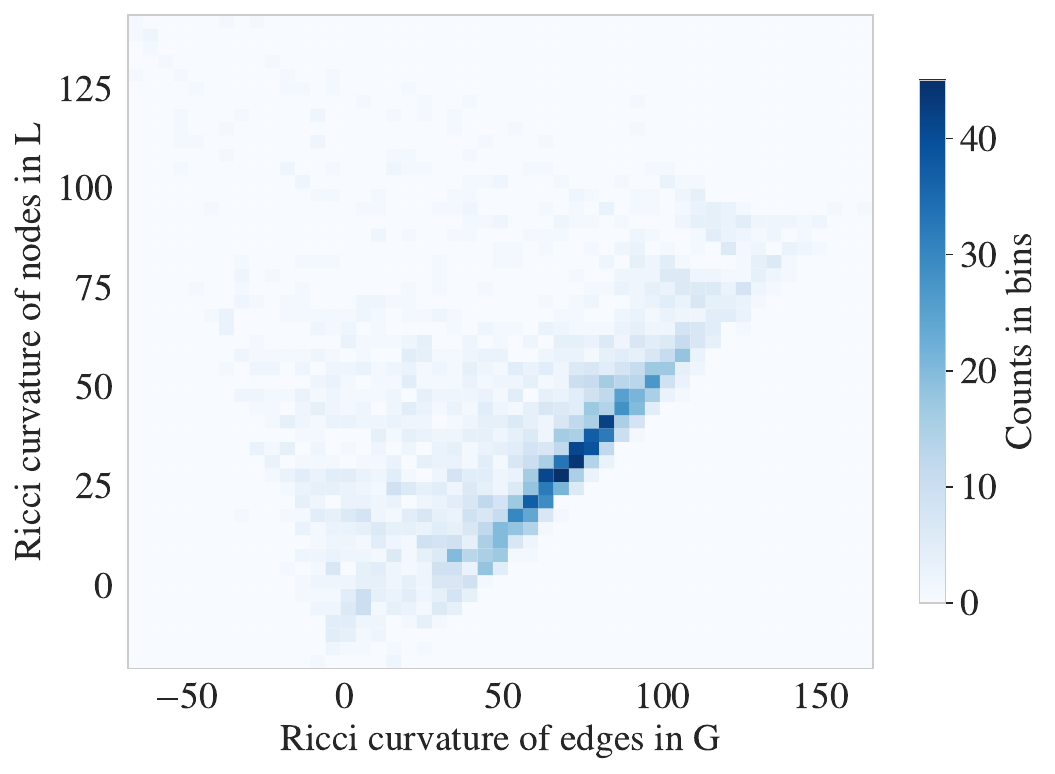}
		\end{tabular}
		\caption{Correlation of the distributions of the classical ORC (left), the approximation of ORC from the average of the upper and lower bounds (middle left), the classical FRC (middle right), and the augmented FRC (right) \textbf{for edges in $G$ versus that for vertices in $L$} where the networks are obtained from (top) the planted SBM of size $n=100$, two equally-sized blocks, $p_{out}=0.02$ and $p_{in} = 0.4$; (middle) the 2-d RGG of size $n=100$, and radius $r = 0.3$; (bottom) the ``Worm" data set. Notice that the range of curvature values varies between curvature notions, but the shape of the distributions has a close resemblence.
		}
		\label{fig:line-curv_e-v}
	\end{figure}
	\begin{figure}[ht]
		\centering
		\hspace*{-1.5em}
		\begin{tabular}{cccc}
			\includegraphics[width=.24\textwidth]{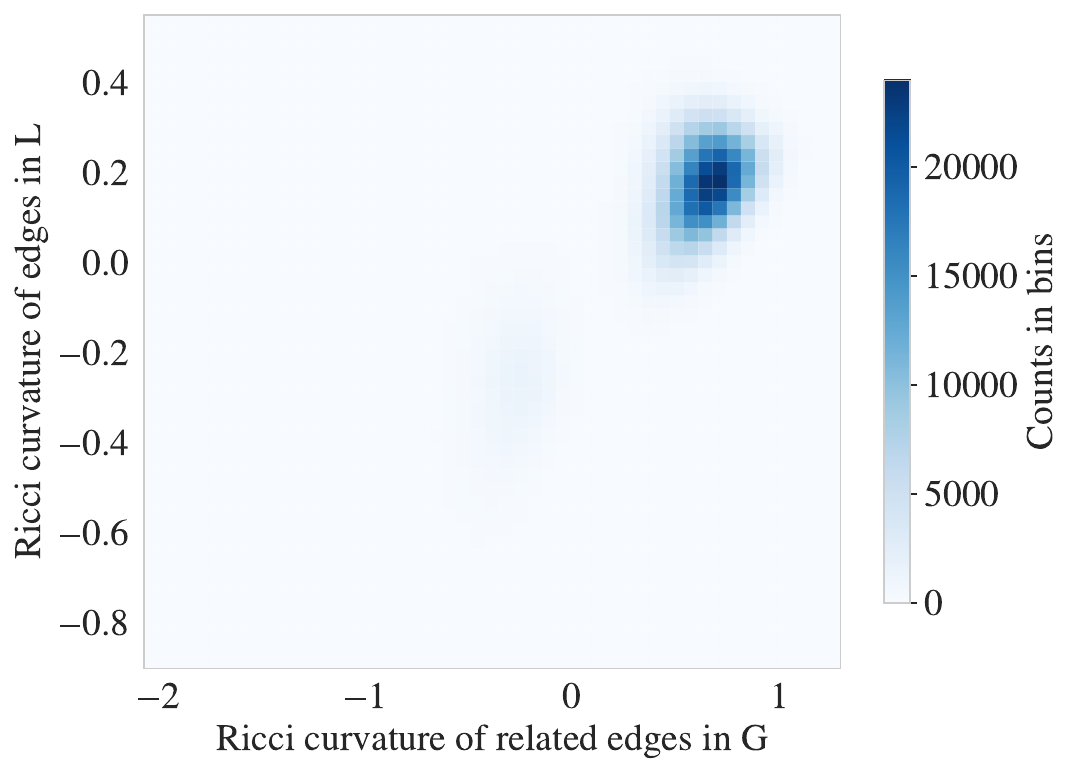} & \includegraphics[width=.24\textwidth]{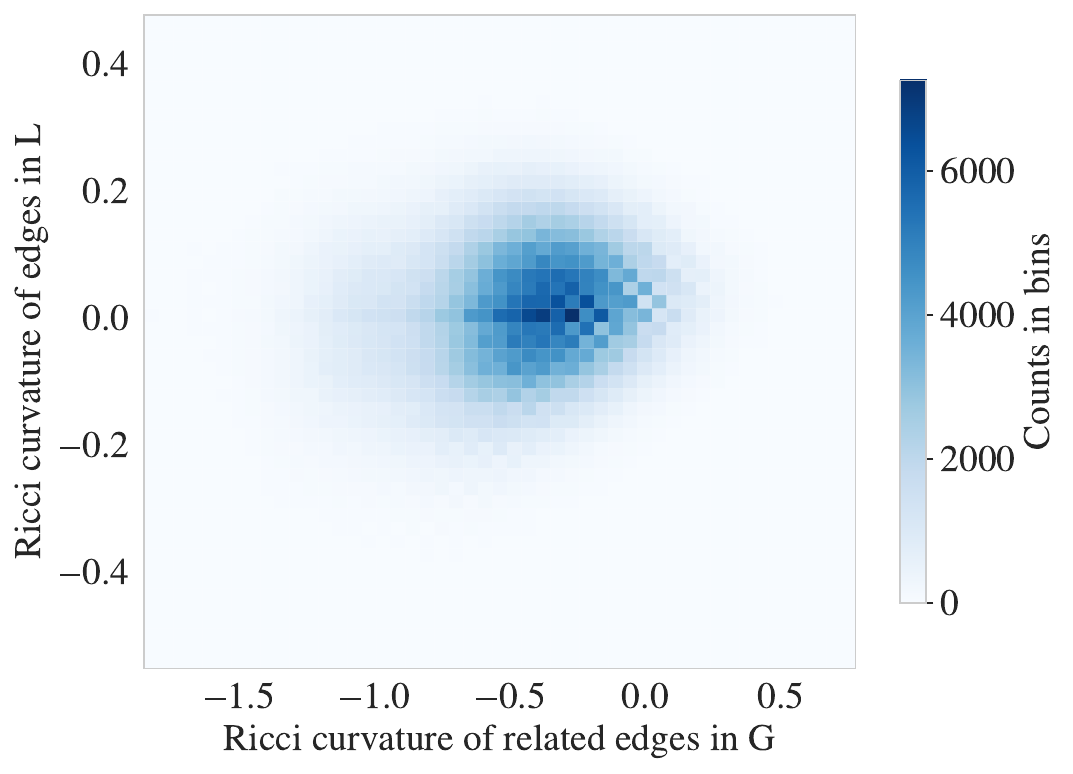} & \includegraphics[width=.24\textwidth]{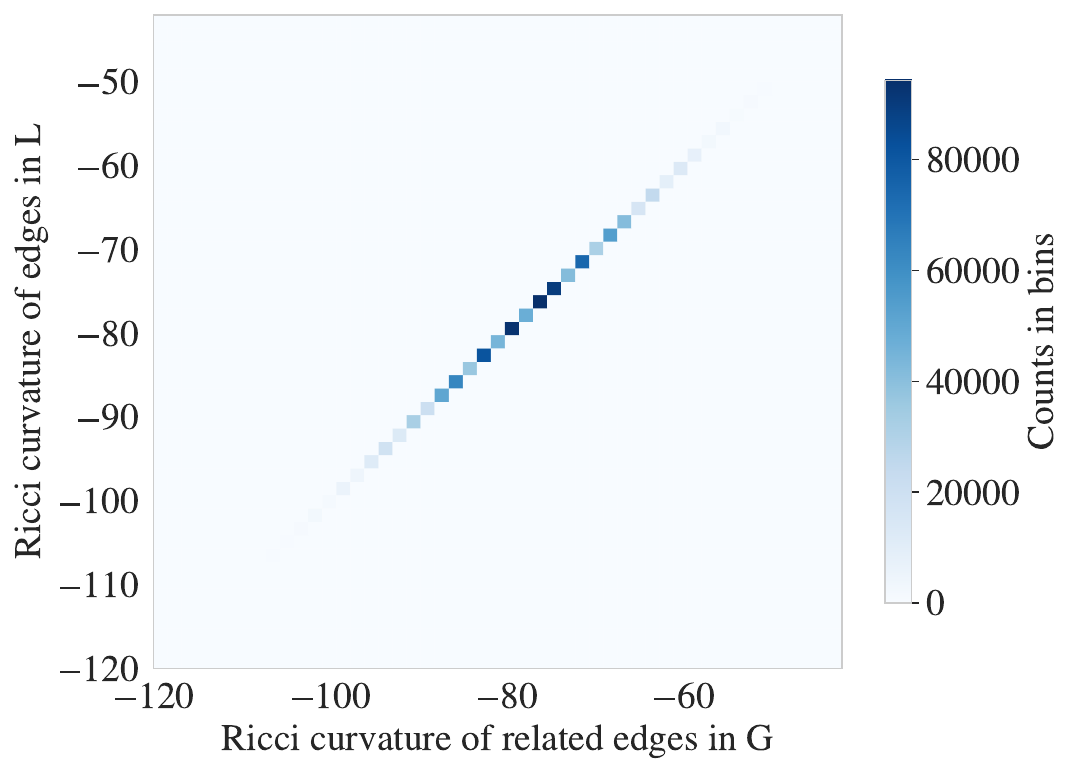} & \includegraphics[width=.24\textwidth]{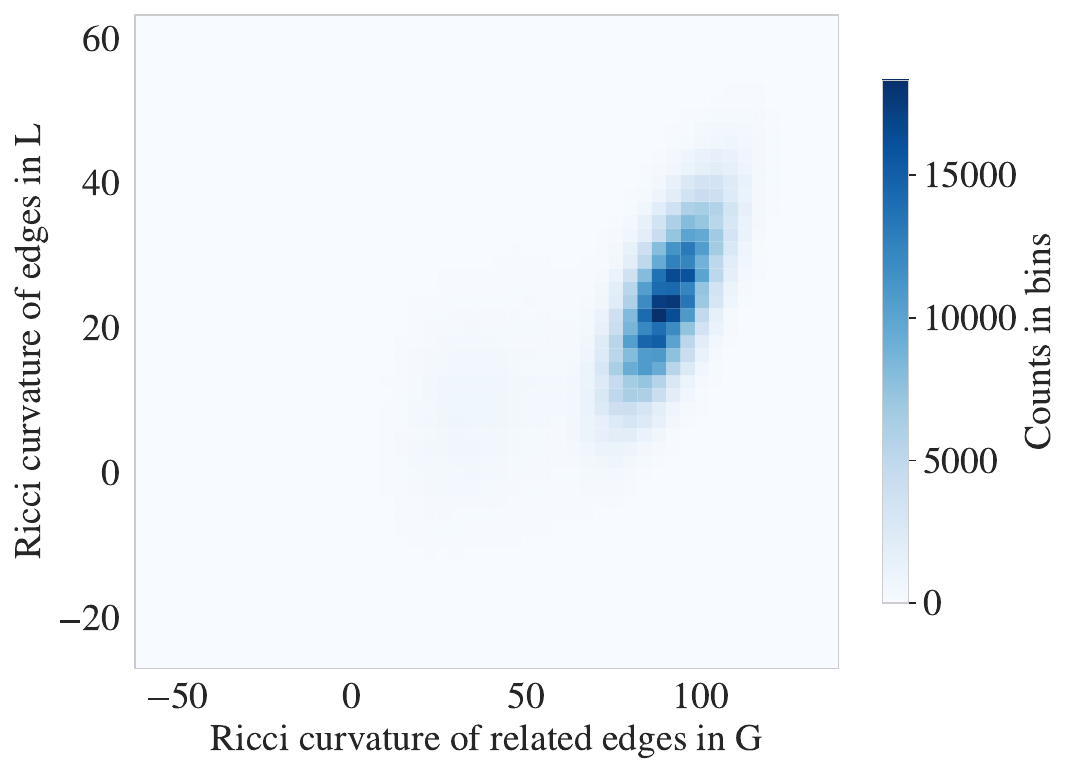}\\
			\includegraphics[width=.24\textwidth]{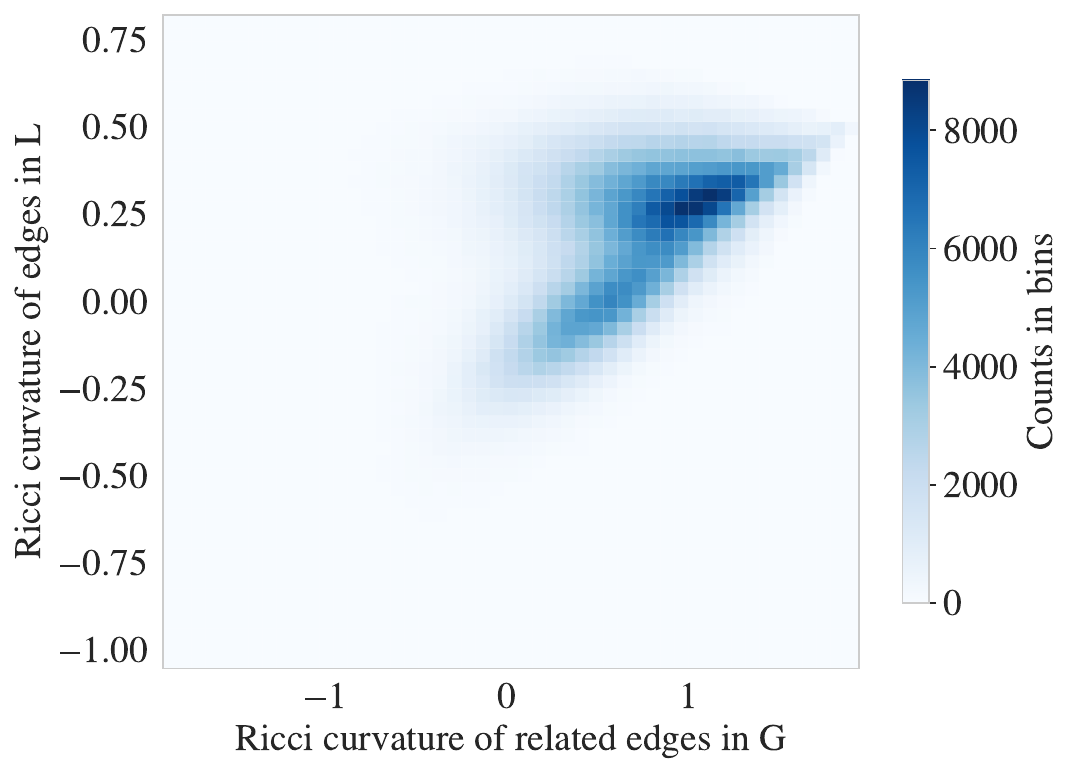} & \includegraphics[width=.24\textwidth]{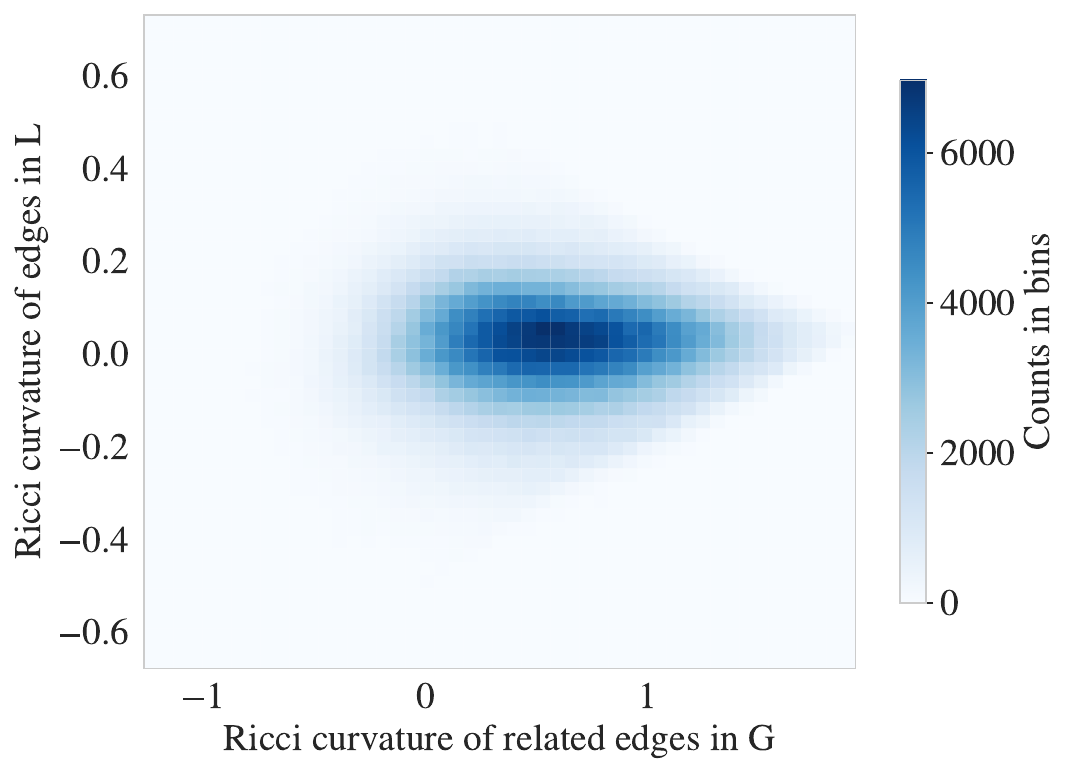} & \includegraphics[width=.24\textwidth]{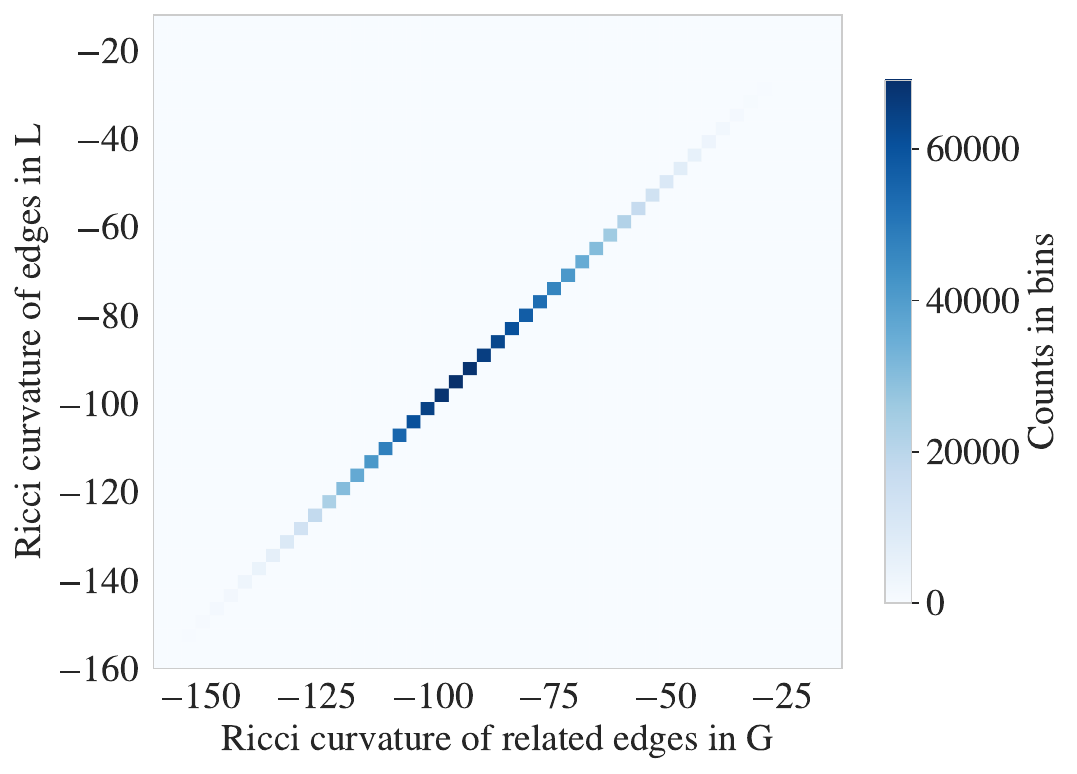} & \includegraphics[width=.24\textwidth]{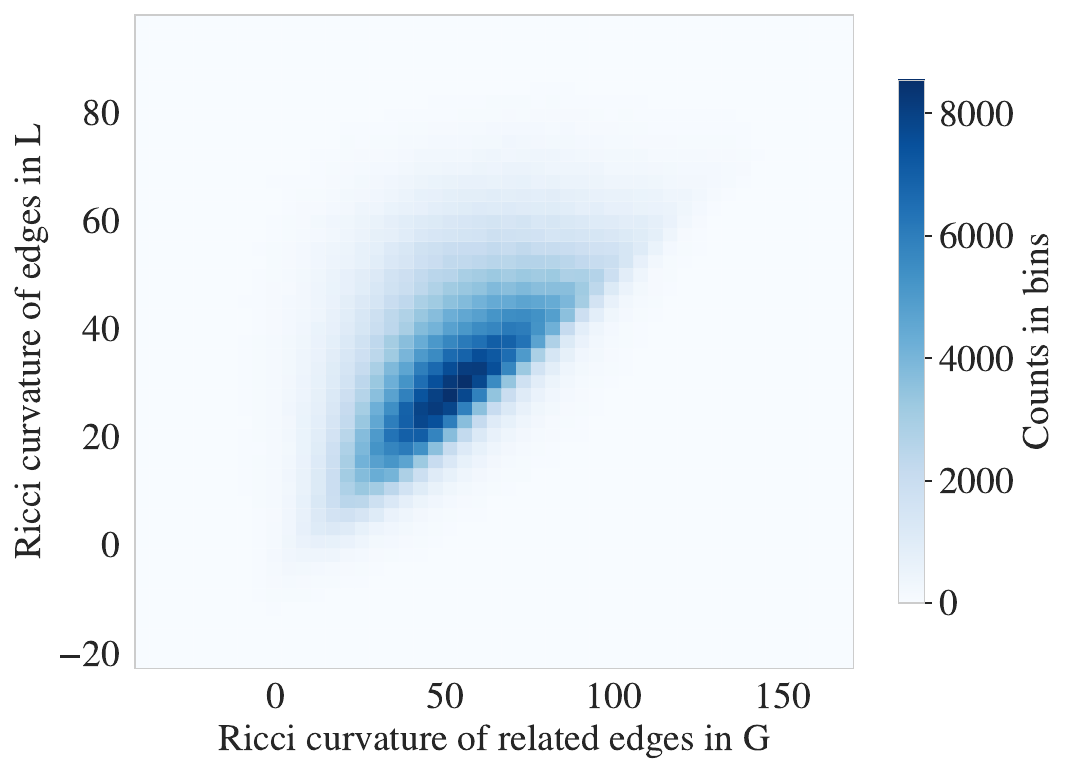}\\
			\includegraphics[width=.24\textwidth]{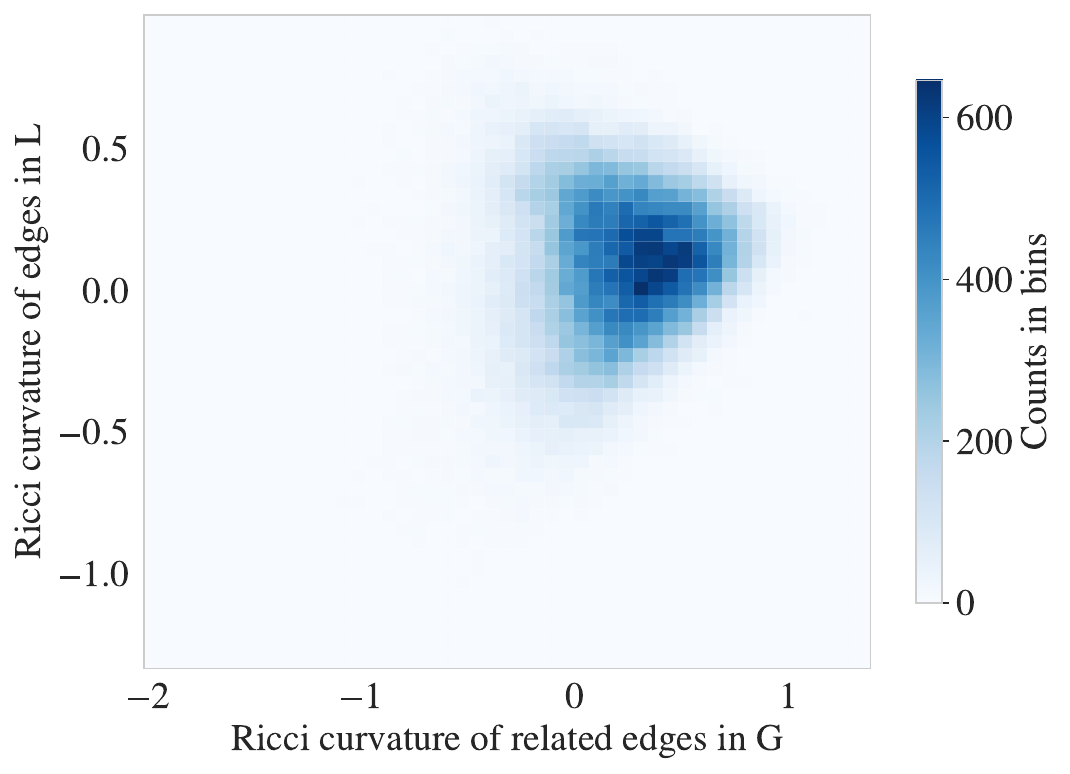} & \includegraphics[width=.24\textwidth]{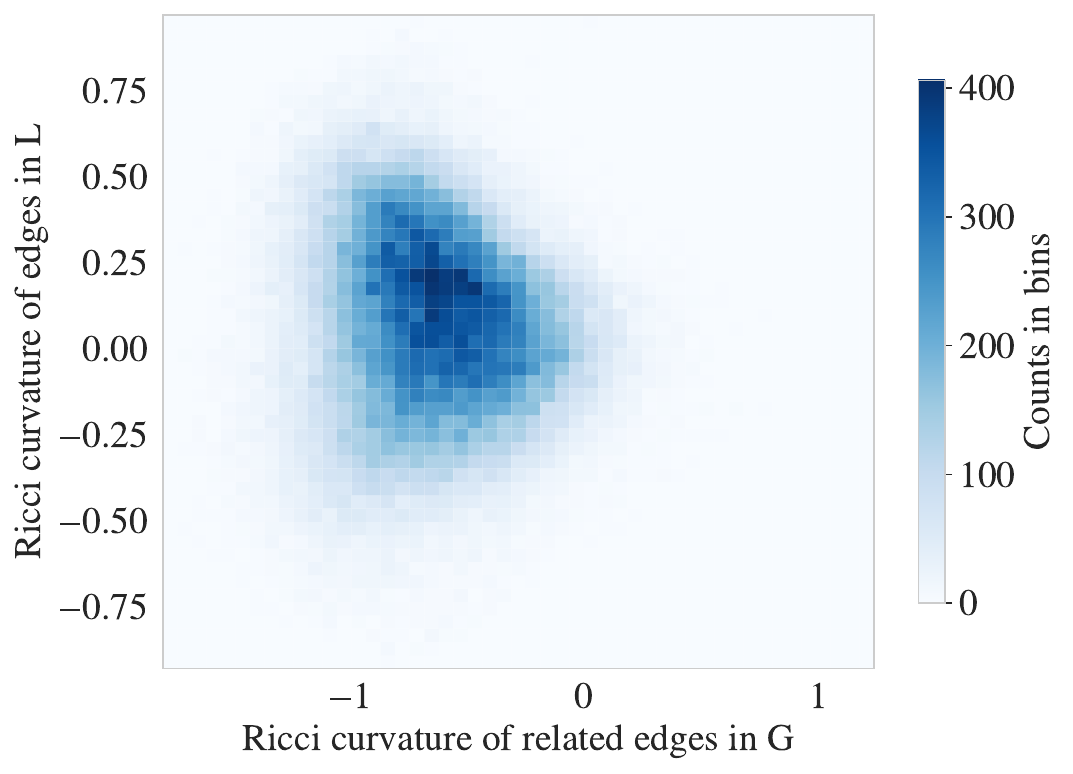} & \includegraphics[width=.24\textwidth]{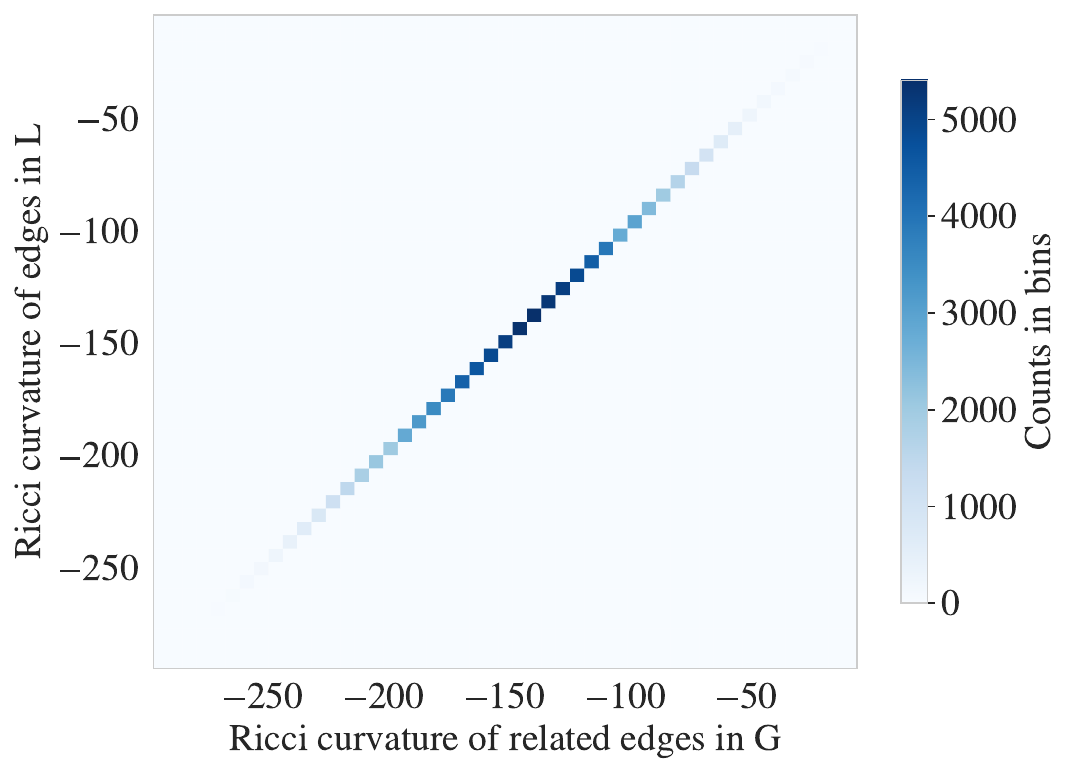} & \includegraphics[width=.24\textwidth]{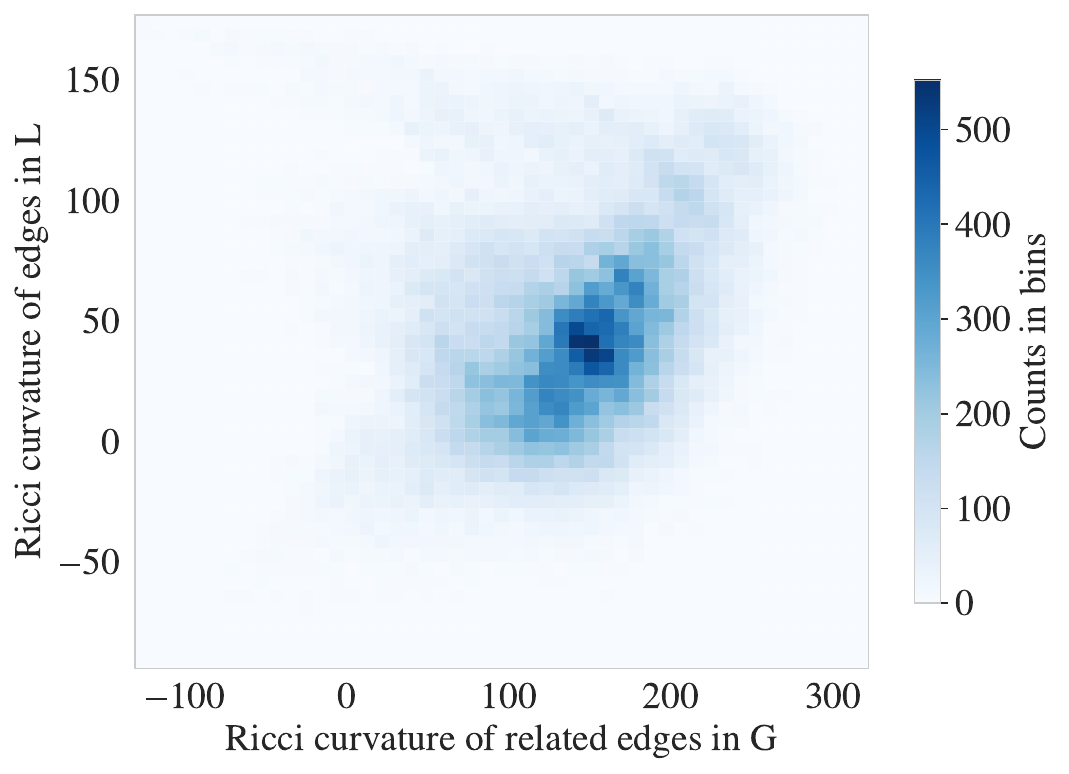}
		\end{tabular}
		\caption{Correlation of the distributions of the classical ORC (left), the approximation of ORC from the average of the upper and lower bounds (middle left), the classical FRC (middle right), and the augmented FRC (right) \textbf{for edges in $G$ versus that for edges in $L$} where the networks are obtained from (top) the planted SBM of size $n=100$, two equally-sized blocks $p_{out}=0.02$ and $p_{in} = 0.4$; (middle) the 2-d RGG of size $n=100$, and radius $r = 0.3$; (bottom) the ``Worm" data set. Notice that the range of curvature values varies between curvature notions, but the shape of the distributions has a close resemblance.
		}
		\label{fig:line-curv_e-e}
	\end{figure}

	\subsection{Ollivier's Ricci Curvature}
	
	In this section, we establish theoretical evidence for the observed relationship between Ollivier's curvature (ORC) of a graph and its line graph.
	
	\subsubsection{Unweighted case}
	
	Consider an unweighted graph $G$, and denote ORC in $G$ and $L(G)$ by $\Ric_O^G$, $\Ric_O^{L(G)}$, respectively.
	We prove combinatorial upper and lower bounds for ORC in $L(G)$, which can be computed solely from structural information in the underlying graph $G$. 
	These results are the analogues of Equations~\eqref{eq:ollilow}, \eqref{eq:olliup}, but make use of the relationships between the line graph $L(G)$ and its base graph to avoid computing curvatures in the line graph.

	\begin{lem}[Bounds on line graph ORC]
		Let $d_u$ denote the degree of the vertex $u\in V$ in the graph $G$, and $d_{\{u,v\}}$ the degree of the vertex $\{u,v\}\in E$ in the line graph $L(G)$. Then we have the following bounds on the ORC of edges in $L(G)$:
		\begin{enumerate}
			\item Upper bound on ORC:
			\begin{equation}
				\Ric_O^{L(G)}(\{\{u,v\},\{v,w\}\})\leq \frac{d_v-2+\mathbf{1}_E(\{u,w\})}{(d_u \vee d_w)+d_v-2} \; ;
			\end{equation}
			\item Lower bound on ORC:
			\begin{align}
				\Ric_O^{L(G)}(\{\{u,v\},\{v,w\}\})&\geq 
				-\left(1-\frac{1}{d_u+d_v-2}-\frac{1}{d_v+d_w-2}-\frac{d_v-2+\mathbf{1}_E(\{u,w\})}{(d_u \wedge d_w)+ d_v -2}\right)_+ \nonumber \\ 
				&-\left(1-\frac{1}{d_u+d_v-2}-\frac{1}{d_v+d_w-2}-\frac{d_v-2+\mathbf{1}_E(\{u,w\})}{(d_u \vee d_w)+d_v-2}\right)_+ \nonumber \\
				&+\frac{d_v -2+\mathbf{1}_E(\{u,w\})}{(d_u \vee d_w)+d_v-2} \; . 
			\end{align}
		\end{enumerate}
		Here, $\mathbf{1}_E(\{u,w\})=1$, if $u,w$ are connected by an edge in $G$, and  $\mathbf{1}_E(\{u,w\})=0$ otherwise.
	\end{lem}
	\begin{proof}  Let $\Delta(u,v)$ denote the number of triangles in $G$ which include the edge $\{u,v\}$, and $\Delta(\{u,v\},\{v,w\})$ the number of triangles in $L(G)$, which include $\{\{u,v\},\{v,w\}\}$. Both bounds follow from a simple combinatorial computation: For the upper bound, we have
		\begin{align*}
			\Ric_O^{L(G)}(\{\{u,v\},\{v,w\}\})&\leq \frac{\Delta(\{u,v\},\{v,w\})}{d_{\{u,v\}} \vee d_{\{v,w\}}}
			= \frac{d_v-2+\mathbf{1}_E(\{u,w\})}{(d_u \vee d_w)+d_v-2} \; .
		\end{align*}
		The lower bound follows from
		\begin{align*}
			\Ric_O^{L(G)}(\{\{u,v\},\{v,w\}\})&\geq -\left(1-\frac{1}{d_{\{u,v\}}}-\frac{1}{d_{\{v,w\}}}-\frac{\Delta(\{u,v\},\{v,w\})}{d_{\{u,v\}} \wedge d_{\{v,w\}}}\right)_+\\
			&\quad -\left(1-\frac{1}{d_{\{u,v\}}}-\frac{1}{d_{\{v,w\}}}-\frac{\Delta(\{u,v\},\{v,w\})}{d_{\{u,v\}} \vee d_{\{v,w\}}}\right)_+ + \frac{\Delta(\{u,v\},\{v,w\})}{d_{\{u,v\}} \vee d_{\{v,w\}}}\\
			&=-\left(1-\frac{1}{d_u+d_v-2}-\frac{1}{d_v+d_w-2}-\frac{d_v-2+\mathbf{1}_E(\{u,w\})}{(d_u \wedge d_w)+d_v-2}\right)_+\\
			&\quad -\left(1-\frac{1}{d_u + d_v -2}-\frac{1}{d_v+d_w-2}-\frac{d_v-2+\mathbf{1}_E(\{u,w\})}{(d_u \vee d_w)+d_v-2}\right)_+\\
			&\quad +\frac{d_v-2+\mathbf{1}_E(\{u,w\})}{(d_u \vee d_w)+d_v-2} \; .
		\end{align*}
	\end{proof}
	A crucial feature of this lemma is that it does not require the computation of triangle counts in $L(G)$: It only requires evaluating the degrees of the constituent vertices, and the existence of the edge $\{u,w\}\in E$. The latter results in an additional triangle in the base graph $G$ and, since $L(K_3)$ is isomorphic to $K_3$, an additional triangle in the line graph.
	
	\subsubsection{Weighted case}
	\label{sec:line-weighted}
	Importantly, we can show a version of these bounds in the weighted case, too. Let $\omega(\{u,v\},\{v,w\})$ denote the weight of the edge $\{\{u,v\},\{v,w\}\}\in\mathcal{E}$. Then for an edge $\{e_1,e_2\}\in\mathcal{E}$, Theorem~\ref{thm:weightollibnds} gives us
	
	\begin{enumerate}
		\item Upper bound on \orc:
		\begin{equation*}
			\Ric_O(\{e_1,e_2\})\leq 1- \max \left\{\sum_e \frac{d_{L(G)}(e,\mathcal{N})}{\omega(e_1,e_2)}(m_1(e)-m_2(e))_+, \sum_e \frac{d_{L(G)}(e,\mathcal{P})}{\omega(e_1,e_2)}(m_2(e)-m_1(e))_+ \right\}.
		\end{equation*}
		\item Lower bound on \orc:
		\begin{align*}
			\Ric_O(\{e_1,e_2\})&\geq 1- \sum_\ell \frac{\omega(\ell,e_1)}{\omega(e_1,e_2)} m_1(\ell)-\sum_r \frac{\omega(r,e_2)}{\omega(e_1,e_2)} m_2(r)\\
			&\quad-\sum_c \left[\frac{\omega(c,e_2)}{\omega(e_1,e_2)}(m_1(c)-m_2(c))_+ +\frac{\omega(c,e_1)}{\omega(e_1,e_2)}(m_2(c)-m_1(c))_+\right]\\
			&\quad- \left|L_1+X_1 - X_2 - \sum_c (m_2(c)-m_1(c))_+\right| \; .
		\end{align*}
	\end{enumerate}
	
	\paragraph{Line Graph edge weights.} 
	The classical construction of the line graph produces an unweighted graph; while node weights may be inherited from edge weights in the original graph, it is not clear how one may impose edge weights that align with the geometric information in the underlying graph. In this section, we propose one such choice of edge weights: 
	Let $\omega_{e_1,e_2}=\sqrt{\omega_{e_1}\omega_{e_2}}$ denote weights on the line graph edges, which arise from edge weights $\omega_{e_1},\omega_{e_2}$ in the original graph (node weights in the line graph). Here, $e_1,e_2\in E$, $e_1\sim e_2$, as usual. We can define a measure on node neighborhoods in the line graph as
	\begin{equation*}
		m_{\{u,v\}}^{\alpha,p}(\{x,y\})=\begin{cases}
			\alpha&\text{if }\{x,y\}=\{u,v\}\\
			\frac{1-\alpha}{C} \exp(-\omega_{\{u,w\}}^{p/2} \omega_{\{u,v\}}^{p/2})&\text{if }\{x,y\}=\{u,w\}\\
			\frac{1-\alpha}{C} \exp(-\omega_{\{v,w\}}^{p/2} \omega_{\{u,v\}}^{p/2})&\text{if }\{x,y\}=\{v,w\}\\
			0&\text{otherwise} \; .
		\end{cases}  
	\end{equation*}
	Here, the normalizing factor is given by
	\begin{align*}
		C&=\sum_{\{x,y\}\sim\{u,v\}} \exp(-\omega_{\{x,y\}}^{p/2}\omega_{\{u,v\}}^{p/2})\\
		&= \sum_{r: \{u,r\}\in E} \exp(-\omega_{\{u,r\}}^{p/2}\omega_{\{u,v\}}^{p/2}) + \sum_{r: \{v,r\}\in E} \exp(-\omega_{\{v,r\}}^{p/2}\omega_{\{u,v\}}^{p/2}) \; .
	\end{align*}

	Setting $s_{w}(e)=(m_{\{v,w\}}(e)-m_{\{u,v\}}(e))_+$, and $s_u(e)=(m_{\{u,v\}}(e)-m_{\{v,w\}}(e))_+,$ we may write the lower bound as:
	\begin{align*}
		\Ric_O(\{u,v\},\{v,w\})&\geq 1- \sum_{\substack{\ell\sim u\\ \ell\ne w}} \frac{\omega_{\{u,\ell\}}^{1/2}}{\omega_{\{v,w\}}^{1/2}} m_{\{u,v\}}(\{u,\ell\})
		- \sum_{\substack{r\sim w\\ r\ne u}}\frac{\omega_{\{w,r\}}^{1/2}}{\omega_{\{u,v\}}^{1/2}}m_{\{v,w\}}(\{w,r\})\\
		&-\sum_{e=\{v,c\},\{u,w\}} \omega_{e}^{1/2}\left(\frac{s_w(e)}{\omega_{\{v,w\}}^{1/2}}+\frac{s_u(e)}{\omega_{\{u,v\}}^{1/2}}\right)\\
		&-\left|\sum_{\substack{\ell\sim u\\ \ell\ne w}} m_{\{u,v\}}(\{u,\ell\})+\alpha- m_{\{v,w\}}(\{u,v\})-\sum_{e=\{v,c\},\{u,w\}} s_w(e)\right| \; .
	\end{align*}
	
	Note that the upper bound for ORC in the line graph still requires one to compute the line graph. However, we can always use the trivial upper bound $\Ric_O(\{e_1,e_2\})\leq 1$ if an approximation will suffice.

	\paragraph{Summary.}
	In the applications discussed in the next section, we again approximate ORC as the arithmetic mean of the upper and lower bounds for the sake of computational efficiency:
	\begin{defn}[Approximate ORC in the Line Graph]
		\begin{equation}\label{eq:orc-approx-line-graph}
			\widehat{\Ric_O}^{L(G)}(\{u,v\}) := \frac{1}{2} \Big( \big(\Ric_O^{L(G)} \big)^{up}(\{u,v\}) + \big(\Ric_O^{L(G)} \big)^{low}(\{u,v\}) \Big) \; .
		\end{equation}
	\end{defn}
	
	In Figure~\ref{fig:ORC-L-approx}, we compare this approximation to the original ORC for some simulated graphs.
	
	\begin{figure}
		\centering
		\hspace*{-2em}
		\begin{tabular}{cccc}
			\includegraphics[width=.24\textwidth]{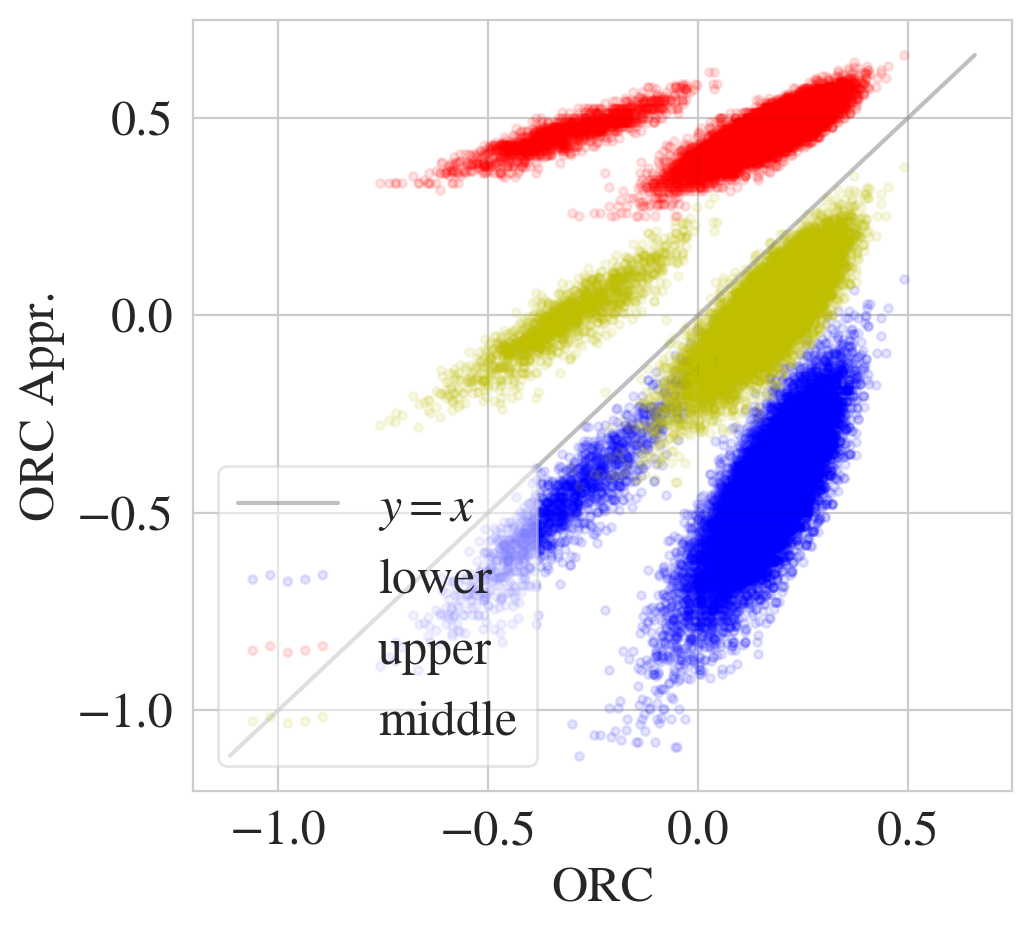} & \includegraphics[width=.24\textwidth]{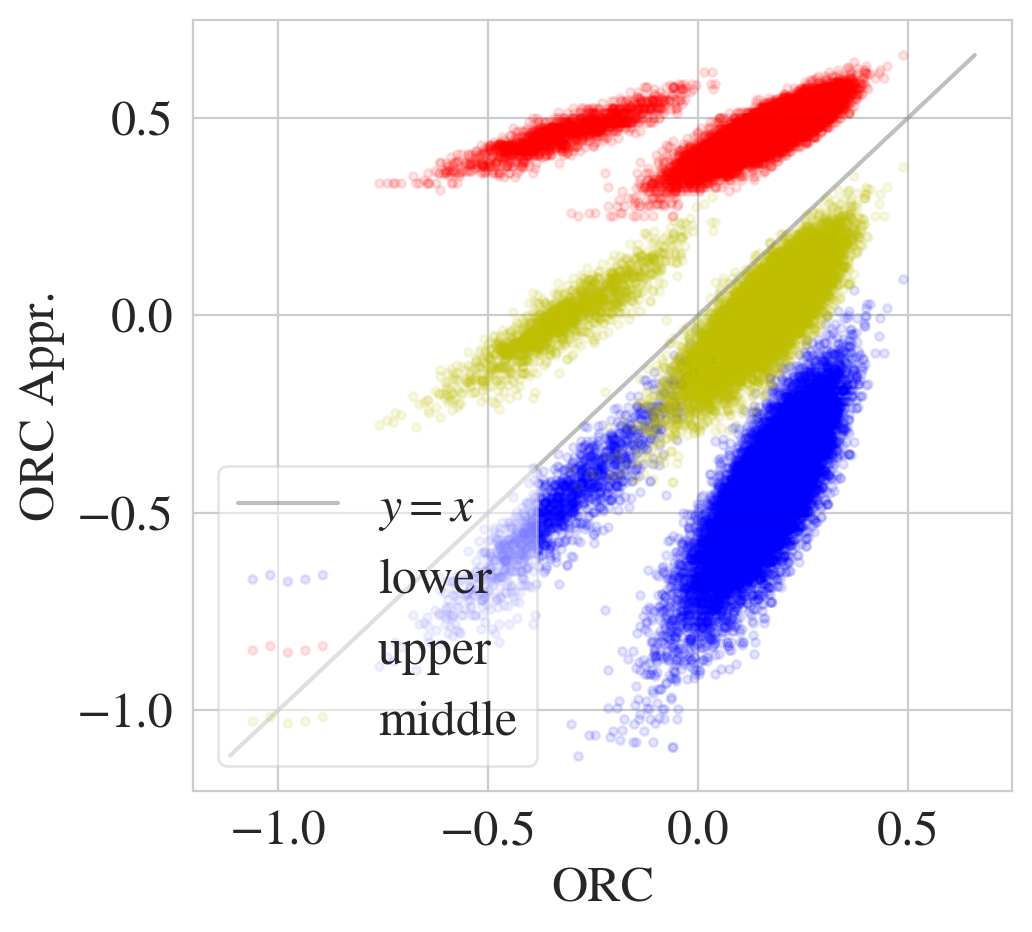} & \includegraphics[width=.24\textwidth]{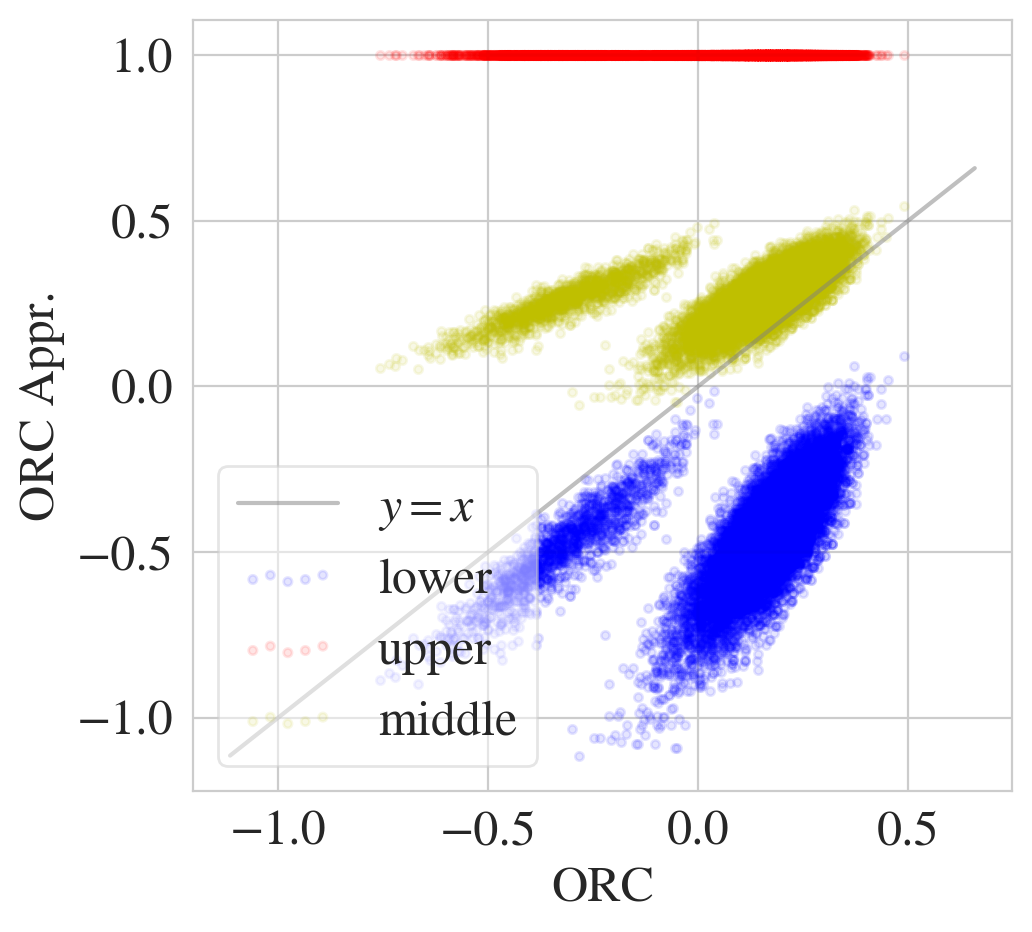} & \includegraphics[width=.24\textwidth]{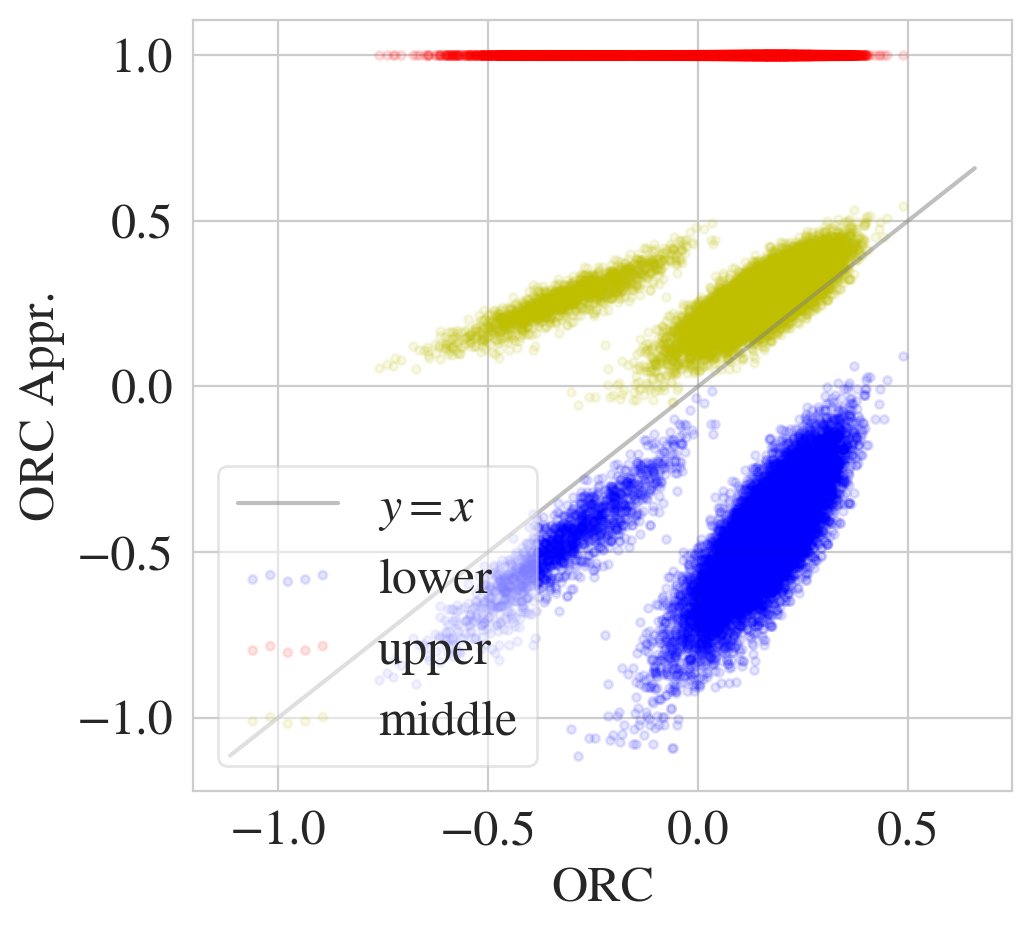}\\
			\includegraphics[width=.24\textwidth]{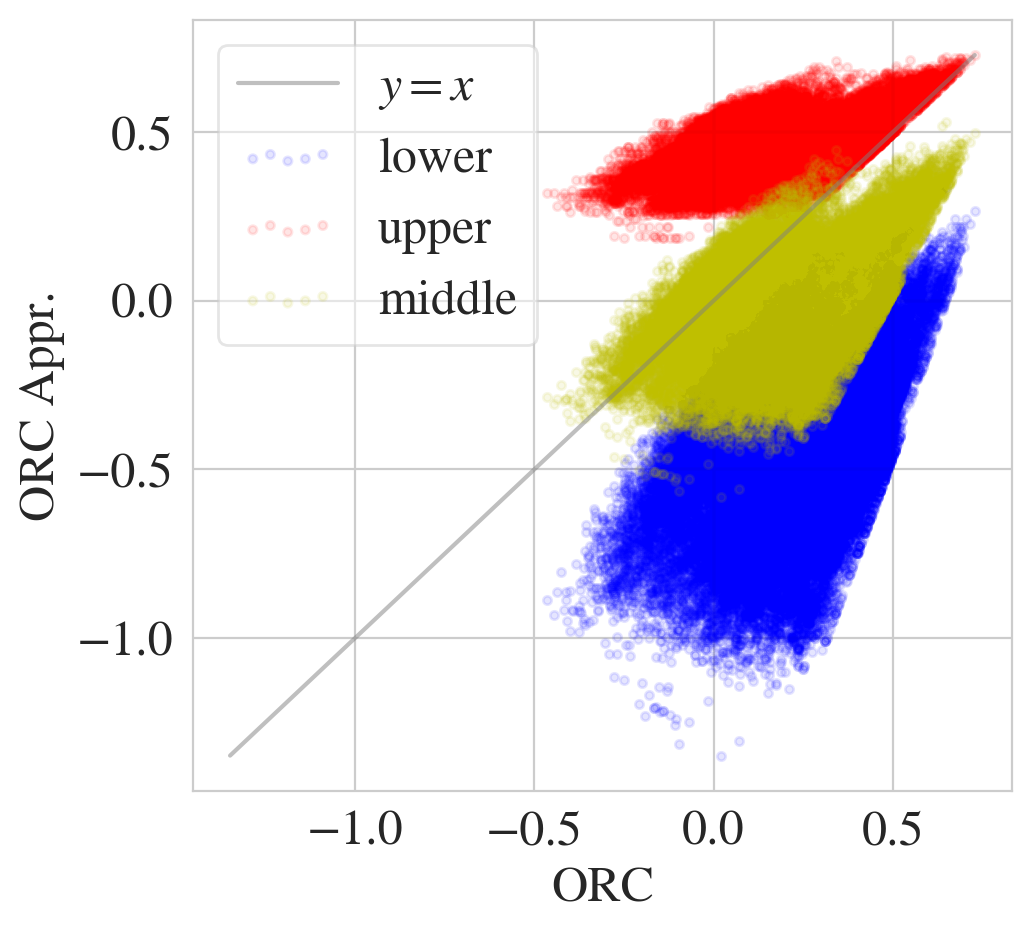} & \includegraphics[width=.24\textwidth]{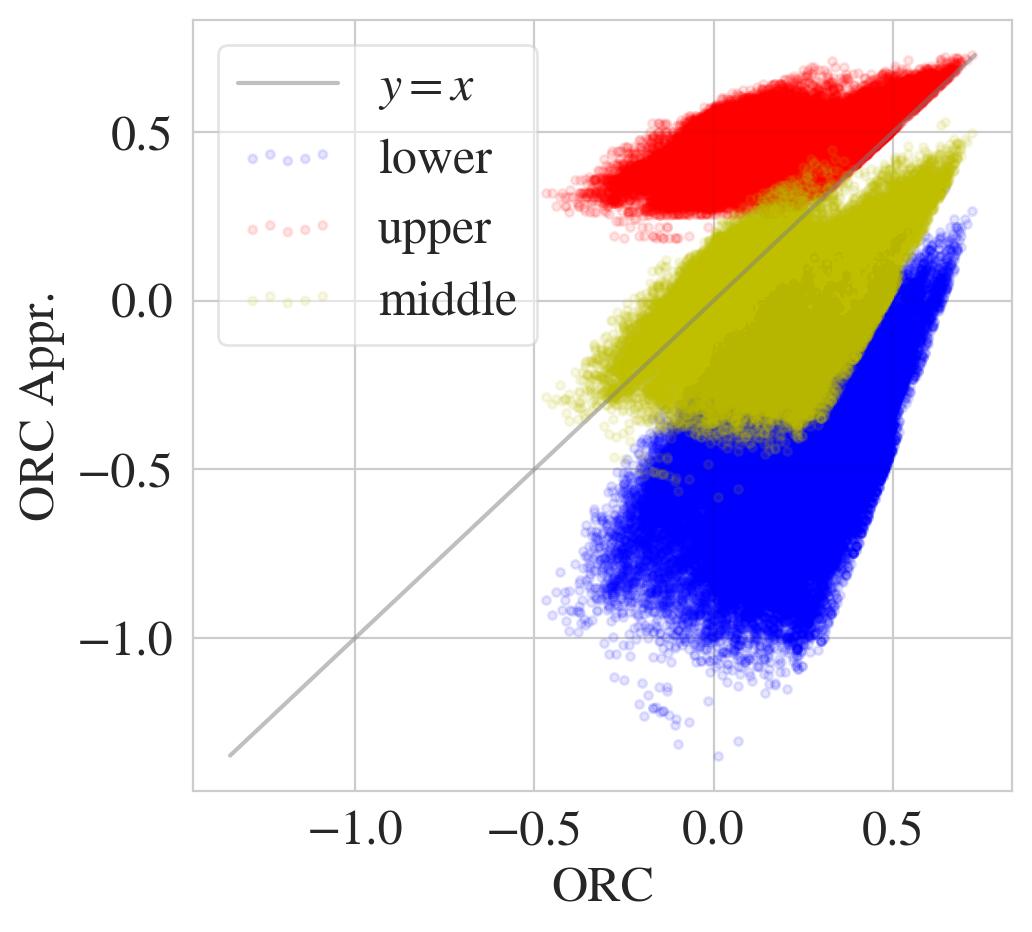} & \includegraphics[width=.24\textwidth]{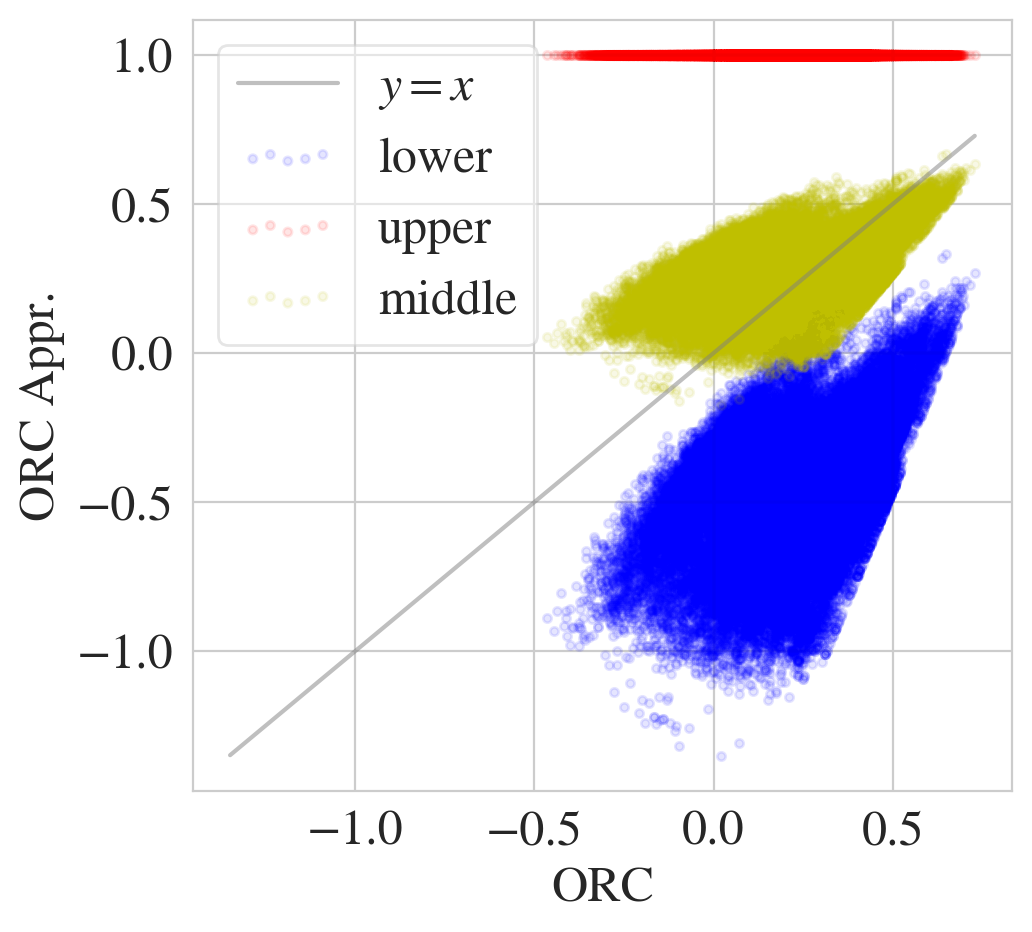} & \includegraphics[width=.24\textwidth]{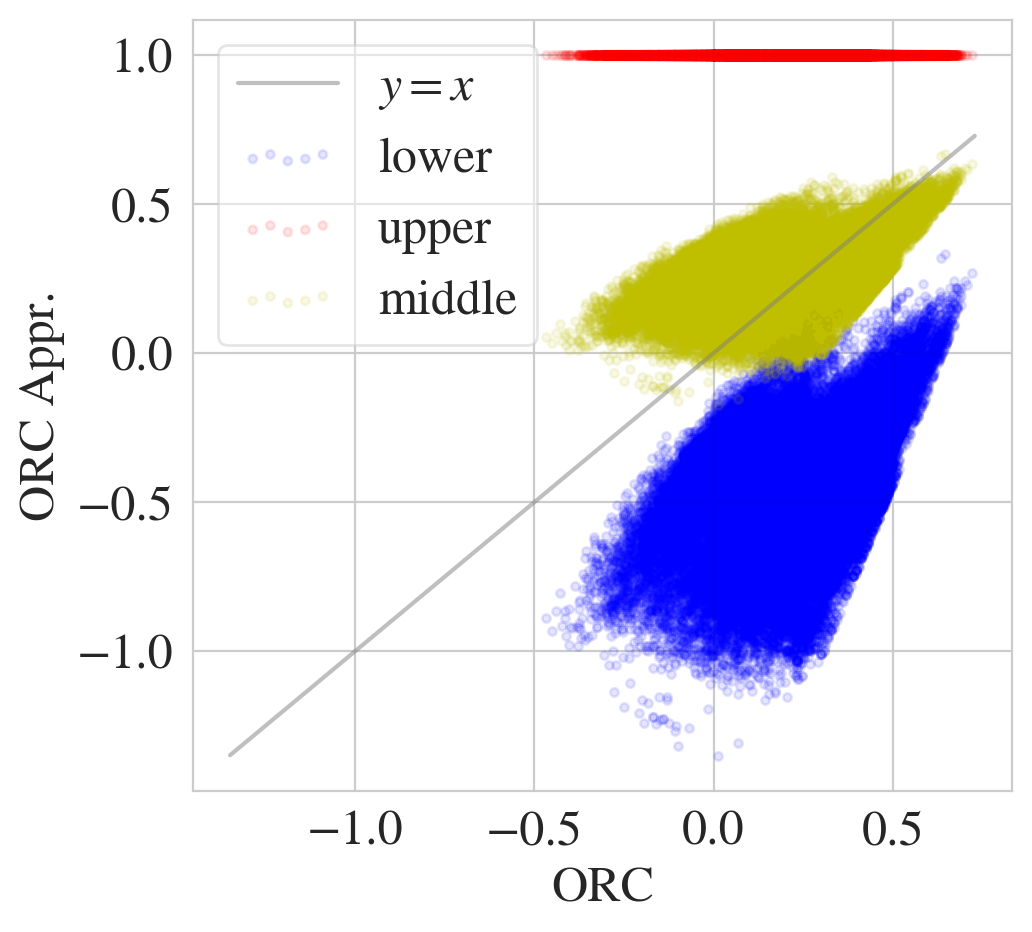}
		\end{tabular}
		\caption{Scatter plot of the ORC on $L$ obtained from optimization versus its approximation with $L$ constructed (left two columns) and without (right two columns), where the optimization is done by solving the earth mover's distance exactly (left most and middle right) and the approximately with Sinkhorn (middle left and right most), and the original graph $G$ is obtained from the planted SBM of size $n=100$, two equally-sized blocks, $p_{out}=0.02$ and $p_{in}=0.4$ (top) and the $2$-d RGG of size $n=100$ and radius $r = 0.4$ (bottom).} 
		\label{fig:ORC-L-approx}
	\end{figure}

	\subsection{Forman's Ricci Curvature}
	
	In this section, we establish a series of relationships between Forman's curvature of a graph and its line graph. We further discuss how higher order structures in the line graph may be incorporated in the curvature computation.

	\subsubsection{Unweighted case} 
	Consider an unweighted graph $G$, and denote FRC in $G$ and $L(G)$ by $\Ric_F^G$, $\Ric_F^{L(G)}$, respectively. At first, we only consider the 1-complex curvature notion (Eqs.~\eqref{eq:frc-e}, \eqref{eq:frc-v}), which does not consider contributions of higher-order structures, such as triangles or quadrangles. The argument specifies, whether node- or edge-level curvature is considered. 
	We begin by investigating the relationship of edge-level Forman curvature in a graph $G$ and its line graph $L(G)$, which is given by the following result:
	\begin{lem}
		Edge-level FRC in $G$ relates to edge-level FRC in $L$ as
		\begin{equation}\label{eq:lgfrc_edge}
			\Ric_{F}^{L(G)}(\{\{u,v\},\{v,w\}\}) =\Ric_F^G(\{u,v\})+\Ric_F^G(\{v,w\}) \; .
		\end{equation}
	\end{lem}
	\begin{proof}
		The relationship arises from simple combinatorial arguments.
		Recalling Lem.~\ref{lem:degree}, we have
		\begin{align*}
			\Ric_{F}^{L(G)}(\{\{u,v\},\{v,w\}\})&=4-d_{\{u,v\}}-d_{\{v,w\}} \notag\\
			&=4-(d_u - d_v -2)-(d_v+ d_w -2)\notag\\
			&=8- d_u - d_v - d_v - d_w \notag\\
			&=\Ric_F^G(\{u,v\})+\Ric_F^G(\{v,w\}),
		\end{align*}
		which gives the claim.
	\end{proof}
	
	\begin{lem}
		Edge-level FRC in $G$ relates to node-level FRC in $L$ as
		\begin{equation}\label{eq:lgfrc_vtx}
			\Ric_F^{L(G)}(\{u,v\}) =\Ric_F^G(u)+\Ric_F^G(v)-\Ric_F^G(\{u,v\})^2 \; .
		\end{equation}
	\end{lem}
	Importantly, this lemma allows us to compute node-level FRC in the line graph with respect to quantities in the original graph only. Specifically, the right hand side can be computed from edge-level FRC only, utilizing Eq.~\eqref{eq:frc-v}.
	%
	\begin{proof}
		Similarly,
		\begin{align*}    
			\Ric_F^{L(G)}(\{u,v\})&=\sum_{\{\{u,v\},\{v,w\}\}\in \mathcal{E}}\Ric_F^{L(G)}(\{\{u,v\},\{v,w\}\})  +\sum_{\{\{u,v\},\{u,w\}\}\in \mathcal{E}}\Ric_F^{L(G)}(\{\{u,v\},\{u,w\}\})\notag\\
			&=\sum_{\substack{\{v,w\}\in E\\ w\neq u}} \Ric_F^G(\{u,v\})+\Ric_F^G(\{v,w\}) + \sum_{\substack{\{u,w\}\in E \\ w\neq v}}\Ric_F^G(\{u,v\})+\Ric_F^G(\{u,w\})\notag\\
			&=d_{\{u,v\}} \Ric_F^G(\{u,v\})+\sum_{\substack{\{v,w\}\in E \\ w\neq u}}\Ric_F^G(\{v,w\})+\sum_{\substack{\{u,w\}\in E \\ w\neq v}}\Ric_F^G(\{u,w\})\notag\\
			&=(d_u + d_v -2)\Ric_F^G(\{u,v\})+\Ric_F^G(v)-\Ric_F^G(\{u,v\})+\Ric_F^G(u)-\Ric_F^G(\{u,v\})\notag\\
			&=\Ric_F^G(u)+\Ric_F^G(v)-\Ric_F^G(\{u,v\})^2.
		\end{align*}
	\end{proof}
	
	\subsubsection{Weighted case} 
	
	When the vertex weights $\omega_u=1$ for all $u\in V$, and using the line graph edge weights $\omega_{\{e_1,e_2\}}=\omega_{e_1}\omega_{e_2},$ we get the following formula for the relationship between the edge-level FRC in $G$ versus the edge-level FRC in $L(G)$:
	
	\begin{lem}
		\label{lem:lgfrc_weighted}
		\begin{align*}
			\Ric_F^{L(G)}(\{\{u,v\},\{v,w\}\})&= \omega_{\{u,v\}}\left(2- \sqrt{\frac{\omega_{\{v,w\}}}{\omega_{\{u,v\}}}}(2-\Ric_F^G(\{u,v\}))\right)\\
			&+\omega_{\{v,w\}}\left(2- \sqrt{\frac{\omega_{\{u,v\}}}{\omega_{\{v,w\}}}}(2-\Ric_F^G(\{v,w\}))\right).
		\end{align*}
	\end{lem}

	\subsubsection{Incorporating higher-order structures} \label{sec:DC-frc-higher}
	Finally, we discuss how to account for contributions of higher-order structures, such as triangles or quadrangles in the computation of curvature in the line graph. As discussed in sec.~\ref{sec:graph-curv}, the curvature contribution of such structures is crucial in curvature-based clustering. We have demonstrated that the successful application of the FRC-based approach to single-membership community detection relies on the consideration of triangular faces, such as in our 2-complex FRC notion (Eq.~\eqref{eq:frc-f}). In the following, we will illustrate the necessity of including contributions of \emph{quadrangular faces} in the curvature contribution on the line graph. 
	We first define a 2-complex FRC notion on the line graph, which accounts for curvature contributions of triangular faces.
	
	Let $\Gamma(v) = \{v': \{v,v'\}\in E\}$ denote the set of neighbors for a vertex $v \in V$ in $G$.
	Suppose $\{u,v\}, \{v,w\}\in E$ denote neighboring edges in $G$ and $(\{u,v\}, \{v,w\})\in \mathcal{E}$ the corresponding edge in the line graph.
	We define the set of neighbours for $\{u,v\}, \{v,w\}$ as follows:
	\begin{align*}
		\Gamma \big(\{u,v\}\big) &:= \Big\{\{u,u'\}: u'\in \Gamma (u)\backslash\{v\}\Big\} \cup \Big\{\{v,v'\}: v'\in\Gamma (v)\backslash\{u\}\Big\}\\
		\Gamma \big(\{v,w\}\big) &:= \Big\{\{v,v'\}: v'\in \Gamma (v)\backslash\{w\}\Big\} \cup \Big\{\{w,w'\}: w'\in\Gamma (w)\backslash\{v\}\Big\} \; .
	\end{align*}
	The set of triangles in the line graph, which contain $\{u,v\}, \{v,w\}$ is then given by
	\begin{align*}
		\mathcal{F}_{\Delta} &:= \Big\{\big(\{u,v\}, \{v,w\}, e\big): \; e\in \Gamma (\{u,v\})\cap\Gamma (\{v,w\})\Big\}\\
		&= \Big\{\big(\{u,v\}, \{v,w\}, e\big): \; e\in \big\{\{v,v'\}: v'\in\Gamma (v)\backslash\{u,w\}\big\} \cup \big(\{\{u,w\}\}\cap E \big)\Big\} \; . 
	\end{align*}
	Hence, there are only two possible structures in $G$, which give rise to triangles in its line graph $L$: (i) edges incident to the common node of the two edges under consideration, or (ii) the other endpoints of the two edges are connected with each other. The 2-complex FRC with triangle contributions can then be written as follows:
	\begin{defn}[2-complex FRC in the Line Graph ($k \leq 3$)]
		\begin{align*}
			&\Ric_{F} \big(\{\{u,v\}, \{v,w\}\}\big) \\
			&\qquad =\ \w_{\{u,v\}, \{v,w\}} \Biggl( \sum_{f\in \mathcal{F}_{\Delta}}\frac{\w_{\{u,v\}, \{v,w\}}}{\w_f}  + \frac{\w_{\{u,v\}}}{\w_{\{u,v\}, \{v,w\}}} + \frac{\w_{\{v,w\}}}{\w_{\{u,v\}, \{v,w\}}} \Big.\\
			&\qquad \qquad \left.- \Biggl[
			\sum_{\substack{e\in \Gamma \big(\{u,v\} \big)\\ \big(\{u,v\}, \{v,w\}, e \big) \notin \mathcal{F}_{\Delta}}}\frac{\w_{\{u,v\}}}{\sqrt{\w_{\{u,v\}, \{v,w\}} \cdot \w_{\{u,v\}, e}}} \right.  + \sum_{\substack{e\in \Gamma \big(\{v,w\} \big)\\ \big(\{u,v\}, \{v,w\}, e \big) \notin \mathcal{F}_{\Delta}}}\frac{\w_{\{v,w\}}}{\sqrt{\w_{\{u,v\}, \{v,w\}}\cdot \w_{\{v,w\}, e}}}
			\Biggr] \Biggr) \; .
		\end{align*}
	\end{defn}
	\noindent In the special case of an unweighted network (i.e., $\w_e = 1$ for all $e\in V$ and $\w_{\{e,e'\}} = 1$ for all $\{e,e'\}\in\mathcal{E}$), we have
	\begin{align}
		\Ric_{F} \Big(\{\{u,v\}, \{v,w\}\}\Big) 
		&= 2 + \sum_{f\in \mathcal{F}_{\Delta}}\frac{1}{\w_f} - 
		\left\vert \Big\{e\in \Gamma \big(\{u,v\} \big): \big(\{u,v\}, \{v,w\}, e\big) \notin \mathcal{F}_{\Delta} \Big\} \right\vert \nonumber \\
		&\qquad - \left\vert \Big\{e\in\Gamma (\{v,w\}): \{\{u,v\}, \{v,w\}, e\}\notin \mathcal{F}_{\Delta} \Big\} \right\vert \; .
	\end{align}
	We can specialize this notion to the case of unweighted graphs. Consider first the case where, similar to sec.~\ref{sec:graph-curv}, we include only higher-order contributions from triangles. 
	In an unweighted graph, we may assign each triangle the same weight $\w_{\Delta}$, since we typically assume that face weights are a function of edge weights.
	This simplifies the above equation as follows:
	\begin{lem}[2-complex FRC in the Line Graph (triangles only)]
		\begin{equation}
			\Ric_{F} (\{\{u,v\}, \{v,w\})
			= \frac{d_v}{\w_{\Delta}} - d_u - d_w + 4 - \frac{2}{\w_{\Delta}} + \left(\frac{1}{\w_{\Delta}} + 2\right)\mathbf{1}_E(\{u,w\})\; .
		\end{equation}
	\end{lem}
	\begin{proof}
		\begin{align*}
			\Ric_{F} &(\{\{u,v\}, \{v,w\}\}) \\
			&= \frac{\vert \mathcal{F} \vert}{\w_{\Delta}} + 2 - 
			\left\vert \Big\{e\in \Gamma \big(\{u,v\} \big): \big(\{u,v\}, \{v,w\}, e\big) \notin \mathcal{F}_{\Delta} \Big\} \right\vert \nonumber \\
			&\qquad - \left\vert \Big\{e\in\Gamma (\{v,w\}): \{\{u,v\}, \{v,w\}, e\}\notin \mathcal{F}_{\Delta} \Big\} \right\vert\\ 
			&= \frac{d_v-2 + \mathbf{1}_E(\{u,w\})}{\w_{\Delta}} + 2 - \big(d_u - 1 -\mathbf{1}_E(\{u,w\}) + d_w - 1 -\mathbf{1}_E(\{u,w\})\big)\\
			&= \frac{d_v}{\w_{\Delta}} - d_u - d_w + 4 - \frac{2}{\w_{\Delta}} + \left(\frac{1}{\w_{\Delta}} + 2\right)\mathbf{1}_E(\{u,w\})\; .
		\end{align*}
	\end{proof}
	The lemma implies that the triangles-only 2-complex FRC in the line graph can still be largely determined by the degree of the nodes in the original graph $G$; the local triangle count enters via the terms $\mathbf{1}_E(\{u,w\}$ (which is $1$ if $\{u,w\}\in E$ and $0$ otherwise). A simple calculation shows that $\w_{\Delta}=\sqrt{3}/4$, i.e., $1/\w_\Delta+2=(4\sqrt{3}+6)/3\approx 3.309$, while node degree could be as high as $O(\vert V \vert)$. With that, a triangle increases the curvature value by $1/\w_{\Delta} + 2$.
	However, our experimental results (see Table \ref{tab:mmbmfrc-nmi}) demonstrate that clustering methods based on the triangles-only 2-complex FRC can have low accuracy. While considering triangles was sufficient for FRC-based community detection on the original graph, we need to incorporate additional higher-order information to perform community detection on the line graph with high accuracy. Our experiments suggest that accounting for curvature contributions of quadrangles leads to a significant improvement in the performance of FRC-based mixed-membership community detection (see Table \ref{tab:mmbmfrc-nmi}). We found that incorporating curvature contributions of $k$-faces with $k \geq 5$ does not improve performance, while significantly increasing the computational burden. This is expected, given the decreasing frequency of such structures (see Fig.~\ref{fig:k-faces}). These design choices lead to the following 2-complex FRC notion, which we will utilize in our FRC-based clustering method below:
	The set of quandrangles in the line graph, which contains $\{u,v\}, \{v,w\}$ is given by
	\begin{align*}
		\mathcal{F}_{\square} &:= \Big\{\big(\{u,v\}, \{v,w\}, \{w,x\}, \{x,u\}\big): \; x\in (\Gamma (w)\cap\Gamma (u))\backslash\{v\}\Big\}\; . 
	\end{align*}
	Hence, only quadrangles in the original graph $G$ can give rise to quadrangles in its line graph $L$.
	
	\begin{defn}[2-complex FRC in the Line Graph ($k\leq 4$)]\label{def:frc-f4}
		\begin{align*}
			&\Ric_{F} \Big(\{\{u,v\}, \{v,w\}\}\Big) \\
			&\qquad =\ \w_{\{u,v\}, \{v,w\}} \Biggl( \sum_{f\in \mathcal{F}_{\Delta}}\frac{\w_{\{u,v\}, \{v,w\}}}{\w_f} + \sum_{f\in \mathcal{F}_{\square}}\frac{\w_{\{u,v\}, \{v,w\}}}{\w_f} + \frac{\w_{\{u,v\}}}{\w_{\{u,v\}, \{v,w\}}} + \frac{\w_{\{v,w\}}}{\w_{\{u,v\}, \{v,w\}}} \Big.\\
			&\qquad \left.- \Big\vert
			\sum_{\big(\{u,v\}, \{v,w\}, e_1, e_2 \big) \in \mathcal{F}_{\square}}\frac{\sqrt{\w_{\{u,v\}, \{v,w\}} \cdot \w_{e_1, e_2}}}{w_f} \right. \\
			&\qquad \qquad \left. -\sum_{\substack{e\in \Gamma \big(\{u,v\} \big)\\ \big(\{u,v\}, \{v,w\}, e \big) \notin \mathcal{F}_{\Delta}}}\frac{\w_{\{u,v\}}}{\sqrt{\w_{\{u,v\}, \{v,w\}} \cdot \w_{\{u,v\}, e}}} \right.  - \sum_{\substack{e\in \Gamma \big(\{v,w\} \big)\\ \big(\{u,v\}, \{v,w\}, e \big) \notin \mathcal{F}_{\Delta}}}\frac{\w_{\{v,w\}}}{\sqrt{\w_{\{u,v\}, \{v,w\}}\cdot \w_{\{v,w\}, e}}}
			\Big\vert \Biggr) \; .
		\end{align*}
	\end{defn}
	To determine weights of quadrangles, we utilize again Heron's formula: If we treat the two edges with the largest weights to be parallel, we can compute the weight of a quadrangular face via the area of the two triangles that form the trapezoid.  
	Specifically, let $f = (e_i,e_j,e_k,e_l)$ denote a quadrangle in $G$, i.e., $e_i \sim e_j$, $e_j \sim e_k$, $e_k \sim e_l$, $e_l \sim e_i$ while $e_i \not\sim e_k$ and $e_j \not\sim e_l$, where $\w_{e_i} \ge \w_{e_j} \ge \w_{e_k} \ge \w_{e_l}$, and we consider three different cases here: (i) $\w_{e_i} > \w_{e_j}$, i.e., when the two largest weights (or lengths) are not the same, (ii) $\w_{e_i} = \w_{e_j}$ and $\w_{e_k} > \w_{e_l}$, i.e., when the two largest weights are the same while the other two are different, and (iii) $\w_{e_i} = \w_{e_j}$ and $\w_{e_k} = \w_{e_l}$, i.e., when the two largest are the same while the others are also the same. In case (i), we assume that $e_i,e_j$ are parallel to each other, and then the area of the quadrangle can be obtained via Heron's formula:
	\begin{align}
		\label{equ:quad-lp}
		\w_{f} &:= \sqrt{s(s-(\w_{e_i}-\w_{e_j}))(s-\w_{e_k})(s-\w_{e_l})}\left(1 + 2\frac{\w_{e_j}}{\w_{e_i}-\w_{e_j}}\right) \\
		s &= \frac{(\w_{e_i}-\w_{e_j}) + \w_{e_k}+\w_{e_l}}{2} \; .
	\end{align}  
	In case (ii), we assume that $e_k,e_l$ are parallel to each other, and similarly
	\begin{align}
		\label{equ:quad-sp}
		\w_{f} &:= \sqrt{s(s-(\w_{e_k}-\w_{e_l}))(s-\w_{e_i})(s-\w_{e_j})}\left(1 + 2\frac{\w_{e_l}}{\w_{e_k}-\w_{e_l}}\right) \\
		s &= \frac{(\w_{e_k}-\w_{e_l}) + \w_{e_i}+\w_{e_j}}{2} \; .
	\end{align}  
	Finally, in case (iii), we assume that $f$ is a rectangle, thus
	\begin{align}
		\label{equ:quad-rect}
		\w_{f} &:= \w_{e_i}\w_{e_k} \; .
	\end{align}

	\section{Algorithms}
	\label{sec:algorithm}
	\subsection{ORC-based approach}
	
	The curvature-based community detection algorithms (Alg.~\ref{alg.1}) is build on the fact that \emph{bridges} between communities are more negatively curved than edges within communities. We have discussed earlier that this can be observed in both graphs (in the single-membership case) and their line graphs (in the case mixed-membership communities). In this section, we provide a detailed discussion of Alg.~\ref{alg.1} with $\kappa$ denoting Ollivier's Ricci curvature (ORC). We begin by giving some intuition for the differences in curvature values of bridges and internal edges, before formalizing the argument below. By construction, ORC is closely linked to the behavior of two random walks starting at neighboring vertices~\citep{Ol2,jost2014ollivier}. Informally, they are more likely to draw apart if the edge between them has negative ORC, and to draw closer together otherwise. Random walks that start at nodes adjacent to a bridge (\emph{bridge nodes}), typically move into the communities to which the respective nodes belong. As a consequence, they draw apart quickly. Random walks that start at non-bridge nodes, i.e., nodes adjacent to internal edges, are more likely to stay close to each other, as they remain within the same community. Hence, we expect bridges to have much lower ORC than internal edges, a fact that is easily confirmed empirically (see, e.g., Fig.~\ref{fig:weights}).
	
	\begin{rmk}\normalfont
		We note that an alternative notion of Ollivier's curvature (Eq.~\eqref{eq:orc-e}) was given by~\citet{lin-yau}, where node neighborhoods are endowed with the measure
		\begin{equation}\label{eq:orc-lazy}
			m_u^{\alpha}(z) := \begin{cases}
				\alpha, &{u=z} \\
				\frac{1-\alpha}{\vert \mathcal{N}_u \vert}, &{u \sim z} \\
				0, &{\rm else}
			\end{cases} \; .
		\end{equation}
		instead of the uniform measure. A weighted version of this curvature was considered in~\citep{ni2019community}. Notice that for $\alpha>0$, this curvature notion connects to \emph{lazy} random walks starting at neighboring nodes. Here, the parameter $\alpha$ could be seen as controlling whether the random walk is likely to revisit a node, which in turn relates to a distinction between exploration of node neighborhoods in the style of ``breadth-first search'' vs. 	``depths-first search''. A suitable choice of $\alpha$ is highly dependent on the topology of the graph; hence, replacing ORC with Eq.~\eqref{eq:orc-lazy} provides an additional means for incorporating side information on the structure of the underlying graph. However, in this work we restrict ourselves to the classical ORC notion; an exploration of notion~\eqref{eq:orc-lazy} is beyond the scope of the paper.
	\end{rmk}
	\subsubsection{Algorithm}
	\paragraph{Single-Membership Community Detection.}
	We implement Alg.~\ref{alg.1} via ORC by setting $\kappa(\cdot)=\widehat{\Ric_O}(\cdot)$, i.e., the curvature is chosen to be the combinatorial ORC approximation proposed in sec.~\ref{sec:orc}. Notice that as the edge weights evolve under the Ricci flow, the graph is weighted, even if the input graph was unweighted. Consequently, we use weighted ORC curvature, constructing the approximation from the bounds in Thm.~\ref{thm:weightollibnds}. We further choose equally-spaced cut-off points $\{x_i\}_{i=0}^{n_f}$, where $x_0 = \max_{\{u,v\}\in E}w_{u,v}^T,\, x_1 = x_0 - \delta,\, \dots,\, x_{n_f} = ((x_0-1)\mod \delta + 1)$, and the step size for the cut-off points is set to be $\delta = 0.025$. Other hyperparameter choices include a constant step size $\nu=1$, $\epsilon=10^{-4}$, drop threshold $\epsilon_d=0.1$, and stopping time $T=10$, i.e., we evolve edge weights under Ricci flow for 10 iterations.
	Instances of Alg.~\ref{alg.1} via ORC were previously considered in~\citep{ni2019community,sia2019ollivier,gosztolai2021unfolding}. Both approaches computed ORC either exactly (via the earth mover's distance) or approximately (via Sinkhorn's algorithm). However, due to the higher computational cost of either variant of the ORC computation, the approaches were  limited in their scalability (see discussion sec.~\ref{sec:orc-comp} below).
	
	
	\paragraph{Mixed-Membership Community Detection.}
	As in the single-membership case, we implement Alg.~\ref{alg.1} via our combinatorial ORC approximation (Eq.~\eqref{eq:orc-approx-line-graph}). The input graph $G$ is the line graph of the underlying graph, which is constructed before the clustering procedure is started. All edge weights are initially set to $1$, i.e., the input graph is unweighted. Hyperparameter choices are analogous to the single-membership case. Instance of Alg.~\ref{alg.1} via ORC were first considered in~\citep{tian2022mixed}. Again, due to the high computational cost of the ORC computation, the approach was limited in its scalability.
	
	
	\subsubsection{Theoretical results.}\label{sec:orc-theory}
	\paragraph{Single-Membership case.} To give theoretical intuition for curvature-based community detection algorithms in the single-membership case, we considered the following model class of graphs with community structure:
	\begin{defn}\label{def:G-ab}
		Let $G_{a,b}$ ($a > b \geq 2$) be a graph constructed as follows:
		\begin{enumerate}
			\item Construct a complete graph with $b$ vertices $\{v_1,\ldots,v_b\}$.
			\item For each $i=1,\ldots,b$, introduce vertices $u_{ij}$, $j=1,\ldots,a$, and make all possible connections between vertices in $\{v_i\}\cup\{u_{ij}\}_{j=1}^a$, resulting in $b$ copies of complete graphs $K_{a+1}$.
		\end{enumerate}
		
		Notice that $G_{a,b}$ is a stochastic block model with $b$ blocks of size $a+1$.  
	\end{defn}
	\begin{wrapfigure}{R}{0.48\linewidth}
		\centering
		\includegraphics[scale=0.35]{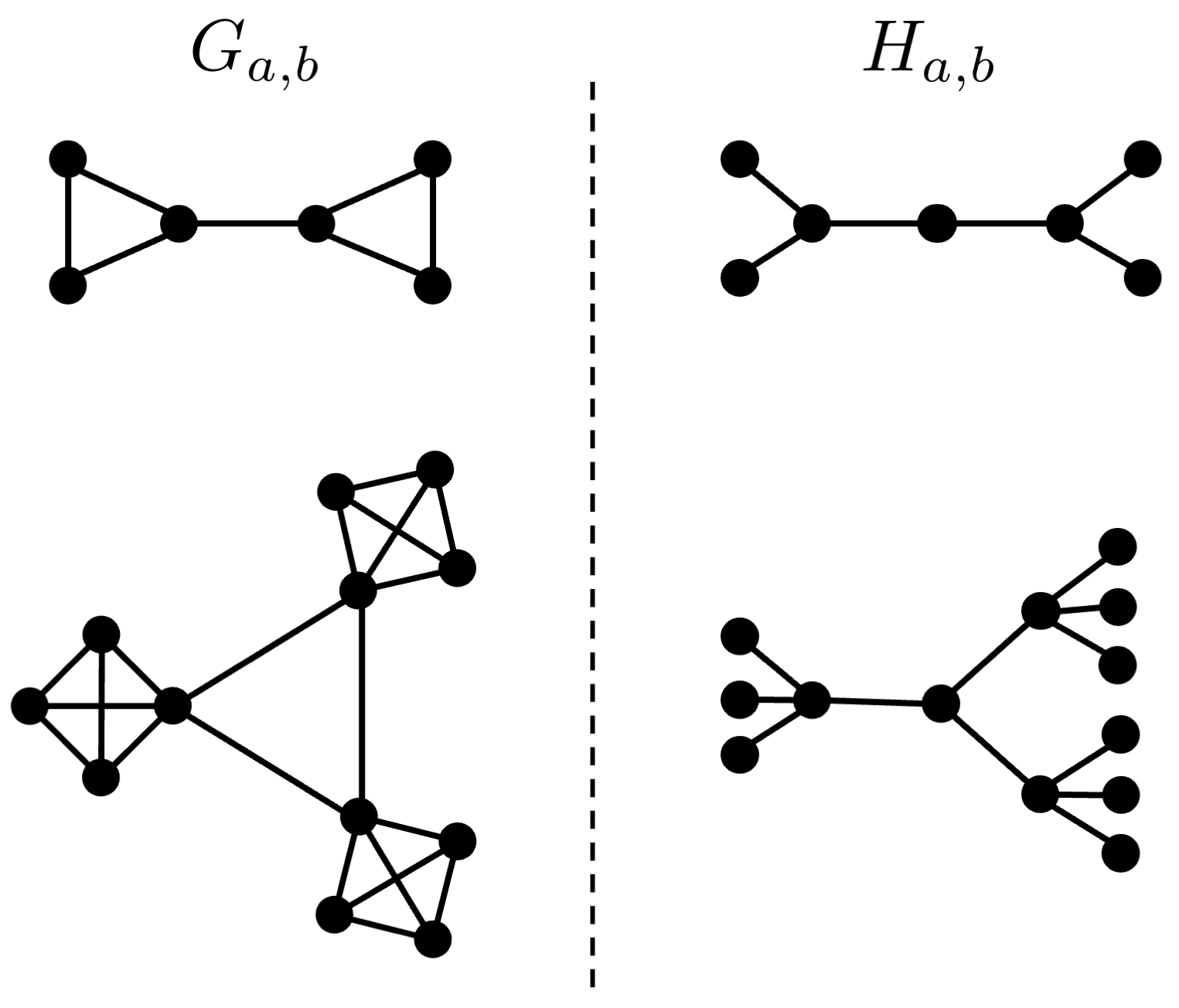}
		\caption{Examples of graphs of the model classes $G_{a,b}$ (left; top: $a=b=2$, bottom: $a=b=3$) and $L_{a,b}$ (right; top: $a=b=2$, bottom: $a=b=3$). Notice that for the top and bottom pairs of graphs, $G_{a,b}$ is the dual of $L_{a,b}$.}
		\label{fig:orc-theory}
		\vspace{-0.7em}
	\end{wrapfigure}

	\noindent We show examples of graphs of the form $G_{a,b}$ in Fig.~\ref{fig:orc-theory}.
	Notice that in a graph $G_{a,b}$ we have three types of edges: 
	\begin{enumerate}
		\item \emph{bridges} between communities, we denote vertices adjacent to bridges as \emph{bridge nodes} and all others as \emph{internal nodes};
		\item \emph{internal edges} that connect a bridge node and an internal node;
		\item \emph{internal edges} that connect two internal nodes.
	\end{enumerate}
	As before, we consider the following mass distribution on the node neighborhoods:
	\begin{equation}\label{eq:G_ab-mass}
		m_x = \begin{cases}
			\frac{1}{C} e^{-w_{x,y}},  &x \sim y \\
			0, &{\rm else}
		\end{cases} \; .
	\end{equation}
	Here, $C$ denotes the normalizing function. This choice follows from the more general scheme in Eq.~\eqref{eq:orc-lazy} by setting $\alpha=0$, $p=1$.  We assume that at initialization, all edges are assigned weight 1.  
	\begin{theorem}\label{thm:ORC-guarantee}
		Let the edge weights in $G_{a,b}$ evolve under the Ricci flow $w_i^{t+1} \gets (1-\kappa_i^t)w_i^t$ (where $\kappa_i^t$ denotes ORC for the edge type $i$ using the weights at time $t$).
		\begin{enumerate}
			\item  The weights of internal edges (type (3)) contract asymptotically, i.e.,  $w_3^{t} \rightarrow 0$ as $t \rightarrow \infty$.  
			\item The weights of bridges are larger than those of internal edges of types (2) and (3), i.e., $w_1^{t} > w_2^{t}$ and $w_1^t>w_3^t$ for all $t>0$.
		\end{enumerate}
	\end{theorem}
	\begin{rmk}
		We note that the class of graphs $G_{a,b}$ was previously considered in~\citep{ni2019community}, which also give asymptotic guarantees for the evolution of edge weights under the Ricci flow.  However,  they assume a different mass distribution, hence, their asymptotic results are not directly applicable to our setting.  We therefore include a proof for our setting below.
	\end{rmk}
	We will present the full proof of this theorem in Appendix~\ref{sec:ORC-guarantee-proof}, but we give a sketch of the proof here. The argument proceeds by induction, first calculating the exact ORC curvatures of the three edge types at iteration 1, and then calculating the exact ORC curvatures of the three edge types at iteration $t$, given the weights $w_i^t$. In each of these steps, we construct transport plans, as well as 1-Lipschitz functions with equal objective function values to show there optimality. This proves that we have the exact values of the Wasserstein distances. We use these exact formulas to show that the evolution of the weights has the properties given in the statement.

	\paragraph{Mixed-Membership case.} We modify the model class introduced in Definition~\ref{def:G-ab} to account for mixed-membership community structure:
	\begin{defn}
		Let $L_{a,b}$ ($a > b \geq 2$) be a graph constructed as follows:
		\begin{enumerate}
			\item Construct a star-shaped graph,  consisting of a center node of degree $b$ and $b$ leave nodes of degree 1.
			\item Replace each leave  node with a star-shaped graph with $a+1$ vertices.
		\end{enumerate}
		Notice that $L_{a,b}$ consists of $b$ blocks (communities) of size $a+2$ with one common node, i.e., the center node is a member of each of the $b$ communities.
	\end{defn}
	It is easy to see that the dual of a graph $L_{a,b}$ is a graph of the form $G_{a,b}$. Consequently, we recover guarantees for curvature-based clustering on the line graph from Theorem~\ref{thm:ORC-guarantee}.

	\subsubsection{Computational Considerations}\label{sec:orc-comp}
	The two key bottlenecks in ORC instances of Alg.~\ref{alg.1} are (1) the computation of the $W_1$-distance in the ORC computation (both single- and mixed-membership case), and (2) the construction of the line graph (mixed-membership case). Following our discussion in sec.~\ref{sec:orc}, we can circumvent bottleneck (1) by approximating ORC with the arithmetic mean of its upper and lower bounds, which are given by Eq.~\eqref{eq:orc-approx} in the original graph (use in single-membership case) and Eq.~\eqref{eq:orc-approx-line-graph} in the line graph (use in mixed-membership case). This reduces the complexity of Alg.~\ref{alg.1} to $O(\vert E \vert)$, in comparison with previous approaches~\citep{ni2019community,sia2019ollivier}, which achieved at best $O(\vert E \vert^2)$. The second bottleneck may also be circumvented in special cases. We discuss the details in sec.~\ref{sec:exp-mmcd} below.
	
	\subsection{FRC-based approach}
	Finally, we discuss the second instance of Alg.~\ref{alg.1}, where $\kappa$ denotes Forman's Ricci curvature (FRC). Again, we observe that, typically, bridges have lower FRC than internal edges in communities, which can be exploited for curvature-based community detection. As in the previous section, we first provide some informal intuition on this observation, before formally introducing the algorithm. Notice that bridge nodes have typically a high node degree, as they connect not only to other nodes within the same community, but also to nodes in other communities. It is easy to see from the definition (recall, $\Ric_F(e)=4-d_{v_1}-d_{v_2}$) that, as a result, bridges have usually lower curvature than internal edges in unweighted networks. This imbalance in curvature values is iteratively reinforced as edge weights evolve under the Ricci flow. In particular, we see that, in the case of bridges, the 2-complex FRC (Eq.~\eqref{eq:frc-f})
	\begin{align*}
		\Ric_{F} (e) &:= \w_{e} \left( \Big(\Big(\sum_{f \sim e}\frac{\w_{e}}{\w_{f}} \Big) + \frac{\w_{v_1}}{\w_{e}} + \frac{\w_{v_2}}{\w_{e}} \Big)
		- \sum_{ \hat{e} \parallel e}\left\vert
		\sum_{f \sim \hat{e},e} \frac{\sqrt{\w_{e} \cdot \w_{\hat{e}}}}{\w_f} - \sum_{v \sim e, \hat{e}} \frac{\w_v}{\sqrt{\w_e \cdot \w_{\hat{e}}}} \right\vert \right)  
	\end{align*}
	is increasingly dominated by the second term, due to the high degree of the adjacent vertices and, consequently, the large number of parallel edges. This is reinforced as the edge weight increases under Ricci flow, which decreases the first term, resulting in negative curvature with high absolute value. In contrast, for internal edges, the second term has a smaller magnitude due to the smaller degrees of the adjacent vertices. This is again reinforced as the edge weight decreases under Ricci flow.

	\subsubsection{Algorithm}
	
	\paragraph{Single-Membership Community Detection.}
	We implement Alg.~\ref{alg.1} via FRC by setting $\kappa(\cdot)=\Ric_F(\cdot)$, chosen to be the 2-complex FRC with triangle contributions (Eq.~\eqref{eq:frc-2-complex}). In the absence of given edge weights, we initialize edge weights to 1 as before. Adaptive step sizes are chosen proportional to the inverse of the highest absolute curvature of any edge in a given iteration, i.e., $\nu_t = 1/(1.1\times \max_{\{u,v\}\in E}|\kappa_{u,v}^t|)$. In our experiments below, we evolve edge-weights for $T=10$ iterations under the Ricci flow. To infer communities, cut-off points are chosen as $\{x_i\}_{i=0}^{n_f}$, where $x_0 = \max_{\{u,v\}\in E}w_{u,v}^T,\, x_1 = \max_{\{u,v\}\in E, w_{u,v}^T < x_0}w_{u,v}^T,\, \dots,\, x_{n'_f} = w^T_{q},\, x_{n'_f+1} = x_{n'_f} - \delta,\, \dots,\, x_{n_t} = ((x_{n'_f}-1.1w_{\min})\mod\delta) + 1.1w_{\min}$. Here, $w^T_{q}$ is the $q$th quantile of all weights; we set $q = 0.999$, step size $\delta=0.25$, and $w_{\min} = \min_{\{u,v\}\in E}w_{u,v}^T$. Note that this differs from the uniformly spaced cut-off points in the ORC case. Under FRC-based Ricci flow, the magnitude of some edge weights grows rapidly, resulting in a fat-tailed edge weight distribution. Hence, adjusting the spacing between cut-off points accordingly reduces the number of cut-off points needed to achieve high accuracy in cluster assignments.
	We remark that a special case of Alg.~\ref{alg.1} was previously proposed in~\citep{weber2018detecting}, although considering the classical FRC notion (Eq.~\eqref{eq:frc-e}) instead of the 2-complex FRC and utilizing a simpler scheme to identify bridges between community using FRC.


	\paragraph{Mixed-Membership Community Detection.}
	We implement Alg.~\ref{alg.1} via FRC on the line graph, by setting $\kappa(\cdot,\cdot)=\Ric_F(\cdot,\cdot)$, chosen to be the 2-complex FRC with triangle \emph{and quadrangle} contributions (Definition \ref{def:frc-f4}). In the mixed-membership case, bridges between communities appear only in the line graph. Hence, we first construct the line graph before the actual clustering procedure is started. The input to Alg.~\ref{alg.1} is the pre-computed line graph. All edge weights are initially set to $1$, i.e., the input graph is unweighted. The adaptive step sizes, stopping time and cut-off points are chosen analogously to the single-membership case. To the best of our knowledge, FRC-based approaches have not been considered previously in this mixed-membership setting.

	\subsubsection{Computational Considerations}
	Forman's Ricci curvature for edges is given by a simple combinatorial formula, which can be computed in linear time. The variants considered in this work (2-complex Forman curvature with triangle and quadrangle contributions) require an additional subroutine, which computes the number of triangles that include the respective edge. Counting the number of triangles in a subgraph can be done in $O(\vert V \vert \cdot \vert E \vert)$, i.e., in this case $\vert V \vert$ denotes the number of vertices and $\vert E \vert$ the number of edges in the 2-hop neighborhood. Similarly, the number of quadrangles in the same 2-hop neighborhood can be computed in $O(\vert V \vert^2 \cdot \vert E \vert)$. Since this subroutine requires only local information, it is easily parallelizable.

	\section{Experiments}
	\label{sec:exp}
	\subsection{Network Data}
	\label{sec:graph_models}
	
	\paragraph{Synthetic Data.}
	In order to explore the interaction between discrete Ricci curvature and the structure of complex networks, we have considered both Stochastic Block Model and Random Geometric Graphs in our previous analysis. In this section, we systematically examine the performance of the algorithms on Single- and mixed-membership Stochastic Block Models, where the ground-truth community membership is available. We have introduced the two variants of the stochastic block model above in sec.~\ref{sec:overview}. 
	
	\paragraph{Real network data.}
	To test our methods on real data, we consider two data sets, (i) collaboration networks from DBLP~\citep{Yangdblpdata2018}, and (ii) ego-networks from Facebook~\citep{Leskovec2012ego}, for which ground truth labels are available through SNAP~\citep{snapnets}. Specifically, the collaboration network is constructed by a comprehensive list of research papers in computer science provided by the DBLP computer science bibliography. Here, an edge from one author to another indicates that they have published at least one joint paper. There are intrinsic communities defined by the publication venue, e.g., journals or conferences. Here, we maintain the choice of two communities, and randomly select two venues whose sizes are over $n=100$ and have at least one node of mixed membership. We report results for two networks: ``DBLP-1" from publication venues no.~1347 and 1892, and ``DBLP-2" from publication venues no.~1347 and 2459. While in the second data set (ego-networks), all the nodes are friends of one central user, and the friendship circles set by this user can be used as ground truth communities. We carried out the preprocessing in a similar manner as \citep{zhang2020overlap}, and then select two networks for which the modularity of the ground-truth communities is greater than $0.5$: ``FB-1" for no.~414, and ``FB-2" for no.~1684. To better understand the characteristics of the different real networks, we provide the following summary statistics for each network (see Table \ref{tab:real-stats}): (i) average node degree $d$, (ii) degree heterogeneity measured by the standard deviation of node degrees over $d$, (iii) the actual proportion of nodes with mixed membership, and (iv) the modularity of the ground-truth community. We note that the Facebook networks tend to be denser, with more nodes of mixed membership, while the collaboration networks tend to have relatively more heterogeneous degrees.  
	\begin{table}[htb]
		\centering
		\begin{tabular}{l|cccccccc}
			\hline
			& $n$ & $|E|$ & $|\mathcal{E}|$ & $k$ & $d$ & $\sigma_d/d$ & $\hat{\pi}_o$ & Modularity\\
			\hline
			DBLP-1 & $213$ & $591$ & $5489$ & $2$ & $5.54$ & $0.93$ & $0.005$ & $0.40$\\
			DBLP-2 & $237$ & $670$ & $5935$ & $2$ & $5.65$ & $0.86$ & $0.004$ & $0.45$\\
			FB-1 & $128$ & $1593$ & $45635$ & $3$ & $24.89$ & $0.44$ & $0.055$ & $0.53$\\
			FB-2 & $621$ & $12399$ & $718267$ & $5$ & $39.93$ & $0.69$ & $0.004$ & $0.52$\\
			\hline
		\end{tabular}
		\caption{Summary statistics of the real networks.}
		\label{tab:real-stats}
	\end{table}

	\subsection{Evaluating clustering accuracy}\label{apx:nmi}
	To evaluate the accuracy of a clustering in the single- and mixed-membership settings, we introduce the following two notions of \emph{Normalized Mutual Information} (NMI). When each node can only belong to one community, the classic NMI is defined by considering each clustering result as a random variable and then comparing the information contained in them \citep{strehl2002nmi}. Specifically, let $X,Y$ be the random variables described by two different cluster labels. Let $I(X,Y)$ denote the mutual information between $X$ and $Y$, and let $H(X), H(Y)$ denote the entropy of $X,Y$, respectively. Then NMI is defined as 
	\begin{align}
		NMI(X,Y) = \frac{I(X,Y)}{mean(H(X), H(Y))},
	\end{align}
	where $mean(H(X),H(Y))$ denote the mean of the two, and we consider the arithmetic mean in our experiments, $mean(H(X),H(Y)) = (H(X) + H(Y))/2$. The NMI ranges from $0$ to $1$, where value $1$ corresponds to perfect matching, or exact recovery if one of the clustering result is ground truth.

	We note that the classic NMI is not well-defined for mixed-membership community detection, thus we later consider the following extended NMI \citep{Lancichinetti2009nmi}. Specifically, for communities $C_1, C_2, \dots, C_k$, we now express the community membership of each node $i$ as a binary vector of length $k$. $(z_i)_l = 1$ if node $i$ belongs to $C_l$; $(z_i)_l = 0$ otherwise. The $l$-th entry of this vector can then be viewed as a random variable $Z_l$, whose probability distribution is given by $P(Z_l = 1) = n_l/n$ and $P(Z_l = 0) = 1 - P(Z_l = 1)$, where $n_l=|C_l|$ and $n=|V|$. The same holds for the random variable $Y_h$ associated with another set of communities $C'_1, C'_2,\dots, C'_{k'}$. Both the empirical marginal probability distribution $P_{Z_l}$ and the joint probability distribution $P(Z_l, Y_h)$ are used to further define entropy $H(\textbf{Z})$ and $H(Z_l, Y_h)$. The conditional entropy of $Z_l$ given $Y_h$ is defined as $H(Z_l|Y_h) = H(Z_l, Y_h) - H(Y_h)$. The entropy of $Z_l$ with respect to the entire vector $\textbf{Y}$ is based on the best matching between $Z_l$ and the component of $\textbf{Y}$ given by
	\begin{equation*}
		H(Z_l|\textbf{Y}) = \min_{h\in \{1,2,\dots,k'\}}H(Z_l|Y_h).
	\end{equation*}
	The normalized conditional entropy of $\textbf{Z}$ with respect to $\textbf{Y}$ is 
	\begin{equation*}
		H(\textbf{Z}|\textbf{Y}) = \frac{1}{k}\sum_{l}\frac{H(Z_l|Y)}{H(Z_l)}.
	\end{equation*}
	In the same way, we can define $H(\textbf{Y}|\textbf{Z})$. Finally, the extended NMI for two sets of communities $C_1, C_2,\dots, C_k$ and $C'_1, C'_2,\dots,C'_{k'}$ is given by
	\begin{equation*}
		NMI(\textbf{Z}|\textbf{Y}) = 1 - [H(\textbf{Z}|\textbf{Y}) + H(\textbf{Y}|\textbf{Z})]/2 \; .
	\end{equation*}
	Hence to use the extended NMI, we convert the label vector $y$ to a binary assignment by first normalizing it by its 2-norm and then thresholding its element by $0.8/k$.

	\subsection{Single-Membership Community Detection}
	
	\paragraph{Benchmarking.} For the first set of experiments, we explore the performance of curvature-based methods on planted SBMs with two equally-sized blocks and various parameters. We start with graphs with fixed probabilities $p_{in}=0.1, p_{out}=0.01$, while varying the network size $n$ from $100$ up to $5000$, and then we also consider graphs with fixed size $n=1000$ and probability $p_{out} = 0.01$, while changing the probability $p_{in}$ from $0.05$ to $0.2$. For each set of parameters, we generate $n_s = 10$ graphs and run our algorithms for both ORC and FRC, together with other popular community detection methods, the Louvain algorithm (``Louvain") and Spectral clustering (``Spectra"). Specifically, for ORC, the optimal transport can be obtained by solving the earth mover's distance exactly (``ORC-E") or approximately with Sinkhorn (``ORC-S"), while we can also approximate it by the mean of the upper and lower bounds we have developed (``ORC-A"). For FRC, we can consider the 1-d version as in Eq.~\eqref{eq:frc-e} where no face information is included (``FRC-1"), the augmented version with triangular faces as in Eq.~\eqref{eq:frc-f} (``FRC-2") and also the one with further quadrangular faces as in Definition \ref{def:frc-f4} (``FRC-3"). For the Louvain algorithm, we run it for $100$ times, construct the co-occurrence matrix of the clustering results, and then run the algorithm one more time on the graph constructed by the co-occurence matrix, in order to obtain consistent communities. We report the mean NMI, runtime and their standard deviations (SD) of the methods. To maintain the detectability of the communities, we require the generated graphs to have modularity grater than $0.4$ with the ground-truth communities. Our experimental results (see Tables \ref{tab:sbm1-nmi} and \ref{tab:sbm2-nmi}) demonstrate that the ORC approach outperforms the reference methods, and successfully recovers single-membership community structure, with an NMI greater than $0.8$, as we vary both the network size and the density via $p_{in}$. We only show the results from ORC-E, since ORC-S has almost the same performance as ORC-E. More importantly, the approximation of ORC we have proposed manages to maintain the important information in ORC, since the corresponding community detection algorithm outperforms the Spectral clustering, and can have comparable performance to the Louvain algorithm, especially as graphs become larger or denser. Being computationally very efficient overall, the FRC approach can also have comparable performance to Spectral clustering, and retrieve the single-membership community structure when the graph is sufficiently large or dense. Here, only the results from FRC-2 is included for the reasons that we will discuss later. 
	\begin{table}[htb]
		\centering
		\begin{tabular}{c|ccccc}
			\hline
			$n$    & Louvain           & Spectra           & ORC-E             & ORC-A             & FRC-2             \\ \hline
			$100$  & $0.517$ ($0.097$) & $0.186$ ($0.163$) & $0.860$ ($0.117$) & $0.354$ ($0.070$) & $0.329$ ($0.066$) \\
			$500$  & $0.998$ ($0.006$) & $0.418$ ($0.378$) & $1$ ($0$)         & $0.662$ ($0.177$) & $0.279$ ($0.134$) \\
			$1000$ & $1$ ($0$)         & $0.526$ ($0.461$) & $1$ ($0$)         & $0.990$ ($0.015$) & $0.437$ ($0.263$) \\
			$5000$ & $1$ ($0$)         & $0.672$ ($0.463$) & $1$ ($0$)         & $1$ ($0$)         & $0.911$ ($0.263$) \\\hline              
		\end{tabular}
		\caption{Mean (SD) of NMI from different methods on planted two-block SBMs with $p_{in} = 0.1$, $p_{out}=0.01$, and varying sizes $n$.}
		\label{tab:sbm1-nmi}
	\end{table}
	\begin{table}[htb]
		\centering
		\begin{tabular}{c|ccccc}
			\hline
			$n$    & Louvain            & Spectra           & ORC-E              & ORC-A              & FRC-2              \\ \hline
			$100$  & $1.135$ ($0.068$) & $0.027$ ($0.059$) & $1.910$ ($0.096$) & $1.877$ ($0.093$) & $0.727$ ($0.123$)  \\
			$500$  & $14.95$ ($0.294$) & $0.039$ ($0.025$) & $8.692$ ($0.305$) & $6.328$ ($0.205$) & $5.005$ ($0.367$)  \\
			$1000$ & $44.93$ ($0.739$) & $0.211$ ($0.311$) & $34.48$ ($1.207$) & $19.82$ ($0.795$) & $17.48$ ($1.075$) \\
			$5000$ & $1013$ ($17.61$) & $518.3$ ($810.7$)  & $2744$ ($17.88$)  & $1080$ ($7.679$) & $394.2$ ($7.373$) \\
			\hline
		\end{tabular}
		\caption{Mean (SD) of runtime (in seconds) from different methods on planted two-block SBMs with $p_{in} = 0.1$, $p_{out}=0.01$, and varying sizes $n$.}
		\label{tab:sbm1-time}
	\end{table}
	
	\begin{table}[htb]
		\centering
		\begin{tabular}{c|ccccc}
			\hline
			$p_{in}$ & Louvain   & Spectra           & ORC-E             & ORC-A             & FRC-2             \\ \hline
			$0.05$   & $0.999$ ($0.003$) & $0.087$ ($0.175$) & $0.992$ ($0.011$) & $0.202$ ($0.010$) & $0.177$ ($0.013$) \\
			$0.1$    & $1$ ($0$)         & $0.446$ ($0.440$) & $1$ ($0$)         & $0.993$ ($0.005$) & $0.275$ ($0.182$) \\
			$0.15$   & $1$ ($0$)         & $1$ ($0$)         & $1$ ($0$)         & $1$ ($0$)         & $0.990$ ($0.031$) \\
			$0.2$    & $1$ ($0$)         & $1$ ($0$)         & $1$ ($0$)         & $1$ ($0$)         & $1$ ($0$)         \\ \hline
		\end{tabular}
		\caption{Mean (SD) of NMI from different methods on planted two-block SBMs with $n=1000$, $p_{out}=0.01$, and varying probabilities $p_{in}$.}
		\label{tab:sbm2-nmi}
	\end{table}
	
	\begin{table}[htb]
		\centering
		\begin{tabular}{c|ccccc}
			\hline
			$p_{in}$ & Louvain             & Spectra             & ORC-E               & ORC-A              & FRC-2              \\ \hline
			$0.05$   & $48.96$ ($1.163$) & $0.105$ ($0.070$) & $14.50$ ($0.861$)  & $13.39$ ($0.78$) & $11.14$ ($1.447$) \\
			$0.1$    & $44.39$ ($0.572$) & $0.158$ ($0.125$) & $34.67$ ($0.649$)  & $20.40$ ($0.39$) & $18.12$ ($0.735$) \\
			$0.15$   & $50.88$ ($0.537$) & $0.998$ ($1.125$) & $71.12$ ($1.541$)  & $28.50$ ($0.45$) & $27.39$ ($0.992$) \\
			$0.2$    & $59.99$ ($1.107$) & $7.004$ ($6.655$) & $129.0$ ($4.211$) & $37.21$ ($0.78$) & $39.79$ ($1.702$) \\ \hline
		\end{tabular}
		\caption{Mean (SD) of runtime (in seconds) from different methods on planted two-block SBMs with $n=1000$, $p_{out}=0.01$, and varying probabilities $p_{in}$.}
		\label{tab:sbm2-time}
	\end{table}
	
	\paragraph{Rationale for ORC approximation.} ORC-S has been considered in the literature in order to improve the efficiency of ORC computation, but we notice that the improvement is not clear when the graph is not sufficiently large or dense (see Table \ref{tab:sbm3-time}). However, the approximation we have proposed almost always significantly improve the computational efficiency (see also Tables \ref{tab:sbm1-time} and \ref{tab:sbm2-time}). We also mention here that the current computations are only allocated with moderate level of memory, while we also observed that ORC-S can improve the efficiency when large amount of memory is available, as in the analysis of real data in sec.~\ref{sec:exp-mmcd}.
	\begin{table}[htb]
		\centering
		\begin{tabular}{c|ccc}
			\hline
			$p_{in}$ & ORC-E     & ORC-S     & ORC-A     \\ \hline
			$0.1$ & $2785$   & $5859$   & $1013$   \\
			$0.2$ & $16876$  &  $68658$  & $2933$     \\
			\hline
		\end{tabular}
		\caption{Runtime (in seconds) from different ORC-based methods on one sample of planted two-block SBMs with $n=5000$, $p_{out} = 0.01$, and varying probabilities $p_{in}$.}
		\label{tab:sbm3-time}
	\end{table}
	
	\paragraph{Rationale for 2-complex FRC.} For FRC, one can continue including higher order faces in the formulation, thus more information in the graphs, in order to further improve the performance of community detection. Our experimental results (see Tables \ref{tab:sbm4-nmi} and \ref{tab:sbm4-time}) suggest that this is the case to include triangular faces, however, as for quadrangular faces, the improvement of accuracy is not clear, while the drop in efficiency is significant. More importantly, further including quadrangular faces in FRC may harm the performance of community detection, especially when the graph is relatively dense. This could be expected from that FRC-2 has the strongest correlation with the clustering coefficients, compared with FRC-1 and FRC-3 (see Fig.~\ref{fig:clust-curv_v}). 
	\begin{table}[htb]
		\centering
		\begin{tabular}{c|ccc}
			\hline
			$p_{in}$ & FRC-1             & FRC-2             & FRC-3             \\ \hline
			$0.05$ & $0.107$ ($0.087$) & $0.177$ ($0.013$) & $0.220$ ($0.100$) \\
			$0.1$  & $0.106$ ($0.087$) & $0.229$ ($0.118$) & $0.077$ ($0.155$) \\
			$0.15$ & $0.073$ ($0.089$) & $0.990$ ($0.013$) & $0.019$ ($0.056$) \\
			$0.2$  & $0.054$ ($0.083$) & $1$ ($0$)         & $0.092$ ($0.092$) \\
			\hline
		\end{tabular}
		\caption{Mean (SD) of NMI from different FRC-based methods on planted two-block SBMs with $n=1000$, $p_{out}=0.01$ and varying probabilities $p_{in}$.}
		\label{tab:sbm4-nmi}
	\end{table}
	\begin{table}[htb]
		\centering
		\begin{tabular}{c|ccc}
			\hline
			$p_{in}$ & FRC-1               & FRC-2               & FRC-3              \\ \hline
			$0.05$ & $10.95$ ($1.386$)  & $11.14$ ($1.447$)  & $18.28$ ($15.99$) \\
			$0.1$  & $16.73$ ($1.103$)  & $17.48$ ($1.075$)  & $202.7$ ($5.996$) \\
			$0.15$ & $23.35$ ($0.353$)  & $27.39$ ($0.992$)  & $101.4$ ($133.7$) \\
			$0.2$  & $29.91$ ($0.571$)  & $39.79$ ($1.702$)  & $111.9$ ($1.680$) \\
			\hline
		\end{tabular}
		\caption{Mean (SD) of runtime (in seconds) from different FRC-based methods on planted two-block SBMs with $n=1000$, $p_{out}=0.01$ and varying probabilities $p_{in}$.}
		\label{tab:sbm4-time}
	\end{table}
	\begin{figure}[htb]
		\centering
		\begin{tabular}{cccc}
			\includegraphics[width=.24\textwidth]{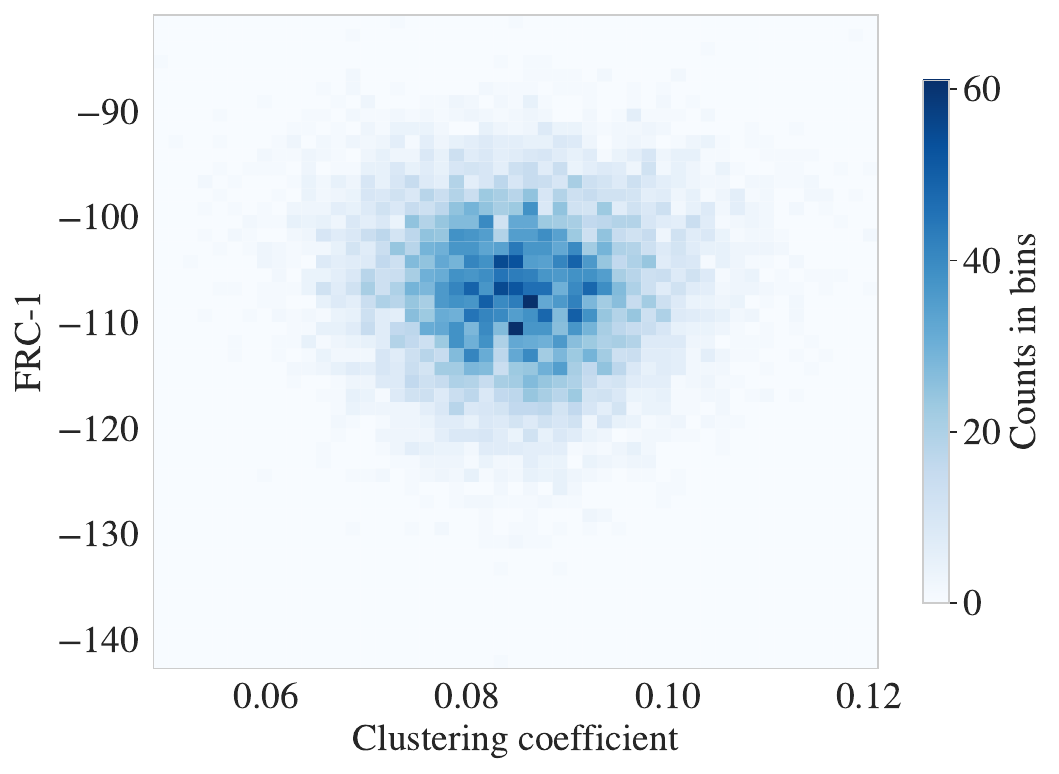} & \includegraphics[width=.24\textwidth]{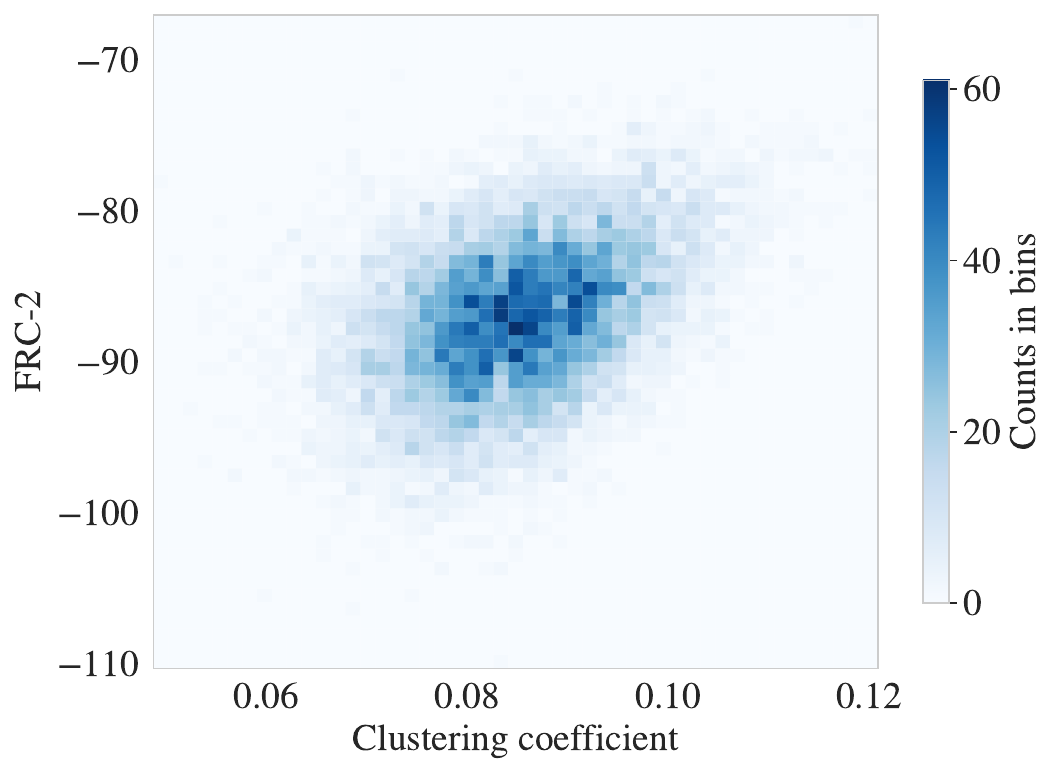} & \includegraphics[width=.24\textwidth]{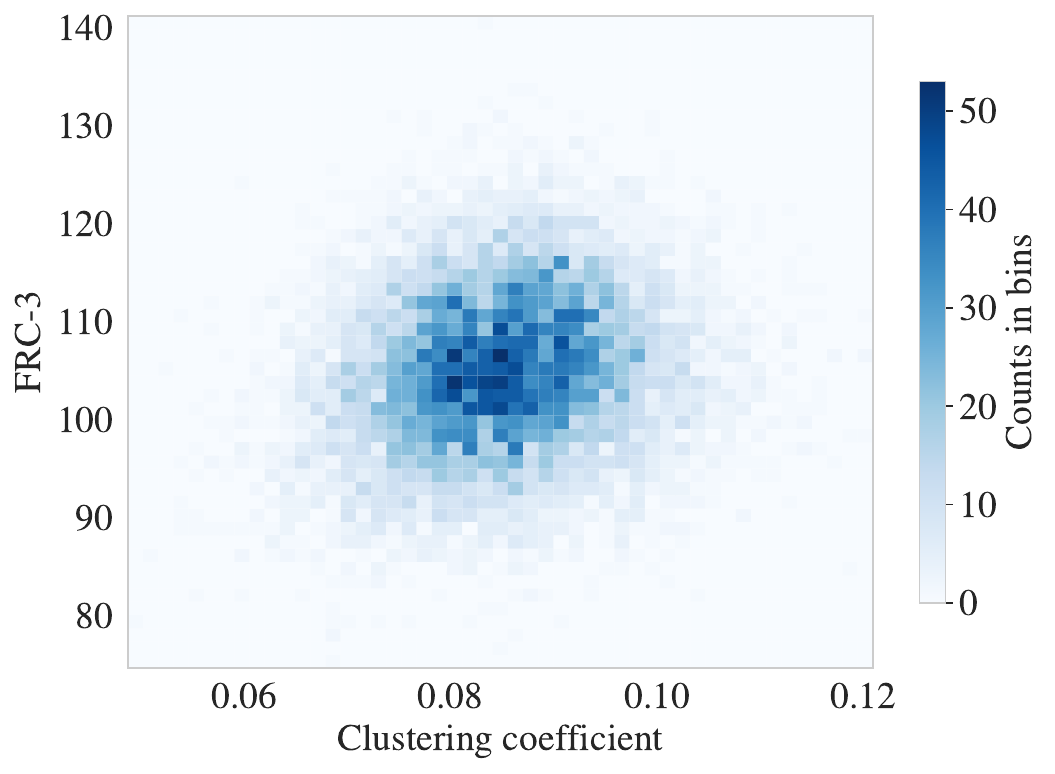} & \includegraphics[width=.24\textwidth]{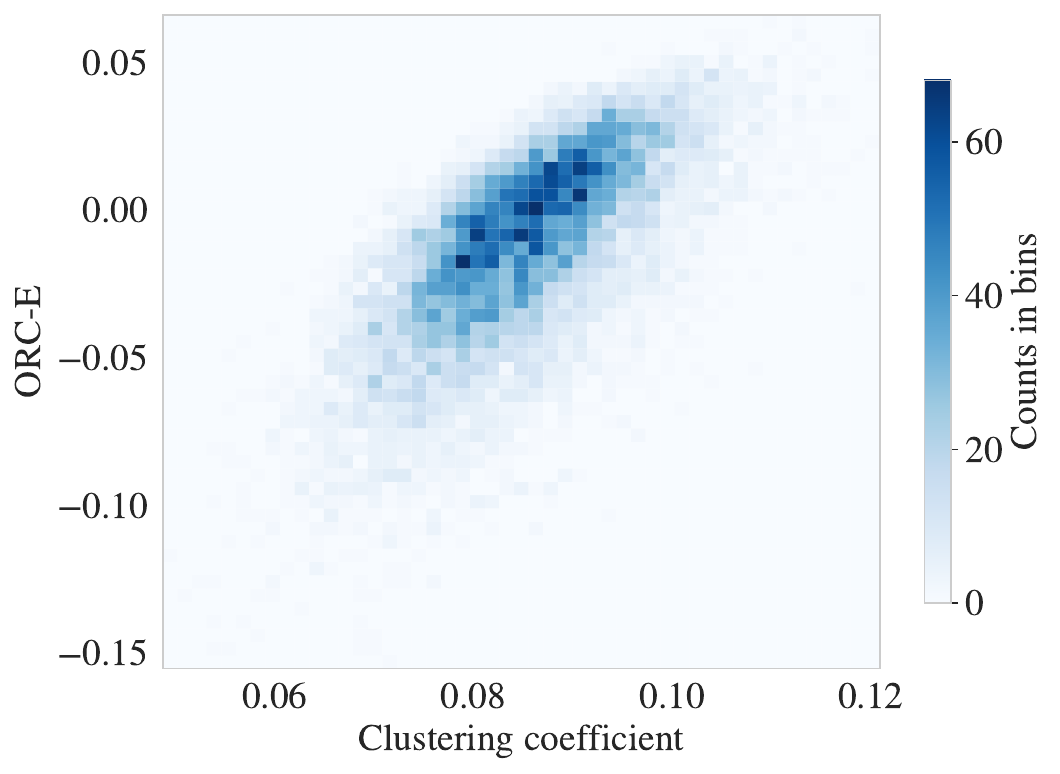}
		\end{tabular}
		\caption{2-d histogram of FRC-1 (left), FRC-2 (middle left), FRC-3 (middle right) and ORC-E (right) versus the clustering coefficients of nodes in $G$, where $n_s = 10$ networks are generated from the planted two-block SBM of size $n=1000$, $p_{in}=0.1$ and $p_{out}=0.01$.}
		\label{fig:clust-curv_v}
	\end{figure}
	
	\subsection{Mixed-Membership Community Detection}\label{sec:exp-mmcd}
	
	\paragraph{Benchmarking.} To demonstrate the performance of our algorithms, we generate graphs from planted MMBs with two (almost) equally-sized blocks and different configurations. We first fix the probabilities $p_{in} = 0.1$, $p_{out} = 0$ and the number of mixed-membership nodes to be $n_o = 1$, while varying the network size from $50$ up to $500$, and then we fix the network size $n=300$, $p_{out}=0$ and $n_o=1$, while varying the probability $p_{in}$ from $0.05$ up to $0.2$. For each set of parameters, we construct $n_s = 10$ graphs, and run our algorithms for both ORC and FRC, together with the Louvain algorithm and Spectral clustering, in the same way as in the single-membership setting, but on the line graphs $L(G)$. We report the mean NMI, runtime, and their standard deviations (SD) of the methods. For the purposes of detectability of communities, we exclude graphs that have modularity no greater than $0.4$ with the ground-truth communities. Our experimental results (see Tables \ref{tab:mmbm1-nmi} and \ref{tab:mmbm2-nmi}) demonstrate that both ORC-based and FRC-based approaches outperform the reference methods, where the ORC approach successfully recovers the mixed-membership community structure in most cases as we vary the network size $n$ and the probability $p_{in}$, while the FRC approach can retrieve the mixed membership when the network is sufficiently large or dense. Since again the performance of ORC-S is almost the same as ORC-E, we only include the results from ORC-E. For the reasons that we have demonstrated in sec.~\ref{sec:DC-frc-higher} and will discuss later, we only include the results from FRC-3. 
	\begin{table}[htb]
		\centering
		\begin{tabular}{c|ccccc}
			\hline
			$n$    & Louvain           & Spectra           & ORC-E             & ORC-A             & FRC-3             \\ \hline
			$50$   & $0.342$ ($0.030$) & $0.169$ ($0.035$) & $0.671$ ($0.157$) & $0.774$ ($0.207$) & $0.701$ ($0.190$) \\
			$100$  & $0.351$ ($0.028$) & $0.088$ ($0.007$) & $0.894$ ($0.186$) & $0.720$ ($0.187$) & $0.771$ ($0.203$) \\
			$300$  & $0.529$ ($0.051$) & $0.044$ ($0.022$) & $0.999$ ($0.002$) & $0.625$ ($0.079$) & $0.739$ ($0.218$) \\
			$500$  & $0.553$ ($0.023$) & $0.033$ ($0.029$) & $1$ ($0$)         & $0.648$ $(0.130)$ & $0.896$ ($0.185$) \\
			\hline              
		\end{tabular}
		\caption{Mean (SD) of the extended NMI from different methods on planted two-block MMBs with $p_{in} = 0.1$, $p_{out}=0$, $n_o=1$ node of mixed membership, and varying sizes $n$.}
		\label{tab:mmbm1-nmi}
	\end{table}
	
	\begin{table}[htb]
		\centering
		\begin{tabular}{c|ccccc}
			\hline
			$n$    & Louvain            & Spectra           & ORC-E            & ORC-A             & FRC-3              \\ \hline
			$50$   & $1.315$ ($0.202$) & $0.012$ ($0.010$)  & $1.504$ ($0.261$) & $0.964$ ($0.072$) & $0.613$ ($0.065$) \\
			$100$  & $10.21$ ($0.866$) & $0.049$ ($0.023$) & $4.861$ ($1.070$) & $3.220$ ($0.377$) & $1.582$ ($0.115$)  \\
			$300$  & $504.4$ ($39.61$) & $606.3$ ($286.6$) & $106.1$ ($10.14$) & $62.80$ ($4.352$) & $94.77$ ($23.04$) \\
			$500$  & $3318$ ($199.1$) & $25332$ ($697.7$) & $628.1$ ($33.08$) & $347.1$ ($19.43$) & $1256$ ($776.1$) \\
			\hline
		\end{tabular}
		\caption{Mean (SD) of runtime (in seconds) from different methods on planted two-block MMBs with $p_{in} = 0.1$, $p_{out}=0$, $n_o=1$ node of mixed membership, and varying sizes $n$.}
		\label{tab:mmbm1-time}
	\end{table}
	
	\begin{table}[htb]
		\centering
		\begin{tabular}{c|ccccc}
			\hline
			$p_{in}$    & Louvain           & Spectra           & ORC-E          & ORC-A            & FRC-3             \\ \hline
			$0.05$  & $0.289$ ($0.046$) & $0.059$ ($0.007$) & $0.998$ ($0.003$) & $0.602$ ($0.096$) & $0.746$ ($0.184$) \\
			$0.1$   & $0.529$ ($0.051$) & $0.044$ ($0.022$) & $0.999$ ($0.002$) & $0.625$ ($0.079$) & $0.739$ ($0.218$) \\
			$0.15$  & $0.585$ ($0.027$) & $0.146$ ($0.285$) & $1$ ($0$)         & $0.746$ ($0.137$) & $0.857$ ($0.155$) \\
			$0.2$   & $0.584$ ($0.020$) & $0.023$ ($0.029$) & $1$ ($0$)         & $0.714$ ($0.140$) & $0.814$ ($0.174$) \\
			\hline              
		\end{tabular}
		\caption{Mean (SD) of the extended NMI from different methods on planted two-block MMBs with $n = 300$, $p_{out}=0$, $n_o = 1$ node of mixed membership, and varying probabilities $p_{in}$.}
		\label{tab:mmbm2-nmi}
	\end{table}
	
	\begin{table}[htb]
		\centering
		\begin{tabular}{c|ccccc}
			\hline
			$p_{in}$    & Louvain     & Spectra      & ORC-E     & ORC-A      & FRC-3      \\ \hline
			$0.05$   & $99.73$ ($4.836$) & $15.58$ ($25.62$)  & $22.33$ ($3.60$) & $19.94$ ($1.327$) & $12.08$ ($0.520$) \\
			$0.1$  & $504.4$ ($39.61$) & $606.3$ ($286.6$) & $106.1$ ($10.14$) & $62.81$ ($4.352$) & $94.77$ ($23.04$)  \\
			$0.15$  & $1221$ ($59.99$) & $3359$ ($330.0$) & $325.9$ ($17.29$) & $139.0$ ($8.562$) & $411.8$ ($115.0$) \\
			$0.2$  & $2129$ ($76.42$) & $9902$ ($249.0$) & $723.0$ ($26.64$) & $267.1$ ($17.63$) & $243.8$ ($378.2$) \\
			\hline
		\end{tabular}
		\caption{Mean (SD) of runtime (in seconds) from different methods on planted two-block MMBs with $n = 300$, $p_{out}=0$, $n_o = 1$ node of mixed membership, and varying probabilities $p_{in}$.}
		\label{tab:mmbm2-time}
	\end{table}
	
	Furthermore, we demonstrate the performance of our algorithms on real data via the two data sets described in sec.~\ref{sec:graph_models}. As in the synthetic networks, the ORC-based method outperforms the reference methods in all four different real networks, while the FRC-based method can also outperform or have comparable performance to the reference methods. The results from the proposed approximation of ORC have more variance, which is expected since the real data has more variance in the structure of the networks built and some may lie in the region where the approximation is relatively far from the exact value, but the approximated version can also outperform or have comparable performance in most cases. Additionally, we also observe an improvement of runtime from obtaining ORC by Sinkhorn optimization in ``FB-2" dataset (ORC-E: $5523$s; ORC-S: $4941$s; ORC-A: $2530$s), since we allocate much larger memory for the computation ($250$G here instead of $25$G for others in this section).  
	\begin{table}[htb]
		\centering
		\begin{tabular}{c|ccccc}
			\hline
			& Louvain  & Spectra  & ORC-E  & ORC-A  & FRC-3  \\ \hline
			DBLP-1  & $0.490$       & $0.077$       & $0.544$     & $0.137$     & $0.330$  \\
			DBLP-2  & $0.191$       & $0.087$       & $0.597$     & $0.424$     & $0.445$  \\
			FB-1    & $0.527$       & $0.078$       & $0.872$     & $0.680$     & $0.680$  \\
			FB-2    & $0.347$       & $0.073$       & $0.652$      & $0.242$     & $0.287$   \\
			\hline              
		\end{tabular}
		\caption{Extended NMI from different methods on the real data.}
		\label{tab:real-nmi}
	\end{table}
	
	\paragraph{Towards avoiding line graph construction.} We note that only the upper bound requires global information of the line graph, while local information is sufficient to compute the lower bound. Hence, one step towards avoiding the construction of line graph is to consider an upper bound that can be obtained without the information of the whole graph while performing sufficiently well in the clustering tasks. Specifically, here we consider $1$ as the upper bound, and denote the method using the average of $1$ and the lower bound by ``ORC-A1". Our numerical results indicate that ORC-A1 can further improve the efficiency from ORC-A while not significantly affecting the accuracy (see Table \ref{tab:mmbm-a1}, where the NMI information is ignored because their values are close and all three methods can find the ground-truth communities when $p_{in}=0.2$).
	\begin{table}[htb]
		\centering
		\begin{tabular}{c|ccc}
			\hline
			$p_{in}$ & ORC-E     & ORC-A     & ORC-A1     \\ \hline
			$0.1$ & $426.8$   & $182.2$   & $144.6$   \\ 
			$0.2$ & $2878$  &  $732.7$  & $530.0$     \\
			\hline
		\end{tabular}
		\caption{Runtime (in seconds, including line graph construction) from different ORC-based methods on one sample of planted two-block MMBs with $n=500$, $p_{out} = 0$, $n_o=1$ node of mixed membership, and varying probabilities $p_{in}$.}
		\label{tab:mmbm-a1}
	\end{table}
	
	\paragraph{Rationale for 2-complex FRC in line graph.} We have shown in sec.~\ref{sec:DC-frc-higher} that FRC-2 in line graph can largely depend on the degree of nodes in the original graph, while the performance of FRC-1 for the task of community detection is limited, hence the corresponding algorithm may not lead to satisfactory results. Our experimental results (see Table \ref{tab:mmbmfrc-nmi}) verifies the above statement, where FRC-2 can have comparable performance to FRC-3 when node degree is moderate, but the performance drops significantly and approaches FRC-1 as the network becomes denser.  
	\begin{table}[htb]
		\centering
		\begin{tabular}{c|ccc}
			\hline
			$p_{in}$ & FRC-1             & FRC-2             & FRC-3             \\ \hline
			$0.1$    & $0.315$ ($0.266$) & $0.816$ ($0.231$) & $0.854$ ($0.184$) \\
			$0.2$    & $0.269$ ($0.254$) & $0.856$ ($0.162$) & $0.895$ ($0.148$) \\
			$0.3$    & $0.328$ ($0.337$) & $0.859$ ($0.183$) & $0.982$ ($0.049$) \\
			$0.4$    & $0.560$ ($0.296$) & $0.673$ ($0.136$) & $0.978$ ($0.048$) \\
			\hline
		\end{tabular}
		\caption{Mean (SD) of the extended NMI from different FRC-based methods on planted two-block MMBs with $n=100$, $p_{out}=0$, $n_o=1$ node of mixed membership, and varying probabilities $p_{in}$.}
		\label{tab:mmbmfrc-nmi}
	\end{table}
	
	\paragraph{Rationale for detecting mixed-membership communities on the line graph.} In the mixed-membership setting, it is still possible to apply algorithms that are designed for the single-membership case, but the performance varies. Our experimental results (see Table \ref{tab:mmbmgl-nmi}) indicate that the improvement from applying the algorithm on the line graph is more significant when the graph is sparser, even while there is only a few nodes of mixed membership, and also when there are more nodes of mixed membership. We should notice that applying algorithms directly on the original graph cannot detect mixed-membership nodes, thus there is an upper bound on the accuracy that these methods can achieve and how far this upper bound is from $1$ depends on the portion of mixed-membership nodes.
	\begin{table}[htb]
		\centering
		\begin{tabular}{cc|ccc}
			\hline
			$p_{in}$ & $n_{o}$ & ORC-E ($G$) / ($L$) & ORC-A ($G$) / ($L$) & FRC-2 ($G$) / FRC-3 ($L$) \\ \hline
			$0.05$   & $1$     & $0.993$ / $0.999$ & $0.470$ / $0.648$ & $0.314$ / $0.679$      \\
			$0.15$   & $1$     & $0.993$ / $1$     & $0.993$ / $0.714$ & $0.784$ / $0.865$       \\
			$0.15$   & $30$    & $0.773$ / $0.969$ & $0.493$ / $0.497$ & $0.176$ / $0.385$       \\
			\hline
		\end{tabular}
		\caption{Mean of the extended NMI from different methods on the original graph $G$ versus its line graph $L$, where the graphs are generated from planted two-block MMBs with $n=300$, $p_{out}=0.01$, and varying probability $p_{in}$ and number of mixed-membership node(s).}
		\label{tab:mmbmgl-nmi}
	\end{table}

	\color{black}
	
	\section{Comparison of Curvature Notions}
	\label{sec:comparison}
	
	In this final section we discuss differences and commonality between the two curvature notions that we considered.  Specifically,  we adress geometric considerations, i.e., which structural features are best captured by which curvature notion and how this is reflected in their performance in clustering tasks.
	
	\subsection{Relationship between 2-complex FRC and ORC.}
	There is already strong numerical~\citep{tannenbaum,samal2018comparative} and theoretical evidence~\citep{jost2021characterizations} of a close relationship between the 2-complex FRC and ORC in the prior literature.  Specifically,~\citet{jost2021characterizations} proved the following relationship between the two notions:
	\begin{theorem}[\citet{jost2021characterizations}]
		\label{thm:jost-muench}
		Olliviers and Forman$^\ast$ curvature for the edges of a graph $G$ are related as
		\begin{equation}
			\Ric_O(e)= \max_K \; \Ric_F^\ast(e) \; ,
		\end{equation}
		where the maximum is taken over all complexes $K$ that have $G$ as a 1-skeleton.
	\end{theorem}
	\begin{rmk}\normalfont
		The Forman$^\ast$ curvature considered by~\citet{jost2021characterizations} deviates from Forman's orginal notion. However, one can recover the original notion as a special case via a specific choice of the cell weights. For details, see~\citet[sec. 7.4]{jost2021characterizations}.
	\end{rmk}
	Theorem~\ref{thm:jost-muench} implies that FRC and ORC coincide on the edges,  when maximizing FRC over the choice of 2-cells whose contributions are included in curvature computation.  The proof of this result proposes a means to ``translate'' between the contributions of 2-cells and the cost of transport plans. 
	
	Our numerical results (Fig.~\ref{fig:clust-curv_v}) confirm that the correlation between FRC and ORC becomes stronger, if FRC is ``augmented", i.e., if the contributions of higher-order structures (triangles, quadrangles) are taking into account for the curvature computation. In this paper, we considered a 2-complex notion of FRC in the original graph $G$ (previously considered in~\citep{WSJ2}) and ``approximate'' the contribution of $k$-faces, by including curvature contribution of faces up to order $K$ (see Figs.~\ref{fig:exp-frcorc-L-30}, \ref{fig:exp-frcorc-L-60} and \ref{fig:exp-frcorc-L-90} in appendix \ref{app:exp-msbm}).
	Since the frequency of $k$-faces decreases sharply as $k$ increases (illustrated in Fig.~\ref{fig:k-faces}), it is enough to consider small values of $K$. In our application to mixed-membership community detection, we considered $K=4$. We found that this parameter choice balances fast computation (complexity increases with $K$) and performance (accuracy increases with $K$). Our numerical experiments confirmed that the performance of the FRC-based community detection algorithm closer resembles that of the ORC-based community detection algorithm as $K$ increases.
	
	We have previously noticed in Figs.~\ref{fig:line-curv_e-v} and \ref{fig:line-curv_e-e} that the range of curvature values varies between ORC and classical FRC, while the shapes of the distributions resemble each other. Notice that the range of curvature values between the 2-complex FRC and ORC is much better aligned, which is in agreement with their much more similar performance in downstream tasks.
	
	\subsection{Differences in structural features captured by FRC and ORC.}
	
	\subsubsection{Relationship with clustering coefficient}
	\label{sec:clustering-coeff}
	The \emph{clustering coefficient}, initially introduced by~\citet{watts-strogatz}, is a network-theoretic measure, which is correlated with the ability of clustering algorithms to detect communities.  Formally,  it measures how close a node neighborhood is to a clique (i.e., a fully connected subgraph).  It is defined as the ratio of the number of edges between neighbors of a node $u \in V$ and the number of edges in the clique, if the neighborhood of $u$ were fully connected:
	\begin{defn}[Clustering coefficient]
		Let $A=(a_{ij})_{1 \leq i,j, \leq n}$ denote the adjacency of a graph $G$ and $u \in V$ the node whose connectivity is characterized by the $i$th row in the adjacency. Then
		\begin{equation}
			C(u):= \frac{\sum_{v \sim u} \#(u,v)}{d_u (d_u -1)} = \frac{\sum_{j,k} a_{ij} a_{jk} a_{ki}}{d_u (d_u -1)} \; .
		\end{equation}
	\end{defn}
	It has been previously observed empirically~\citep{tannenbaum,samal2018comparative} that node-level ORC is highly correlated with the clustering coefficient. FRC is also correlated with the clustering coefficient, although the degree of correlation depends significantly on the type of higher-order structures that is included (see Fig.~\ref{fig:clust-curv_v}). We see that FRC correlates best with the clustering coefficient, if triangle contributions are included in the curvature computation, i.e., for our notion of 2-complex FRC. This is also the FRC notion that resembles ORC the closest. Importantly, the close relationship between ORC and the clustering coefficient can also be established theoretically. In particular,~\citet{jost2014ollivier} showed the following result:
	\begin{theorem}[\citet{jost2014ollivier}, Cor. 1]
		\begin{equation*}
			\frac{d_u -1}{d_u} C(u) \geq \Ric_O(u) \geq -2 + \frac{d_u -1}{d_u \vee \max_{v \sim u} d_v} C(u)\; .
		\end{equation*}
	\end{theorem}
	\noindent To the best of our knowledge, no theoretical connection between FRC and the clustering coefficient is known.
	
	\subsubsection{Node clustering vs. edge clustering.}
	Notice that Alg.~\ref{alg.1} when applied to the underlying graph $G$ performs a \emph{node clustering}, whereas it performs an \emph{edge clustering} when applied to the line graph $L(G)$. Specifically, if Alg.~\ref{alg.1} is applied to the line graph, then we obtain a mixed-membership label vector $y$ for each vertex $v$ by computing $y_l(v)=\frac{1}{|E_v|}\sum_{e\in E_v} \chi_{C_l}(e)$, where $\chi_{C_l}$ is the indicator for the cluster $C_l$. This assignment is based on an \emph{edge clustering} in the line graph, i.e., in Eq.~\eqref{eq:labeling} above we assume that edge clusters correspond to node clusters. Hence, our approach solves single-membership community detection via node clustering and mixed-membership community detection via edge clustering. Both clustering approaches utilize the implicit assumption that the underlying graph is \emph{homophilic}, i.e., similar nodes are more densely connected than dissimilar nodes. Node clustering identifies clusters of nodes by finding densely connected subgraphs. We have seen above that dense connectivity increases curvature; consequently, our curvature-based algorithm finds clusters by identifying regions of low curvature. In contrast, edge clustering identifies groups of edges that connect a set of similar nodes. Due to the higher connectivity among sets of similar nodes, those edges have a higher number of neighboring edges (i.e., edges that are adjacent to the same node). Consequently, those edges form a densely connected subgraph in the line graph. Applying curvature-based clustering on the line graph allows us to identify these densely connected subgraphs, due to their higher curvature. Throughout the paper, we have provided empirical and theoretical evidence that all curvature notions considered here have higher value in densely connected subgraphs of any graph. However, our study has demonstrated that the difference of curvature values of edges between communities, which is directly linked to the performance of curvature-based clustering, is more pronounced for some notions. In particular, our analysis shows that ORC highlights this effect best. Among the different FRC variants studied here, we showed that incorporating higher-order faces (triangles in the original graph, triangles and quadrangles in the line graph) leads to the best results. 
	The effect of Alg.~\ref{alg.1} may also be interpreted as a \emph{graph coarsening} scheme. When applied to a graph $G$, our approach uncovers meso-scale structure, such as communities. The interpretation of Alg.~\ref{alg.1} as coarsening the line graph is perhaps less intuitive. Nevertheless, following again the homophily perspective above, we see that curvature-based edge clustering groups edges that connect similar nodes, uncovering crucial meso-scale structure in the original graph. 
	
	\subsubsection{Computational Considerations.} 
	An often cited argument in the previous literature is that variants of 2-complex FRC are preferable over ORC in applications, due to their faster computation. In particular, the known links between the spectrum of the graph Laplacian and clustering coefficients~\citep{jost2014ollivier} mentioned above make ORC a natural choice for curvature-based clustering. 
	Previous approaches have either computed ORC exactly via the Hungarian method ($O(\vert E \vert^3)$) or approximated it via Sinkhorn's method ($O(\vert E \vert^2$), both of which are significantly more expensive than computing Forman's curvature ($O(\vert E \vert)$). We propose above an approximation via the arithmetic mean of upper and lower curvature bounds. We illustrate above that this approximation is on par or, in certain instances even superior, to the Sinkhorn approximation. Importantly, our proposed approximation can be computed in $O(\vert E \vert)$, i.e., has complexity that is comparable to that of FRC. Hence, computational considerations do not provide a clear preference for our approximate ORC approach over FRC. In addition, we propose above a variant of Alg.~\ref{alg.1} via ORC, which does not require the computationally expensive construction of the line graph.
	
	\section{Discussion}
	In this paper, we have proposed a unifying framework for curvature-based node clustering algorithms and systematically investigated the strength and weakness of different curvature notions, specifically variants of Forman’s and Ollivier’s Ricci curvature. In addition to single-membership community structure, which is the focus of most existing work, we have also considered the mixed-membership setting and extended curvature-based clustering approaches to this case. To this end, we have further studied the discrete curvature on the line graph, and related this to the curvature on the original. Furthermore, we have proposed an effective approximation of Ollivier’s notion of discrete curvature to overcome its scalability issues. This construction may be of independent interest. We emphasize that the results derived in this paper provide insights into the relationship of discrete curvatures on a graph and that of its dual, which may be of independent interest.
	
	\noindent Throughout the experiments, we observe consistently convincing performance of the curvature-based methods. Specifically, the ORC-based approach can outperform classical community detection methods, the Louvain algorithm and Spectral Clustering, in both synthetic and real networks, to detect either single- or mixed-membership community structure. Although the experimental results indicate quadratic complexity of the ORC-based method, the proposed approximation's performance is comparable to classical methods and, in relatively large or dense graphs, even computationally superior. The results agree with the linear complexity of the approximated ORC-based approach, as demonstrated in Sec.~\ref{sec:orc-comp}. The FRC-based approach, which has a lower complexity, can outperform Spectral Clustering in most cases, and when the graph is sufficiently large or dense, it can also outperform or have comparable performance to the Louvain algorithm. Note that we consider different variants of Forman’s Ricci curvature for single- and mixed-membership node clustering problems, where we only consider triangular faces in the former while incorporating the quadrangular faces in the latter. We also observe that FRC with triangular faces has a stronger relationship with the clustering coefficients in SBMs, while FRC with further quadrangular faces has a slightly stronger relationship with the ORC in MMBs. As byproduct, the results provide some experimental evidence to relate Forman’s to Ollivier’s Ricci curvature. 
	
	\noindent Several extensions of this work could be avenues for future investigation: In the variant of Alg.~\ref{alg.1} presented here, several crucial hyperparameters are chosen heuristically. Specifically, we believe that the choice of cut-off points could be improved, potentially utilizing geometric information, such as \emph{curvature gaps}~\citep{gosztolai2021unfolding}. The choice of the stopping criterion ($T$) could be guided by established heuristics in the community detection literature. For instance, the Girvan-Newman algorithm constructs a dendrogram to inform the stopping criterion and edge threshold $\Delta$, when the number of communities is not known a priori. Our mixed-membership approaches require the construction of the line graph, which adds significant computational cost when applied to large-scale graphs. Approaches that avoid the construction of the line graph, perhaps utilizing the relationship between the Ricci curvature of a graph and its dual, could remove this bottleneck. We have also seen that using $1$ as the upper bound for ORC could already improve the efficiency in our experiments. These approximations may prove computationally beneficial in both the mixed-membership and single-membership settings.  While Alg.~\ref{alg.1} allows for clustering in both unweighted and weighted networks, it cannot be applied to directed networks in its current form. An extension to directed networks would significantly widen the range of possible applications. Further study of the behavior of discrete curvatures in various random graph models could serve to improve both theoretical understanding of these graphs, as well as to inform the development of improved algorithms for certain graph distributions.

\section*{Acknowledgements}
Part of the computations in this paper were run on the FASRC Cannon cluster supported by the FAS Division of Science Research Computing Group at Harvard University. Part of the computations were enabled by resources provided by the Swedish National Infrastructure for Computing (SNIC) at the PDC Center for High Performance Computing, KTH Royal Institute of Technology, partially funded by the Swedish Research Council through grant agreement no.~2018-05973. We would also like to acknowledge the use of the University of Oxford Mathematical Institute compute facility in carrying out the preliminary work.
		
YT is funded by the Wallenberg Initiative on Networks and Quantum Information (WINQ). Part of the work was done when YT was at Mathematical Institute, University of Oxford, where she was funded by the EPSRC Centre for Doctoral Training in Industrially Focused Mathematical Modelling (EP/L015803/1) in collaboration with Tesco PLC. 
Part of this work was done while MW was at the University of Oxford, supported by a Hooke Research Fellowship.
Over the course of this work, ZL has been funded under NIH grant MH-19-147, NSF HDR TRIPODS grant 1934979, and NSF DMS grant 2113099. He would also like to acknowledge support from the Naval Engineering Education Consortium and Microsoft Research.

	
	\newpage
	
	\appendix
	\section{ORC on weighted graphs}
	
	\subsection{Proofs of Results in Section~\ref{sec:orc}}
	
	\textbf{Proof of Lemma~\ref{lem:up}:}
	Consider the subgraph of $G$ given represented schematically in Figure~\ref{fig:schema1}. 
	
	\begin{figure}[h]
		\centering
		\begin{tikzpicture}
			\path[fill=blue!30] plot[smooth cycle] coordinates 
			{(0,-0.4) (4.2,-0.4) (4.2,0.4) (0.8,0.4) (2.3,1.7) (1.9,2.4) (0,0.4) (-1.8,2.4) (-2.4,1.6)};
			\path[fill=red!30,opacity=0.7] plot[smooth cycle] coordinates 
			{(-0.2,-0.4) (4,-0.4) (6.4,1.6) (5.8,2.4) (4,0.4) (2.1,2.4) (1.7,1.7) (3.2,0.4) (-0.2,0.4)};
			
			\node [shape=circle,draw=black] (x) at (0,0) {$x$};
			\node [shape=circle,draw=black] (y) at (4,0) {$y$};
			\node [shape=circle,draw=black] (l) at (-2,2) {$\ell$};
			\node [shape=circle,draw=black] (c) at (2,2) {$c$};
			\node [shape=circle,draw=black] (r) at (6,2) {$r$};
			
			\draw[thick,black] (x) -- (l);
			\draw[thick,black] (x) -- (c);
			\draw[thick,black] (c) -- (y);
			\draw[thick,black] (y) -- (r);
			\draw[thick,black] (x) -- (y);
		\end{tikzpicture}
		\caption{Schematic of the subgraph $H$ of $G$, focusing on $x,y$ and their neighbors. Note that typically there are several nodes $\ell_1,\ldots,\ell_{j_1}$, $c_1,\ldots,c_{j_2}$, and $r_1,\ldots,r_{j_3}$, satisfying $\mathrm{deg}(x)=j_1+j_2+1, \mathrm{deg}(y)=j_2+j_3+1.$ The support of the $m_x$ measure is shown in blue, $m_y$ in red.}
		\label{fig:schema1}
	\end{figure}
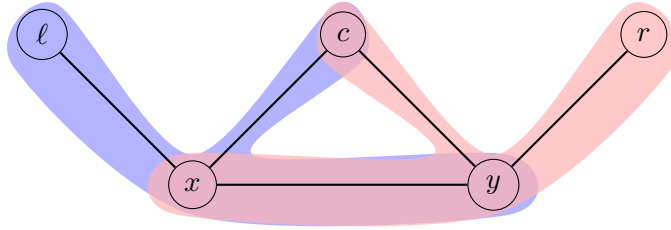
	
	Now define the following quantities:
	\begin{align*}
		L_x&= \sum_\ell m_x(\ell) & L_y &= \sum_\ell m_y(\ell)\; (=0)\\
		X_x&= m_x(x) & X_y &= m_y(x)\\
		C_x&= \sum_c m_x(c) & C_y &= \sum_c m_y(c)\\
		Y_x&= m_x(y) & Y_y &= m_y(y)\\
		R_x&= \sum_r m_x(r)\; (=0) & R_y &= \sum_r m_y(r).
	\end{align*}
	If we consider transport plans restricted to $H$, then all of the $L_x$ mass must be transported to $x$, and all of the $R_y$ mass must be transported from $y$. As such, we may consider a modified diagram as in Figure~\ref{fig:schema2}.
	
	\begin{figure}[h]
		\centering
		\begin{tikzpicture}
			\path[fill=blue!30] plot[smooth cycle] coordinates 
			{(-0.2,-0.4) (4.2,-0.4) (4.2,0.4) (0.8,0.4) (2.3,1.7) (1.9,2.4) (-0.2,0.4)};
			\path[fill=red!30,opacity=0.7] plot[smooth cycle] coordinates 
			{(-0.2,-0.4) (4.2,-0.4) (4.2,0.4) (2.1,2.4) (1.7,1.7) (3.2,0.4) (-0.2,0.4)};
			
			\node [shape=circle,draw=black] (x) at (0,0) {$x$};
			\node [shape=circle,draw=black] (y) at (4,0) {$y$};
			\node [shape=circle,draw=black] (c) at (2,2) {$c$};
			
			\draw[thick,black] (x) -- (c);
			\draw[thick,black] (c) -- (y);
			\draw[thick,black] (x) -- (y);
		\end{tikzpicture}
		\caption{Schematic of the subgraph $H'$ of $H$, focusing on $x,y$ and their common neighbors. The support of the $\tilde{m}_x$ measure is shown in blue, $\tilde{m}_y$ in red.}
		\label{fig:schema2}
	\end{figure}
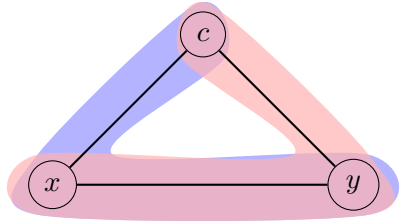
	
	Here we define $\tilde{m}_x, \tilde{m}_y$ as follows:
	$$ \tilde{m}_x(v) =\begin{cases} L_x+X_x &\text{if }v=x\\ m_x(c)&\text{if }v=c\\ Y_x+R_x\; (=Y_x) &\text{if }v=y,\end{cases}\quad \tilde{m}_y(v) =\begin{cases} L_y+X_y\;(=X_y) &\text{if }v=x\\ m_y(c)&\text{if }v=c\\ Y_y+R_y &\text{if }v=y.\end{cases}$$
	
	Then we have $$W_1(m_x,m_y) \leq \sum_\ell \omega(\ell,x) m_x(\ell) + \sum_r \omega(r,y) m_y(r) + W_1(\tilde{m}_x,\tilde{m}_y),$$ and this is tight when $x,y$ have no unshared neighbors. Indeed, if we only consider the subgraph $H$ of $G$, this upper bound is the exact value of $W_1(m_x,m_y)$. Define $\tilde{X}_x, \tilde{C}_x, \tilde{Y}_x, \tilde{X}_y, \tilde{C}_y, \tilde{Y}_y$ analogously for these new measures.
	
	Now for each vertex of type $c$, if $\tilde{m}_x(c)-\tilde{m}_y(c)>0$, we will send this amount to $y$ along $(c,y)$. Otherwise, we will send $\tilde{m}_y(c)-\tilde{m}_x(c)$ from $x$ along $(c,x)$. This results in the updated (signed) measure
	$$\hat{m}_x(v)= \begin{cases}\tilde{X}_x - \sum_{c} (\tilde{m}_y(c)-\tilde{m}_x(c))_+&\text{if }v=x\\ \tilde{m}_y(c)&\text{if }v=c\\ \tilde{Y}_x+\sum_c (\tilde{m}_x(c)-\tilde{m}_y(c))_+&\text{if }v=y.\end{cases}$$
	Finally, if $\hat{m}_x(x)-\tilde{X}_y>0$, send this amount to $y$ along the edge $(x,y)$. Otherwise, send $\hat{m}_x(y)-\tilde{Y}_y= -(\hat{m}_x(x)-\tilde{X}_y)$ from $y$ to $x$ along the $(x,y)$ edge.
	
	This gives the upper bound
	\begin{align*}
		W_1(\tilde{m}_x,\tilde{m}_y)&\leq \sum_c (\tilde{m}_x(c)-\tilde{m}_y(c))_+\omega(c,y)+(\tilde{m}_y(c)-\tilde{m}_x(c))_+\omega(c,x)\\
		&+\left|\tilde{X}_x-\tilde{X}_y -\sum_c (\tilde{m}_y(c)-\tilde{m}_x(c))_+\right|\omega(x,y).
	\end{align*}
	
	Finally, we have 
	\begin{align*}
		W_1(m_x,m_y)&\leq \sum_\ell \omega(\ell,x) m_x(\ell)+\sum_r \omega(r,y) m_y(r)\\
		&\quad+\sum_c \omega(c,y)(m_x(c)-m_y(c))_+ +\omega(c,x)(m_y(c)-m_x(c))_+\\
		&\quad+ \left|L_x+X_x - X_y - \sum_c (m_y(c)-m_x(c))_+\right|\omega(x,y)=U_{x,y}.
	\end{align*}
	\hfill$\square$
	
	For completeness, we state and prove the following lower bound on the Wasserstein distance:
	\begin{lem}
		\label{lem:wasslowerbound}
		For any $1$-Lipschitz function $h(v)$ on $G$, we have
		\begin{equation*}
			W_1(m_x,m_y)\geq \EE_{v\sim m_x}[h(v)] - \EE_{v\sim m_y}[h(v)]= \sum_{v\in N(x)\cup N(y)} h(v) (m_x(v)-m_y(v)) \; .
		\end{equation*}
	\end{lem}
	\begin{proof}
		Given a coupling $\pi$ for $m_x$ and $m_y$, and a 1-Lipschitz function $h:V\rightarrow\RR$, we have
		\begin{align*}
			\EE_{(u,w)\sim\pi}[d_G(u,w)]&=  \sum_{\substack{u\sim N(x)\\ w\sim N(y)}} d_G(u,w) \pi(u,w)\\
			&\geq \sum_{\substack{u\sim N(x)\\w\sim N(y)}} (h(u)-h(w))\pi(u,w)\\
			&=\sum_{v\in N(x)\cup N(y)} h(v) (m_x(v)-m_y(v)). 
		\end{align*}
		Since the choice of coupling was arbitrary, we conclude that $W_1(m_x,m_y)$ satisfies the same lower bound.
	\end{proof}
	
	Now we may prove Lemma~\ref{lem:low}.\\
	
	\textbf{Proof of Lemma~\ref{lem:low}:}
	Let $V_H= N(x)\cup N(y)$ be the set of vertices from the subgraph $H$ of $G$ that we considered in the upper bound for $W_1(m_x,m_y)$. A difficulty that arises in defining a 1-Lipschitz function $h$ on $G$ is that there are more edges and paths in $G$ than just those in $H$. These ``short-circuits" in the larger graph mean that a 1-Lipschitz function defined on $H$ does not necessarily extend to a 1-Lipschitz function on $G$. We define a partition of $V_H$ into three subsets $\mathcal{P}, \mathcal{Z}, \mathcal{N}$ based on whether $m_x(v)-m_y(v)$ is Positive, Zero, or Negative. Observe that for a vertex $v\in\mathcal{P}$, the excess mass $m_x(v)-m_y(v)$ must be transported away from $v$ in order to have a valid transport plan. In fact, it must be transported at least enough to reach $\mathcal{N}$, since 
	$$0=\sum_{v\in V} m_x(v) - m_y(v)= \sum_{v\in\mathcal{P}} (m_x(v)-m_y(v)) + \sum_{v\in\mathcal{N}} (m_x(v)-m_y(v))$$ implies that $m_x(\mathcal{P})-m_y(\mathcal{P})= m_y(\mathcal{N})-m_x(\mathcal{N}).$ In other words, all excess mass from $m_x$ in $\mathcal{P}$ must be transported to vertices in $\mathcal{N}$, or we cannot have a valid transport plan. 
	
	Motivated by this observation, we define the distance from a vertex $v\in V$ to a set of vertices $S\subseteq V$ as $d_G(v,S)= \min_{u\in S} d_G(v,u)$, where $d_G$ is the shortest path distance in $G$. Thus, $d_G(v,S)$ gives the distance in $G$ from the vertex $v$ to the nearest element in the set $S$. 
	
	Consider the following 1-Lipschitz functions parameterized by $\lambda\in[0,1]$: 
	\begin{equation}
		h_\lambda(v)= \begin{cases}
			\lambda d_G(v,\mathcal{N})&\text{if }v\in \mathcal{P}\\
			-(1-\lambda) d_G(v,\mathcal{P})&\text{if }v\in\mathcal{N}\\
			\lambda d_G(v,\mathcal{N})-(1-\lambda)d_G(v,\mathcal{P})&\text{otherwise.}
		\end{cases}
	\end{equation}
	Note that the third equation works for all $v\in V$, since $d_G(v,\mathcal{P})=0$ for $v\in\mathcal{P}$ and similarly for $\mathcal{N}$, but we split this into cases to show the behavior on the different sets more clearly.
	
	These functions are all 1-Lipschitz since they are convex combinations of the two 1-Lipschitz functions $d_G(\cdot,\mathcal{N})$ and $-d_G(\cdot,\mathcal{P})$. Let's show that these are 1-Lipschitz: Let $S\subseteq V$, and let $u,v\in V$. Let $w_u,w_v\in S$ be the nearest vertices in $S$ to $u,v$ respectively, so that $d_G(u,w_u)=d_G(u,S)$ and similarly for $w_v$. Then
	\begin{align*}
		d_G(v,S)&\leq d_G(v,w_u) \leq d_G(v,u)+d_G(u,w_u)\\
		&=d_G(u,v)+d_G(u,S)\\
		d_G(u,S)&\leq d_G(u,w_v)\leq d_G(u,v)+d_G(v,w_v)\\
		&=d_G(u,v)+d_G(v,S)\\
		\Longrightarrow\quad & |d_G(v,S)-d_G(u,S)|\leq d_G(u,v).
	\end{align*}
	
	But since the objective function $\EE_{v\sim m_x} [h(v)] - \EE_{v\sim m_y} [h(v)]$ is linear in $h$, we see that the maximizer of this quantity over $\{h_\lambda: \lambda\in[0,1]\}$ must be achieved at $h_0$ or $h_1$. This gives the following lower bound for $W_1(m_x,m_y):$
	
	\begin{align*}
		W_1(m_x,m_y)&\geq \max_{\lambda\in\{0,1\}}\sum_{v\in V_H} h_\lambda(v)(m_x(v)-m_y(v))\\
		&= \sum_{v\in \mathcal{P}} h_\lambda(v) \underbrace{(m_x(v)-m_y(v))}_{>0}+ \sum_{v\in \mathcal{Z}} h_\lambda(v) \underbrace{(m_x(v)-m_y(v))}_{=0}\\
		&\quad + \sum_{v\in \mathcal{N}} h_\lambda(v) \underbrace{(m_x(v)-m_y(v))}_{< 0}\\
		&= \max\left\{ \sum_{v\in\mathcal{P}} d_G(v,\mathcal{N})(m_x(v)-m_y(v)), \sum_{v\in\mathcal{N}} d_G(v,\mathcal{P})(m_y(v)-m_x(v))\right\}.
	\end{align*}
	\hfill$\square$
	
	We note that vertices of type $\ell$ will always belong to $\mathcal{P}$, and vertices of type $r$ will always belong to $\mathcal{N}$. Vertices of type $c$ may fall into either set, and there will typically be some of these vertices in both sets. When we set $\alpha=0$, then we always have $x\in\mathcal{N}$ and $y\in\mathcal{P}$, though if we choose $\alpha\geq 1/2$, $x$ cannot be in $\mathcal{N}$, and $y$ cannot be in $\mathcal{P}$.

	\subsection{Note on combinatorial ORC bound by Jost-Liu}
	\begin{rmk}\label{rmk:jost-liu-issue}
		We note that the lower bound for $W_1(m_x,m_y)$ given in \cite{jost2014ollivier} for weighted graphs using the measure $m_x(v)= w_{vx}/d_x$ is actually incorrect. Consider the weighted graph $G$ shown in Figure~\ref{fig:counterexample}, which also summarizes the measures $m_x, m_y$.
		\begin{figure}[h]
			\begin{subfigure}{0.48\textwidth}
				\centering
				\begin{tikzpicture}
					\node [circle, draw] (l) at (-2,2) {$\ell$};
					\node [circle, draw] (x) at (0,0) {$x$};
					\node [circle, draw] (c) at (2,2) {$c$};
					\node [circle, draw] (y) at (4,0) {$y$};
					
					\draw (l) -- (x) node[pos=0.5,fill=white] {3};
					\draw (x) -- (c) node[pos=0.5,fill=white] {2};
					\draw (c) -- (y) node[pos=0.5,fill=white] {2};
					\draw (l) -- (c) node[pos=0.5,fill=white] {1/5};
					\draw (x) -- (y) node[pos=0.5,fill=white] {1};
				\end{tikzpicture}
			\end{subfigure}
			\begin{subfigure}{0.48\textwidth}
				\centering
				\begin{tabular}{c|c|c|c|c}
					& $x$ & $y$ & $\ell$ & $c$  \\
					\hline
					$m_x$ & 0 & 1/6 & 1/2 & 1/3\\
					\hline
					$m_y$ & 1/3 & 0 & 0 & 2/3\\
					\hline
					$h$ & -6/5 & -1/5 & 1 & 4/5
				\end{tabular}
			\end{subfigure}
			\caption{The weighted graph $G$, the measures $m_x, m_y$ from \cite{jost2014ollivier}, and our 1-Lipschitz function $h$.}
			\label{fig:counterexample}
		\end{figure}
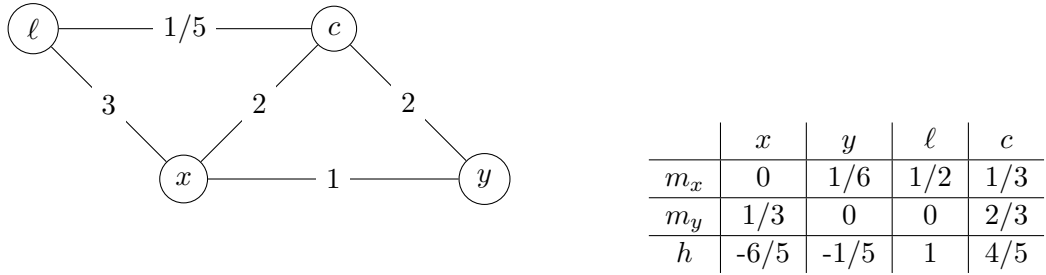
		An optimal transport plan sends 1/6 mass from $y$ to $x$ and from $c$ to $x$, and 1/2 mass from $\ell$ to $c$, for a cost of 3/5. We can verify the optimality of this transport plan with the 1-Lipschitz function $h$ given in the table in Figure~\ref{fig:counterexample}, since this also achieves the value 3/5. As such, the exact ORC of the edge $\{x,y\}$ in $G$ is $1-3/5=2/5.$ However, the upper bound on this ORC from \cite{jost2014ollivier} is 
		$$\frac{w_{cx}}{d_x} \wedge \frac{w_{cy}}{d_y} = \frac{1}{3} \wedge \frac{2}{3} = \frac{1}{3}< \frac{2}{5}.$$
	\end{rmk}
	
	\subsection{Proof of Theorem~\ref{thm:ORC-guarantee}}
	\label{sec:ORC-guarantee-proof}
	
	\begin{proof}
		Due to the homogeneity of the graphs $G_{a,b}$, which is preserved under the Ricci flow, it is sufficient to analyze the evolution of one edge of types (1), (2), (3).
		In the following, $d_i, D_i,w_i^t,\kappa_i^t$ denote the shortest-path distance, Wasserstein-1 distance, weight in iteration $t$, and ORC curvature of an edge of type ($i$) in iteration $t$, respectively.  Notice that here
		\begin{equation}
			\label{eq:flowweights}
			w_i^{t+1}=(1 - \kappa_i^t)w_i^t=w_i^t - \Big(1- \frac{D_i}{d_i} \Big)w_i^t
			= D_i \; ,
		\end{equation}
		since $d_i=w_i^t$. Throughout the proof, $x,y$ denote the vertices adjacent to the edge under consideration.
		
		\paragraph{Iteration 1.}	In the first iteration, all edges have the same weight (i.e., $1$), hence,  Eq.~\eqref{eq:G_ab-mass} reduces to the uniform distribution.  
		\begin{itemize}
			\item Edges of type (3) have symmetric neighborhoods,  with uniform mass distributions of $m_v = \frac{1}{a}$ on all neighbors of $x,y$.  Thus, the total transportation cost of moving mass from $x$'s neighbors to $y$'s neighbors is $D_3=\frac{1}{a}d_3$ and therefore $w_3^1=D_3=\frac{1}{a}$, since $d_3=1$.
			\item For edges of type (1),  note that $d_x=d_y=a+(b-1)$, hence $m_v=\frac{1}{a+b-1}$ for all neighbors of $x,y$.  By construction, $x,y$ have $b-2$ common neighbors; due to the symmetry of the neighborhoods,  no mass has to be moved for those neighbors.  The optimal transportation plan for the remaining mass is as follows: Move the remaining mass $\frac{a}{a+b-1}$ from the neighbors of $x$ (which are not neighbors of $y$) to $x$.  This incurs a cost of $\frac{a}{a+b-1} d_2$, since this mass is transported along edges of type (2).  Leave a mass of $\frac{1}{a+b-1}$ on $x$ and transport the remainder to $y$, which incurs a cost of $\frac{a-1}{a+b-1} d_1$, since it is transported along an edge of type (1).  Lastly,  we distribute mass $\frac{a}{a+b-1}$ from $y$ uniformly to its neighbors, which incurs a cost of $\frac{a}{a+b-1} d_2$, since it is transported along edges of type (2). In total,  the transportation cost amounts to
			\begin{equation*}
				D_1 = \frac{2a}{a+b-1} d_2 + \frac{a-1}{a+b-1} d_1 \; ,
			\end{equation*}
			To prove that this is an optimal flow, consider the following 1-Lipschitz function $h$:
			$$
			h(v)=
			\begin{cases}
				2& \text{if } v\sim x, v\not\sim y\\
				1&\text{if }v=x\\
				0&\text{if }v=y\text{ or }v\sim x,y\\
				-1&\text{if }v\sim y, v\not\sim x.
			\end{cases}
			$$
			Then $$\EE_{v\sim m_x}[h(v)]- \EE_{v\sim m_y}[h(v)]=2\frac{a}{a+b-1}+1\frac{(-1)}{a+b-1}-1\frac{(-a)}{a+b-1}=\frac{3a-1}{a+b-1}. $$
			This implies $w_3^1=D_3=\frac{3a-1}{a+b-1}$, since $d_1=d_2=1$.
			\item  For edges of type (2), let $x$ denote the internal node and $y$ the bridge node. By construction, $x,y$ have $a-1$ common neighbors (internal nodes). We want to leave a mass of $\frac{1}{a+b-1}$ on each of those nodes (corresponding to the mass distribution in the neighborhood of $y$), and transport the excess mass to $x$ and the unshared neighbors of $y$ in an optimal way. Observe that $\frac{1}{a}>\frac{1}{a+b-1}$, so we have a total of $(a-1)\left(\frac{1}{a}-\frac{1}{a+b-1}\right)$ excess mass to move. If this exceeds $\frac{1}{a+b-1},$ we send that amount to $x$, and the remaining amount to $y$; otherwise we send all of this mass to $x$. We can simplify this condition as follows:
			$$
			(a-1)\left(\frac{1}{a}-\frac{1}{a+b-1}\right)-\frac{1}{a+b-1}=\frac{(a-1)(b-1)-a}{a(a+b-1)}\geq0
			$$
			if and only if $(a-1)(b-1)-a= a(b-2)-(b-1)\geq0$. When $b=2$, this never holds, and when $b\geq 3$, this always holds, since $a(b-2)-(b-1)\geq a-b+1\geq0$, given that $a\geq b$. Now let us consider these cases in turn.\\
			
			\underline{Case $b=2$:} We send all of the excess mass from the common vertices to $x$, amounting to $\frac{a-1}{a(a+1)}<\frac{1}{a+1}$, meaning that we must send an additional $\frac{1}{a(a+1)}$ mass from $y$ to $x$, and the remaining $\frac{1}{a+1}$ is sent to $y$'s single unshared neighbor. The total cost is 
			$$D_2 = \frac{a-1}{a(a+1)}d_3+ \frac{1}{a(a+1)}d_2+\frac{1}{a+1}d_1.$$ 
			To see that this is an optimal flow, consider the following 1-Lipschitz function $h$:
			$$
			h(v)= \begin{cases}
				1&\text{if }v\sim x,y\text{ or }v=y\\
				0&\text{otherwise}.
			\end{cases}
			$$
			This achieves a dual objective function value of $\frac{2}{a+1}$, so $w_2^1=D_2= \frac{2}{a+1},$ since $d_1=d_2=d_3=1.$
			\\
			
			\underline{Case $b\geq 3$:} We send $\frac{1}{a+b-1}$ of the excess mass from the common vertices to $x$, and the remainder $\frac{(a-1)(b-1)-a}{a(a+b-1)}$ to $y$. Since there is already a mass of $\frac{1}{a}$ at $y$, this gives a total mass of $\frac{(a-1)(b-1)-a+(a+b-1)}{a(a+b-1)}=\frac{b-1}{a+b-1} $ to be sent to $y$'s unshared neighbors, as required. The total cost is
			$$D_2 = \frac{1}{a+b-1}d_3+ \frac{(a-1)(b-1)-a}{a(a+b-1)}d_2+\frac{b-1}{a+b-1}d_1.$$
			To see that this is an optimal flow, consider the following 1-Lipschitz function $h$:
			$$
			h(v) = \begin{cases}
				1&\text{if }v\sim x,y\\
				0&\text{if }v=x,y\\
				-1&\text{if }v\sim y, v\not\sim x.
			\end{cases}
			$$
			This achieves a dual objective function value 
			$$\EE_{v\sim m_x}[h(v)]-\EE_{v\sim m_y}[h(v)]= 1\frac{(a-1)(b-1)}{a(a+b-1)}-1\frac{-(b-1)}{a+b-1}=\frac{(2a-1)(b-1)}{a(a+b-1)},$$
			which gives $w_2^1 = D_2 = \frac{(2a-1)(b-1)}{a(a+b-1)}=\frac{2a-1}{a+b-1}\left(\frac{b-1}{a}\right)$, since $d_1=d_2=d_3=1.$
		\end{itemize}
		
		We summarize these calculations in the following table:
		\begin{center}
			\begin{tabular}{c|c|c|c}
				Type& 1 & 2 & 3\\
				\hline
				$b=2$ & $\frac{3a-1}{a+1}$ & $\frac{2}{a+1}$ & $\frac{1}{a}$\\
				\hline
				$b\geq 3$ & $\frac{3a-1}{a+b-1}$ & $\frac{(2a-1)(b-1)}{a(a+b-1)}$ & $\frac{1}{a}$
			\end{tabular}
		\end{center}
		
		In both cases, it is easy to check that $w_1^1>w_2^1>w_3^1$, using $a\geq b\geq 2$.  
		
		\paragraph{Iteration t+1.} In round $t\geq1$, the mass distributions on the neighborhood of a bridge node are given by
		\begin{equation*}
			m_x(y) = \begin{cases}
				\frac{1}{C_b} e^{-w_1^{t}}, &y \; {\rm is \; bridge \; node} \\
				\frac{1}{C_b} e^{-w_2^{t}}, &{\rm else} 
			\end{cases} \; ,
		\end{equation*}
		and for an internal node by
		\begin{equation*}
			m_x(y) = \begin{cases}
				\frac{1}{C_{in}} e^{-w_2^{t}}, &y \; {\rm is \; bridge \; node} \\
				\frac{1}{C_{in}} e^{-w_3^{t}}, &{\rm else}
			\end{cases} \; .
		\end{equation*}
		Here the normalizing constants for bridge nodes is given by
		\begin{equation*}
			C_b = a e^{-w_2^{t}} + (b-1) e^{-w_1^{t}} \; ,
		\end{equation*}
		and for internal nodes as
		\begin{equation*}
			C_{in}=e^{-w_2^{t}} + (a-1) e^{-w_3^{t}} \; .
		\end{equation*}
		By way of induction, we suppose that $w_1^t>w_2^t$ and $w_1^t>w_3^t$, where we proved this in the case $t=1$ above. 
		
		We again consider the three edge types separately.  Given $w_i^t$,  edge weights are updated as follows in iteration $t$:
		\begin{itemize}
			\item For edges of type (3),  due to the symmetry of the neighborhoods, we again only move mass from $x$ to $y$ at a total cost of $D_3=\frac{w_3^{t}}{C_{in}} e^{-w_3^{t}}$.  Hence,  the updated weights are $w_3^{t+1}=\frac{w_3^t}{C_{in}} e^{-w_3^t}$ by Equation~\eqref{eq:flowweights}.
			\item For edges of type (1), the transport plan is analogous to the one in iteration 1: We move a mass of $\frac{a}{C_b}e^{-w_2^t}$ from the neighbors of $x$ to $x$, leave a mass of $\frac{1}{C_b}e^{-w_1^t}$ there, and transport the remaining mass $\frac{a}{C_b}e^{-w_2^t} - \frac{1}{C_b}e^{-w_1^t}	$ to $y$.  Picking up a mass of $\frac{1}{C_b} e^{-w_1^t}$ previously on $y$, we distribute the mass among $y$'s neighbors.  The total transportation cost is
			\begin{equation*}
				D_1 = \frac{a}{C_b} e^{-w_2^t} d_2 + \Big( \frac{a}{C_b} e^{-w_2^t} - \frac{1}{C_b} e^{-w_1^t}	\Big)d_1+\frac{a}{C_b} e^{-w_2^t} d_2 \; .
			\end{equation*}
			To prove that this is an optimal flow, consider the following 1-Lipschitz function $h$:
			$$
			h(v)=
			\begin{cases}
				w_2^t+w_1^t& \text{if } v\sim x, v\not\sim y\\
				w_1^t&\text{if }v=x\\
				0&\text{if }v=y\text{ or }v\sim x,y\\
				-w_2^t&\text{if }v\sim y, v\not\sim x.
			\end{cases}
			$$
			Then $$\EE_{v\sim m_x}[h(v)]- \EE_{v\sim m_y}[h(v)]=(w_2^t+w_1^t)\frac{a}{C_b}e^{-w_2^t}+w_1^t\frac{(-1)}{C_b}e^{-w_1^t}-w_2^t\frac{(-a)}{C_b}e^{-w_2^t}. $$
			Consequently, the updated weights are $w_1^{t+1}=\frac{2a w_2^t + a w_1^t}{C_b} e^{-w_2^t}  - \frac{w_1^t}{C_b} e^{-w_1^t}$.
			\item For edges of type (2), let $x$ denote the internal node and $y$ the bridge node. At any of the $a-1$ common neighbors of $x,y$, we have
			$$m_x(v) = \frac{e^{-w_3^t}}{C_{in}},\quad m_y(v) = \frac{e^{-w_2^t}}{C_b}.$$ We define the following quantity
			\begin{align*}
				c_t =m_x(v)-m_y(v)&= \frac{e^{-w_3^t}C_b-e^{-w_2^t}C_{in}}{C_b C_{in}}\\
				&= \frac{ae^{-w_2^t-w_3^t}+(b-1)e^{-w_1^t-w_3^t}-e^{-2w_2^t}-(a-1)e^{-w_2^t-w_3^t}}{C_b C_{in}}\\
				&=\frac{(b-1)e^{-w_1^t-w_3^t}+e^{-w_2^t}(e^{-w_3^t}-e^{-w_2^t})}{C_b C_{in}}.
			\end{align*}
			When $w_2^t>w_3^t$, this is always positive, but if $w_2^t<w_3^t$, this could be negative. 
			
			If $c_t\geq 0$, then as before, we want to transport the excess mass from these common neighbors to $x$ and the unshared neighbors of $y$ in an optimal way. If the total excess mass exceeds $m_y(x)$, we send that amount to $x$, and the remaining amount to $y$; otherwise we send all of this mass to $x$. This condition may be written as
			$$ r_t:= (a-1)\frac{e^{-w_3^t}}{C_{in}}-a\frac{e^{-w_2^t}}{C_b}\geq 0. $$ Clearly, $r_t=(a-1)c_t - \frac{e^{-w_2^t}}{C_b}<(a-1)c_t$.

			\underline{Case $c_t, r_t\geq 0$:} We send $m_y(x)$ mass from the common neighbors to $x$, and the remaining portion to $y$. Then we send all of the mass at $y$ to its unshared neighbors, for a cost of 
			$$
			D_2 = \frac{e^{-w_2^t}}{C_b}d_3 + r_t d_2 + (b-1)\frac{e^{-w_1^t}}{C_b}d_1.
			$$
			Consider the 1-Lipschitz function $h$ given by
			$$
			h(v) = \begin{cases}
				w_2^t&\text{if }v\sim x,y\\
				w_2^t-w_3^t&\text{if }v=x\\
				0&\text{if }v=y\\
				-w_1^t&\text{if }v\sim y, v\not\sim x.
			\end{cases}
			$$
			This gives a dual objective function value
			$$ w_2^t (a-1)c_t+(w_2^t-w_3^t)\frac{-e^{-w_2^t}}{C_b}-w_1^t\frac{-(b-1)e^{-w_1^t}}{C_b},$$
			which matches the formula for $D_2$ above when we take into consideration that $d_j = w_j^t$ at iteration $t$. 
			
			\underline{Case $c_t\geq0, r_t\leq 0$:} We send all of the excess mass from the common neighbors to $x$, and get the remaining portion from $y$. We then send the rest of the mass at $y$ to its unshared neighbors, for a total cost of 
			$$
			D_2 = (a-1)c_t d_3+ (-r_t)d_2+(b-1)\frac{e^{-w_1^t}}{C_b}d_1.
			$$
			Consider the 1-Lipschitz function $h$ given by 
			$$
			h(v) = \begin{cases}
				w_3^t&\text{if }v\sim x,y\\
				0&\text{if }v=x\\
				w_2^t&\text{if }v=y\\
				w_2^t-w_1^t&\text{if }v\sim y, v\not\sim x.
			\end{cases}
			$$
			To show that this function is 1-Lipschitz, note that if $v$ is a common neighbor of $x,y$, then $|h(v)-h(y)|=|w_2^t-w_3^t|\leq w_2^t$: either $w_2^t>w_3^t$, in which case $|w_2^t-w_3^t|=w_2^t-w_3^t\leq w_2^t$, or else $w_2^t<w_3^t$, and $w_3^t-w_2^t\leq w_2^t$ follows by the triangle inequality since $w_3^t=W_1(m_v,m_x)\leq W_1(m_v,m_y)+W_1(m_y,m_x)=2w_2^t$, where $W_1$ is the Wasserstein-1 distance, and $m_x,m_v,m_y$ are the measures at the vertices $x,v,y$ with the weights from iteration $t-1$ (or when $t=1$, $w_3^0=1<2=w_2^0$).
			
			This gives a dual objective function value
			$$
			w_3^t (a-1)c_t+w_2^t\frac{e^{-w_2^t}}{C_{in}}+(w_2^t-w_1^t)\frac{-(b-1)e^{-w_1^t}}{C_b}.
			$$
			This gives the equality when we note that 
			$$-r_t= a\frac{e^{-w_2^t}}{C_b}-(a-1)\frac{e^{-w_3^t}}{C_{in}}= 1- \frac{(b-1)e^{-w_1^t}}{C_b} - \left(1- \frac{e^{-w_2^t}}{C_{in}}\right)= \frac{e^{-w_2^t}}{C_{in}}-\frac{(b-1)e^{-w_1^t}}{C_b},$$ where the second equality follows from the definitions of $C_b$ and $C_{in}$.
			
			\underline{Case $c_t, r_t\leq 0$:} If $c_t\leq 0$, then $r_t< (a-1)c_t\leq 0$. Since $m_y$ places more mass at the common vertices than $m_x$, the only vertex on which $m_x$ places more mass than $m_y$ must be $y$. Thus, we send all of the needed mass from $y$ to $x$, the common neighbors, and the unshared neighbors of $y$, for a cost of
			$$D_2= \frac{e^{-w_2^t}}{C_b}d_2+(a-1)(-c_t)d_2+(b-1)\frac{e^{-w_1^t}}{C_b}d_1.$$
			Consider the 1-Lipschitz function $h$ given by 
			$$
			h(v) = \begin{cases}
				-w_2^t&\text{if }v\sim x,y\\
				-w_2^t&\text{if }v=x\\
				0&\text{if }v=y\\
				-w_1^t&\text{if }v\sim y, v\not\sim x.
			\end{cases}
			$$
			This gives a dual objective function value
			$$
			-w_2^t (a-1)c_t+(-w_2^t)\left(\frac{-e^{-w_2^t}}{C_b}\right)+(b-1)(-w_1^t)\left(\frac{-e^{-w_1^t}}{C_b}\right),
			$$
			which matches the formula for $D_2$ above.
		\end{itemize}
		
		We summarize these calculations in the following table:
		\begin{center}
			\begin{tabular}{c|c}
				Type $i$ & $w_i^{t+1}$\\
				\hline
				1 & $w_1^t\frac{ae^{-w_2^t}-e^{-w_1^t}}{C_b}+2w_2^t\frac{ae^{-w_2^t}}{C_b}$\\
				\hline
				2 ($c_t,r_t\geq 0$) & $w_1^t \frac{(b-1)e^{-w_1^t}}{C_b}+w_2^t r_t + w_3^t \frac{e^{-w_2^t}}{C_b}$\\
				\hline
				2 ($c_t\geq 0, r_t\leq 0$) & $w_1^t\frac{(b-1)e^{-w_1^t}}{C_b}+w_2^t(-r_t)+w_3^t (a-1)c_t$\\
				\hline
				2 ($c_t, r_t\leq 0$) & $w_1^t\frac{(b-1)e^{-w_1^t}}{C_b}+w_2^t(-r_t)$ \\
				\hline
				3& $w_3^t\frac{e^{-w_3^t}}{C_{in}}$
			\end{tabular}
		\end{center}
		
		We are now ready to prove the main two theorem statements. First, we show that $w_3^t$ is decaying to 0. Since $a>b\geq 2$, we have that $a-1\geq 2$. Then since $C_{in}\geq (a-1)e^{-w_3^t}$, we get
		$$
		w_3^{t+1}=w_3^t \frac{e^{-w_3^t}}{C_{in}}\leq w_3^t \frac{1}{a-1}.
		$$
		This gives $w_3^{t+1}\leq [a (a-1)^t]^{-1}$, which tends to 0 exponentially quickly as $t\rightarrow\infty$.
		
		Now we wish to show that $w_1^t>w_2^t$ and $w_1^t>w_3^t$ for all $t\geq 1$. From the proof above, we know that $w_3^t<2w_2^t$. We will also need the inequalities
		\begin{equation}
			\label{eq:thm21-intermed1}
			ae^{-w_2^t}-be^{-w_1^t}> 0
		\end{equation}
		which holds from $a>b$ and the induction hypothesis $w_1^t>w_2^t$; and
		\begin{align}
			&\frac{ae^{-w_2^t}}{C_b}> r_t \label{eq:thm21-intermed2}\\
			\Leftrightarrow\quad& C_{in}ae^{-w_2^t}-C_{in}C_b r_t> 0\notag\\
			\Leftrightarrow\quad& 2ae^{-2w_2^t}+a(a-1)e^{-w_2^t-w_3^t}-(a-1)(b-1)e^{-w_1^t-w_3^t}>0,\notag
		\end{align}
		which holds since $e^{-w_2^t}>e^{-w_1^t}$ and $a>b-1$. Note that this implies 
		$$\frac{ae^{-w_2^t}}{C_b}\geq \frac{a-1}{2}\frac{e^{-w_3^t}}{C_{in}}\geq \frac{e^{-w_3^t}}{C_{in}}.$$
		
		We first show $w_1^{t+1}>w_3^{t+1}$.
		\begin{align*}
			w_1^{t+1}-w_3^{t+1}&= w_1^t \frac{ae^{-w_2^t}-e^{-w_1^t}}{C_b}+2w_2^t \frac{ae^{-w_2^t}}{C_b}- w_3^t\frac{e^{-w_3^t}}{C_{in}}\\
			&\geq w_1^t\frac{(b-1)e^{-w_1^t}}{C_b}>0,
		\end{align*}
		using Equation~(\ref{eq:thm21-intermed1}), $w_3^t<2w_2^t$, and the implication of Equation~(\ref{eq:thm21-intermed2}) just above.
		
		Since Equation~(\ref{eq:thm21-intermed1}) clearly shows that the coefficient of $w_1^t$ in $w_1^{t+1}-w_2^{t+1}$ will be nonnegative in any of the three cases (it is the same in all of them), we argue as follows:
		$$
		w_1^{t+1}-w_2^{t+1}\geq \begin{cases}
			2w_2^t\frac{ae^{-w_2^t}}{C_b}-w_2^t r_t - w_3^t\frac{e^{-w_2^t}}{C_b}&c_t,r_t\geq0\\
			2w_2^t\frac{ae^{-w_2^t}}{C_b}+w_2^t r_t - w_3^t(a-1)c_t& c_t\geq 0, r_t\leq 0\\
			2w_2^t\frac{ae^{-w_2^t}}{C_b}+w_2^t r_t& c_t,r_t\leq 0.
		\end{cases}
		$$
		
		The quantity in the first subcase may be bounded below as follows, using $w_3^t<2w_2^t$:
		$$
		w_2^t\left(2\frac{ae^{-w_2^t}}{C_b} - r_t\right)-w_3^t\frac{e^{-w_2^t}}{C_b}\geq w_2^t\left((a-2)\frac{e^{-w_2^t}}{C_b}+ \frac{ae^{-w_2^t}}{C_b}-r_t \right).
		$$
		This lower bound is always positive by Equation~(\ref{eq:thm21-intermed2}).
		
		For the second subcase, observe that $(a-1)c_t=r_t+\frac{e^{-w_2^t}}{C_b}\leq \frac{e^{-w_2^t}}{C_b}$, so we may lower bound the quantity there (again using $w_3^t<2w_2^t$) by
		$$
		w_2^t\left(\frac{ae^{-w_2^t}}{C_b}+\frac{(a-1)e^{-w_3^t}}{C_{in}}\right)-w_3^t\frac{e^{-w_2^t}}{C_b}\geq w_2^t\left(\frac{(a-2)e^{-w_2^t}}{C_b}+\frac{(a-1)e^{-w_3^t}}{C_{in}}\right)>0.
		$$
		
		The third subcase is simplest of all, since the quantity above simplifies to 
		$$
		w_2^t\left(\frac{ae^{-w_2^t}}{C_b}+\frac{(a-1)e^{-w_3^t}}{C_{in}}\right)\geq0.
		$$
		
		This completes the proof.

	\end{proof}

	\section{Data sets}\label{apx:data}
	
	\paragraph{Real-world data.}
	We consider both synthetic networks generated from the random graph models (see sec.~\ref{sec:line-graph-curv}), and real networks\footnote{Obtained from NeuroData's open source data base~\url{https://neurodata.io/project/connectomes/}, accessed on 2 May 2022. The database hosts animal connectomes produced using data from a multitude of labs, across model species, using different modalities. }, which characterize the connectomes of a set of different animal species. A connectome is a specific, cell-to-cell mapping of axonal tracts between neurons, created from cellular data obtained, for instance, via electron microscopy. Here specfically, we consider networks constructed from mixed species of cats (\citep{deReus2013cat}, ``Cat"), male C.~elegans (\citep{jarrel2012worm}, ``Worm"), and rhesus (\citep{harriger2012macaque}, ``Macaque"). 
	Table~\ref{tab:animals} provides summary statistics for the real networks in our study. 
	\begin{table}[htb]
		\centering
		\begin{tabular}{c|cccccc}
			\hline
			& $n$ & $|E|$ & $|\mathcal{E}|$ & $d$ & $\sigma_d/d$ & Modularity \\
			\hline
			Cat     & $65$ & $730$ & $18859$ & $22.46$ & $0.44$ & $0.27$ \\
			Worm & $269$ & $2902$ & $93893$ & $21.58$ & $0.74$ & $0.36$ \\
			Macaque & $242$ & $3054$ & $114991$ & $25.24$ & $0.73$ & $0.36$\\
			\hline
		\end{tabular}
		\caption{Summary statistics of the real networks.}
		\label{tab:animals}
	\end{table}
	\begin{figure}[h]
		\centering
		\begin{tabular}{ccc}
			\includegraphics[width=.3\textwidth]{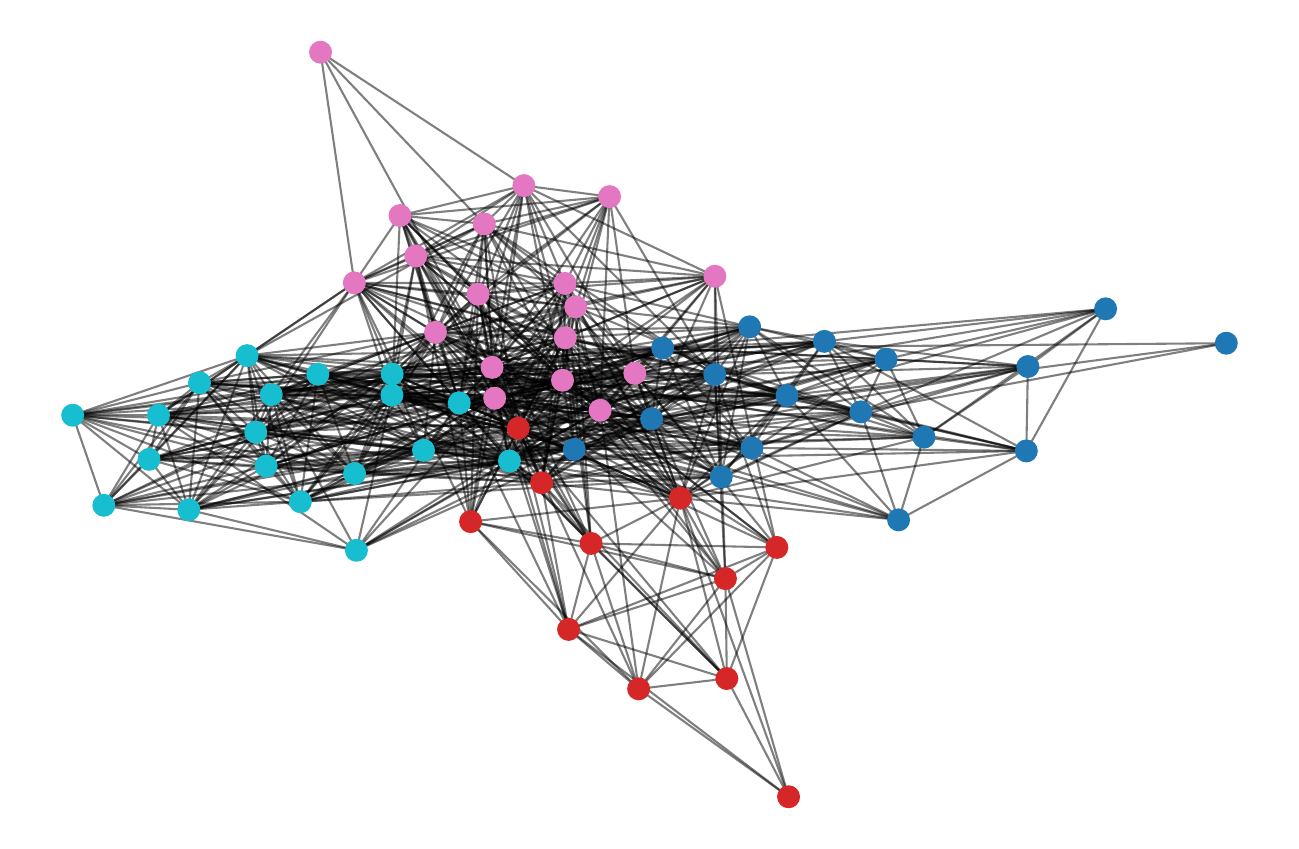} & \includegraphics[width=.3\textwidth]{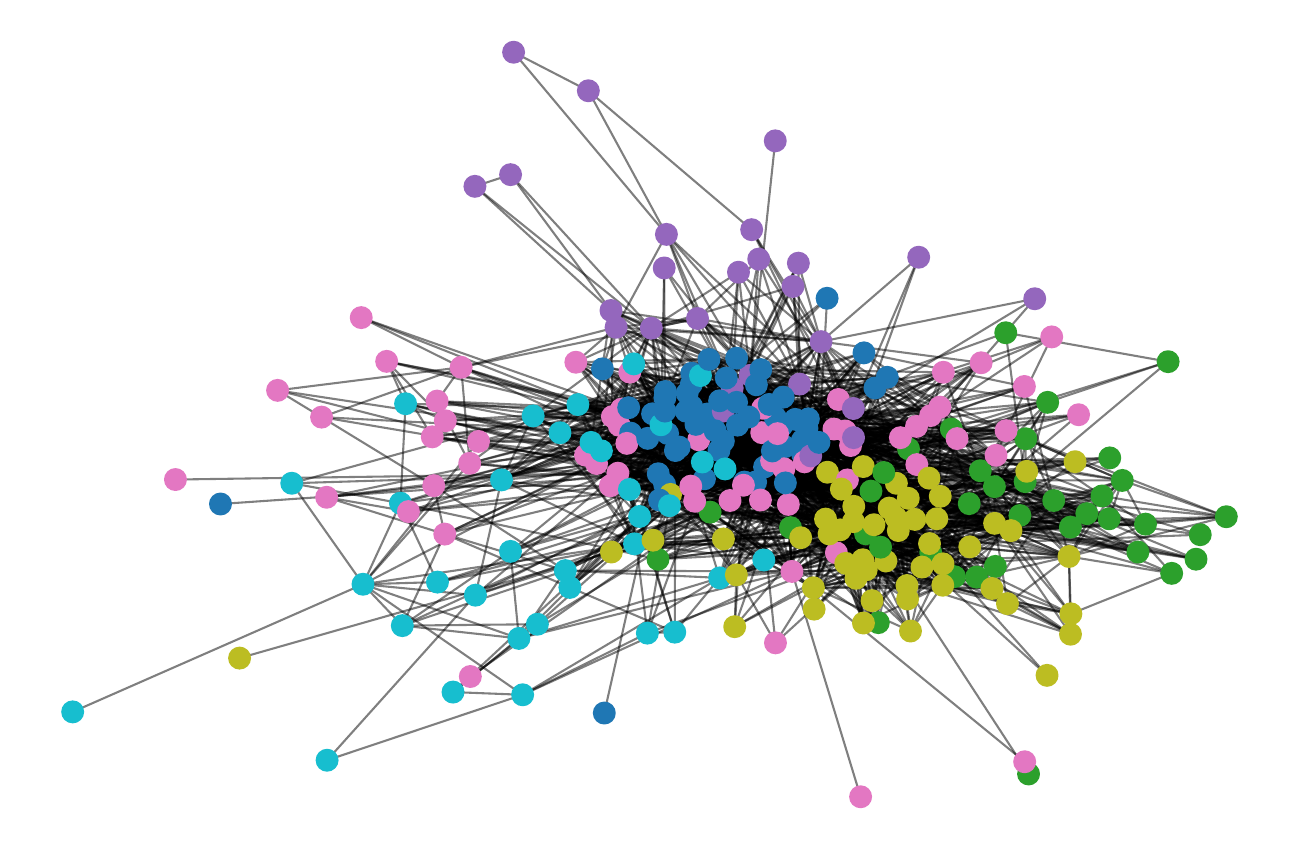} & \includegraphics[width=.3\textwidth]{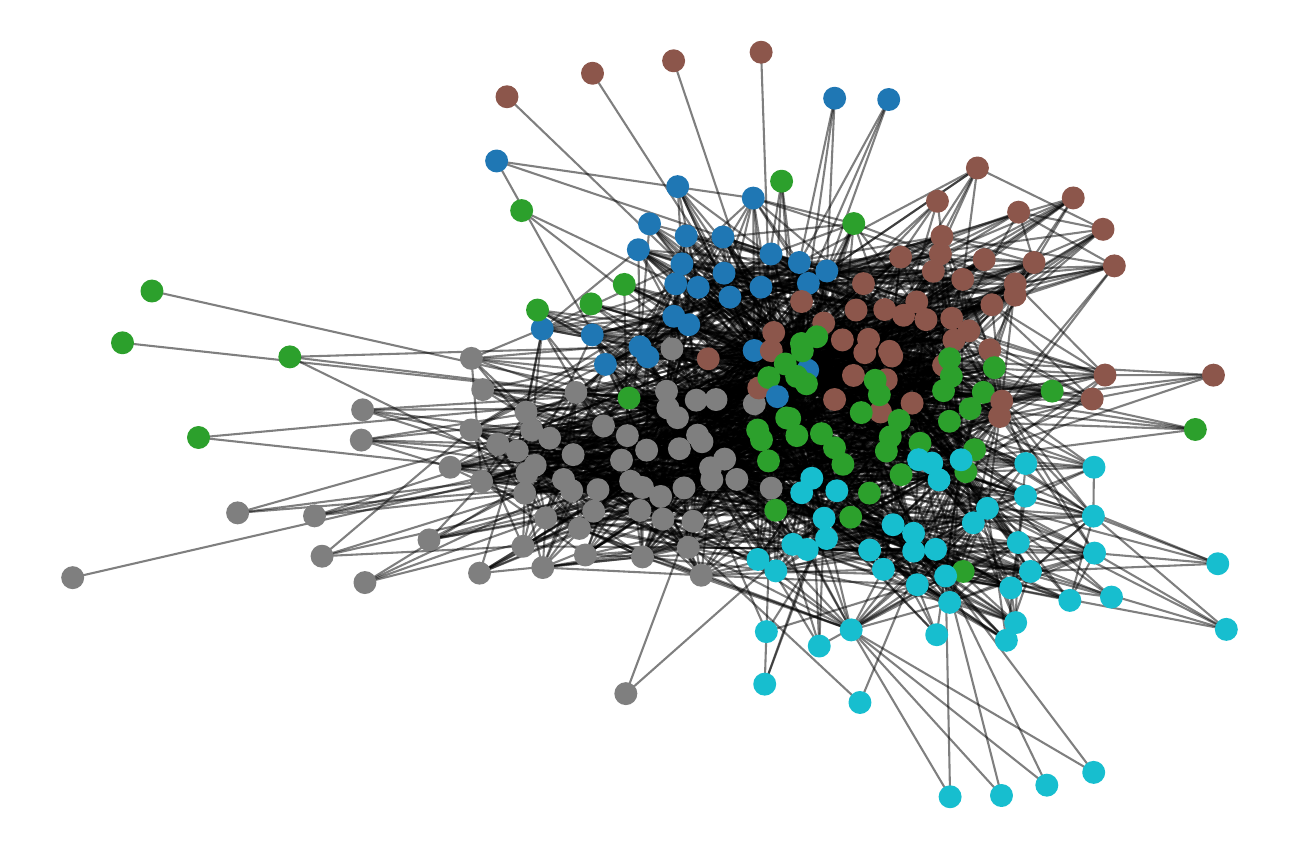}
		\end{tabular}
		\caption{The real networks, ``Cat" (left), ``Worm" (middle) and ``Macaque" (right), where we apply the Louvain algorithm to detect communities and vertices in different communities are shown in different colors.}
		\label{fig:graphs_real}
	\end{figure}
	
	\paragraph{Random Geometric Graphs.}
	In our experiments, we consider Random Geometric Graphs (RGGs) which incorporate the spatial perspective~\citep{Penrose_RGG_2003}. For completeness, we define the RGG below.
	
	\begin{defn}[Random Geometric Graphs]
		In this model, vertices are placed in the unit cube uniformly at random, and two vertices will be connected by an edge if the distance between the vertices is less than a previously specified value called radius. Specifically, let the unit cube be $[0,1)^{n_d}$ where $n_d$ defines its dimension, the radius be $r\in (0,1)$, and the distance function be $d(x,y),\ \forall x,y\in [0,1)^{n_d}$. Then there will be an edge between vertices $v_i$ and $v_j$ if 
		\begin{equation*}
			d(x_i, x_j) < r \; ,
		\end{equation*}
		where $x_i, x_j$ are the positions of the vertices $v_i, v_j$, respectively, in the unit cube $[0,1)^{n_d}$.
	\end{defn}
	Specifically, in our experiments, we start from the RGGs corresponding to $2$-d unit cubes.
	
	\section{Experiments}
	\subsection{Synthetic data}
	\label{app:exp-msbm}
	In Sec.~\ref{sec:exp-mmcd}, we observe that the performance of the algorithm via augmented FRC is not comparable with the performance of the algorithm via ORC, which raises the important problem of relating FRC to ORC. We find that FRC-3 has a stronger linear relationship with ORC as the graphs become denser; see Figs.~\ref{fig:exp-frcorc-L-30}, \ref{fig:exp-frcorc-L-60} and \ref{fig:exp-frcorc-L-90}. This is potentially because there are more and more triangular and quadrangular faces in the graphs when more edges are present, and these characteristics become more dominant for the ORC. Meanwhile, we find that, when $p_{in}$ is relatively small, all three variants of FRC cannot really separate the edges in $L$ between the nodes in the same community versus the others initially, which might be the reason of a relatively worse performance of the algorithm via FRC in the experiments. However, FRC-3 can separate these two types of edges when $p_{in}$ is relatively large, which indicates that the performance is expected to be better in the dense regime, where ORC might have computational issues.
	
	\subsection{Node neighborhood measures}
	We give some brief computational evidence that the choice of $\alpha$ in the node neighborhood measure as in Eq.~\eqref{eq:nodemeasure} has little impact on the performance of our curvature-based algorithms. Over a range of different choices of $\alpha$, and for two real-data networks, we see that the results are highly conserved.
	\begin{table}[!ht]
		\centering
		\begin{tabular}{c|ccccccc}
			\hline
			Data\textbackslash $\alpha$  & $0$ & $0.1$ & $0.2$ & $0.3$ & $0.4$ & $0.5$ & $0.6$\\
			\hline
			DBLP-1 & $0.54$ & $0.54$ & $0.54$ & $0.55$ & $0.54$ & $0.54$ & $0.48$\\
			DBLP-2 & $0.54$ & $0.56$ & $0.60$ & $0.60$ & $0.60$ & $0.60$ & $0.60$  \\
			\hline
		\end{tabular}
		\caption{Extended NMI of ORC-E via line graph curvature on the real data, where we choose different values of $\alpha$ in Eq.~\eqref{eq:nodemeasure}.}
		\label{tab:exp-alpha}
	\end{table}

	\begin{figure}[p]
		\centering
		\begin{tabular}{ccc}
			\includegraphics[height=.2\textheight]{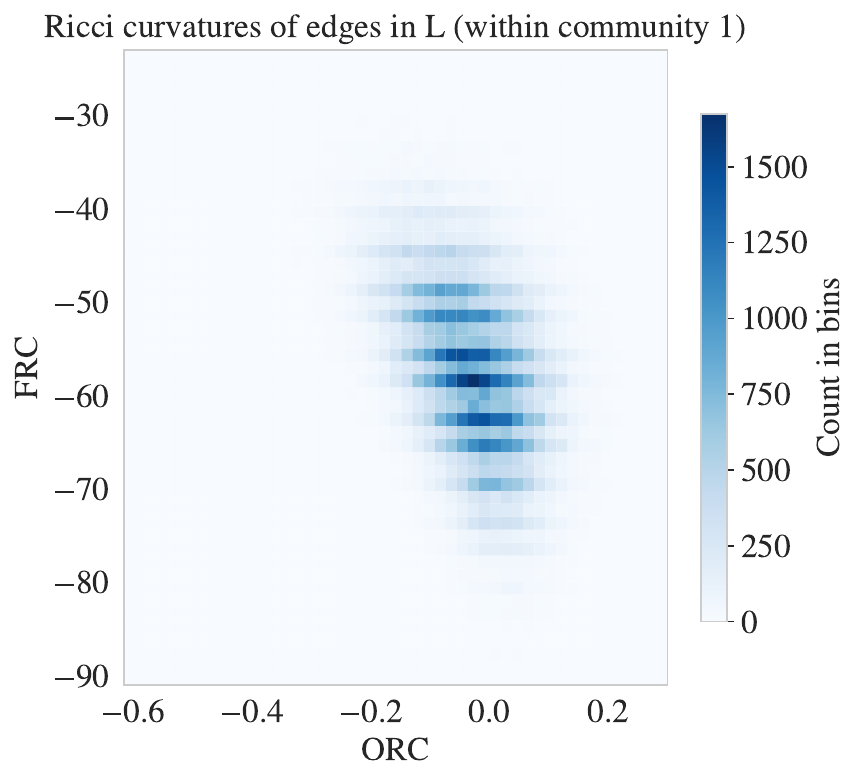} & \includegraphics[height=.2\textheight]{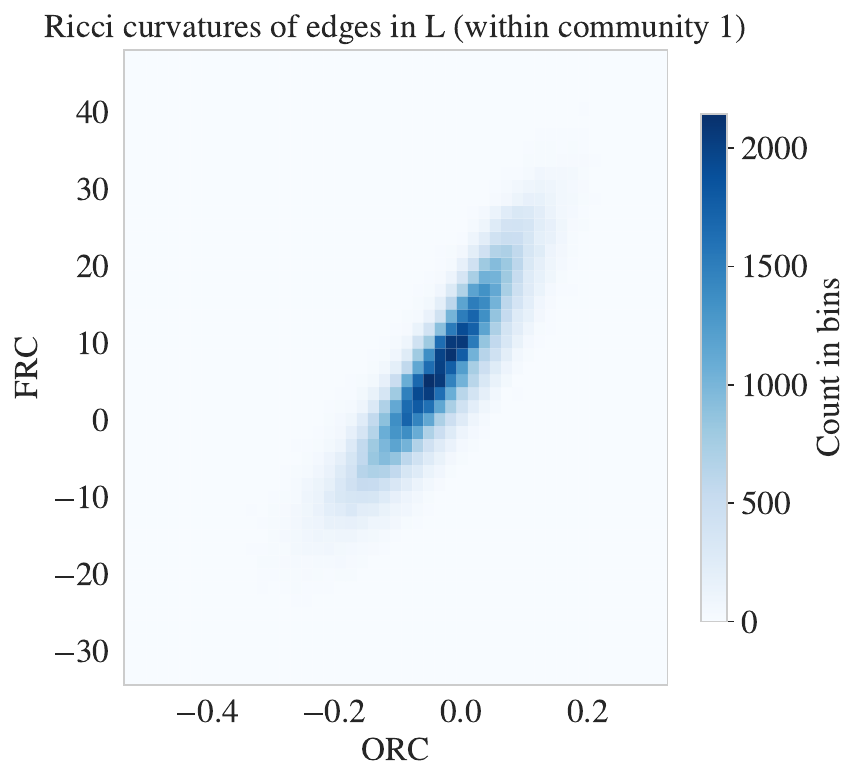} & \includegraphics[height=.2\textheight]{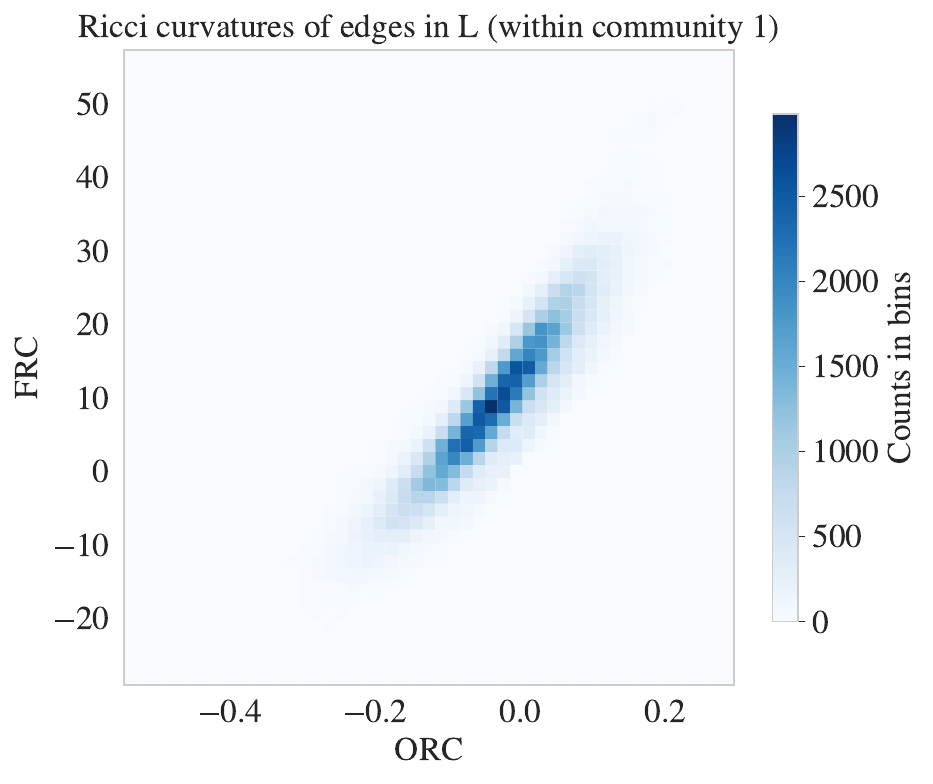}\\
			\includegraphics[height=.2\textheight]{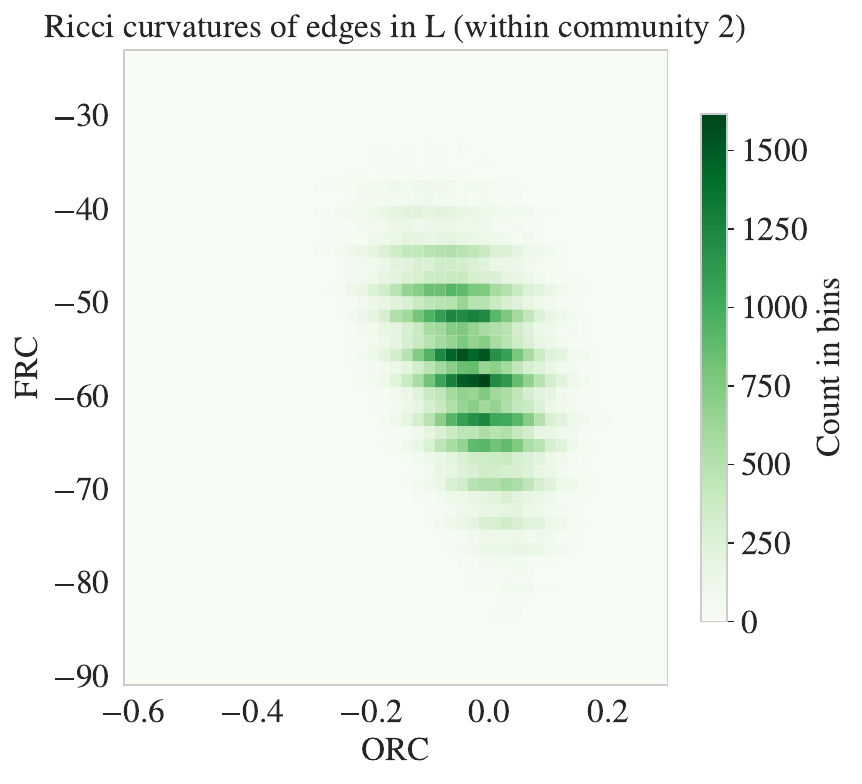} & \includegraphics[height=.2\textheight]{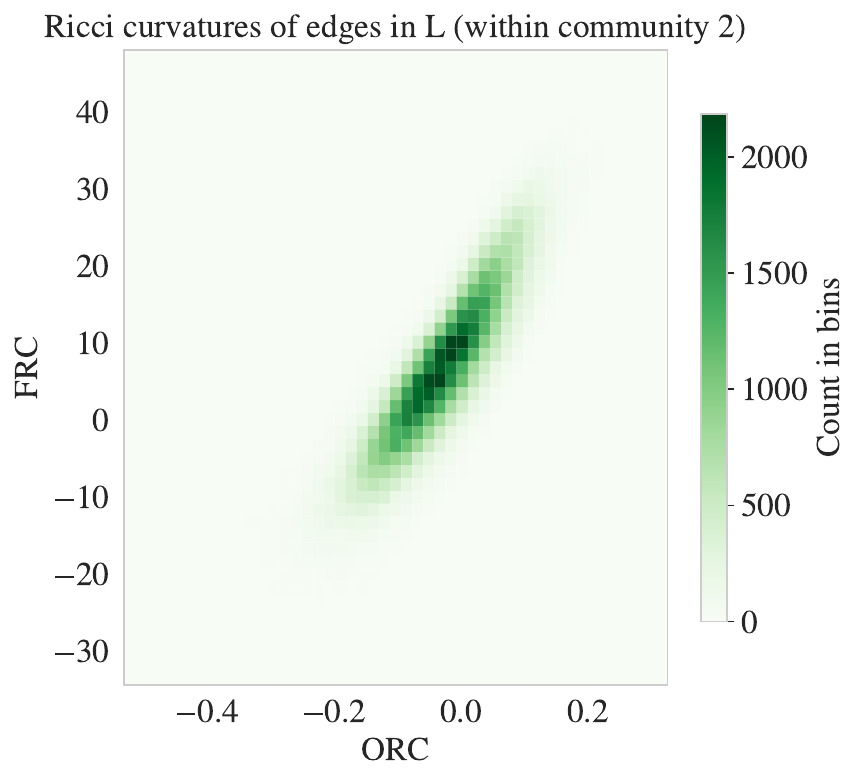} & \includegraphics[height=.2\textheight]{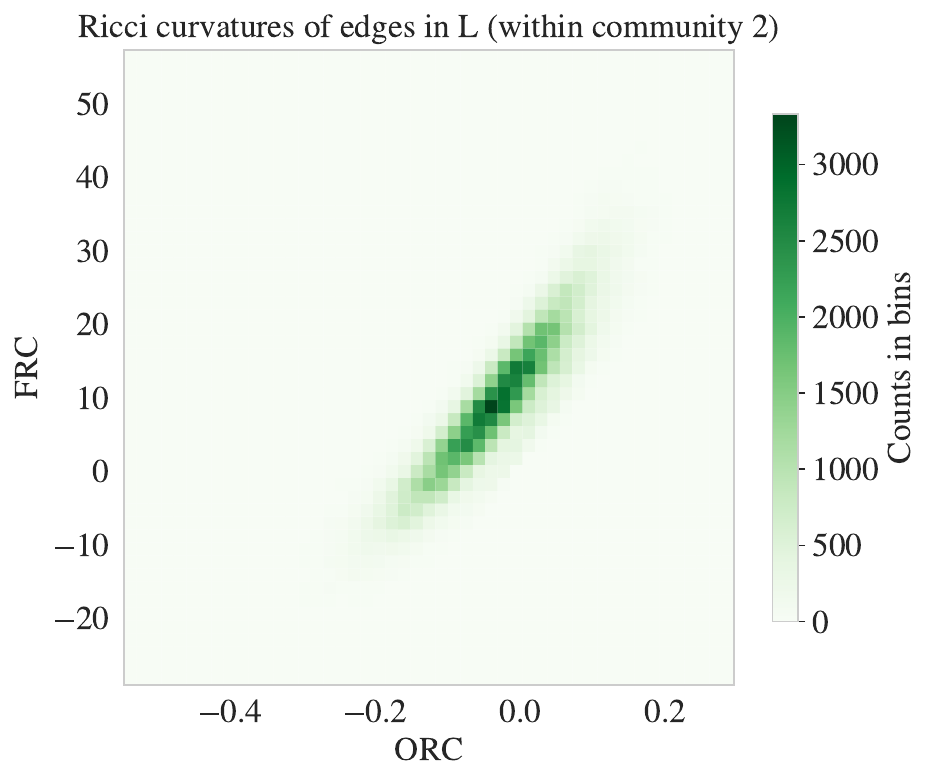}\\
			\includegraphics[height=.2\textheight]{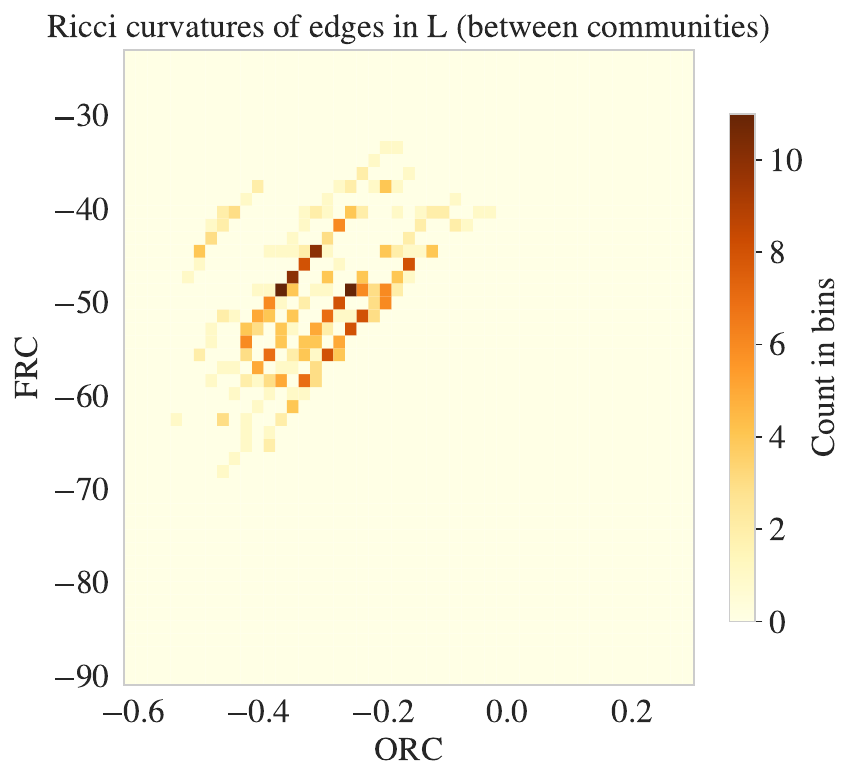} & \includegraphics[height=.2\textheight]{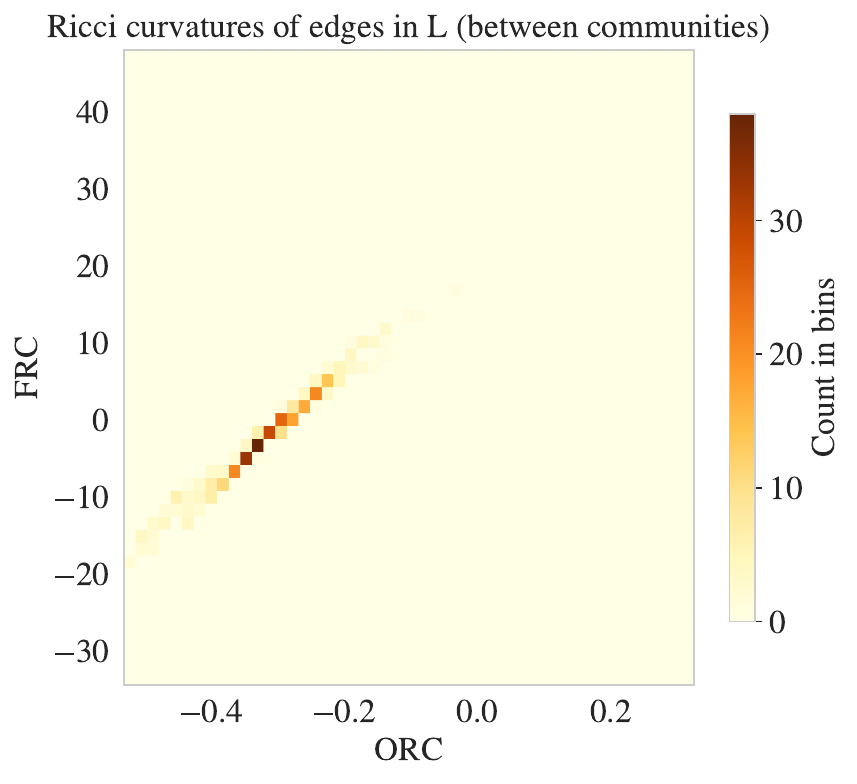} & \includegraphics[height=.2\textheight]{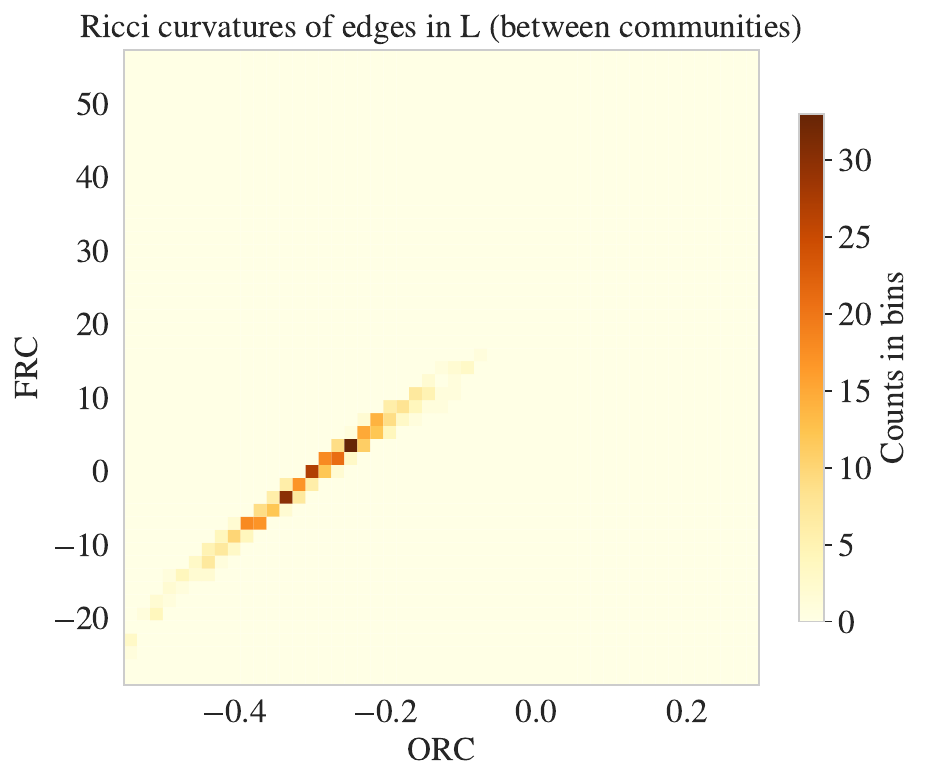}
		\end{tabular}
		\caption{2-d histograms of the ORC-E versus the FRC-1 (left), FRC-2 (middle) and FRC-3 (right) for edges in $L$, where the networks are obtained from the planted two-block MMB of size $n=300$, $p_{in} = 0.1$, $p_{out} = 0$, $n_o = 1$ node of mixed membership, and we generate $n_s=10$ networks for the results. (Row 1 and 2: edges between nodes within communities 1, in blue and 2, in green, respectively, in $L$; Row 3: edges between nodes in the different communities, in red).}
		\label{fig:exp-frcorc-L-30}
	\end{figure}
	
	\begin{figure}[p]
		\centering
		\begin{tabular}{ccc}
			\includegraphics[height=.2\textheight]{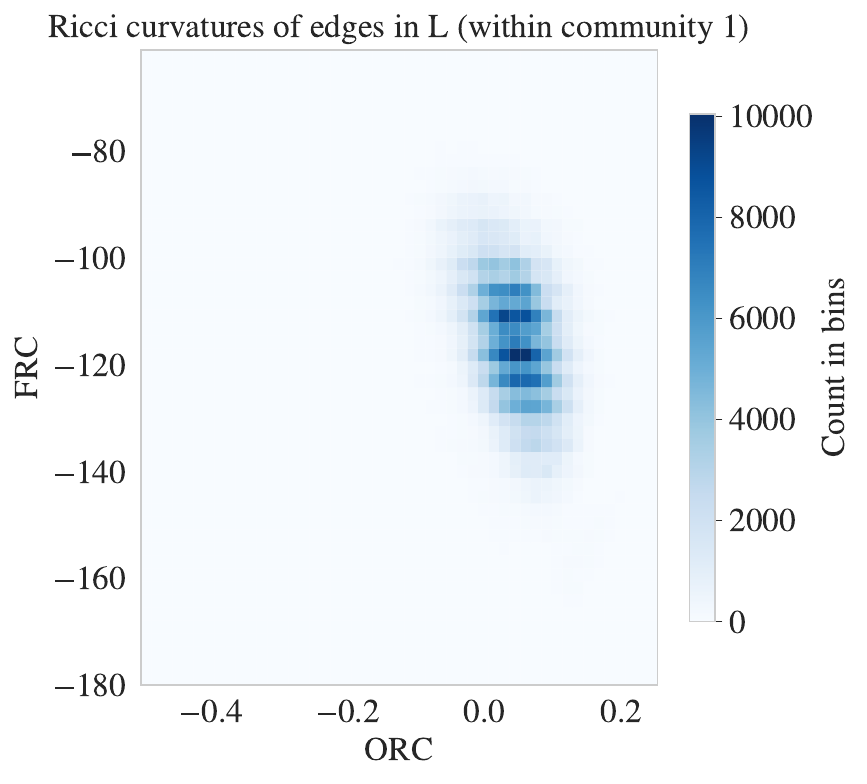} & \includegraphics[height=.2\textheight]{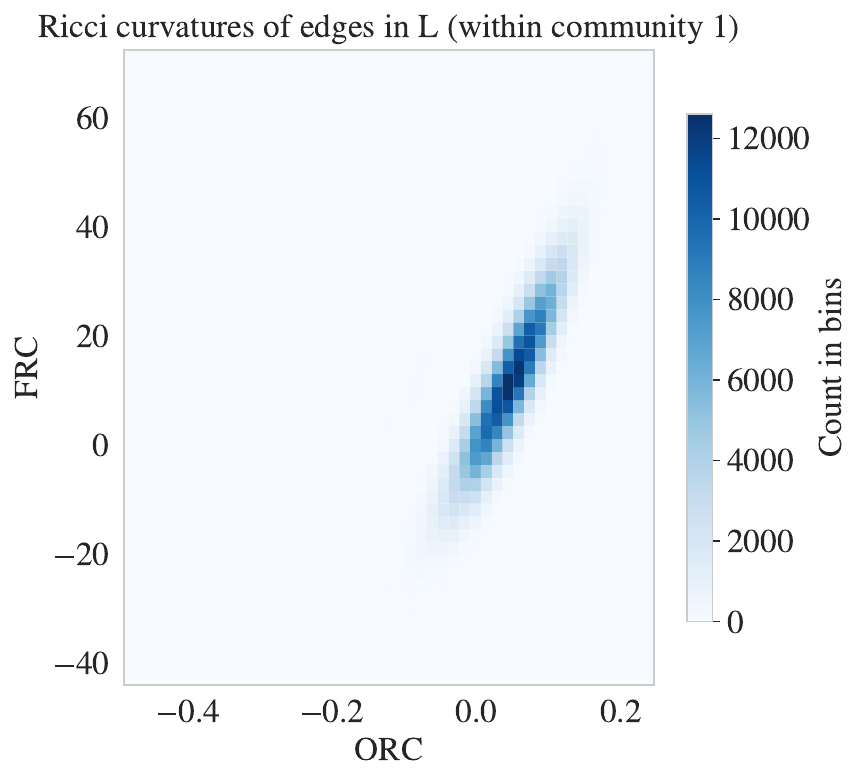} & \includegraphics[height=.2\textheight]{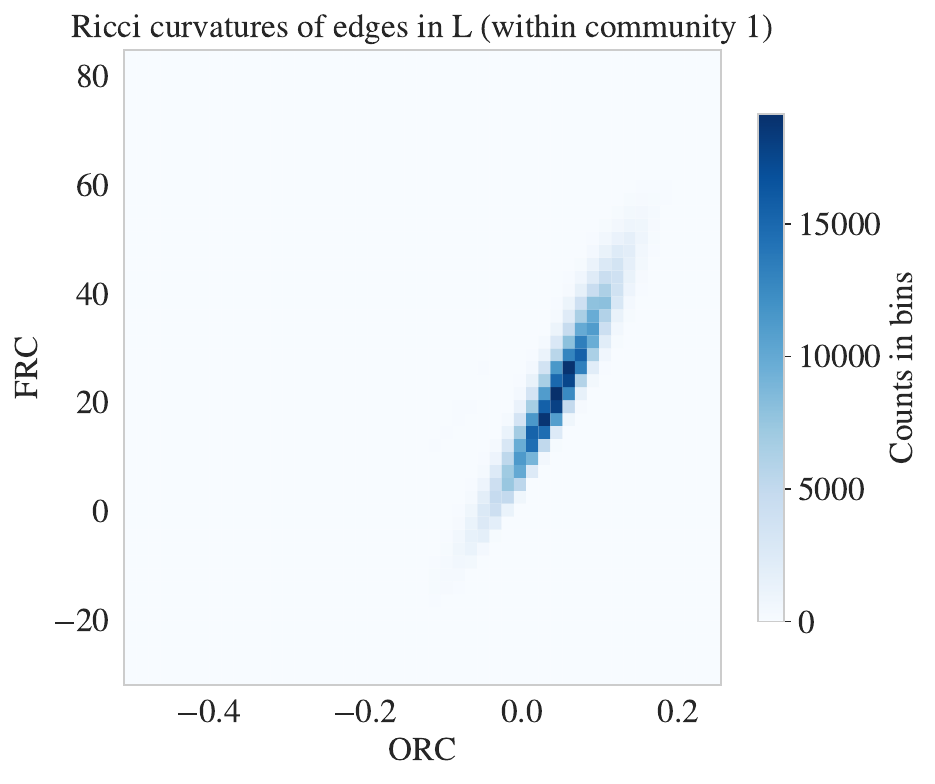}
			\\
			\includegraphics[height=.2\textheight]{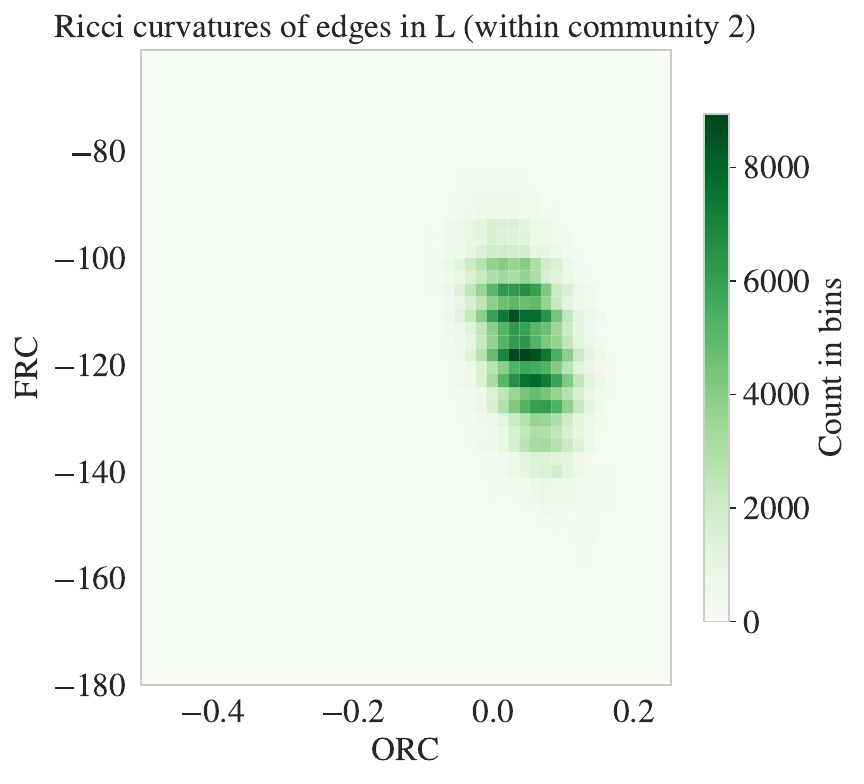} & \includegraphics[height=.2\textheight]{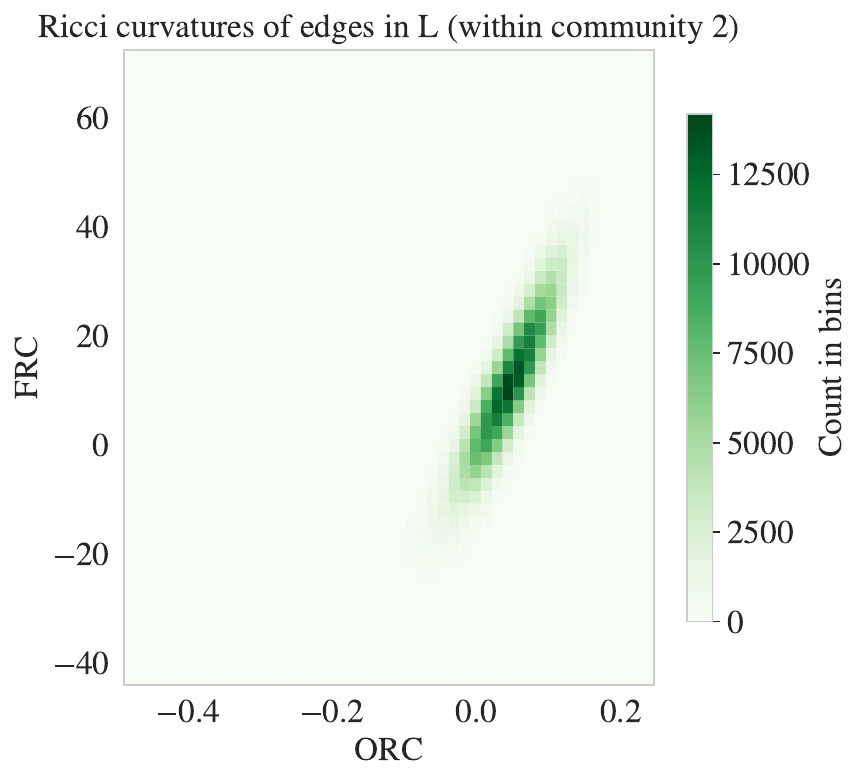} & \includegraphics[height=.2\textheight]{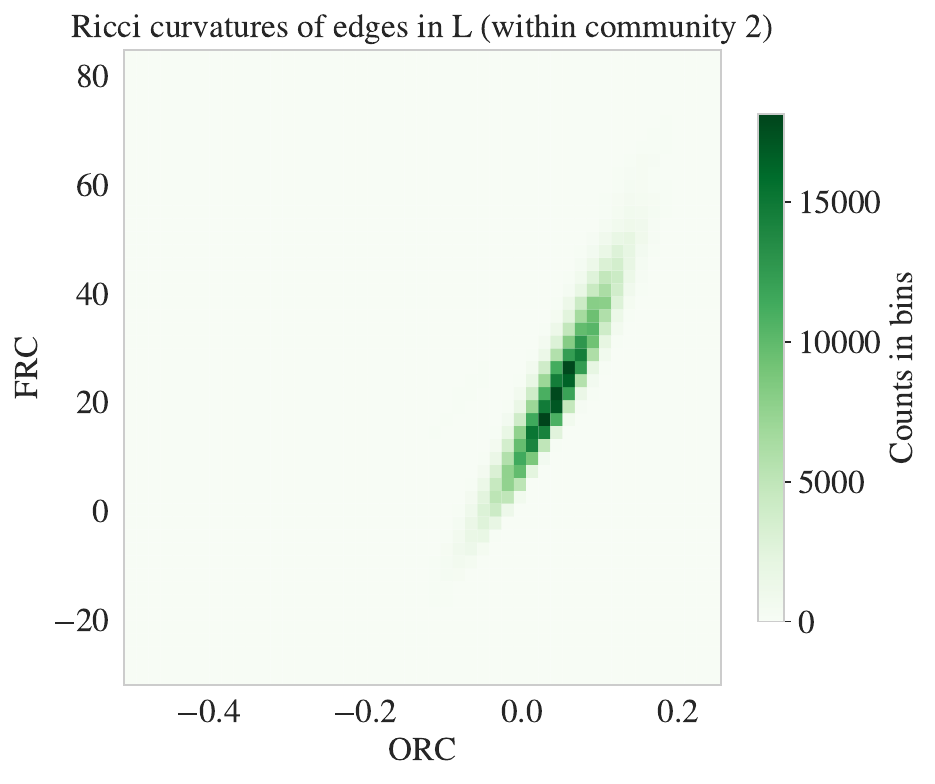}
			\\
			\includegraphics[height=.2\textheight]{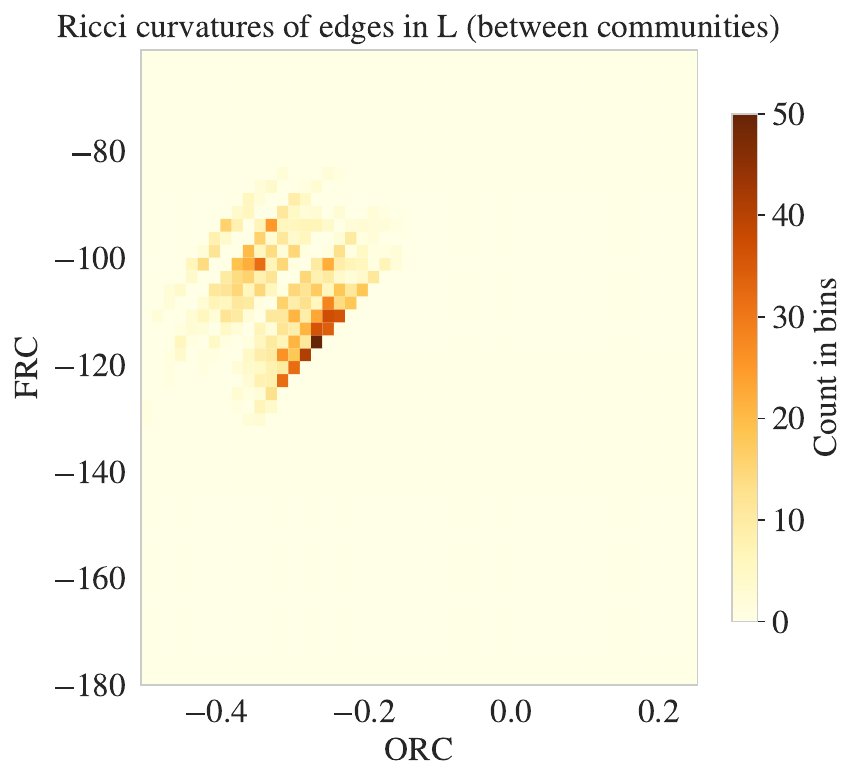} & \includegraphics[height=.2\textheight]{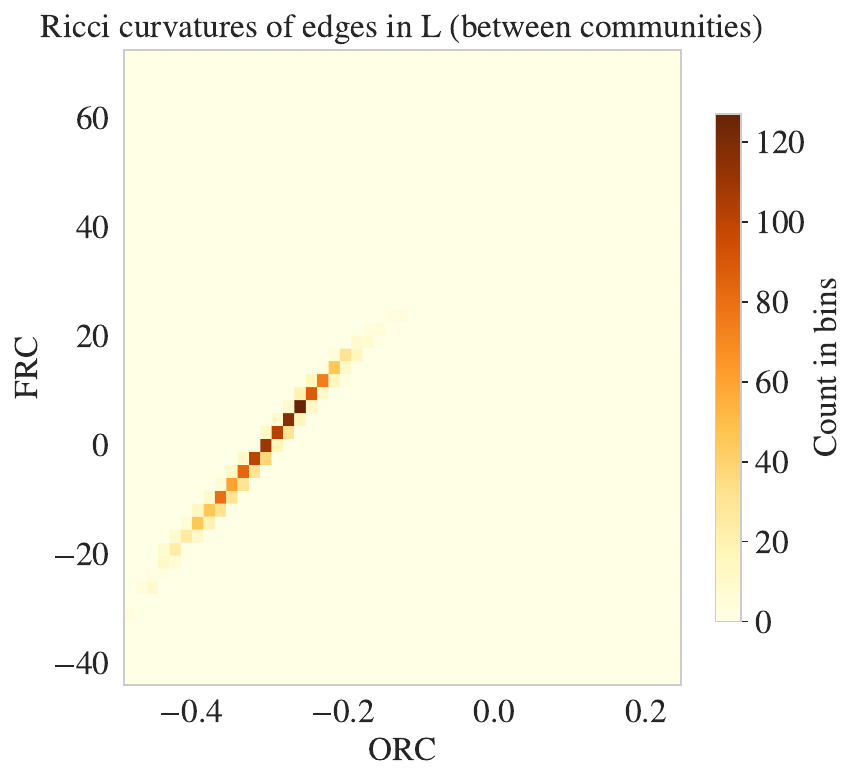} & \includegraphics[height=.2\textheight]{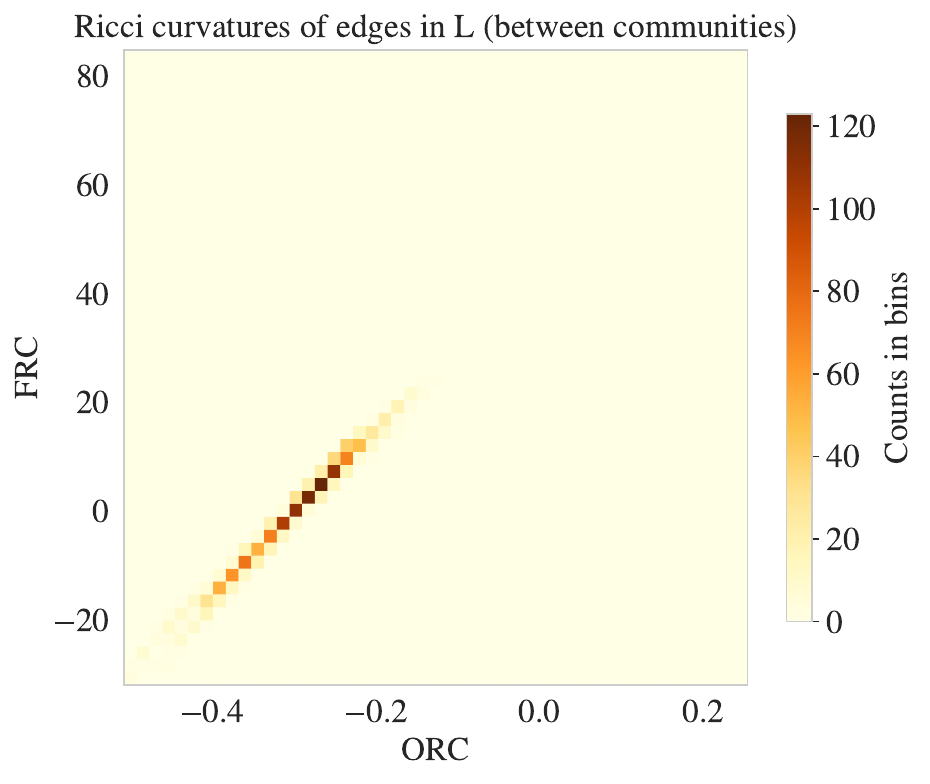}
		\end{tabular}
		\caption{2-d histograms of the ORC-E versus the FRC-1 (left), FRC-2 (middle) and FRC-3 (right) for edges in $L$, where the networks are obtained from the planted two-block MMB of size $n=300$, $p_{in} = 0.2$, $p_{out} = 0$, $n_o = 1$ node of mixed membership, and $n_s=10$.}
		\label{fig:exp-frcorc-L-60}
	\end{figure}
	
	\begin{figure}[p]
		\centering
		\begin{tabular}{ccc}
			\includegraphics[height=.2\textheight]{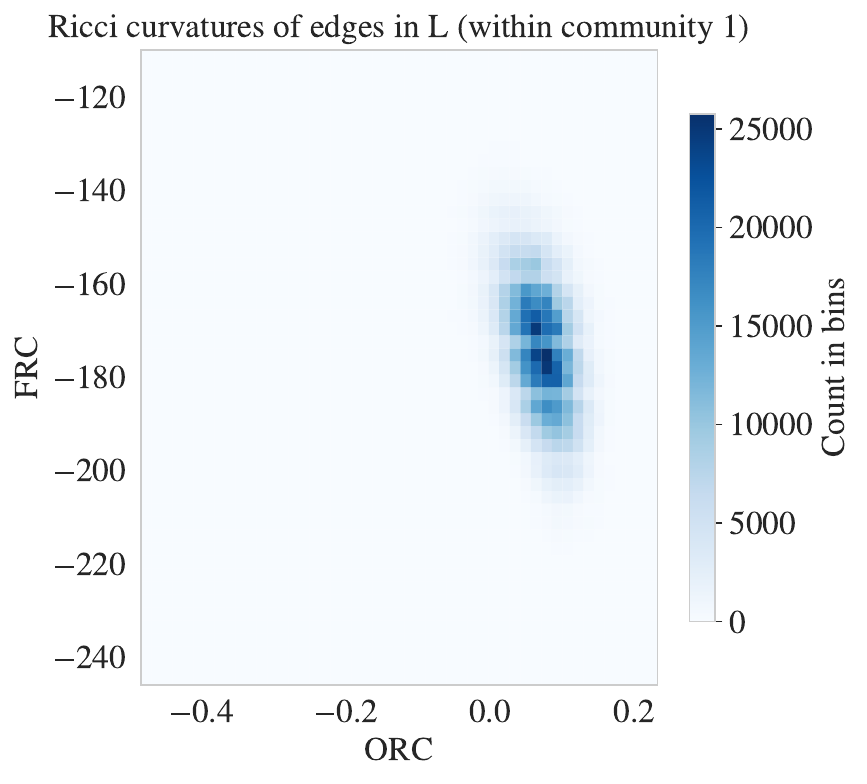} & \includegraphics[height=.2\textheight]{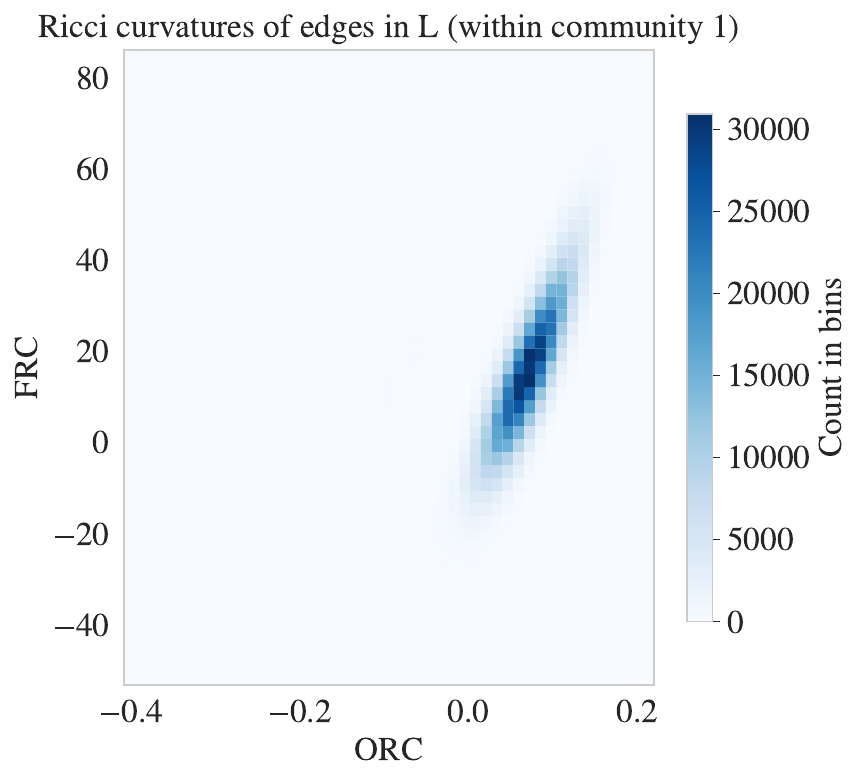} & \includegraphics[height=.2\textheight]{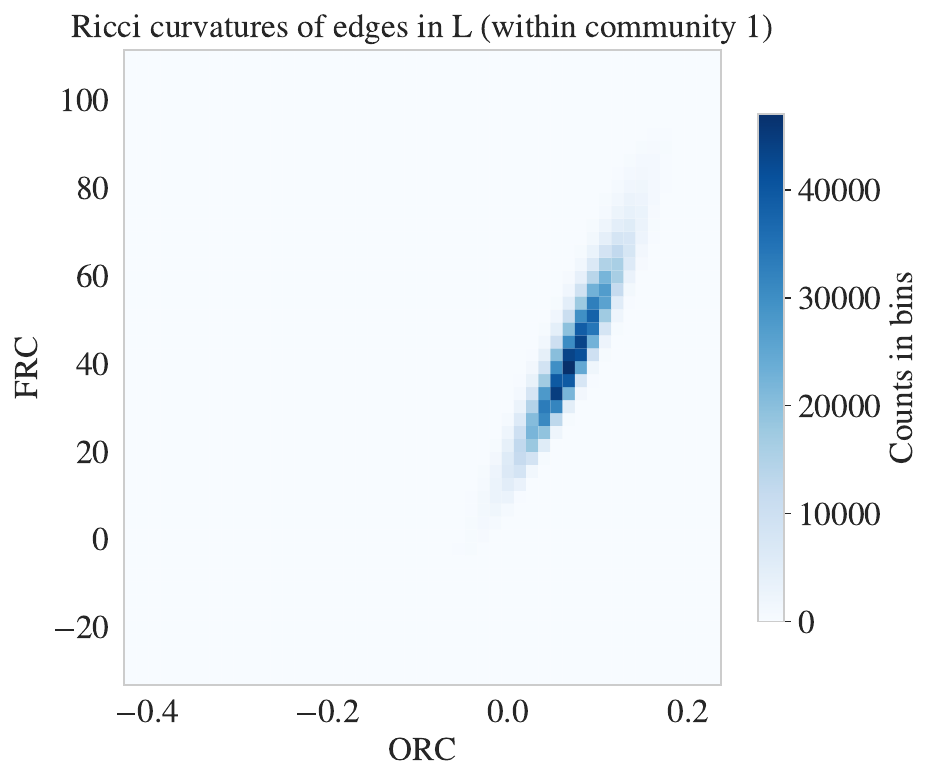}\\
			\includegraphics[height=.2\textheight]{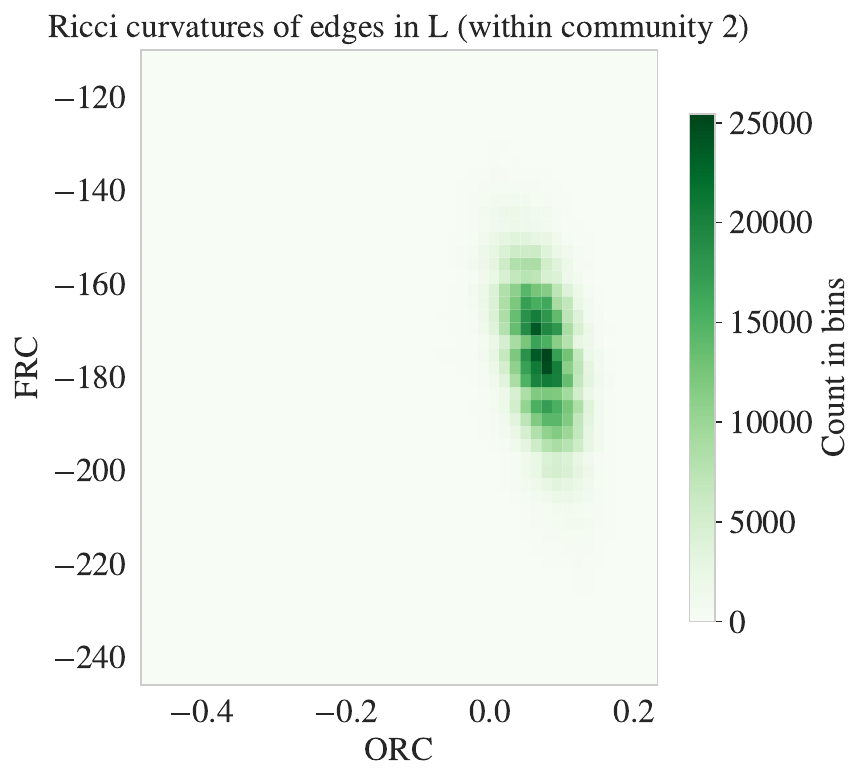} & \includegraphics[height=.2\textheight]{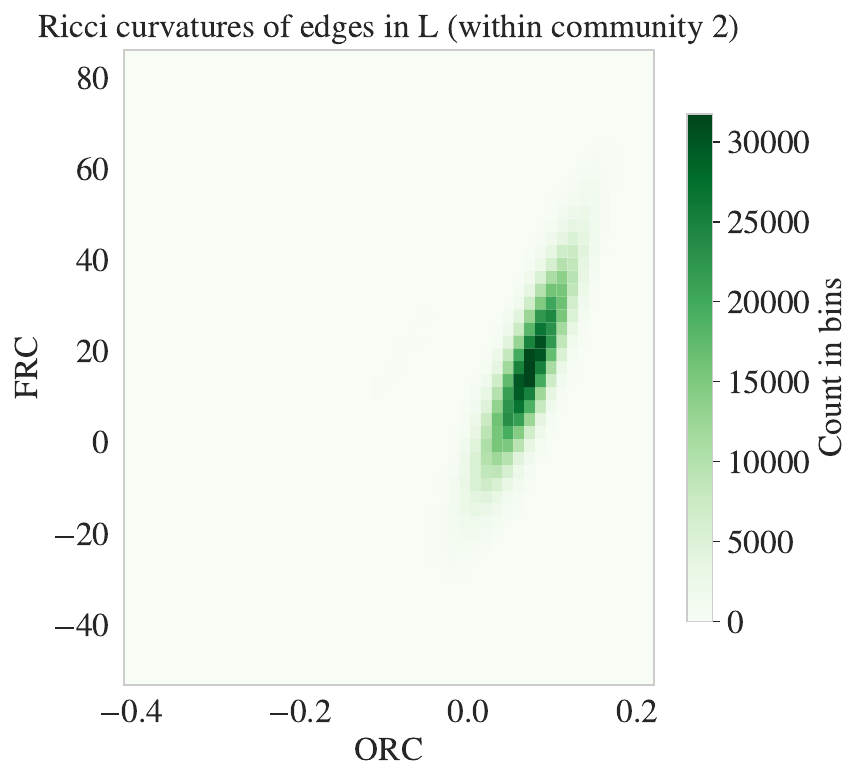} & \includegraphics[height=.2\textheight]{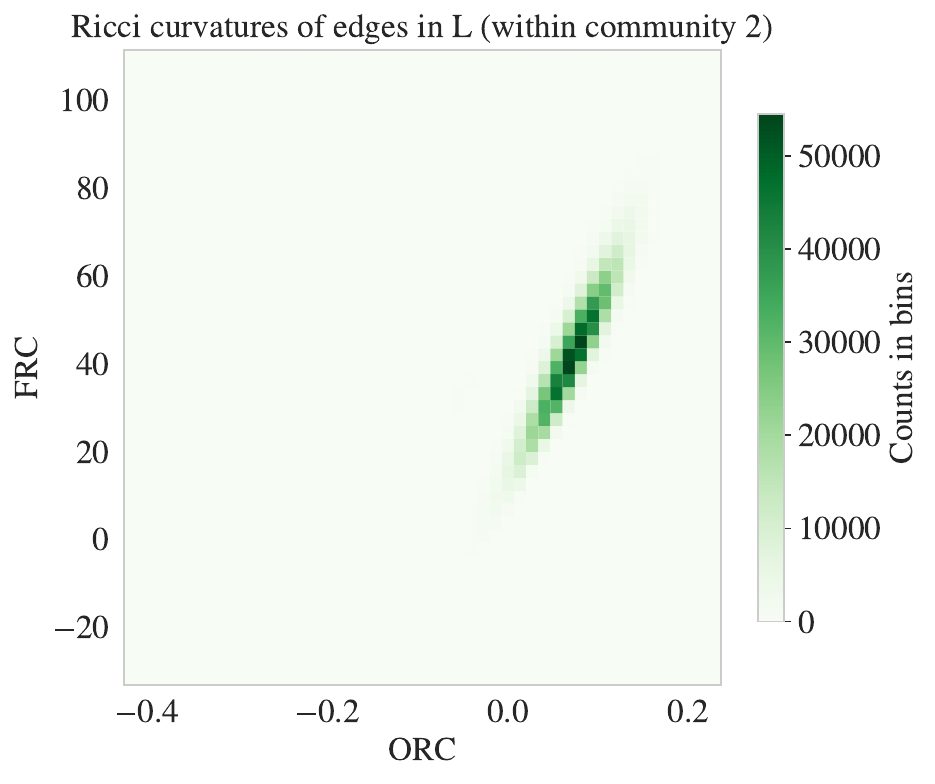}\\
			\includegraphics[height=.2\textheight]{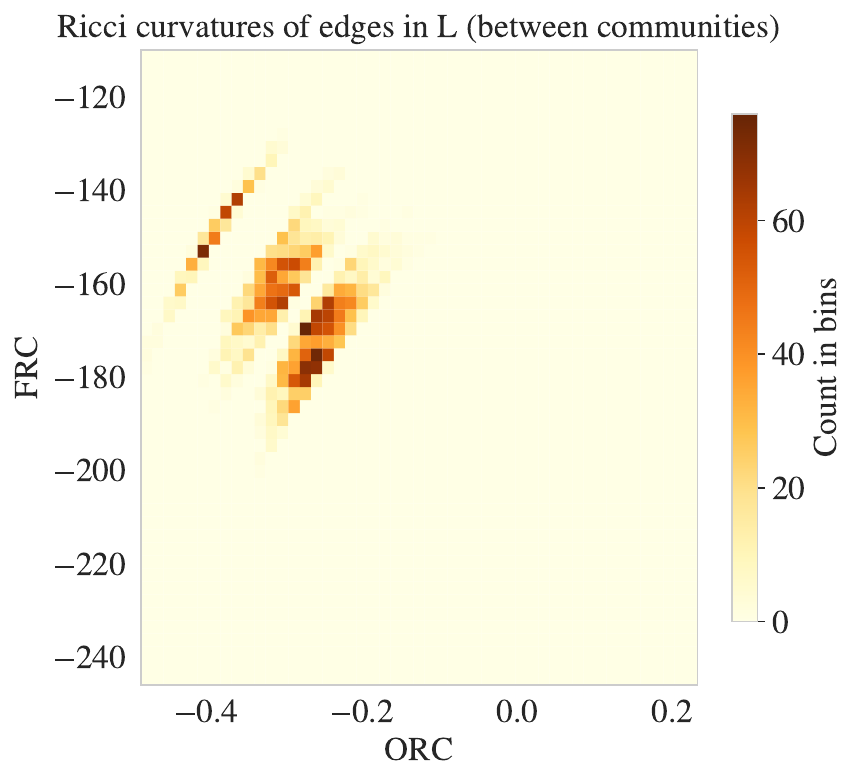} & \includegraphics[height=.2\textheight]{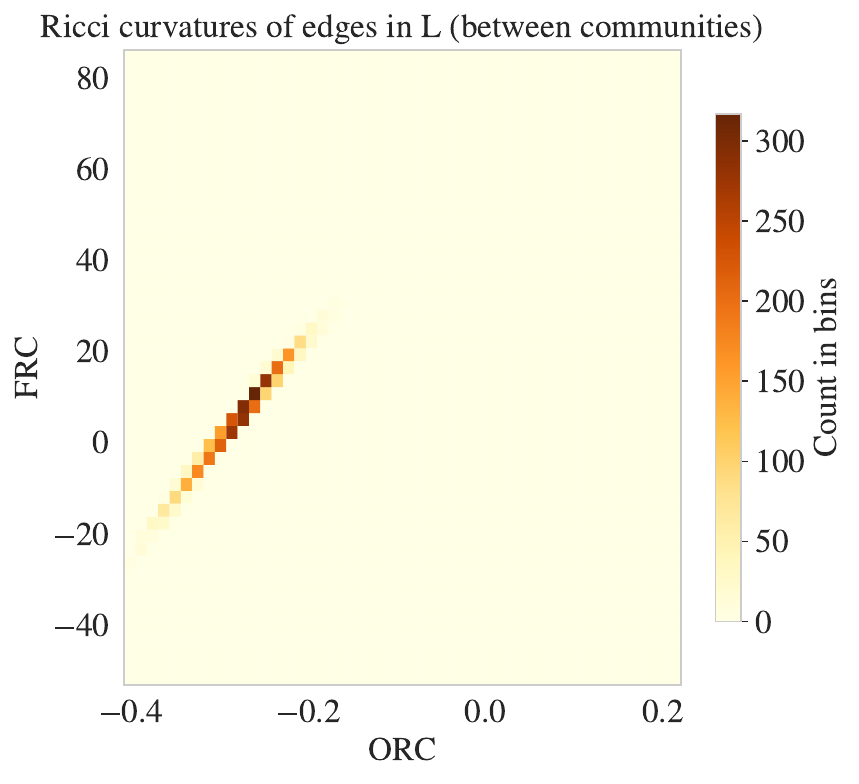} & \includegraphics[height=.2\textheight]{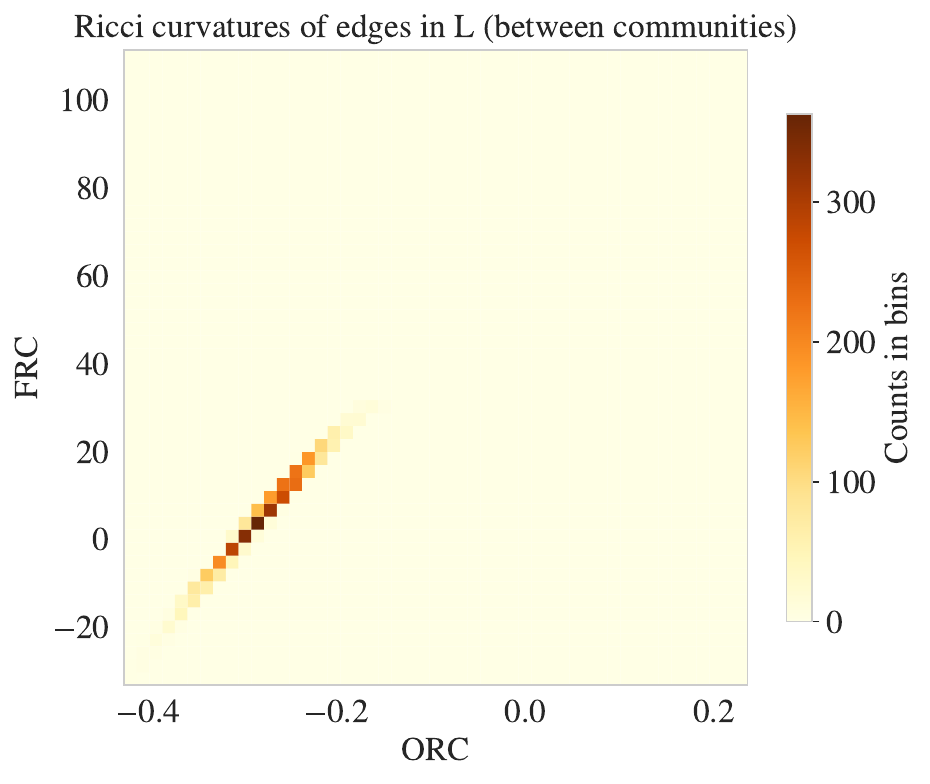}
		\end{tabular}
		\caption{2-d histograms of the ORC-E versus the FRC-1 (left), FRC-2 (middle) and FRC-3 (right) for edges in $L$, where the networks are obtained from the planted two-block MMB of size $n=300$, $p_{in} = 0.3$, $p_{out} = 0$, $n_o = 1$ node of mixed membership, and $n_s=10$.}
		\label{fig:exp-frcorc-L-90}
	\end{figure}
	
\newpage
\bibliographystyle{plainnat}
\bibliography{ref}

\begin{thebibliography}{64}
\providecommand{\natexlab}[1]{#1}
\providecommand{\url}[1]{\texttt{#1}}
\expandafter\ifx\csname urlstyle\endcsname\relax
  \providecommand{\doi}[1]{doi: #1}\else
  \providecommand{\doi}{doi: \begingroup \urlstyle{rm}\Url}\fi

\bibitem[Abbe(2017)]{abbe2017community}
E.~Abbe.
\newblock Community detection and stochastic block models: recent developments.
\newblock \emph{Journal of Machine Learning Research}, 18\penalty0
  (1):\penalty0 6446--6531, 2017.

\bibitem[Airoldi et~al.(2008)Airoldi, Blei, Fienberg, and
  Xing]{airoldi2008mixed}
E.~Airoldi, D.~Blei, S.~Fienberg, and E.~Xing.
\newblock Mixed membership stochastic blockmodels.
\newblock \emph{Advances in neural information processing systems}, 21, 2008.

\bibitem[Battiston et~al.(2020)Battiston, Cencetti, Iacopini, Latora, Lucas,
  Patania, Young, and Petri]{battiston2020networks}
F.~Battiston, G.~Cencetti, I.~Iacopini, V.~Latora, M.~Lucas, A.~Patania,
  J.~Young, and G.~Petri.
\newblock Networks beyond pairwise interactions: structure and dynamics.
\newblock \emph{Physics Reports}, 874:\penalty0 1--92, 2020.

\bibitem[Benson et~al.(2016)Benson, Gleich, and Leskovec]{benson2016higher}
A.~Benson, D.~Gleich, and J.~Leskovec.
\newblock Higher-order organization of complex networks.
\newblock \emph{Science}, 353\penalty0 (6295):\penalty0 163--166, 2016.

\bibitem[Bloch(2014)]{bloch2014combinatorial}
E.~Bloch.
\newblock Combinatorial {R}icci curvature for polyhedral surfaces and posets.
\newblock \emph{arXiv:1406.4598}, 2014.

\bibitem[Blondel et~al.(2008)Blondel, Guillaume, Lambiotte, and
  Lefebvre]{Blondel_2008}
V.~Blondel, J.~Guillaume, R.~Lambiotte, and E.~Lefebvre.
\newblock Fast unfolding of communities in large networks.
\newblock \emph{Journal of Statistical Mechanics: Theory and Experiment},
  2008\penalty0 (10), oct 2008.

\bibitem[Cheeger(1969)]{Cheeger}
J.~Cheeger.
\newblock A lower bound for the smallest eigenvalue of the laplacian.
\newblock In \emph{Proceedings of the Princeton conference in honor of
  Professor S. Bochner}, pages 195--199, 1969.

\bibitem[Chen et~al.(2017)Chen, Li, and Bruna]{chen2017supervised}
Z.~Chen, X.~Li, and J.~Bruna.
\newblock Supervised community detection with line graph neural networks.
\newblock \emph{arXiv:1705.08415}, 2017.

\bibitem[de~Reus and van~den Heuvel(2013)]{deReus2013cat}
M.~de~Reus and M.~van~den Heuvel.
\newblock Rich club organization and intermodule communication in the cat
  connectome.
\newblock \emph{Journal of Neuroscience}, 33\penalty0 (32):\penalty0
  12929--12939, 2013.

\bibitem[Elumalai et~al.(2022)Elumalai, Yadav, Williams, Saucan, Jost, and
  Samal]{elumalai2022graph}
P.~Elumalai, Y.~Yadav, N.~Williams, E.~Saucan, J.~Jost, and A.~Samal.
\newblock Graph {R}icci curvatures reveal atypical functional connectivity in
  autism spectrum disorder.
\newblock \emph{Scientific reports}, 12\penalty0 (1):\penalty0 1--19, 2022.

\bibitem[Erbar and Maas(2012)]{erbar2012ricci}
M.~Erbar and J.~Maas.
\newblock Ricci curvature of finite markov chains via convexity of the entropy.
\newblock \emph{Archive for Rational Mechanics and Analysis}, 206:\penalty0
  997--1038, 2012.

\bibitem[Evans and Lambiotte(2010)]{evans_line_2010}
T.~Evans and R.~Lambiotte.
\newblock Line graphs of weighted networks for overlapping communities.
\newblock \emph{The European Physical Journal B}, 77\penalty0 (2):\penalty0
  265--272, September 2010.

\bibitem[Fiedler(1973)]{fiedler1973algebraic}
M.~Fiedler.
\newblock Algebraic connectivity of graphs.
\newblock \emph{Czechoslovak mathematical journal}, 23\penalty0 (2):\penalty0
  298--305, 1973.

\bibitem[Forman(2003)]{forman2003bochner}
R.~Forman.
\newblock Bochner's method for cell complexes and combinatorial {R}icci
  curvature.
\newblock \emph{Discrete and Computational Geometry}, 29\penalty0 (3):\penalty0
  323--374, 2003.

\bibitem[Girvan and Newman(2002)]{girvan-newman}
M.~Girvan and M.~Newman.
\newblock Community structure in social and biological networks.
\newblock \emph{Proceedings of the National Academy of Sciences}, 99\penalty0
  (12):\penalty0 7821--7826, 2002.

\bibitem[G{\'o}mez et~al.(2009)G{\'o}mez, Jensen, and
  Arenas]{gomez2009analysis}
S.~G{\'o}mez, P.~Jensen, and A.~Arenas.
\newblock Analysis of community structure in networks of correlated data.
\newblock \emph{Physical Review E}, 80\penalty0 (1):\penalty0 016114, 2009.

\bibitem[Gosztolai and Arnaudon(2021)]{gosztolai2021unfolding}
A.~Gosztolai and A.~Arnaudon.
\newblock Unfolding the multiscale structure of networks with dynamical
  {O}llivier-{R}icci curvature.
\newblock \emph{Nature Communications}, 12\penalty0 (1):\penalty0 4561, 2021.

\bibitem[Gromov(1987)]{gromov}
M.~Gromov.
\newblock Hyperbolic groups.
\newblock In \emph{Essays in Group Theory}, volume~8. Mathematical Sciences
  Research Institute Publications. Springer, NY, 1987.

\bibitem[Harriger et~al.(2012)Harriger, van~den Heuvel, and
  Sporns]{harriger2012macaque}
L.~Harriger, M.~van~den Heuvel, and O.~Sporns.
\newblock Rich club organization of macaque cerebral cortex and its role in
  network communication.
\newblock \emph{PLOS ONE}, 7\penalty0 (9):\penalty0 1--13, 09 2012.

\bibitem[Holland et~al.(1983)Holland, Laskey, and Leinhardt]{Holland_SBM_1983}
P.~Holland, K.~Laskey, and S.~Leinhardt.
\newblock Stochastic blockmodels: First steps.
\newblock \emph{Social Networks}, 5\penalty0 (2):\penalty0 109--137, 1983.

\bibitem[Hopkins and Steurer(2017)]{hopkins_bayesian_2017}
S.~Hopkins and D.~Steurer.
\newblock Bayesian estimation from few samples: community detection and related
  problems, September 2017.

\bibitem[Jarrell et~al.(2012)Jarrell, Wang, Bloniarz, Brittin, Xu, Thomson,
  Albertson, Hall, and Emmons]{jarrel2012worm}
T.~Jarrell, Y.~Wang, A.~Bloniarz, C.~Brittin, M.~Xu, J.~Thomson, D.~Albertson,
  D.~Hall, and S.~Emmons.
\newblock The connectome of a decision-making neural network.
\newblock \emph{Science}, 337\penalty0 (6093):\penalty0 437--444, 2012.

\bibitem[Jin et~al.(2017)Jin, Ke, and Luo]{jin2017estimating}
J.~Jin, Z.~Ke, and S.~Luo.
\newblock Estimating network memberships by simplex vertex hunting.
\newblock \emph{arXiv:1708.07852}, 2017.

\bibitem[Jost and Liu(2014)]{jost2014ollivier}
J.~Jost and S.~Liu.
\newblock Ollivier’s {R}icci curvature, local clustering and
  curvature-dimension inequalities on graphs.
\newblock \emph{Discrete \& Computational Geometry}, 51\penalty0 (2):\penalty0
  300--322, 2014.

\bibitem[Jost and M{\"u}nch(2021)]{jost2021characterizations}
J.~Jost and F.~M{\"u}nch.
\newblock Characterizations of forman curvature.
\newblock \emph{arXiv:2110.04554}, 2021.

\bibitem[Krzakala et~al.(2013)Krzakala, Moore, Mossel, Neeman, Sly,
  Zdeborov{m\'a}, and Zhang]{krzakala2013spectral}
F.~Krzakala, C.~Moore, E.~Mossel, J.~Neeman, A.~Sly, L.~Zdeborov{m\'a}, and
  P.~Zhang.
\newblock Spectral redemption in clustering sparse networks.
\newblock \emph{Proceedings of the National Academy of Sciences}, 110\penalty0
  (52):\penalty0 20935--20940, 2013.

\bibitem[Kuhn(1955)]{kuhn1955hungarian}
H.~Kuhn.
\newblock The hungarian method for the assignment problem.
\newblock \emph{Naval research logistics quarterly}, 2\penalty0 (1-2):\penalty0
  83--97, 1955.

\bibitem[Lancichinetti et~al.(2009)Lancichinetti, Fortunato, and
  Kert{\'{e}}sz]{Lancichinetti2009nmi}
A.~Lancichinetti, S.~Fortunato, and J.~Kert{\'{e}}sz.
\newblock Detecting the overlapping and hierarchical community structure in
  complex networks.
\newblock \emph{New Journal of Physics}, 11\penalty0 (3):\penalty0 033015,
  2009.

\bibitem[Leal et~al.(2021)Leal, Restrepo, Stadler, and Jost]{leal2021forman}
W.~Leal, G.~Restrepo, P.~F. Stadler, and J.~Jost.
\newblock Forman--{R}icci curvature for hypergraphs.
\newblock \emph{Advances in Complex Systems}, 24\penalty0 (01):\penalty0
  2150003, 2021.

\bibitem[Lehot(1974)]{lehot1974optimal}
P.~Lehot.
\newblock An optimal algorithm to detect a line graph and output its root
  graph.
\newblock \emph{Journal of the ACM (JACM)}, 21\penalty0 (4):\penalty0 569--575,
  1974.

\bibitem[Leskovec and Krevl(2014)]{snapnets}
J.~Leskovec and A.~Krevl.
\newblock {SNAP Datasets}: {Stanford} large network dataset collection.
\newblock \url{http://snap.stanford.edu/data}, June 2014.

\bibitem[Leskovec and Mcauley(2012)]{Leskovec2012ego}
J.~Leskovec and J.~Mcauley.
\newblock Learning to discover social circles in ego networks.
\newblock In F.~Pereira, C.J. Burges, L.~Bottou, and K.Q. Weinberger, editors,
  \emph{Advances in Neural Information Processing Systems}, volume~25. Curran
  Associates, Inc., 2012.

\bibitem[Lin et~al.(2011)Lin, Lu, and Yau]{lin-yau}
Y.~Lin, L.~Lu, and S.~Yau.
\newblock {Ricci curvature of graphs}.
\newblock \emph{Tohoku Mathematical Journal}, 63\penalty0 (4):\penalty0 605 --
  627, 2011.

\bibitem[Lubberts et~al.(2021)Lubberts, Athreya, Park, and
  Priebe]{lubberts2021beyond}
Z.~Lubberts, A.~Athreya, Y.~Park, and C.~E. Priebe.
\newblock Beyond the adjacency matrix: random line graphs and inference for
  networks with edge attributes.
\newblock \emph{arXiv:2103.14726}, 2021.

\bibitem[Lubold et~al.(2023)Lubold, Chandrasekhar, and
  McCormick]{lubold2023identifying}
S.~Lubold, A.~Chandrasekhar, and T.~McCormick.
\newblock Identifying the latent space geometry of network models through
  analysis of curvature.
\newblock \emph{Journal of the Royal Statistical Society Series B: Statistical
  Methodology}, 85\penalty0 (2):\penalty0 240--292, 2023.

\bibitem[Ni et~al.(2019)Ni, Lin, Luo, and Gao]{ni2019community}
C.~Ni, Y.~Lin, F.~Luo, and J.~Gao.
\newblock Community detection on networks with {r}icci flow.
\newblock \emph{Scientific reports}, 9\penalty0 (1):\penalty0 1--12, 2019.

\bibitem[Ollivier(2009)]{Ol2}
Y.~Ollivier.
\newblock Ricci curvature of markov chains on metric spaces.
\newblock \emph{Journal of Functional Analysis}, 256\penalty0 (3):\penalty0
  810--864, 2009.

\bibitem[Ollivier(2010)]{Ol}
Y.~Ollivier.
\newblock A survey of {R}icci curvature for metric spaces and {M}arkov chains.
\newblock In \emph{Probabilistic Approach to Geometry}, pages 343--381.
  Mathematical Society of Japan, 2010.

\bibitem[Painter et~al.(2019)Painter, Daniels, and Jost]{jost_social}
D.~T. Painter, B.~C. Daniels, and J.~Jost.
\newblock Network analysis for the digital humanities: Principles, problems,
  extensions.
\newblock \emph{ISIS}, 110\penalty0 (3):\penalty0 538--554, 2019.

\bibitem[Penrose(2003)]{Penrose_RGG_2003}
M.~Penrose.
\newblock \emph{Random Geometric Graphs}.
\newblock Oxford studies in probability. Oxford University Press, 2003.

\bibitem[Pouryahya et~al.(2017)Pouryahya, Mathews, and Tannenbaum]{tannenbaum}
M.~Pouryahya, J.~Mathews, and A.~Tannenbaum.
\newblock Comparing three notions of discrete {R}icci curvature on biological
  networks.
\newblock \emph{arXiv:1712.02943}, 2017.

\bibitem[Samal et~al.(2018)Samal, Sreejith, Gu, Liu, Saucan, and
  Jost]{samal2018comparative}
A.~Samal, R.~Sreejith, J.~Gu, S.~Liu, E.~Saucan, and J.~Jost.
\newblock Comparative analysis of two discretizations of {R}icci curvature for
  complex networks.
\newblock \emph{Scientific reports}, 8\penalty0 (1):\penalty0 1--16, 2018.

\bibitem[Sandhu et~al.(2016)Sandhu, Georgiou, and Tannenbaum]{sandhu2016ricci}
R.~Sandhu, T.~Georgiou, and A.~Tannenbaum.
\newblock Ricci curvature: An economic indicator for market fragility and
  systemic risk.
\newblock \emph{Science advances}, 2\penalty0 (5):\penalty0 e1501495, 2016.

\bibitem[Saucan and Weber(2018)]{hypernets}
E.~Saucan and M.~Weber.
\newblock Forman’s {R}icci curvature-from networks to hypernetworks.
\newblock In \emph{International conference on complex networks and their
  applications}, pages 706--717. Springer, 2018.

\bibitem[Saucan et~al.(2018)Saucan, Samal, Weber, and Jost]{saucan2018discrete}
E.~Saucan, A.~Samal, M.~Weber, and J.~Jost.
\newblock Discrete curvatures and network analysis.
\newblock \emph{MATCH}, 80\penalty0 (3):\penalty0 605--622, 2018.

\bibitem[Sia et~al.(2019)Sia, Jonckheere, and Bogdan]{sia2019ollivier}
J.~Sia, E.~Jonckheere, and P.~Bogdan.
\newblock Ollivier-{R}icci curvature-based method to community detection in
  complex networks.
\newblock \emph{Scientific reports}, 9\penalty0 (1):\penalty0 1--12, 2019.

\bibitem[Sinkhorn and Knopp(1967)]{sinkhorn1967concerning}
R.~Sinkhorn and P.~Knopp.
\newblock Concerning nonnegative matrices and doubly stochastic matrices.
\newblock \emph{Pacific Journal of Mathematics}, 21\penalty0 (2):\penalty0
  343--348, 1967.

\bibitem[Spielman and Teng(1996)]{spielman1996spectral}
D.~Spielman and S.~Teng.
\newblock Spectral partitioning works: Planar graphs and finite element meshes.
\newblock In \emph{Proceedings of 37th conference on foundations of computer
  science}, pages 96--105. IEEE, 1996.

\bibitem[Strehl and Ghosh(2002)]{strehl2002nmi}
A.~Strehl and J.~Ghosh.
\newblock Cluster ensembles – a knowledge reuse framework for combining
  multiple partitions.
\newblock \emph{Journal of Machine Learning Research}, 3:\penalty0 583--617,
  2002.

\bibitem[Tannenbaum et~al.(2015)Tannenbaum, Sander, Zhu, Sandhu, Kolesov,
  Reznik, Senbabaoglu, and Georgiou]{tannenbaum2015ricci}
A.~Tannenbaum, C.~Sander, L.~Zhu, R.~Sandhu, I.~Kolesov, E.~Reznik,
  Y.~Senbabaoglu, and T.~Georgiou.
\newblock Ricci curvature and robustness of cancer networks.
\newblock \emph{arXiv:1502.04512}, 2015.

\bibitem[Tian et~al.(2023)Tian, Lubberts, and Weber]{tian2022mixed}
Y.~Tian, Z.~Lubberts, and M.~Weber.
\newblock Mixed-membership community detection via line graph curvature.
\newblock In \emph{Proceedings of the 1st NeurIPS Workshop on Symmetry and
  Geometry in Neural Representations}, volume 197 of \emph{Proceedings of
  Machine Learning Research}, pages 219--233. PMLR, 03 Dec 2023.

\bibitem[von Luxburg(2007)]{vonLuxburg2007spectra}
U.~von Luxburg.
\newblock A tutorial on spectral clustering.
\newblock \emph{Statistics and Computing}, 17\penalty0 (1):\penalty0 395--416,
  2007.

\bibitem[Watts and Strogatz(1998)]{watts-strogatz}
D.~Watts and S.~Strogatz.
\newblock Collective dynamics of ‘small-world’networks.
\newblock \emph{nature}, 393\penalty0 (6684):\penalty0 440--442, 1998.

\bibitem[Weber(2020)]{pmlr-v108-weber20a}
M.~Weber.
\newblock Neighborhood growth determines geometric priors for relational
  representation learning.
\newblock In \emph{Proceedings of the Twenty Third International Conference on
  Artificial Intelligence and Statistics}, volume 108 of \emph{Proceedings of
  Machine Learning Research}, pages 266--276. PMLR, 26--28 Aug 2020.

\bibitem[Weber and Nickel(2018)]{weber2018curvature}
M.~Weber and M.~Nickel.
\newblock Curvature and representation learning: Identifying embedding spaces
  for relational data.
\newblock In \emph{NeurIPS Relational Representation Learning}, 2018.

\bibitem[Weber et~al.(2017{\natexlab{a}})Weber, Saucan, and Jost]{WSJ1}
M.~Weber, E.~Saucan, and J.~Jost.
\newblock Characterizing complex networks with {F}orman-{R}icci curvature and
  associated geometric flows.
\newblock \emph{Journal of Complex Networks}, 5\penalty0 (4):\penalty0
  527--550, 2017{\natexlab{a}}.

\bibitem[Weber et~al.(2017{\natexlab{b}})Weber, Saucan, and Jost]{WSJ2}
M.~Weber, E.~Saucan, and J.~Jost.
\newblock Coarse geometry of evolving networks.
\newblock \emph{Journal of Complex Networks}, 6\penalty0 (5):\penalty0
  706--732, 2017{\natexlab{b}}.

\bibitem[Weber et~al.(2017{\natexlab{c}})Weber, Stelzer, Saucan, Naitsat,
  Lohmann, and Jost]{weber2017curvature}
M.~Weber, J.~Stelzer, E.~Saucan, A.~Naitsat, G.~Lohmann, and J.~Jost.
\newblock Curvature-based methods for brain network analysis.
\newblock \emph{arXiv:1707.00180}, 2017{\natexlab{c}}.

\bibitem[Weber et~al.(2018)Weber, Jost, and Saucan]{weber2018detecting}
M.~Weber, J.~Jost, and E.~Saucan.
\newblock Detecting the coarse geometry of networks.
\newblock In \emph{NeurIPS Relational Representation Learning}, 2018.

\bibitem[Whitney(1932)]{whitney1932congruent}
H.~Whitney.
\newblock Congruent graphs and the connectivity of graphs.
\newblock \emph{American Journal of Mathematics}, 54\penalty0 (1):\penalty0
  150--168, 1932.

\bibitem[Xie et~al.(2013)Xie, Kelley, and Szymanski]{xie2013overlapping}
J.~Xie, S.~Kelley, and B.~K. Szymanski.
\newblock Overlapping community detection in networks: The state-of-the-art and
  comparative study.
\newblock \emph{ACM Computing Surveys (csur)}, 45\penalty0 (4):\penalty0 1--35,
  2013.

\bibitem[Yang and Leskovec(2013)]{yang2013overlapping}
J.~Yang and J.~Leskovec.
\newblock Overlapping community detection at scale: a nonnegative matrix
  factorization approach.
\newblock In \emph{Proceedings of the sixth ACM international conference on Web
  search and data mining}, pages 587--596, 2013.

\bibitem[Yang and Leskovec(2015)]{Yangdblpdata2018}
J.~Yang and J.~Leskovec.
\newblock Defining and evaluating network communities based on ground-truth.
\newblock \emph{Knowledge and Information Systems}, 42:\penalty0 181--213,
  2015.

\bibitem[Zhang et~al.(2020)Zhang, Levina, and Zhu]{zhang2020overlap}
Y.~Zhang, E.~Levina, and J.~Zhu.
\newblock Detecting {Overlapping} {Communities} in {Networks} {Using}
  {Spectral} {Methods}.
\newblock \emph{SIAM Journal on Mathematics of Data Science}, 2\penalty0
  (2):\penalty0 265--283, January 2020.
\newblock Publisher: Society for Industrial and Applied Mathematics.

\end{thebibliography}

\end{document}